\documentclass[11pt,a4paper,twoside]{article}

\usepackage[T1]{fontenc}
\usepackage[utf8]{inputenc}

\usepackage{amsmath}
\usepackage{mathtools}
\usepackage{bm}
\usepackage{amsthm}

\usepackage{newtxtext}
\usepackage{newtxmath}

\usepackage[
  a4paper,
  margin=28mm,
  headheight=14pt
]{geometry}

\usepackage{microtype}
\usepackage{setspace}
\usepackage{graphicx}
\usepackage{booktabs}
\usepackage{siunitx}
\usepackage{float}
\usepackage{tabularx}
\usepackage{array}
\usepackage{adjustbox}
\usepackage{subcaption}
\usepackage[ruled,vlined]{algorithm2e}

\usepackage{tikz}
\usetikzlibrary{3d}
\usetikzlibrary{arrows.meta}

\usepackage[
  font=small,
  labelfont=bf,
  labelsep=period
]{caption}

\usepackage[round,authoryear]{natbib}
\newtheorem{theorem}{Theorem}[section]
\newtheorem{lemma}[theorem]{Lemma}
\newtheorem{proposition}[theorem]{Proposition} 
\newtheorem{corollary}[theorem]{Corollary}

\newtheorem{remark}[theorem]{Remark}

\usepackage{titlesec}

\titleformat{\section}
  {\normalfont\large\bfseries}
  {\thesection.}
  {0.6em}
  {}

\titleformat{\subsection}
  {\normalfont\normalsize\bfseries}
  {\thesubsection.}
  {0.6em}
  {}

\titleformat{\subsubsection}
  {\normalfont\normalsize\itshape}
  {\thesubsubsection.}
  {0.6em}
  {}

\titleformat{\paragraph}
  {\normalfont\normalsize\bfseries}
  {\theparagraph}
  {1em}
  {}

\titlespacing*{\paragraph}
  {0pt}
  {3.25ex plus 1ex minus .2ex}
  {1.5ex plus .2ex}
  
\usepackage{fancyhdr}

\newcommand{\shorttitle}{Mixture-aware closure of the N-phase Navier--Stokes--Cahn--Hilliard mixture model}
\newcommand{\shortauthors}{Ten Eikelder and Brunk}

\fancypagestyle{plain}{%
  \fancyhf{}%
  \fancyfoot[C]{\thepage}%
}
\usepackage{authblk}

\usepackage{xcolor}
\definecolor{darkgreen}{rgb}{0, 0.7, 0}
\definecolor{orange}{rgb}{0.98, 0.6, 0.01}
\definecolor{napiergreen}{rgb}{0.16, 0.5, 0.0}

\usepackage{pifont}

\usepackage{accents}
\newlength{\dhatheight}

\newcommand{\R}{\mathbb{R}}

\numberwithin{equation}{section}
\allowdisplaybreaks

\def\div{\operatorname{div}}

\def\la{\langle}
\def\ra{\rangle}

\def\nn{\nonumber}

\usepackage{mathtools}

\DeclarePairedDelimiter{\snorm}{|}{|}

\def\vv{\mathbf{v}}

\def\dt{\partial_t}

\def\ddt{\frac{\mathrm{d}}{\mathrm{d}t}}

\def\softd{{\leavevmode\setbox1=\hbox{d}%
\hbox to 1.05\wd1{d\kern-0.4ex{\char039}\hss}}}%cstocs

\newcommand{\mA}{\alpha}
\newcommand{\mB}{{\beta}}

\newcommand{\bJ}{\mathbf{J}}

\newcommand{\trho}{\tilde{\rho}}

\usepackage[
  colorlinks=true,
  linkcolor=blue,
  citecolor=blue,
  urlcolor=blue
]{hyperref}

\usepackage[nameinlink,capitalise,noabbrev]{cleveref}

\newcommand{\email}[1]{\href{mailto:#1}{\texttt{#1}}}

\makeatletter
\newcommand{\myref}[1]{\cref{#1}\mynameref{#1}{\csname r@#1\endcsname}}
\newcommand{\Myref}[1]{\Cref{#1}\mynameref{#1}{\csname r@#1\endcsname}}

\newcolumntype{R}[2]{%
    >{\adjustbox{angle=#1,lap=\width-(#2)}\bgroup}%
    l%
    <{\egroup}%
}

\newcommand{\thickhline}{%
    \noalign {\ifnum 0=`}\fi \hrule height 1pt
    \futurelet \reserved@a \@xhline
}

\newcolumntype{"}{@{\hskip\tabcolsep\vrule width 1pt\hskip\tabcolsep}}
\makeatother

\title{\vspace{-1.0cm}
\bfseries
Mixture-aware closure of the N-phase\\[0.25em]
Navier--Stokes--Cahn--Hilliard mixture model
}

\author[1]{Marco F.P. ten Eikelder\thanks{\email{marco.eikelder@tu-darmstadt.de}}}
\author[2]{Aaron Brunk\thanks{\email{abrunk@uni-mainz.de}}}

\affil[1]{Institute for Mechanics, Computational Mechanics Group, Technical University of Darmstadt, Germany}
\affil[2]{Institute of Mathematics, Johannes Gutenberg-University Mainz, Germany}

\date{}

\begin{document}

\maketitle
\thispagestyle{plain}

\begin{abstract}
Diffuse-interface (phase-field) models are widely used to describe multiphase mixtures and their interfacial dynamics. In multiphase settings, however, the constitutive closure should remain meaningful across different representations of the same mixture.
Existing N-phase phase-field constructions commonly enforce reduction only when a phase is absent (restriction to a face of the Gibbs simplex), but do not address the natural requirement that physically identical phases can be merged without changing the governing equations. This requires characterizing thermodynamically admissible, mixture-aware constitutive closures that are consistent with merging identical phases at the PDE level.

Here, we show that, under a small set of structural axioms, PDE-level reduction consistency uniquely fixes the admissible free-energy structure to an ideal-mixing contribution to an ideal-mixing contribution, a symmetric mean-field interaction term, and a constant-coefficient quadratic gradient penalty. The same requirement constrains the Onsager mobility matrix to a pairwise-exchange form with bilinear degeneracy in the volume fractions, yielding a thermodynamic closure that includes Maxwell–Stefan-type mobilities as a special case. These results provide a consistent closure for N-phase Navier–Stokes–Cahn–Hilliard mixture models and, in the bulk-only setting, for multiphase Maxwell–Stefan diffusion systems. Numerical experiments confirm the predicted mixture-aware reduction properties and illustrate the capabilities of the N-phase Navier–Stokes–Cahn–Hilliard framework in representative multiphase-flow computations.

\end{abstract}

\section{Introduction}\label{sec: Introduction}

Navier--Stokes--Cahn--Hilliard (NSCH) models are widely used to describe viscous, incompressible mixture flows with diffuse interfaces. The origin of this class of models can be traced back to the coupling proposed by \cite{hohenberg1977theory} in 1977, now commonly referred to as model H. Early formulations of NSCH models for two-phase flow\footnote{Throughout this paper, we follow the terminology commonly used in the fluid-dynamics and diffuse-interface literature, where the term \emph{phase} denotes a fluid constituent, such as water or oil. We note that this differs from the stricter terminology in physics, where a phase usually refers to the thermodynamic state of a material, such as liquid, gas, or solid. In the present work, the phases are assumed to be incompressible in the sense that each constituent has a fixed thermodynamic state and a constant specific density. Different phases may, however, correspond to different fluid states and therefore have different material parameters.} were largely restricted to matching phase densities, which limits applicability to many practical situations \citep{gurtinmodel}. Over the past decades, numerous extensions of NSCH models for two-phase flow to non-matching densities have been proposed, including influential contributions by, for example, \cite{lowengrub1998quasi,abels2012thermodynamically}. Although these works aim to describe the same underlying physics, the resulting formulations differ in their balance-law structure, their dissipation statements, and, in some cases, their behavior in limiting single-phase regimes. A unified framework rooted in continuum mixture theory has recently clarified these relationships and identified a single NSCH model for two-phase flow that is invariant to the choice of fundamental variables \cite{eikelder2023unified}; benchmark computations and a divergence-conforming discretization supporting this unified perspective are reported in \cite{ten2024divergence}. From a numerical perspective, even in the two-phase flow setting the discretization of NSCH models is an active field of research. For the matching-density case, energy-stable schemes have been proposed in e.g. \cite{feng2006fully,kay2007efficient,chen2016efficient,diegel2017convergence,brunk2023second,brunk2026review}, while for non-matching densities the chosen formulation impacts algorithm design and coupling strength, with many methods focusing on volume-averaged velocity formulations \cite{khanwale2023projection,guillen2014splitting,garcke2016stable,chen2016efficient} and fewer on mass-averaged velocity formulations \cite{lowengrub1998quasi,shen2013mass,giesselmann2015energy,brunk2026simple}.

In the context of $N$-phase flow, NSCH models have been proposed, e.g. by \cite{boyer2014hierarchy,dong2018multiphase}. Most recently, a unified mixture-theory framework for $N$-phase NSCH mixture models with a single mixture momentum equation has been established in \cite{ten2025unified}. While a general $N$-phase framework is now available, its practical use hinges on constitutive choices that remain challenging to specify in a thermodynamically admissible manner. Addressing this gap is the focus of the present work.

In the mixture-theory NSCH setting, these constitutive choices enter through two objects: the Helmholtz free-energy density and the mobility tensor. The free energy determines chemical potentials and capillary forcing, while the mobility tensor relates chemical-potential gradients to diffusive fluxes and thus determines the diffusive dissipation. For $N$-phase systems, specifying these ingredients in a way that is simultaneously thermodynamically admissible and robust under changes in the phase set remains subtle. A prominent line of work enforces \emph{reduction} by requiring that the $N$-phase model restricted to a face of the Gibbs simplex (i.e.\ when one or more phases are absent) reduces to a lower-dimensional model of the same type; see, for instance, the ternary and hierarchy constructions of Boyer and co-workers \cite{boyer2006study,boyer2014hierarchy}. Other multiphase phase-field approaches, often formulated in terms of concentrations or order parameters, propose alternative $N$-phase closures and numerical strategies; representative examples include multiphase CH/NSCH formulations in \cite{kim2005phase,toth2016phase,dong2018multiphase,garcke1999multiphase} and multiphase-field constructions such as \cite{wu2017multiphase,heida2012development,malek2022thermodynamic}.

The present work is concerned with a different, complementary notion of reduction that arises naturally in mixture models: \emph{merging physically identical phases}. In applications, several phases may represent the same physical phase, sharing identical constitutive parameters, or a single phase may be split into sublabels for modeling or numerical convenience. A desirable closure should therefore be invariant under redistributions between such labels and should reduce exactly to the corresponding $(N\!-\!1)$-phase model when two identical phases are merged. Importantly, this does \emph{not} require that the \emph{free-energy value} of the $N$-phase model equals that of the reduced model under merging. Such a requirement would, in general, conflate label bookkeeping with physical mixing and is conceptually delicate in view of the Gibbs paradox: the entropy of mixing depends on distinguishability and changes under purely formal relabelings. Instead, we impose reduction at the level that is physically and computationally decisive for NSCH models, namely the \emph{governing equations}. This PDE-level reduction principle yields a sharp set of structural constraints on admissible free energies and mobilities, while remaining compatible with the mixture-theory NSCH framework \citep{ten2025unified}. This admissible closure overcomes a central barrier to computation and thereby enables the first practical numerical realization of this recently proposed $N$-phase NSCH model. We illustrate this through simulations for the cases $N=2,3,4$.
Finally, we note that, although we present the derivation in the $N$-phase NSCH setting, the reduction-consistency arguments for bulk free energies and mobility closures apply more broadly to multiphase mixture models, including gradient-free Maxwell--Stefan diffusion systems, cf. \cite{Brunk2026MWS}.

The paper is organized as follows. \cref{sec: NSCH model} recalls the $N$-phase NSCH model and summarizes key properties. \cref{sec: mixture-aware closure} develops the reduction-consistent thermodynamic closure. \cref{sec:practical_closure} then introduces a practical admissible closure. \cref{sec: Numerical results} presents numerical verification and illustrative simulations. Finally, \cref{sec: Conclusions} concludes with a discussion and outlook.

\section{The Navier-Stokes Cahn-Hilliard $N$-phase mixture model}\label{sec: NSCH model}

This section is concerned with the Navier-Stokes Cahn-Hilliard model. \cref{subsec:gov eq} presents the governing equations, and \cref{subsection: structural properties} summarizes its structural properties.

\subsection{Governing equations}\label{subsec:gov eq}
We consider the Navier-Stokes Cahn-Hilliard system with non-matching densities given by \citep{ten2025unified}:
\begin{subequations}\label{eq:sys1}
  \begin{align}\partial_t (\rho \vv) + {\rm div} \left( \rho \vv\otimes \vv \right) + \sum_\mB \phi_\mB \nabla \mu_\mB + \nabla \lambda
    %&\nn\\
    - {\rm div} \boldsymbol{\tau}-\rho\mathbf{g} &=~ 0, \label{eq: intro mass: mom}\\
  \partial_t \phi_\mA  + {\rm div}(\phi_\mA  \vv) +\rho_\mA^{-1}{\rm div} \bJ_\mA &=~0,\label{eq: intro mass: mass}
  \end{align}
\end{subequations}
in a spatial domain $\Omega \subset {\R}^d$ (of dimension $d=2,3$), boundary $\partial \Omega$ and unit outward normal $\mathbf{n}$. The state variables are the phase velocity $\vv: \Omega \rightarrow {\R}^d$, volume fraction variables $\phi_\mA: \Omega \rightarrow {\R}, \mA =1,...,N$ and Lagrange multiplier pressure $\lambda: \Omega \rightarrow {\R}$ subject to the initial conditions $\vv(\mathbf{x},0) = \vv_0(\mathbf{x})$ and $\phi_\mA(\mathbf{x},0) = \phi_0(\mathbf{x})$. The constant specific phase densities are $\rho_\mA$, the partial mass densities are $\trho_\mA = \rho_\mA \phi_\mA$, and the mixture density is given by $\rho = \sum_\mA \trho_\mA$. The force vector is represented by $\mathbf{g} = -g \mathbf{j}$, with $g$ the gravitational constant and $\mathbf{j} = \nabla y$ the vertical unit vector and $y$ the vertical coordinate. The viscous stress is $\boldsymbol{\tau} = \nu (2 \nabla^s \vv+\bar{\lambda}({\rm div}\vv) \mathbf{I}) $, where the dynamic viscosity is $\nu$, $\bar{\lambda} = -2/d$, and the symmetric velocity gradient is denoted by $\nabla^s \vv=(\nabla \vv + (\nabla \vv)^T)/2$. The Helmholtz free energy $\Psi$ associated with the system \eqref{eq:sys1} belongs to the constitutive class:
\begin{align}\label{eq: const class Psi}
    \Psi = \Psi\left(\boldsymbol{\phi},\nabla \boldsymbol{\phi}\right),
\end{align}
where the vector of order parameters is $\boldsymbol{\phi}=(\phi_1,\ldots,\phi_N)^{\mathsf T}$.
The associated chemical potentials are
\begin{align}\label{eq: chem pot}
    \mu_\mA = \dfrac{ \partial \Psi}{\partial \phi_\mA} - {\rm div}\dfrac{\partial \Psi}{\partial\nabla \phi_\mA}.
\end{align}
The model consists of the momentum balance equation \eqref{eq: intro mass: mom}, which governs the mixture momentum balance, along with $N$ phase mass balance equations \eqref{eq: intro mass: mass}.

\begin{remark}[Mixture velocity]
The system is formulated in terms of the mass-averaged velocity $\vv$. Alternatively, it can be expressed using other mixture velocities; for details on the volume-averaged velocity formulation, we refer to \cite{ten2025unified}.
\end{remark}

Next, the absence of void spaces implies that the volume fractions are subject to the saturation constraint
\begin{align}\label{eq: saturation constraint}
    \sum_\mA \phi_\mA =1.
\end{align}
We incorporate the saturation constraint through the addition of the following equation:
\begin{align}\label{eq: LM 1}
    {\rm div}\vv + \displaystyle\sum_{\mA}  \rho_\mA^{-1} {\rm div} \bJ_\mA = 0,
\end{align}
which follows from the addition of the mass balance equations \eqref{eq: intro mass: mass} using \eqref{eq: saturation constraint}. The diffusive fluxes are $
  \bJ_\mA =- \sum_\mB \mathbf{M}_{\mA\mB}\nabla g_\mB$, where $\mathbf{M}_{\mA\mB}$ denotes the degenerate mobility matrix, and where $g_\mA = \rho_\mA^{-1}(\mu_\mA + \lambda)$ is an effective chemical potential. 
Incorporating the chemical potential as a separate state variable, we now arrive at a system that contains $2N+2$ equations for $2N+2$ (independent) state variables ($\vv, \left\{\phi_\mA\right\}_{\mA=1,...,N}, \lambda, \left\{\mu_\mA\right\}_{\mA=1,...,N}$):
\begin{subequations}\label{eq:sys2}
  \begin{align}\partial_t (\rho \vv) + {\rm div} \left( \rho \vv\otimes \vv \right) + \sum_\mB \phi_\mB \nabla \mu_\mB + \nabla \lambda
    %&\nn\\
    - {\rm div} \boldsymbol{\tau}-\rho\mathbf{b} &=~ 0, \label{eq:sys2: mom}\\
  \partial_t \phi_\mA  + {\rm div}(\phi_\mA  \vv) +\rho_\mA^{-1}{\rm div} (\bJ_\mA )&=~0,\label{eq:sys2: mass}\\
  {\rm div}\vv + \displaystyle\sum_{\mA}  \rho_\mA^{-1} {\rm div} \bJ_\mA &=~ 0, \label{eq:sys2: div}\\
   \mu_\mA - \dfrac{ \partial \Psi}{\partial \phi_\mA} + {\rm div}\dfrac{\partial \Psi}{\partial\nabla \phi_\mA} &~=0,\label{eq:sys2: chem}
  \end{align}
\end{subequations}
In this formulation, the pressure $\lambda$ acts as a Lagrange multiplier that enforces \eqref{eq:sys2: div}, see also \cite{ten2025unified}.

\begin{remark}[Pure phases]
In the regions of pure phases ($\phi_\mA = 1$ for some $\mA = 1,...,N$), the peculiar velocities vanish ($\bJ_\mA = 0$) since the mobility is assumed degenerate. As a consequence, the phase velocity is divergence-free  (${\rm div} \vv = 0$) in these regions.
\end{remark}

\begin{remark}[Saturation constraint]
  As an alternative, one can directly substitute the saturation constraint \eqref{eq: saturation constraint} to reduce the set of volume fraction state variables to $\left\{\phi_\mA\right\}_{\mA\neq \mB}$ for some fixed $\mB \in \left\{1,...,N\right\}$. This leads to a reduced system in which \eqref{eq:sys2: mom}-\eqref{eq:sys2: chem} constitute $2N+1$ equations for $2N+1$ state variables ($\vv, \left\{\phi_\mA\right\}_{\mA\neq \mB}, \lambda, \left\{\mu_\mA\right\}$). This system is not symmetric in terms of $\phi_\mA$.
\end{remark}
\subsection{Structural properties}\label{subsection: structural properties}
The system possesses underlying structural properties governed by balance laws and dissipation identities. Specifically, under appropriate boundary conditions, it conserves the phase masses, the volume-fraction integrals, and the total mass; preserves the saturation constraint pointwise; and dissipates the total energy
\begin{subequations}    
\begin{align}
   \ddt \int_\Omega \trho_\mA~{\rm d}\Omega =&~ 0, \qquad  \ddt \int_\Omega \phi_\mA~{\rm d}\Omega=0,\qquad \mA = 1,...,N,\\
   \ddt \int_\Omega \rho~{\rm d}\Omega =&~ 0,\\
   \left(\sum_\mA \phi_\mA(\mathbf{x},0)=1\right)& \qquad \Rightarrow \qquad \left(\sum_\mA \phi_\mA(\mathbf{x},t)=1,\quad t>0\right),\\
   \ddt \mathcal{E}\left(\vv,\left\{\phi_\mA\right\}\right) =&~ -\mathcal{D}\left(\vv, \left\{g_\mA\right\}\right)\leq 0, \label{eq: global energy evolution}
\end{align}
\end{subequations}
Here the energy $\mathcal{E}$ and dissipation rate $\mathcal{D}$ are given by
\begin{subequations}
\begin{align}
 \mathcal{E}\left(\vv,\left\{\phi_\mA\right\}\right):=&~\int_\Omega K\left(\vv,\left\{\phi_\mA\right\}\right) + G\left(\left\{\phi_\mA\right\}\right) + \Psi\left(\left\{\phi_\mA\right\},\left\{\nabla \phi_\mA\right\}\right)~{\rm d}\Omega, \label{eq:defEnergy}\\
 K\left(\vv,\left\{\phi_\mA\right\}\right) =&~ \frac{1}{2}\rho\left(\left\{\phi_\mA\right\}\right)\snorm{\vv}^2,\\ G\left(\left\{\phi_\mA\right\}\right) =&~ \rho\left(\left\{\phi_\mA\right\}\right) g y,\\
 \mathcal{D}\left(\vv, \left\{g_\mA\right\}\right):=&~ \displaystyle\int_{\Omega} 2 \nu \left( \nabla^s \vv - \frac{1}{d} ({\rm div} \vv) \mathbf{I}\right):\left(\nabla^s \vv - \frac{1}{d} ({\rm div} \vv) \mathbf{I}\right)+ \nu \left(\bar{\lambda} + \frac{2}{d}\right)\left({\rm div} \vv\right)^2 \nn\\
    &~\quad\quad+ \displaystyle\sum_{\mA,\mB} \left(\nabla g_\mA\right)^T \mathbf{M}_{\mA\mB} \nabla g_\mB~{\rm d}\Omega \geq 0,
 \label{eq:defDissipation}
\end{align}
\end{subequations}
where we decomposed the energy density of $\mathcal{E}$ into the kinetic energy density $K$, the energy density due to gravity $G$ and the free energy density $\Psi$.

\begin{remark}[Conservative form]
  In the absence of gravity ($\mathbf{g}=0$), the model may be written in conservative form by inserting the Korteweg tensor identity:
  \begin{align}\label{eq: korteweg}
    {\rm div} \boldsymbol{\varsigma} =  \sum_\mB \phi_\mB \nabla \mu_\mB,
  \end{align}
where the Korteweg tensor is
\begin{align}\label{eq: korteweg stress}
   \boldsymbol{\varsigma} = &~ \left(\sum_\mA \hat{\mu}_\mA\phi_\mA -\hat{\Psi}\right)\mathbf{I} + \sum_\mA \nabla \phi_\mA \otimes \dfrac{\partial \hat{\Psi}}{\partial \nabla \phi_\mA} 
\end{align}
\end{remark}
In order to be consistent with angular momentum conservation, the Cauchy stress is a symmetric second-order tensor. This is only guaranteed if the Korteweg stress $\boldsymbol{\varsigma}$ yields a symmetric matrix, which restricts the choices that are consistent with thermodynamics.

\subsection{Equilibrium conditions}
In this section we consider the stationary equilibrium conditions of system \eqref{eq:sys2} in absence of gravity and no-slip boundary conditions for the velocity. To this end we consider the standard equilibrium assumptions $\dt\phi_\mA=0$ and $\dt\vv=\bm{0}$. Under these assumptions, also $\frac{{\rm d}}{{\rm d}t}\mathcal{E}=0$ and hence $\mathcal{D}(\vv,\{g_\alpha\})=0$ which yields $\nabla\vv=0$, i.e. $\vv=0$ and using the kernel of $M$ we find $\nabla(g_\mA-g_{\mB})=0$, i.e. $g_\mA-g_\mB= c_{\mA\mB}=\text{const}$. Substitution into the system yields
\begin{subequations}\label{eq:syseq2}
  \begin{align}\sum_\mB \phi_\mB \nabla \mu_\mB + \nabla \lambda
    %&\nn\\
     &=~ 0, \label{eq:syseq2: mom}\\
  \rho_\mA^{-1}{\rm div} \bJ_\mA &=~0,\label{eq:syseq2: mass}\\
   \mu_\mA - \dfrac{ \partial \Psi}{\partial \phi_\mA} + {\rm div}\dfrac{\partial \Psi}{\partial\nabla \phi_\mA} &~=0.\label{eq:syseq2: chem}
  \end{align}
\end{subequations}
Rewriting \eqref{eq:syseq2: mom} using the saturation constraint and $\nabla (g_\mA-g_\mB)=0$ yields
\begin{equation}
  0 = \sum_\mB \phi_\mB \nabla \mu_\mB + \nabla \lambda
    = \sum_\mB \tilde{\rho}_\mB \nabla g_\mB  = \sum_\mB \tilde{\rho}_\mB \nabla g_\alpha =\rho  \nabla g_\alpha. 
\end{equation}
Since the total density is nonzero everywhere, we conclude $\nabla g_\alpha=0$ for all $\alpha$ and hence $g_\alpha = c_\alpha = \textrm{const}$. In particular, it follows that
\begin{align}
\bJ_\mA = 0
\qquad\text{for all }\mA=1,...,N.
\end{align}
Hence, we find that the volume fractions $\phi_\mA$ are characterized by
\begin{align}
\frac{\partial \Psi}{\partial \phi_\alpha} - {\rm div}\frac{\partial \Psi}{\partial \nabla \phi_\alpha} +\lambda = \hat{c}_\mA.
\label{eq:equilibrium_phi}
\end{align}
with $\hat{c}_\mA :=\rho_\alpha c_\alpha$. Equation \eqref{eq:equilibrium_phi} has the form of a constrained Euler--Lagrange equation. Indeed, the equilibrium conditions for $\phi_\mA$ can be interpreted as the following energy minimization problem
\begin{align}
 &\min_{\phi\in H^1(\Omega)^N} \int_\Omega  \Psi(\boldsymbol{\phi},\nabla\boldsymbol{\phi})   - \sum_\mA \hat{c}_\mA (\phi_\alpha-m_\alpha) + \lambda\left(\sum_\mA \phi_\mA - 1\right) \nn\\
 = &\min_{\substack{\phi\in H^1(\Omega)^N,\\ \la \phi,e_\mA \ra=m_\mA}} \int_\Omega  \Psi(\boldsymbol{\phi},\nabla\boldsymbol{\phi})   + \lambda\left(\sum_\mA \phi_\mA - 1\right) \label{eq:1Dmin}
\end{align}
where $m_\mA$ are the mean values of $\phi_\mA$. Hence, we can understand $\hat{c}_\mA$ as Lagrange multipliers for mean value of $\phi_\alpha$ and $\lambda$ as a Lagrange multiplier for the saturation constraint. 

\section{Mixture-aware closure}\label{sec: mixture-aware closure}

Here we derive a \emph{mixture-aware thermodynamic closure}, i.e. admissible constitutive choices for the free-energy density and mobility tensor such that the $N$-phase model reduces exactly to an $(N\!-\!1)$-phase model when two phases are physically identical and merged. To this end we specify the constitutive model class and the axioms imposed on the family of $N$-phase closures.
\subsection{Model class and axioms} \label{sec:axioms}

Throughout, $\Psi$ denotes the Helmholtz free-energy density and $\mathbf{M}$ the mobility tensor of the $N$-phase model. Quantities associated with the reduced $(N\!-\!1)$-phase system are denoted by a hat, e.g. $\hat\Psi$ and $\hat{\mathbf{M}}$. The phrase ``same functional form'' means that $(\Psi,\mathbf{M})$ and $(\hat\Psi,\hat{\mathbf{M}})$ are obtained from a single constitutive form, differing only through the number of phases and the corresponding coefficient sets.

\subsubsection*{A1. Admissible compositions.}
We define the Gibbs simplex for $N$ phases by
\begin{align}
\mathcal{G}:=\Big\{\boldsymbol{\phi}\in{\R}^{N}\ \Big|\ \sum_{\alpha=1}^{N}\phi_\alpha = 1,\ \phi_\beta \ge 0\ \ (\beta=1,\ldots,N)\Big\}.
\end{align}
The order parameters belong to the Gibbs simplex: $\boldsymbol{\phi}\in\mathcal{G}$. We visualize the Gibbs simplex for $N=3$ in \cref{fig:gibbs-triangle}, and provide details in \cref{appendix Gibbs simplex}. Constitutive relations are formulated on the interior of $\mathcal{G}$ (all $\phi_\alpha>0$) and extended to the boundary when needed.

\begin{figure}
  \centering
  \includegraphics[width=0.52\linewidth]{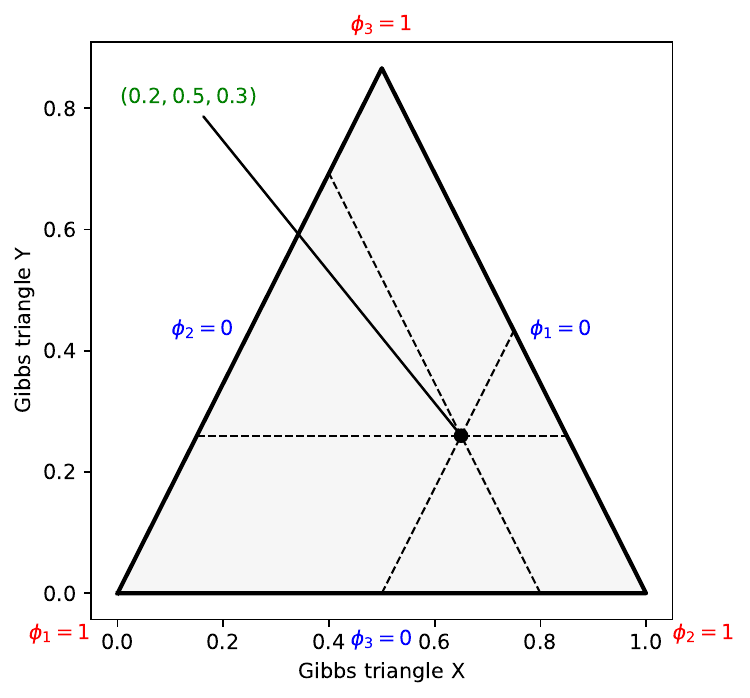}
  \caption{Gibbs triangle (ternary diagram) for $N=3$. Vertices correspond to pure states ($\phi_{\alpha}=1$), edges to binary mixtures (one phase identically zero), and the interior to fully ternary states. The simplex boundaries $\phi_{\alpha}=0$ and vertices $\phi_{\alpha}=1$ are indicated; the point $(\phi_1,\phi_2,\phi_3)=(0.2,0.5,0.3)$ is shown as an example.}
  \label{fig:gibbs-triangle}
\end{figure}

\subsubsection*{A2. Local first-gradient free energy and chemical potentials.}
The Helmholtz free-energy density is local and of first-gradient type, i.e. of the form \eqref{eq: const class Psi} and is assumed sufficiently smooth (at least $C^2$ in $\boldsymbol{\phi}$ and $\nabla\boldsymbol{\phi}$ on $\mathrm{int}(\mathcal{G})$). Chemical potentials are defined in \eqref{eq: chem pot} and are understood up to the standard additive gauge associated with the simplex constraint $\sum_\alpha\phi_\alpha=1$.

\subsubsection*{A3. Constant-capillarity.}
We restrict attention to a constant-capillarity constitutive class in which the second derivatives of $\Psi$ with respect to gradients are independent of the local composition. Equivalently, the tensor
\begin{align}
\frac{\partial^2\Psi}{\partial(\nabla\phi_\alpha)\partial(\nabla\phi_\beta)}
\end{align}
is assumed independent of $\boldsymbol{\phi}$ (while it may depend on material coefficients).

\subsubsection*{A4. Mobility closure and admissibility.}
We impose the following requirements on the mobility tensor:
\begin{align}
\text{(M1) Symmetry:}\quad & \mathbf{M}_{\alpha\beta}(\boldsymbol{\phi})=\mathbf{M}_{\beta\alpha}(\boldsymbol{\phi}), \\
\text{(M2) Positive semidefiniteness:}\quad & \boldsymbol{\xi}^{\mathsf T}\mathbf{M}(\boldsymbol{\phi})\boldsymbol{\xi}\ge 0
\quad\text{for all }\boldsymbol{\xi}\in{\R}^{N},\\
\text{(M3) Constraint compatibility:}\quad & \mathbf{M}(\boldsymbol{\phi})\mathbf{1}=\mathbf{0}, \qquad \mathbf{1}=(1,\ldots,1)^{\mathsf T},\\
\text{(M4) Spatial isotropy:}\quad & \mathbf{M}_{\alpha\beta}(\boldsymbol{\phi}) = M_{\alpha\beta}(\boldsymbol{\phi})\mathbf{I},
\qquad \alpha,\beta=1,\ldots,N,
\label{eq:M_isotropy_axiom}\\
\text{(M5) Degeneracy:}\quad & \text{if } \phi_\mA=0 \text{ or } \phi_\mA=1 \text{ then } \mathbf{M}_{\mA\beta}(\boldsymbol{\phi})=\mathbf{0},
\end{align}
for scalar mobilities $M_{\alpha\beta}(\boldsymbol{\phi})$ and the identity tensor $\mathbf{I}$.

\subsubsection*{A5. Merging identical phases.}
We consider the case in which phases $1$ and $N$ are physically identical. Given $\boldsymbol{\phi}\in\mathcal{G}$, we define the reduced $(N\!-\!1)$-phase composition $\hat{\boldsymbol{\phi}}\in{\R}^{N-1}$ by
\begin{align}\label{eq: merging}
\hat\phi_1 := \phi_1+\phi_N,\qquad
\hat\phi_\alpha := \phi_\alpha \quad (\alpha=2,\ldots,N-1),
\end{align}
with the induced action on gradients
\begin{align}
\nabla\hat\phi_1=\nabla\phi_1+\nabla\phi_N,\qquad
\nabla\hat\phi_\alpha=\nabla\phi_\alpha\quad (\alpha=2,\ldots,N-1).
\end{align}

Reduction consistency is a requirement on the \emph{PDE system}: whenever phases $1$ and $N$ are physically identical, the $N$-phase model must coincide with an $(N-1)$-phase model for the merged variables. We visualize these reduction-consistency requirements in \cref{fig:merge-map,fig:merge-lines-ABC}.
\begin{figure}
    \centering  
\includegraphics[width=0.95\linewidth]{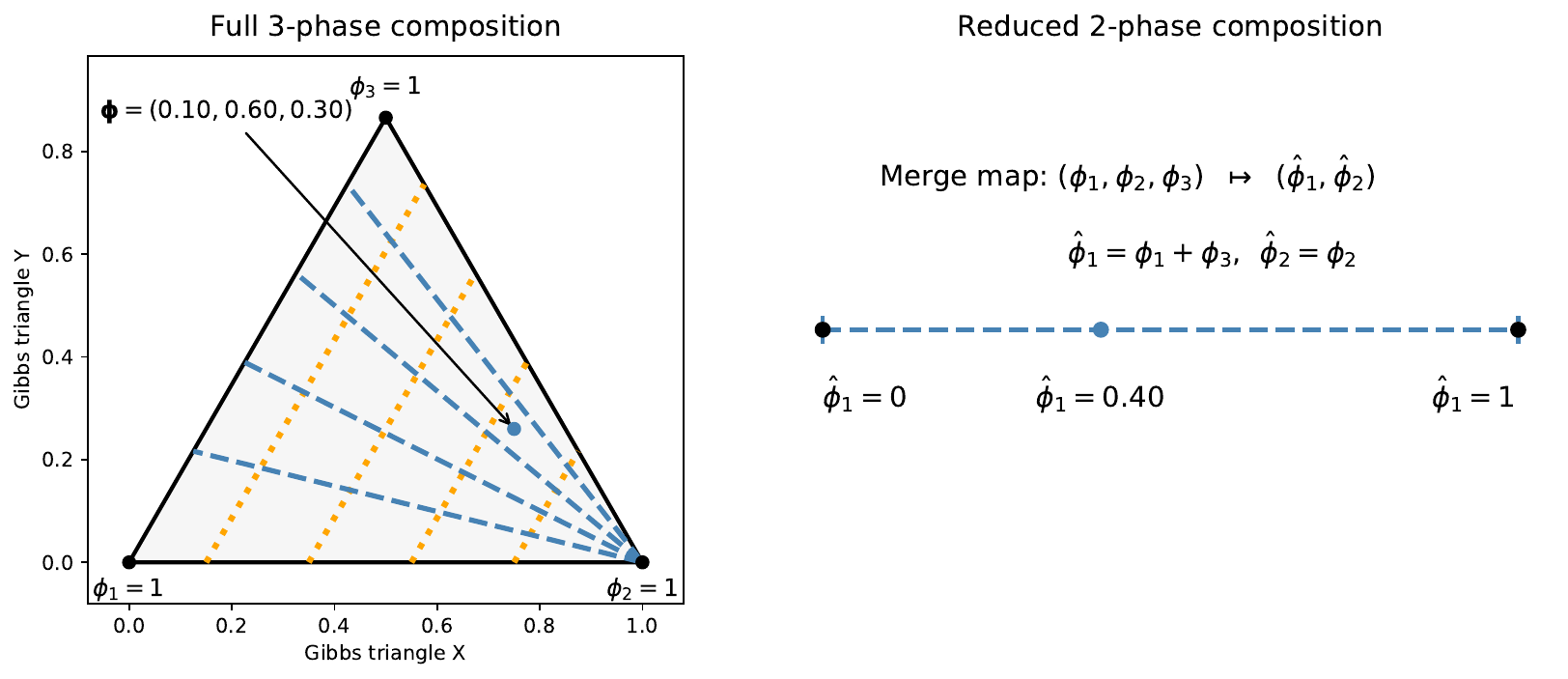}
\captionsetup{justification=raggedright,singlelinecheck=false}\caption{Visualization of the merge (reduction) map for a ternary composition.
Left: Gibbs simplex for $N=3$ with two families of lines. The orange dotted lines are level sets
$\phi_2=\mathrm{const}$; along each such line the reduced variables
$(\hat\phi_1,\hat\phi_2)=(\phi_1+\phi_3,\phi_2)$ are constant, i.e.\ these lines identify all $(\phi_1,\phi_2,\phi_3)$ corresponding to the same reduced two-phase state.
The blue dashed rays emanating from the vertex $\phi_2=1$ represent constant split ratio
$\phi_1/(\phi_1+\phi_3)=\mathrm{const}$, highlighting the degree of freedom that is discarded when phases $1$ and $3$
are merged.
Right: the reduced two-phase composition axis $\hat\phi_1\in[0,1]$,
with the image of the example state marked.}
\label{fig:merge-map}
\end{figure}
\begin{figure}
    \centering    \includegraphics[width=0.45\linewidth]{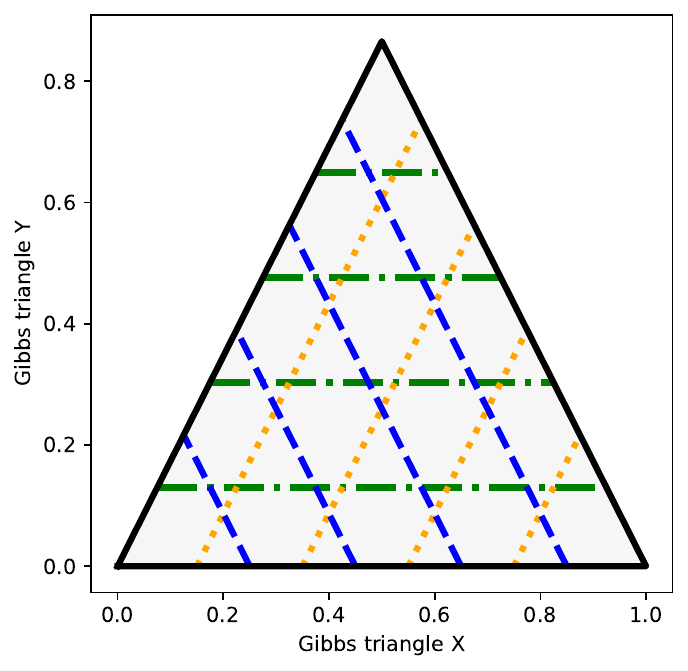}
\captionsetup{justification=raggedright,singlelinecheck=false}\caption{Redistribution lines in the ternary Gibbs simplex along which the reduction-consistency constraints are imposed. Each colored family of lines collects states that correspond to the \emph{same reduced composition} after merging two phases:
orange dotted lines are level sets $\phi_2=\mathrm{const}$ (merging phases $1$ and $3$, so $\hat\phi_1=\phi_1+\phi_3=1-\phi_2$ is fixed);
green dash--dot lines are level sets $\phi_3=\mathrm{const}$ (merging phases $1$ and $2$);
blue dashed lines are level sets $\phi_1=\mathrm{const}$ (merging phases $2$ and $3$).
Along each such line, the reduced $(N\!-\!1)$-phase state is unchanged while the split between the merged phases varies.}
    \label{fig:merge-lines-ABC}
\end{figure}
In particular, the barycentric (mixture) velocity is a physical field independent of labeling; therefore, we identify $\hat{\vv}=\vv$. Moreover, the Lagrange multipliers $\lambda$ and $\hat\lambda$ enforcing the simplex constraint enter the momentum balance only through their gradients, and are defined only up to a constant. As such, we do not require $\hat{\lambda}$ and $\lambda$ to coincide, but work with a consistent gauge for the Lagrange multipliers so that $\nabla \hat{\lambda}=\nabla \lambda$.

The reduced system is closed by $\hat\Psi(\hat{\boldsymbol{\phi}},\nabla\hat{\boldsymbol{\phi}})$ and $\hat{\mathbf{M}}(\hat{\boldsymbol{\phi}})$, with reduced chemical potentials and generalized potentials
\begin{align}
\hat\mu_\alpha := \frac{\partial \hat\Psi}{\partial \hat\phi_\alpha}
-{\rm div}\frac{\partial \hat\Psi}{\partial\nabla\hat\phi_\alpha},
\qquad
\hat g_\alpha := \frac{\hat\mu_\alpha+\hat\lambda}{\rho_\alpha}.
\end{align}

When phases $1$ and $N$ are physically identical, we impose reduction consistency at two levels, corresponding to the places where constitutive choices enter the NSCH system:

\begin{enumerate}
\item \emph{Capillary-force reduction consistency:} for all smooth fields on the interior of $\mathcal{G}$,
\begin{align}
\sum_{\beta=1}^{N}\phi_\beta\nabla\mu_\beta
=
\sum_{\beta=1}^{N-1}\hat\phi_\beta\nabla\hat\mu_\beta,
\label{eq:RC_force}
\end{align}
and hence:
\begin{align}
    \phi_1 \nabla \mu_1 + \phi_N \nabla \mu_N = \hat{\phi}_1 \nabla \hat{\mu}_1.
\end{align}

\item \emph{Diffusive-flux reduction consistency:} for all smooth fields on the interior of $\mathcal{G}$,
\begin{subequations}
\begin{align}
\hat\rho_1^{-1}{\rm div}\hat\bJ_1 =&~ \rho_1^{-1}{\rm div}\bJ_1 + \rho_N^{-1}{\rm div}\bJ_N,\\
\hat\rho_\alpha^{-1}{\rm div}\hat\bJ_\alpha = &~\rho_\alpha^{-1}{\rm div}\bJ_\alpha\ \ (\alpha=2,\ldots,N-1),
\label{eq:RC_div_map}
\end{align}    
\end{subequations}
where the $N$-phase fluxes $\{\bJ_\alpha\}_{\alpha=1}^{N}$ satisfy
\begin{subequations}\label{eq:RC_flux_map}
  \begin{align}
    \bJ_1+\bJ_N = &~ \hat\bJ_1, \label{eq:RC_flux_map 1}\\
    \bJ_\alpha = &~ \hat\bJ_\alpha \quad (\alpha=2,\ldots,N-1), \label{eq:RC_flux_map 2}
  \end{align}
\end{subequations}
with the reduced fluxes $\{\hat\bJ_\alpha\}_{\alpha=1}^{N-1}$ of the same Onsager form
\begin{align}
\hat\bJ_\alpha =- \sum_{\beta=1}^{N-1} \hat{\mathbf{M}}_{\alpha\beta}\nabla \hat g_\beta.
\end{align}
\end{enumerate}

\subsubsection*{A6. Reduction for absent phase.}
We require absent phases not to influence the system. For any $K\ge 1$, define the embedded $(N+K)$-phase state
\begin{align}
\tilde\phi_\alpha=\phi_\alpha \quad (\alpha=1,\ldots,N),
\qquad
\tilde\phi_{N+1}=\cdots=\tilde\phi_{N+K}=0.
\end{align}
Then the $(N+K)$-phase constitutive closures must restrict to the original $N$-phase closures in the sense that
\begin{subequations}
    \begin{align}
\sum_{\alpha=1}^{N+K} \tilde{\phi}_\alpha \nabla \tilde{\mu}_\alpha 
=&~\sum_{\alpha=1}^N \phi_\alpha \nabla \mu_\alpha, 
\label{eq:A3p_fcap_ext}
\\
\tilde{\mathbf{J}}_\alpha
=&~
\mathbf{J}_\alpha
\qquad (\alpha=1,\ldots,N),
\label{eq:A3p_flux_ext_active}
\\
\tilde{\mathbf{J}}_{N+1}
=&~
\cdots
=
\tilde{\mathbf{J}}_{N+K}
=
\mathbf{0}.
\label{eq:A3p_flux_ext_inactive}
\end{align}
\end{subequations}

\subsubsection*{A7. Absent phase remains absent.}
We require that an absent phase remains absent. If $\phi_\mA=0$ initially, then $\phi_\mA=0$ for all later times.

\subsubsection*{A8. Invariant under relabelling.}
We assume that the constitutive closures are \emph{label-invariant}: for every permutation $\pi$ of the phase labels, with associated permutation matrix $P_\pi$, the relabeled state
$P_\pi\boldsymbol{\phi}$ induces the relabeled closures
\begin{align}
\sum_{\alpha=1}^N (P_\pi\boldsymbol{\phi})_\alpha \,
\nabla \mu_\alpha[P_\pi\boldsymbol{\phi}]
=&~
\sum_{\alpha=1}^N \phi_\alpha \,\nabla \mu_\alpha[\boldsymbol{\phi}],
\label{eq:A3p_cap_perm}
\\
(\mathbf{J}[P_\pi \boldsymbol{\phi},\,P_\pi \mathbf{g}])_{\mA}
=&~
(P_\pi \mathbf{J}[\boldsymbol{\phi},\,\mathbf{g}])_{\mA},
\label{eq:A3p_flux_perm}
\end{align}
where $\mathbf{J}=(\mathbf{J}_1,\ldots,\mathbf{J}_N)^{\mathsf T}$ and
$\mathbf{g}=(g_1,\ldots,g_N)^{\mathsf T}$.

\subsubsection*{A9. Identical-phases symmetry.}
When two phases are physically identical, the constitutive closures are invariant under exchanging their labels. In particular, if phases $p$ and $q$ are physically identical then we require:
\begin{align}
\sum_{\alpha=1}^N (\mathbf{P}_{pq}\boldsymbol{\phi})_\alpha \,
\nabla \mu_\alpha[\mathbf{P}_{pq}\boldsymbol{\phi}]
=&~
\sum_{\alpha=1}^N \phi_\alpha \,\nabla \mu_\alpha[\boldsymbol{\phi}],
\label{eq:A7p_cap}
\\
(\mathbf{J}[\mathbf{P}_{pq}\boldsymbol{\phi},\,\mathbf{P}_{pq}\mathbf{g}])_{\mA}
=&~
(\mathbf{P}_{pq}\mathbf{J}[\boldsymbol{\phi},\,\mathbf{g}])_{\mA},
\label{eq:A7p_flux}
\end{align}
where $\mathbf{P}_{pq}$ denotes the permutation matrix that swaps the label $p$ and $q$, and where $\mathbf{J}=(\mathbf{J}_1,\ldots,\mathbf{J}_N)^{\mathsf T}$ and
$\mathbf{g}=(g_1,\ldots,g_N)^{\mathsf T}$.

\subsubsection*{A10. Objectivity and isotropy.}
The constitutive closures are objective and isotropic with respect to spatial rotations. For every rotation $\mathbf{Q}\in SO(d)$, define the rotated phase fields by
\begin{align}
\phi_\alpha^{\mathbf Q}(x)
:=
\phi_\alpha(\mathbf Q^{\mathsf T}x).
\end{align}
We require that the constitutive closures transform covariantly under this change of frame:
\begin{align}
\sum_{\alpha=1}^N \phi_\alpha^{\mathbf Q}(x)\,\nabla \mu_\alpha[\boldsymbol{\phi}^{\mathbf Q}](x)
=&~
\mathbf Q \Big( \sum_{\alpha=1}^N \phi_\alpha \,\nabla \mu_\alpha[\boldsymbol{\phi}] \Big)(\mathbf Q^{\mathsf T}x),
\label{eq:A4p_cap_rot}
\\
(\mathbf{J}[\boldsymbol{\phi}^{\mathbf Q},\,\mathbf g^{\mathbf Q}])_\alpha(x)
=&~
\mathbf Q\,(\mathbf{J}[\boldsymbol{\phi},\,\mathbf g])_\alpha(\mathbf Q^{\mathsf T}x),
\qquad \alpha=1,\ldots,N,
\label{eq:A4p_flux_rot}
\end{align}
Equivalently, if $\boldsymbol{\phi}$ is the phase-field solution of the $N$-phase system, then the rotated field
$\boldsymbol{\phi}^{\mathbf Q}$ solves the same constitutive system in the rotated frame.

\subsection{Derivation of the free energy}\label{sec:free_energy_derivation}
Here we derive the admissible structure of the Helmholtz free-energy density
$\Psi(\boldsymbol{\phi},\nabla\boldsymbol{\phi})$ from the axioms in \cref{sec:axioms}. The derivation proceeds in the
following order. First, we derive the bulk part of the Helmholtz free-energy density with an ideal mixing term and an interaction part. After that we derive the capillarity gradient terms. Our guiding axiom in this derivation is A5 (Merging identical phases).

\subsubsection{Bulk free energy}

Let $\Psi_0(\boldsymbol{\phi}):=\Psi(\boldsymbol{\phi},\mathbf{0})$ denote the bulk (homogeneous) energy.

\subsubsection*{Ideal mixing energy}
We assume that phase $1$ and $N$ are identical and consider the merging of these via \eqref{eq: merging}. This merge affects only how $\hat\phi_1$ is split into $\phi_1$ and $\phi_N$. Therefore, when we compare the $N$-phase and $N-1$-phase capillary-forces via \eqref{eq:RC_force}, the only bulk contributions that can be uniquely determined are those that depend on $\phi_1$ and $\phi_N$ separately (i.e.
is sensitive to the partition $\hat\phi_1=\phi_1+\phi_N$). Indeed, any term that depends on $\phi_1$ and $\phi_N$ only through $\hat{\phi}_1$ cannot be distinguished. As such, we first separate the part of the free energy that is sensitive to the partition. Motivated by A6 (Reduction for absent phase), A8 (Invariant under relabelling), and A9 (Identical-phases symmetry), we consider the ansatz:
\begin{align}\label{eq:one_body_ansatz}
\Psi_{0}^{\text{ideal}}(\boldsymbol{\phi}) := \sum_{\alpha=1}^N f_\alpha(\phi_\alpha), \qquad \hat{\Psi}_{0}^{\text{ideal}}(\boldsymbol{\phi}) := \sum_{\alpha=1}^{N-1} f_\alpha(\hat{\phi}_\alpha),
\end{align}
for $\mathcal{C}^2$ scalar functions $f_\alpha$ that can depend on the phase number only via the material constants of that phase. Denoting the chemical potentials associated with these free energies by $\mu_\alpha^{\text{ideal}}$ and $\hat{\mu}_\alpha^{\text{ideal}}$, we find
\begin{align}\label{eq:mu_one_body}
\mu_\alpha^{\text{ideal}} = f_\alpha'(\phi_\alpha), \quad \alpha  = 1,...,N, \qquad \hat{\mu}_\alpha^{\text{ideal}} = f_\alpha'(\hat{\phi}_\alpha),\quad   \alpha  = 1,...,N-1. 
\end{align}
Applying \eqref{eq:RC_force} to \eqref{eq:one_body_ansatz} provides the condition:
\begin{align}
\phi_1 f_1''(\phi_1)\nabla\phi_1 \;+\; \phi_N f_N''(\phi_N)\nabla\phi_N
=
\hat\phi_1 f_1''(\hat\phi_1)\nabla\hat\phi_1.
\label{eq:RC_one_body_vector}
\end{align}
Using independence of $\nabla\phi_1$ and $\nabla\phi_N$, \eqref{eq:RC_one_body_vector} implies equality of the coefficients of $\nabla\phi_1$ and
$\nabla\phi_N$ separately: 
\begin{subequations}\label{eq:coeff_equalities}
\begin{align}
  \phi_1 f_1''(\phi_1) = &~ \hat\phi_1 f_1''(\hat\phi_1),\\
  \phi_N f_N''(\phi_N) = &~ \hat\phi_1 f_1''(\hat\phi_1).
\end{align}
\end{subequations}
Fixing $\phi_1$ and varying $\phi_N$ subject to $\phi_1+\phi_N = \hat{\phi}_1$, and noting that the material constants of phase $1$ and $N$ coincide,  yields that the scalar quantity
$uf_\alpha''(u)$ is constant for each $\alpha$ on $(0,1)$:
\begin{align}
uf_\alpha''(u)=W_\alpha\qquad\text{for all }u\in(0,1),
\label{eq:ufpp_const}
\end{align}
for constants $W_\alpha$. The solution of \eqref{eq:ufpp_const} is given by:
\begin{align}
f_\alpha(u)=W_\alpha u\log u + a_\alpha u + b_\alpha,
\label{eq:f_solution}
\end{align}
with constants $a_\alpha,b_\alpha$. 
Noting that affine terms do not affect
$\nabla\mu_\alpha$, which is the only way the free energy appears in the NSCH model, allows us to discard these terms in \eqref{eq:RC_force}. As such, the only relevant part is
\begin{align}
\Psi_0^{\text{ideal}}(\boldsymbol{\phi}) =  \sum_{\alpha=1}^N W_\alpha \phi_\alpha\log\phi_\alpha.
\label{eq:Psi_log}
\end{align}

\subsubsection*{Quadratic interaction term} Define the remainder (enthalpic) bulk energy by
\begin{align}
\Psi_{0}^{\text{enth}}(\boldsymbol{\phi})
:=
\Psi_0(\boldsymbol{\phi})-\Psi_0^{\text{ideal}}(\boldsymbol{\phi}).
\label{eq:Psi_rem_def}
\end{align}
Restricted to the homogeneous bulk setting, we have:
\begin{align}
\label{eq:bulk_force_coeffs}
\sum_{\beta=1}^N \phi_\beta \nabla\mu_\beta
=
\sum_{\gamma=1}^N 
\Big(\sum_{\beta=1}^N \phi_\beta \frac{\partial^2\Psi_0}{\partial\phi_\beta\partial\phi_\gamma}\Big)\nabla\phi_\gamma,
\end{align}
Thus reduction consistency \eqref{eq:RC_force} constrains the coefficient functions
\begin{align}
C_\gamma(\boldsymbol{\phi}) := \sum_{\beta=1}^N \phi_\beta \frac{\partial^2\Psi_0}{\partial\phi_\beta\partial\phi_\gamma}(\boldsymbol{\phi}),
\qquad \gamma=1,\dots,N.
\label{eq:C_gamma_def}
\end{align}

Axiom A5 (Merging identical phases), through \eqref{eq:RC_force}, restricts the dependence of $\Psi_{0}^{\text{enth}}$ on redistributions of $\phi_1$ and $\phi_N$ at fixed
$\hat\phi_1=\phi_1+\phi_N$. In particular, it forces all mixed second derivatives of $\Psi_{0}^{\mathrm{enth}}$ to be
constant (independent of $\boldsymbol{\phi}$), and hence $\Psi_{0}^{\mathrm{enth}}$ must be a quadratic form in
$\boldsymbol{\phi}$ up to affine gauge terms. We therefore write
\begin{align}
\Psi_0^{\text{enth}}(\boldsymbol{\phi})
=
-W\sum_{\alpha,\beta=1}^N   \chi_{\alpha\beta}\phi_\alpha\phi_\beta,
\qquad
\chi_{\alpha\beta}=\chi_{\beta\alpha},
\label{eq:Psi_chi}
\end{align}
where the constants $\chi_{\alpha\beta}$ parametrize the interaction matrix, and where we take a common factor $W$ for convenience. Combining \eqref{eq:Psi_log} and
\eqref{eq:Psi_chi} yields the bulk energy (up to affine gauges)
\begin{align}
\Psi_0(\boldsymbol{\phi})
=
\sum_{\alpha=1}^{N}W_\alpha \phi_\alpha\log\phi_\alpha
-W\sum_{\alpha,\beta=1}^{N} \chi_{\alpha\beta}\phi_\alpha\phi_\beta.
\label{eq:Psi0_final}
\end{align}

\subsubsection{Capillarity term}

\subsubsection*{Linear gradient terms}
We next consider possible terms linear in gradients:
\begin{align}
\Psi_{\mathrm{lin}}(\boldsymbol{\phi},\nabla\boldsymbol{\phi})
=
\sum_{\alpha=1}^N \mathbf{b}_\alpha(\boldsymbol{\phi})\cdot\nabla\phi_\alpha,
\label{eq:Psi_lin}
\end{align}
where each $\mathbf{b}_\alpha(\boldsymbol{\phi})\in\mathbb{R}^d$ depends only on the local composition. By A2 (Local first-gradient free energy and chemical potentials), the chemical potentials induced by \eqref{eq:Psi_lin} are
\begin{align}
\mu_\alpha^{\mathrm{lin}}
=
\sum_{\beta=1}^N \mathbf a_{\alpha\beta}(\boldsymbol{\phi})\cdot \nabla\phi_\beta.
\label{eq:mu_lin_A}
\end{align}
where
\begin{align}
\mathbf a_{\alpha\beta}(\boldsymbol{\phi})
:=
\partial_{\phi_\alpha}\mathbf{b}_\beta(\boldsymbol{\phi})
-
\partial_{\phi_\beta}\mathbf{b}_\alpha(\boldsymbol{\phi}),
\qquad
\mathbf a_{\alpha\beta}=-\mathbf a_{\beta\alpha},
\label{eq:Aab_def}
\end{align}
and hence
\begin{align}
\sum_{\alpha=1}^N \phi_\alpha \nabla \mu_\alpha^{\mathrm{lin}}
=
\sum_{\alpha,\beta=1}^N
\phi_\alpha \,
\nabla\!\Big(
\mathbf a_{\alpha\beta}
\cdot \nabla\phi_\beta
\Big).
\label{eq:cap_lin_full}
\end{align}
We now invoke A10 (\emph{Objectivity and isotropy}). For fixed composition $\boldsymbol{\phi}$, each coefficient
$\mathbf a_{\alpha\beta}(\boldsymbol{\phi})\in\mathbb{R}^d$ is a spatial vector-valued constitutive quantity whose arguments are only the scalar variables $\phi_1,\dots,\phi_N$. Under a spatial rotation, these scalar inputs are unchanged, while isotropy requires the constitutive response to rotate covariantly. Therefore
\begin{align}
\mathbf a_{\alpha\beta}(\boldsymbol{\phi})
=
\mathbf Q\,\mathbf a_{\alpha\beta}(\boldsymbol{\phi})
\qquad
\text{for all }\mathbf Q\in SO(d).
\end{align}
The only vector in $\mathbb{R}^d$ invariant under all rotations is the zero vector. Hence
\begin{align}
\mathbf a_{\alpha\beta}(\boldsymbol{\phi})\equiv \mathbf 0
\qquad
\text{for all }\alpha,\beta.
\end{align}
Substituting into \eqref{eq:mu_lin_A} and \eqref{eq:cap_lin_full} yields
\begin{align}
\mu_\alpha^{\mathrm{lin}}=0,
\qquad
\sum_{\alpha=1}^N \phi_\alpha \nabla \mu_\alpha^{\mathrm{lin}}
\equiv \mathbf 0.
\label{eq:no_linear_grad_force}
\end{align}
Thus linear gradient terms do not contribute to the capillary-force closure and may be discarded ($\Psi_{\mathrm{lin}} = 0$).

\subsubsection*{Quadratic gradient terms}
The above reasoning shows that the leading admissible gradient dependence is quadratic. Axiom A3 (Constant-capillarity) specifies that
the Hessian of $\Psi$ with respect to $\nabla\boldsymbol{\phi}$ is independent of $\boldsymbol{\phi}$. This yields the constant-coefficient quadratic form
\begin{align}
\Psi_{\rm quad}(\nabla\boldsymbol{\phi})
=
\frac12\sum_{\alpha,\beta=1}^N \kappa_{\alpha\beta}\nabla\phi_\alpha\cdot\nabla\phi_\beta,
\qquad
\kappa_{\alpha\beta}=\kappa_{\beta\alpha}.
\label{eq:Psi_kappa}
\end{align}

\subsubsection{Final form and structural summary}

Combining \eqref{eq:Psi0_final} and \eqref{eq:Psi_kappa} yields the derived free-energy structure:
\begin{subequations}
\label{eq:Psi_final}
    \begin{align}
\Psi(\boldsymbol{\phi},\nabla\boldsymbol{\phi})
=&~ \Psi_0(\boldsymbol{\phi},\nabla\boldsymbol{\phi})+\frac12\sum_{\alpha,\beta=1}^{N}\kappa_{\alpha\beta}\nabla\phi_\alpha\cdot\nabla\phi_\beta.\\
\Psi_0(\boldsymbol{\phi},\nabla\boldsymbol{\phi})
=&~
\sum_{\alpha=1}^{N}W_\alpha \phi_\alpha\log\phi_\alpha
-W\sum_{\alpha,\beta=1}^{N} \chi_{\alpha\beta}\phi_\alpha\phi_\beta,
\end{align}
\end{subequations}
where the subscript $0$ refers to the standard free energy that does not contain any gradient contributions, whereas the gradient term represents an excess free energy in the interfacial region. Here $\kappa_{\mA\mB}=\kappa_{\mB\mA}$ defines the symmetric (volume-based) capillarity matrix that includes both self- and cross-interaction terms among phases. We note that the symmetry assumptions yield a symmetric Korteweg stress, cf. \eqref{eq: korteweg}.  The associated chemical potential takes the form:
\begin{align}
    \mu_\mA = \partial_{\phi_\mA} \Psi_0 - \sum_{\mB} \kappa_{\mA\mB}\Delta \phi_\mB,
\end{align}
where we used the independence of the state variables $\left\{\phi_\mA\right\}$.

In the incompressible literature it is common to work with capillarity tensors that exclude self-interactions, see e.g. \cite{boyer2006study}. With the aid of the saturation constraint \eqref{eq: saturation constraint} we reformulate the gradient contribution.
\begin{lemma}[Gradient contribution to the free energy]\label{lem: rewriting gradient term}
    The gradient contribution of the free energy \eqref{eq:Psi_final} may be written as
    \begin{subequations}
      \begin{align}
     \tfrac{1}{2}\sum_{\mA,\mB} \kappa_{\mA\mB} \nabla \phi_\mA \cdot \nabla \phi_\mB =&~ -\tfrac{1}{2}\sum_{\substack{\mA,\mB\\\mA < \mB}} \sigma_{\mA\mB} \nabla \phi_\mA \cdot \nabla \phi_\mB,\\
     \sigma_{\mA\mB}=\sigma_{\mB\mA} =&~ \kappa_{\mA\mA}+\kappa_{\mB\mB}-2\kappa_{\mA\mB} \geq 0,
     \end{align}
     \end{subequations}
     where $\sigma_{\mA\mB}$ only accounts for interfacial contributions between different phases.
\end{lemma}

\begin{proof}
    See \cite{surfacetension2026}.
\end{proof}
The chemical potential may now be written as
\begin{align}
    \mu_\mA =&~ \partial_{\phi_\mA} \Psi_0 +\sum_{\substack{\mB\\\mB > \mA}} \sigma_{\mA\mB}\Delta \phi_\mB.
\end{align}
Note that at the continuous level all these formulations are equivalent by using the saturation condition. Under the assumption that the matrix $\kappa_{\mA\mB}$ is symmetric positive semi-definite, \eqref{eq:Psi_final} yields a convex gradient energy, which is more amenable to discretization. We emphasize that the reduced form $\sigma_{\mA\mB}$ without the saturation constraint does not imply a convex energy functional, as this matrix is not symmetric positive semi-definite.

\begin{remark}[Capillarity coefficients are not surface tensions]
    We note that neither $\kappa_{\mA\mB}$ nor $\sigma_{\mA\mB}$ represents a surface tension quantity.
\end{remark}

\begin{remark}[Alternative ideal free-energy forms]
Again, noting that affine terms do not affect the model, the ideal contribution
\begin{align}
\Psi_0^{\mathrm{ideal}}(\boldsymbol{\phi})
=
\sum_{\alpha=1}^N W_\alpha \phi_\alpha \log \phi_\alpha
\end{align}
is for $W_\alpha=\Lambda_\alpha\rho_\alpha$, equivalent to 
\begin{align}
\Psi_0^{\mathrm{ideal}}(\tilde{\boldsymbol\rho})
\sim
\sum_{\alpha=1}^N \Lambda_\alpha \tilde\rho_\alpha \log\!\Big(\frac{\tilde\rho_\alpha}{\rho_0}\Big)
\end{align}
for any fixed reference density $\rho_0>0$, and for $W_\alpha=\Gamma_\alpha\rho_\alpha/M_\alpha$ equivalent to
\begin{align}
\Psi_0^{\mathrm{ideal}}(\mathbf c)
\sim
\sum_{\alpha=1}^N \Gamma_\alpha c_\alpha \log\!\Big(\frac{c_\alpha}{c_0}\Big),
\end{align}
where $c_\alpha=\tilde\rho_\alpha/M_\alpha$ is the molar concentration and $c_0>0$ is a fixed reference concentration.
\end{remark}

\subsection{Derivation of the mobility tensor}\label{sec:mobility_derivation}

In this section we derive a reduction-consistent Onsager mobility tensor $\mathbf{M}(\boldsymbol{\phi})$ for the flux law
in \cref{sec: NSCH model}. We work within the mobility model class of \cref{sec:axioms}, in particular symmetry, positive
semidefiniteness, and constraint compatibility (A4/M1--M3), and we impose the diffusive-flux reduction requirement
(A5).

From \eqref{eq:RC_flux_map} we find:
\begin{subequations}
\begin{align}
   - \sum_{\beta=2}^{N-1} (\mathbf{M}_{1\beta}+ \mathbf{M}_{N \beta}) \nabla \hat g_\beta - (\mathbf{M}_{11}+\mathbf{M}_{N1}) \nabla g_1  - (\mathbf{M}_{1N}+\mathbf{M}_{NN}) \nabla g_N \nn \\
   = - \sum_{\beta=2}^{N-1} \hat{\mathbf{M}}_{1\beta} \nabla \hat g_\beta  - \hat{\mathbf{M}}_{11} \nabla \hat{g}_1,\label{eq: J1N}\\
   - \sum_{\beta=2}^{N-1} \mathbf{M}_{\alpha\beta} \nabla g_\beta - \mathbf{M}_{\alpha1}\nabla g_1  - \mathbf{M}_{\alpha N} \nabla g_N    = - \sum_{\beta=2}^{N-1} \hat{\mathbf{M}}_{\alpha \beta} \nabla \hat g_\beta  - \hat{\mathbf{M}}_{\alpha1} \nabla \hat{g}_1\nn\\
   \text{for} \quad \alpha = 2,...,N-1.\label{eq: Jalpha}
\end{align}    
\end{subequations}
The arbitrariness of $\nabla g_\mA$ provides the conditions:
\begin{subequations}
    \begin{align}
        \mathbf{M}_{1\beta}+ \mathbf{M}_{N \beta} =&~ \hat{\mathbf{M}}_{1\beta}, \quad \text{for}\quad \beta = 2,...,N-1.\\
        \mathbf{M}_{\alpha\beta} = &~ \hat{\mathbf{M}}_{\alpha\beta} \quad \text{for}\quad \alpha = 2,...,N-1.
    \end{align}
\end{subequations}
We have the identities:
\begin{subequations}\label{eq:g_identities_for_mobility}
  \begin{align}
\nabla \hat g_\beta =&~ \nabla g_\beta \quad (\beta=2,\dots,N-1),\\
\phi_1\nabla g_1+\phi_N\nabla g_N =&~ \hat\phi_1\nabla\hat g_1,
  \end{align}
\end{subequations}
and hence:
\begin{align}
\label{eq:ghat1_as_convex_combo}
\nabla\hat g_1=\frac{\phi_1}{\hat\phi_1}\nabla g_1+\frac{\phi_N}{\hat\phi_1}\nabla g_N.
\end{align}
Using \eqref{eq:g_identities_for_mobility}--\eqref{eq:ghat1_as_convex_combo} in \eqref{eq: Jalpha}-\eqref{eq: J1N} gives
\begin{subequations}\label{eq:match_1alpha_equation}
\begin{align}
&-\sum_{\beta=2}^{N-1} (\mathbf{M}_{1\beta}+\mathbf{M}_{N\beta})\nabla g_\beta
-(\mathbf{M}_{11}+\mathbf{M}_{N1})\nabla g_1
-(\mathbf{M}_{1N}+\mathbf{M}_{NN})\nabla g_N
\nonumber\\
&\hspace{1cm}=
-\sum_{\beta=2}^{N-1}\hat{\mathbf{M}}_{1\beta}\nabla g_\beta
-\hat{\mathbf{M}}_{11}\Big(\frac{\phi_1}{\hat\phi_1}\nabla g_1+\frac{\phi_N}{\hat\phi_1}\nabla g_N\Big),
\label{eq:match_1_equation}\\
&-\sum_{\beta=2}^{N-1} \mathbf{M}_{\alpha\beta}\nabla g_\beta
-\mathbf{M}_{\alpha 1}\nabla g_1
-\mathbf{M}_{\alpha N}\nabla g_N
\nonumber\\
&\hspace{1cm}=
-\sum_{\beta=2}^{N-1} \hat{\mathbf{M}}_{\alpha\beta}\nabla g_\beta
-\hat{\mathbf{M}}_{\alpha 1}\Big(
\frac{\phi_1}{\hat\phi_1}\nabla g_1+\frac{\phi_N}{\hat\phi_1}\nabla g_N
\Big),  \quad \text{for} \quad \alpha = 2,...,N-1.
\label{eq:match_alpha_equation}
\end{align}
\end{subequations}
Since \eqref{eq:match_1alpha_equation} holds for all admissible driving fields, the coefficients of the
independent gradients must vanish, i.e.:
\begin{subequations}\label{eq: mob relations}
\begin{align}
%\hat{\mathbf{M}}_{1\beta}=&~\mathbf{M}_{1\beta}+\mathbf{M}_{N\beta}\quad \text{for} \quad\beta=2,\dots,N-1,\\
\mathbf{M}_{11}+\mathbf{M}_{N1}=&~\hat{\mathbf{M}}_{11}\frac{\phi_1}{\hat\phi_1}, \label{eq: mob relations 1}\\
\mathbf{M}_{1N}+\mathbf{M}_{NN}=&~\hat{\mathbf{M}}_{11}\frac{\phi_N}{\hat\phi_1}, \label{eq: mob relations 2}\\
%\mathbf{M}_{\alpha\beta}=&~\hat{\mathbf{M}}_{\alpha\beta} \qquad\qquad \text{for} \quad \beta=2,\dots,N-1,\\
\mathbf{M}_{\alpha 1}=&~\hat{\mathbf{M}}_{\alpha1}\frac{\phi_1}{\hat\phi_1},  \quad \text{for} \quad \alpha = 2,...,N-1, \label{eq: mob relations 3}\\
\mathbf{M}_{\alpha N}=&~\hat{\mathbf{M}}_{\alpha1}\frac{\phi_N}{\hat\phi_1},  \quad \text{for} \quad \alpha = 2,...,N-1.\label{eq: mob relations 4}
\end{align}
\end{subequations}
We proceed by considering variations $(\phi_1,\phi_N) \rightarrow (\phi_1 + \delta,\phi_N-\delta)$ that keep $\hat{\phi}_1$ fixed. The mobilities $\hat{\mathbf{M}}_{\alpha1}$ and $\hat{\mathbf{M}}_{11}$ are unchanged by these variations. Along redistributions $(\phi_1,\phi_N)\mapsto(\phi_1+\delta,\phi_N-\delta)$ that keep
  $\hat\phi_1=\phi_1+\phi_N$ (and thus $\hat{\boldsymbol{\phi}}$) fixed, relations \eqref{eq: mob relations 3}-\eqref{eq: mob relations 4} imply that 
  \begin{itemize}
    \item $\mathbf{M}_{11}+\mathbf{M}_{N1}$ depends linearly on $\phi_1$ (equivalently, $(\mathbf{M}_{11}+\mathbf{M}_{N1})/\phi_1$ is split-independent),
    \item $\mathbf{M}_{1N}+\mathbf{M}_{NN}$ depends linearly on $\phi_N$ (equivalently, $(\mathbf{M}_{1N}+\mathbf{M}_{NN})/\phi_N$ is split-independent),
    \item $\mathbf{M}_{\alpha 1}$ depends linearly on $\phi_1$ for $\alpha=2,\ldots,N-1$ (equivalently, $\mathbf{M}_{\alpha1}/\phi_1$ is split-independent),
    \item $\mathbf{M}_{\alpha N}$ depends linearly on $\phi_N$ for $\alpha=2,\ldots,N-1$ (equivalently, $\mathbf{M}_{\alpha N}/\phi_N$ is split-independent).
  \end{itemize}
These split-linearity statements are consequences of merging the \emph{specific} identical pair $(1,N)$ and therefore constrain only how mobility entries respond to redistributions of $\hat\phi_1=\phi_1+\phi_N$ between $\phi_1$ and $\phi_N$.
In particular, they enforce linear dependence on $\phi_1$ or $\phi_N$ \emph{along such redistributions}, but they do not
(by themselves) determine how the same entries depend on other volume fractions, e.g.\ on $\phi_\alpha$ for
$\alpha\in\{2,\ldots,N-1\}$, since these are held fixed during the redistribution. Symmetry of $\mathbf{M}$ only transfers
the same split-linearity between corresponding rows and columns (e.g.\ $\mathbf{M}_{\alpha1}=\mathbf{M}_{1\alpha}$), but it
does not introduce new dependence on $\phi_\alpha$. To obtain a closed constitutive form for \emph{all} pairs
$(\alpha,\beta)$ that is consistent with performing the same merge operation for arbitrary identical phases and for any
number of phases (A5), we now introduce an extensible, label-covariant mobility form compatible with M1-M4. We choose
$\mathbf{M}_{\alpha\beta}=M_{\alpha\beta}\mathbf{I}$ from pairwise exchange mobilities by prescribing, for $\alpha\neq\beta$,
\begin{align}
M_{\alpha\beta}(\boldsymbol{\phi}) = - m_{\alpha\beta}(\boldsymbol{\phi}) \mathcal{M}(\phi_\alpha,\phi_\beta),
\qquad
m_{\alpha\beta}(\boldsymbol{\phi})=m_{\beta\alpha}(\boldsymbol{\phi})\ge 0,
\label{eq:pairwise_form}
\end{align}
where $m_{\mA\mB}(\boldsymbol{\phi})=m_{\mA\mB}(\hat{\boldsymbol{\phi}})$ and where $\mathcal{M}:(0,1)^2\to{\R}$ is a single symmetric scalar function,
$\mathcal{M}(u,v)=\mathcal{M}(v,u)$, shared by all pairs (label covariance). The diagonal entries are then fixed by the
constraint compatibility $M(\boldsymbol{\phi})\mathbf{1}=\mathbf{0}$:
\begin{align}
M_{\alpha\alpha}(\boldsymbol{\phi}) = \sum_{\beta\neq\alpha} m_{\alpha\beta}(\boldsymbol{\phi}) \mathcal{M}(\phi_\alpha,\phi_\beta).
\label{eq:diag_from_constraint}
\end{align}

We now impose reduction consistency. Let phases $1$ and $N$ be identical. By identical-phases symmetry, $m_{1\mA}=m_{N\mA}$ for all $\mA=2,\dots,N-1$. The off-diagonal relation yields, for each
$\mA=2,\dots,N-1$,
\begin{align}
\hat M_{1\mA}=M_{1\mA}+M_{N\mA}.
\end{align}
Substituting the form \eqref{eq:pairwise_form} and using $\hat\phi_1=\phi_1+\phi_N$ and $\hat\phi_\mA=\phi_\mA$ gives
\begin{align}
-m_{1\mA}\mathcal{M}(\hat\phi_1,\phi_\mA)
=
-m_{1\mA}\mathcal{M}(\phi_1,\phi_\mA)
-m_{1\mA}\mathcal{M}(\phi_N,\phi_\mA).
\end{align}
Cancelling $-m_{1\mA}$ (for $m_{1\mA}>0$) yields the functional equation
\begin{align}
\mathcal{M}(\phi_1+\phi_N,\phi_\mA)=\mathcal{M}(\phi_1,\phi_\mA)+\mathcal{M}(\phi_N,\phi_\mA),
\qquad \mA=2,\dots,N-1.
\label{eq:M_additivity}
\end{align}
Since $\mA$ is arbitrary and $\phi_\mA$ can take any value in $(0,1)$, \eqref{eq:M_additivity} implies that for each fixed
$v\in(0,1)$ the map $u\mapsto \mathcal{M}(u,v)$ is additive on $(0,1)$. By symmetry of $\mathcal{M}$, the same holds in
the second argument. Together with the natural boundary condition $\mathcal{M}(0,v)=0$ (no mobility involving an absent
phase) and mild regularity (e.g.\ continuity), this forces $\mathcal{M}$ to be bilinear:
\begin{align}
\mathcal{M}(u,v)=cuv.
\label{eq:M_bilinear}
\end{align}
The constant $c>0$ may be absorbed into the coefficients $m_{\alpha\beta}$, and we therefore set, without loss of
generality,
\begin{align}
\mathcal{M}(u,v)=uv.
\label{eq:M_uv}
\end{align}
With \eqref{eq:M_uv}, the mobility tensor is
\begin{subequations}
\begin{align}
M_{\alpha\beta}(\boldsymbol{\phi})=&~-m_{\alpha\beta}\phi_\alpha\phi_\beta \quad (\alpha\neq\beta),\\
M_{\alpha\alpha}(\boldsymbol{\phi})=&~\sum_{\beta\neq\alpha} m_{\alpha\beta}\phi_\alpha\phi_\beta
=\phi_\alpha\sum_{\beta\neq\alpha} m_{\alpha\beta}\phi_\beta,
\label{eq:M_final_entries}
\end{align}
\end{subequations}
with $m_{\alpha\beta}=m_{\beta\alpha}\ge 0$.

\subsection{Mixture-aware properties}

We now collect the main mixture-aware consequences of the constitutive structure derived above. The first two results show that the model is reduction-consistent in the two basic situations of interest: merging physically identical phases, and restricting to a face of the Gibbs simplex corresponding to an absent phase.
\begin{theorem}[Reduction by merging identical phases]\label{thm:merge_reduction}
Assume $N\ge 2$ and consider two phases, without loss of generality phases $1$ and $N$, which are identical in the sense that
\begin{subequations}
\begin{align}
\rho_1&=\rho_N,\\
\chi_{1\mA}&=\chi_{N\mA},\qquad \mA=1,\dots,N,\\
\kappa_{1\mA}&=\kappa_{N\mA},\qquad \mA=1,\dots,N,\\
m_{1\mA}&=m_{N\mA},\qquad \mA=2,\dots,N-1.
\end{align}
\end{subequations}
Define the merged variables by
\begin{subequations}
\begin{align}
\widehat\phi_1&:=\phi_1+\phi_N,\\
\widehat\phi_\mA&:=\phi_\mA,\qquad \mA=2,\dots,N-1.
\end{align}
\end{subequations}
Then the $N$-phase NSCH system reduces to the $(N-1)$-phase NSCH system for
\begin{align}
\widehat{\boldsymbol\phi}=(\widehat\phi_1,\widehat\phi_2,\dots,\widehat\phi_{N-1}),
\end{align}
with the reduced coefficients:
\begin{subequations}
\begin{align}
\hat\rho_1&:=\rho_1,\qquad \hat\rho_\mA:=\rho_\mA,\qquad \mA=2,\dots,N-1,\\
\hat\chi_{11}&:=\chi_{11}=\chi_{1N}=\chi_{NN},\\
\hat\chi_{1\mA}&:=\chi_{1\mA}=\chi_{N\mA},\qquad \mA=2,\dots,N-1,\\
\hat\chi_{\mB\mA}&:=\chi_{\mB\mA},\qquad \mB,\mA=2,\dots,N-1,\\
\hat\kappa_{11}&:=\kappa_{11}=\kappa_{1N}=\kappa_{NN},\\
\hat\kappa_{1\mA}&:=\kappa_{1\mA}=\kappa_{N\mA},\qquad \mA=2,\dots,N-1,\\
\hat\kappa_{\mB\mA}&:=\kappa_{\mB\mA},\qquad \mB,\mA=2,\dots,N-1,\\
\hat m_{1\mA}&:=m_{1\mA}=m_{N\mA},\qquad \mA=2,\dots,N-1,\\
\hat m_{\mB\mA}&:=m_{\mB\mA},\qquad \mB,\mA=2,\dots,N-1.
\end{align}
\end{subequations}
\end{theorem}

\begin{remark}
    The above result employs a constant viscosity $\nu$. However, this can be extended without further complications to the widely used linear viscosity case, i.e. $\nu = \sum_\mA \nu_\mA\phi_\mA$ for non-negative partial viscosities $\nu_\mA$. In view of Theorem \ref{thm:merge_reduction} one would require $\nu_1=\nu_N$.
\end{remark}

\begin{theorem}[Reduction for absent phase]\label{thm:face_reduction}

If a solution of the $N$-phase NSCH system satisfies $\phi_\gamma\equiv0$,
then the restriction of the $N$-phase NSCH system to $\mathcal G_\gamma$ coincides with the $(N-1)$-phase NSCH system obtained by deleting phase $\gamma$ and canonically relabeling the remaining indices.
\end{theorem}

\begin{corollary}[Absent phase remains absent]\label{cor:absence_invariance}
If a phase is absent, i.e. $
\phi_\gamma(\cdot,0)\equiv 0$,
then it remains absent, i.e. $
\phi_\gamma(\cdot,t)\equiv 0
\qquad\text{for all }t\ge 0$.
\end{corollary}

The preceding results show that the present closure is mixture-aware in the sense of reduction consistency. For comparison, the next two propositions show that two commonly used multiphase free-energy constructions do not satisfy this requirement.

\begin{proposition}[Boyer--Lapuerta free energies are not mixture-aware]\label{prop: Boyer}
For $N=3$, the Boyer--Lapuerta free energy is not mixture-aware. 
\end{proposition}
The ternary case already exhibits the obstruction clearly. For $N\ge 4$, the energy also seems not mixture-aware. However, we were not able to extract a closed-form free energy from \cite{boyer2014hierarchy}, as the presentation of the constitutive structure seems incomplete.

\begin{proposition}[Steinbach free energy is not mixture-aware]\label{prop: Steinbach}
The free energy structure
\begin{subequations}
   \begin{align}
\Psi =&~ \Psi_0(\boldsymbol{\phi}) + \Psi_{\nabla}(\boldsymbol{\phi},\nabla \boldsymbol{\phi}),\\
\Psi_{\nabla}(\boldsymbol{\phi},\nabla \boldsymbol{\phi}) = &~ \sum_{\alpha<\beta}
\omega_{\alpha\beta} \bigl|\phi_\alpha\nabla \phi_\beta-\phi_\beta\nabla \phi_\alpha\bigr|^2,
\end{align}
\end{subequations}
for coefficients $\omega_{\alpha\beta}\geq 0$ is not mixture-aware.
\end{proposition}
This free energy structure was proposed in \cite{steinbach1996phase} and adopted, for example, in \cite{garcke1999multiphase,ben2014phase}.

The proofs of Theorems  \ref{thm:merge_reduction}-\ref{thm:face_reduction}, Corollary \ref{cor:absence_invariance} and Propositions \ref{prop: Boyer} are \ref{prop: Steinbach} are provided in the Appendix \ref{section: appendix: Proofs}.

\section{Constitutive closure specification}
\label{sec:practical_closure}

Section~\ref{sec: mixture-aware closure} identified the admissible \emph{structure} of the constitutive closure from the mixture-aware axioms. To obtain a concrete closed model, one still has to specify the free-energy and mobility coefficients within this admissible class. In this section, we introduce a convenient parameterization of the free energy and mobility tensor and discuss the resulting constitutive behavior, in particular for distinct and identical phases.

\subsection{Free energy}
\label{sec:practical_free_energy}

To expose explicitly the diffuse-interface thickness scale, we introduce a parameter $\varepsilon_0>0$, choose $W_\alpha=W=\bar{W}/\varepsilon_0$ and $\kappa_{\alpha\beta}=\varepsilon_0 \bar{\kappa}_{\alpha\beta}$ where $\bar W>0$ and $\bar\kappa_{\alpha\beta}$ are independent of $\varepsilon_0$. This yields the free energy
\begin{align}\label{eq: free energy epsilon}
\Psi(\boldsymbol{\phi},\nabla\boldsymbol{\phi})
=
\frac{\bar W}{\varepsilon_0}
\Bigl(
\sum_{\alpha=1}^{N}\phi_\alpha\log\phi_\alpha
-
\sum_{\alpha,\beta=1}^{N}\chi_{\alpha\beta}\phi_\alpha\phi_\beta
\Bigr)
+
\frac{\varepsilon_0}{2}
\sum_{\alpha,\beta=1}^{N}\bar\kappa_{\alpha\beta}\nabla\phi_\alpha\cdot\nabla\phi_\beta.
\end{align}
Here $\bar W$ sets the bulk energy scale, $\chi_{\alpha\beta}$ determines the location and depth of the minima of the homogeneous free energy on the Gibbs simplex, $\bar\kappa_{\alpha\beta}$ determines the energetic cost of gradients, and $\varepsilon_0$ sets the diffuse thickness scale.

The parameterization \eqref{eq: free energy epsilon} separates the bulk and gradient contributions by the standard factors $1/\varepsilon_0$ and $\varepsilon_0$. In the NSCH model, as in diffuse-interface models more generally, the diffuse-interface thickness is an emergent property of the free energy rather than an independently prescribed quantity. The parameter $\varepsilon_0$ controls its overall scale, but the actual interface width also depends on the coefficients $\bar W$ and $\bar\kappa_{\alpha\beta}$, and therefore is not equal to $\varepsilon_0$ in general. In fact, because an $N$-phase model generally contains multiple pairwise interfaces, the diffuse-interface width is interface-dependent. Given the pairwise interface widths $\varepsilon_{\mA\mB}$ for interfaces $\mA$--$\mB$, we define the representative interface width by $\varepsilon = \varepsilon_0 \min_{\alpha<\beta} \varepsilon_{\alpha\beta}$. A precise definition of the corresponding pairwise interface widths $\varepsilon_{\mA\mB}$ is given in \cite{surfacetension2026}.

In practical terms, one chooses $\chi_{\alpha\beta}$ and $\bar\kappa_{\alpha\beta}$ so that the desired bulk states and interfacial couplings are represented, and then chooses $\varepsilon_0$ so that the resulting diffuse interfaces are resolvable on the mesh. The quantitative calibration of $\bar\kappa_{\alpha\beta}$ and $\varepsilon_0$ to target surface tensions and numerical interface widths is treated in \cite{surfacetension2026}. Below we discuss the homogeneous (bulk) free energy behavior.

We split the homogeneous part of the free energy into an entropic and an interaction enthalpic contribution,
\begin{subequations}
\begin{align}
\Psi_0(\boldsymbol{\phi})
=&~
\Psi_{0,\mathrm{ent}}(\boldsymbol{\phi})
+
\Psi_{\mathrm{int}}(\boldsymbol{\phi}),\\
\Psi_{0,\mathrm{ent}}(\boldsymbol{\phi})
=&~
\frac{\bar W}{\varepsilon_0}
\sum_{\alpha=1}^{N}\phi_\alpha\log\phi_\alpha,\\
\Psi_{0,\mathrm{int}}(\boldsymbol{\phi})
=&~
-
\frac{\bar W}{\varepsilon_0}
\sum_{\alpha,\beta=1}^{N}\chi_{\alpha\beta}\phi_\alpha\phi_\beta.
\end{align}
\end{subequations}
The first term represents the ideal-mixing entropy term that favors mixing, whereas the second term contains the coefficient-dependent interaction structure.

\begin{figure}
\centering
\captionsetup[subfigure]{justification=centering}

\makebox[\textwidth][c]{%
\begin{subfigure}{0.49\textwidth}
\centering
\includegraphics[scale=0.25]{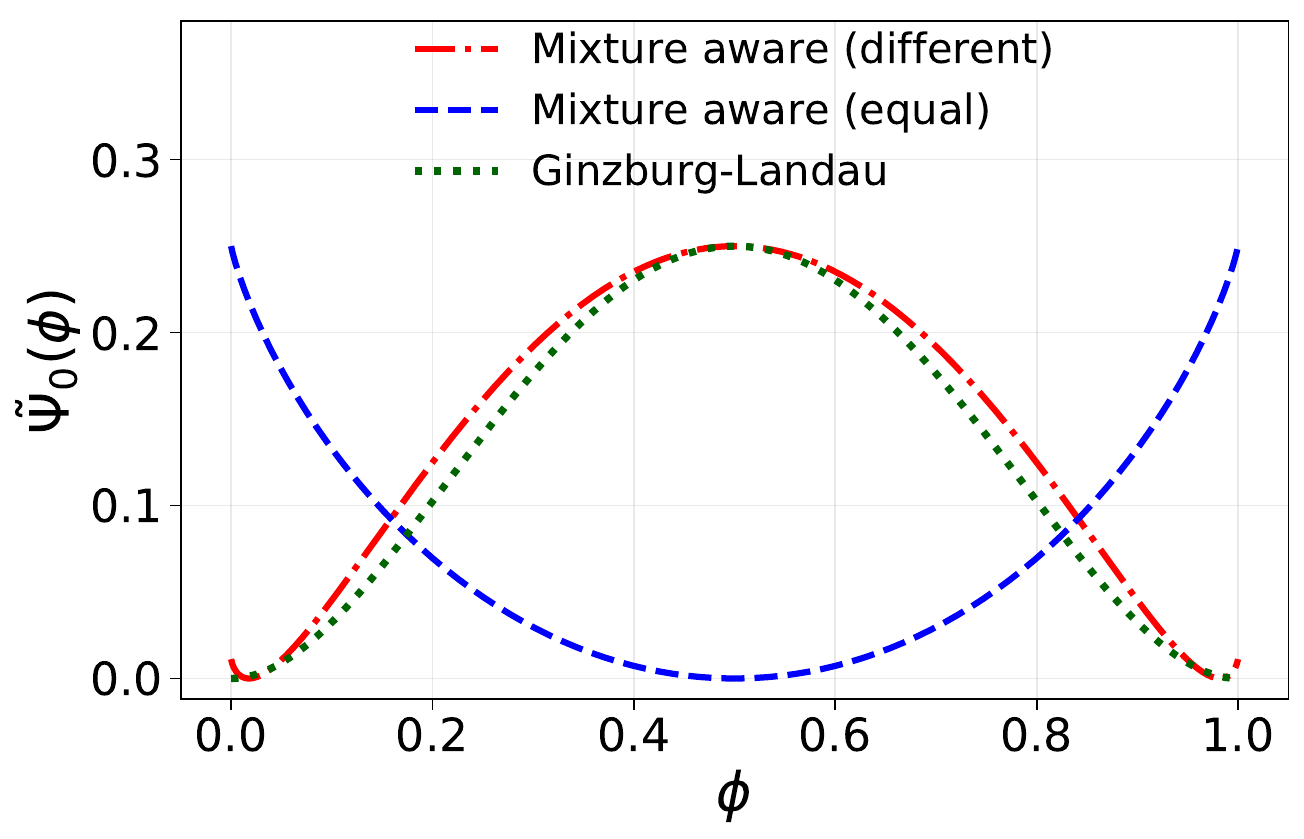}
\caption{Free energies}
\end{subfigure}
}

\medskip
\begin{subfigure}{0.49\textwidth}
\centering
\includegraphics[scale=0.25]{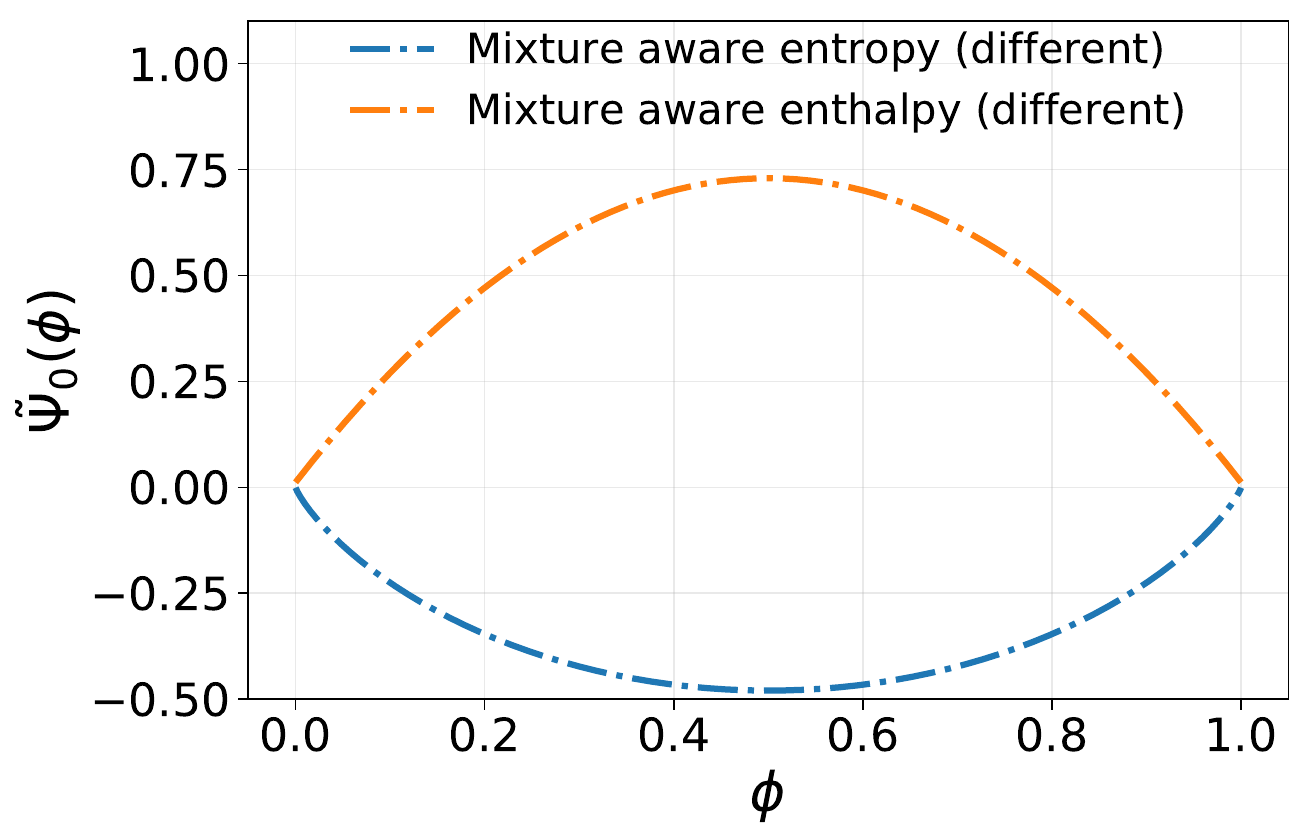}
\caption{Free energy: entropy--enthalpy splitting (different phases)}
\end{subfigure}
\begin{subfigure}{0.49\textwidth}
\centering
\includegraphics[scale=0.25]{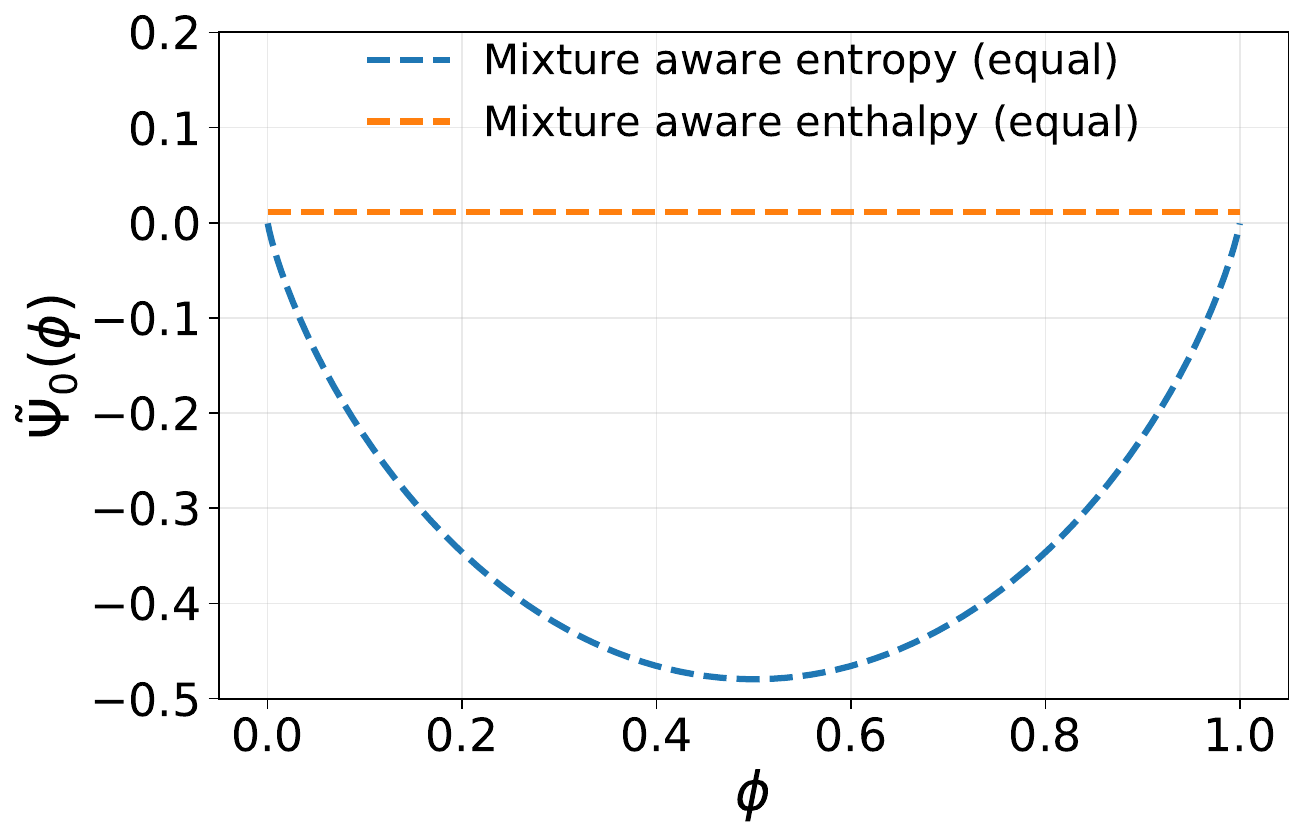}
\caption{Free energy: entropy--enthalpy splitting (identical phases)}
\end{subfigure}
    \caption{Homogeneous binary free-energy densities as functions of $\phi:=\phi_1$ where $\varepsilon_0=1$, $\chi=2.07532$, $\bar{W}=0.691952$. The mixture-aware energies are shown after a constant shift $\tilde{\Psi}_0=\Psi_0+b$, $b={\mathrm const}$, so that their scale roughly matches that of the Ginzburg--Landau potential. (a) Comparison of the Ginzburg--Landau potential $4\phi^2(1-\phi)^2$ with the mixture-aware bulk free energy for identical and different phases (b) Entropy--enthalpy splitting of the free energy for different phases (c) Entropy--enthalpy splitting of the free energy for identical phases.}
\label{fig: 2 phase}
\end{figure}

In \cref{fig: 2 phase}, we compare the homogeneous binary free-energy density of the mixture-aware closure with the quartic Ginzburg--Landau potential. We compare two cases, one representing different phases and one representing identical phases:
\begin{align}
    \boldsymbol{\chi} =\boldsymbol{\chi}^{\mathrm{diff}} =
    \begin{bmatrix}
        \chi_1 & 0 \\ 0 & \chi_2 
    \end{bmatrix}, \qquad     
    \boldsymbol{\chi} = \boldsymbol{\chi}^{\mathrm{identical}} = \begin{bmatrix}
        \chi_1 & \chi_1 \\ \chi_1 & \chi_1 
    \end{bmatrix},
\end{align}
where we take $\chi_1=\chi_2=\chi$ for simplicity.
For the case of different phases the mixture-aware energy exhibits a symmetric double-well structure that is qualitatively very similar to the quartic potential $4\phi^2(1-\phi)^2$. By contrast, in the case of identical phases the profile becomes convex and attains its minimum at equal splitting, $\phi=1/2$. The middle and right panels explain the origin of this difference by separating the homogeneous mixture-aware free energy into its entropic and enthalpic parts. In both cases, the entropic contribution favors mixing and is therefore minimized at $\phi=1/2$, while the enthalpic contribution penalizes mixing. For diagonal $\chi_{\alpha\beta}$, the enthalpic penalty is strong enough to overcome the entropic contribution near the midpoint, resulting in an overall double-well profile. In contrast, for $\chi_{\alpha\beta}=\chi$, the enthalpic contribution becomes constant along the binary edge, so that the remaining variation is entirely entropic and the total free energy is convex.

\begin{figure}
\captionsetup[subfigure]{justification=centering}
\begin{subfigure}{0.33\textwidth}
\centering
\includegraphics[scale=0.38]{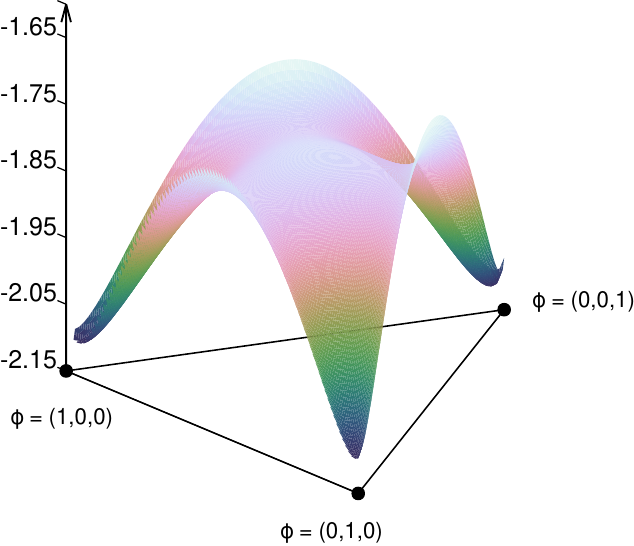}
\caption{Free energy}
\end{subfigure}
\begin{subfigure}{0.33\textwidth}
\centering
\includegraphics[scale=0.38]{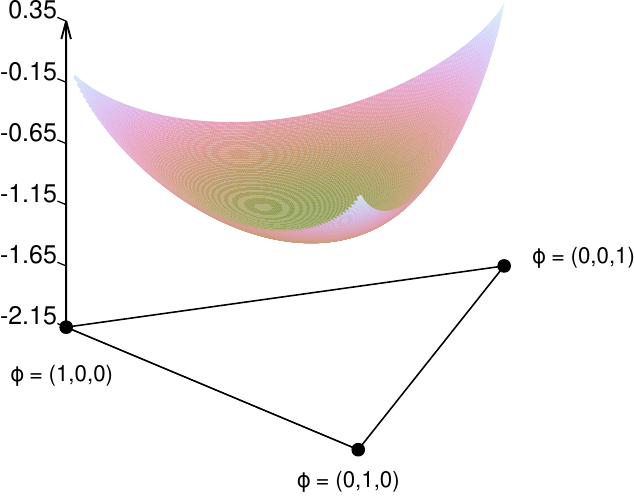}
\caption{Free energy: entropic part}
\end{subfigure}
\begin{subfigure}{0.33\textwidth}
\centering
\includegraphics[scale=0.38]{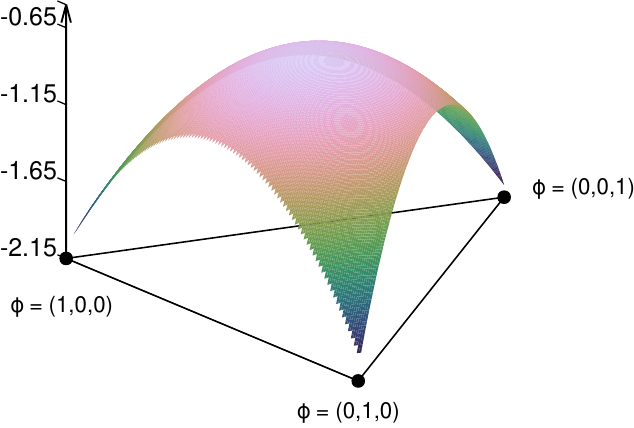}
\caption{Free energy: enthalpic part}
\end{subfigure}
\caption{Homogeneous ternary free-energy densities on the Gibbs simplex for three distinct phases and $\varepsilon_0=1$, $\bar{W}=1$, and $\chi=2.0722$. (a) Mixture-aware bulk free energy $\Psi_0$. (b) Entropic contribution. (c) Enthalpic contribution.}
\label{fig:3 phase}
\end{figure}

\begin{figure}
\captionsetup[subfigure]{justification=centering}
\begin{subfigure}{0.33\textwidth}
\centering
\includegraphics[scale=0.38]{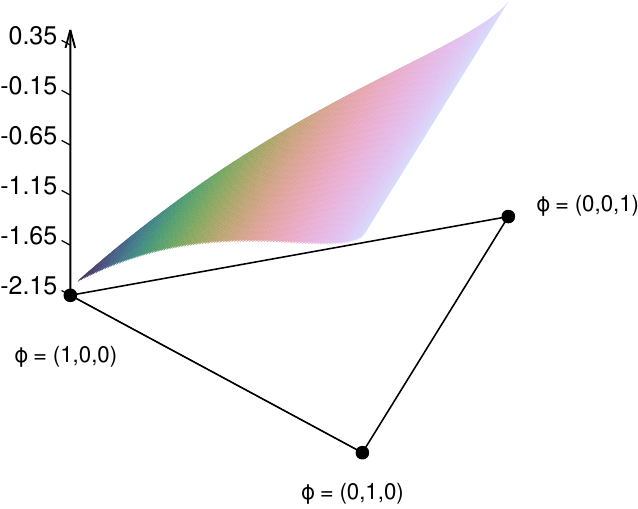}
\caption{Free energy $\Psi_{0,1}$}
\end{subfigure}
\begin{subfigure}{0.33\textwidth}
\centering
\includegraphics[scale=0.38]{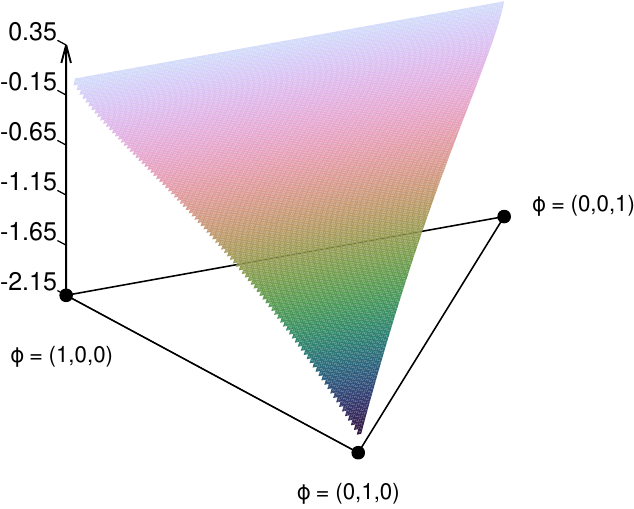}
\caption{Free energy $\Psi_{0,2}$}
\end{subfigure}
\begin{subfigure}{0.33\textwidth}
\centering
\includegraphics[scale=0.38]{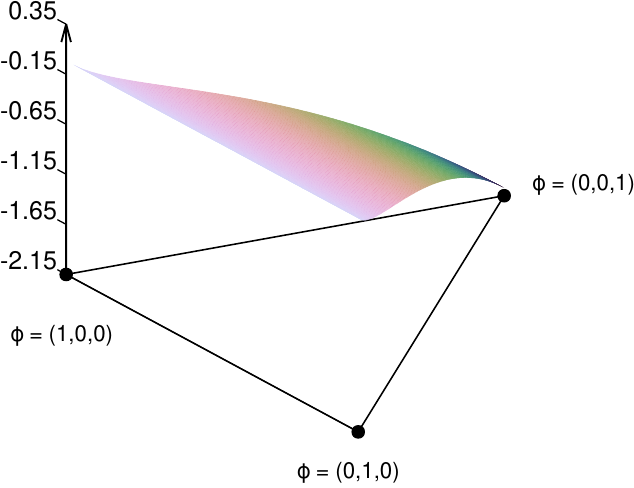}
\caption{Free energy $\Psi_{0,3}$}
\end{subfigure}
\caption{Homogeneous ternary free-energy densities on the Gibbs simplex for three distinct phases and $\varepsilon_0=1$, $\bar{W}=1$, and $\chi=2.0722$. (a)--(c) Partial free energies $\Psi_{0,1}$, $\Psi_{0,2}$, and $\Psi_{0,3}$.}
\label{fig: partial free energies 3 phase}
\end{figure}

Next, we consider the ternary case. We compare a case representing $3$ different phases with a case where phases $1$ and $2$ are identical:
%\begin{subequations}
\begin{align}
    \boldsymbol{\chi} =\boldsymbol{\chi}^{\mathrm{diff}} =
    \begin{bmatrix}
        \chi_1 & 0 & 0\\ 
        0 & \chi_2 & 0 \\ 
        0 & 0 & \chi_3 
    \end{bmatrix}, \qquad     
    \boldsymbol{\chi} = \boldsymbol{\chi}^{\mathrm{identical}} =\begin{bmatrix}
        \chi_1 & \chi_1 & 0\\ 
        \chi_1 & \chi_1 & 0 \\ 
        0& 0 & \chi_3 
    \end{bmatrix},
\end{align}
%\end{subequations}
where again we take $\chi_1=\chi_2=\chi_3=\chi$ for simplicity.

Starting with the case of different phases, the corresponding homogeneous free-energy landscape for three distinct phases, together with its entropic and enthalpic parts, is shown in \cref{fig:3 phase} (see also \cref{fig: heat different} for the heat map of the free energy); its partial free-energy contributions $\Psi_{0,1}$, $\Psi_{0,2}$, and $\Psi_{0,3}$ are shown in \cref{fig: partial free energies 3 phase}. The entropy contribution is smooth and convex toward the interior of the simplex, whereas the interaction part shapes the wells and barriers that determine the preferred bulk states. In the present example, this results in three distinct minima of equal height associated with the pure phases. The partial free-energy plots provide a phase-wise view of this landscape. We emphasize that these are not to be interpreted as independent single-phase energies, but rather as contributions of individual phases to the total bulk free energy. In particular, they make transparent the symmetry or asymmetry of the chosen coefficient set and help identify which phases dominate the energetic barriers between phases. 

Next, for the case of two identical phases, the corresponding homogeneous free-energy landscape, together with its entropic and enthalpic parts, is shown in \cref{fig:3 phase same} (see also \cref{fig: heat equal} for the heat map of the free energy); the associated partial free-energy contributions $\Psi_{0,1}$, $\Psi_{0,2}$, and $\Psi_{0,3}$ are shown in \cref{fig: partial free energies 3 phase same}. The symmetry of the interaction energy is reflected in the Gibbs-simplex plots. Along the edge connecting phases $1$ and $2$, the bulk free energy develops a convex profile with a minimum at equal splitting, $\phi_1=\phi_2=1/2$, exactly as in the binary identical-phase case. By contrast, along the edges connecting phase $3$ to either phase $1$ or phase $2$, one recovers the usual double-well structure associated with two distinct phases. Hence the ternary identical-phase case combines both behaviors in a single Gibbs-simplex landscape: a redistribution trough along the edge between the identical labels, and ordinary binary double wells along the edges involving the distinct phase. Of course, the entropy contribution remains unchanged. The interaction part, however, is constant along the $1$--$2$ edge because labels $1$ and $2$ are energetically indistinguishable, so the variation along this edge is entirely entropic. Along the $1$--$3$ and $2$--$3$ edges, by contrast, the interaction term distinguishes the phases and therefore produces the familiar double-well structure. The partial free-energy plots provide a phase-wise view of the same phenomenon: $\Psi_{0,1}$ and $\Psi_{0,2}$ are symmetric images of one another, while $\Psi_{0,3}$ reflects the role of the distinct third phase.

\begin{figure}
\captionsetup[subfigure]{justification=centering}
\begin{subfigure}{0.33\textwidth}
\centering
\includegraphics[scale=0.38]{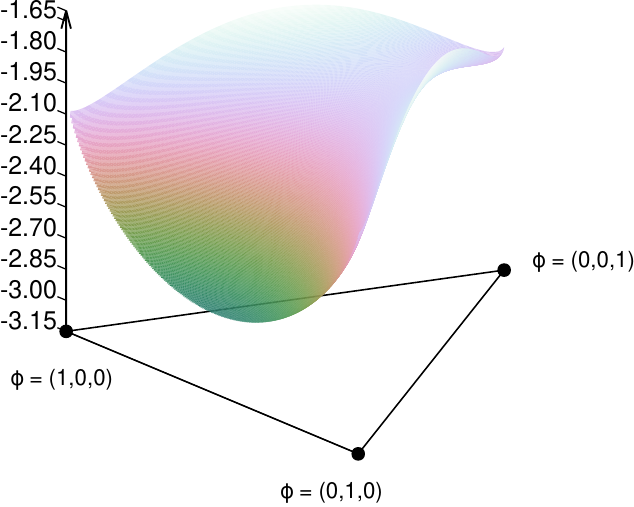}
\caption{Free energy}
\end{subfigure}
\begin{subfigure}{0.33\textwidth}
\centering
\includegraphics[scale=0.38]{figures/gibbs_N3a.pdf}
\caption{Free energy - entropic part}
\end{subfigure}
\begin{subfigure}{0.33\textwidth}
\centering
\includegraphics[scale=0.38]{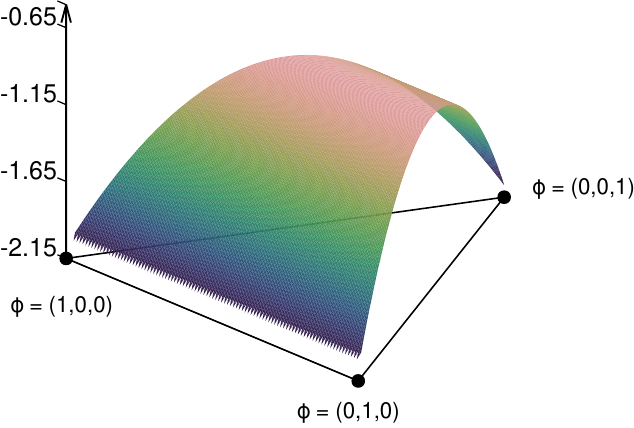}
\caption{Free energy - enthalpic part}
\end{subfigure}
\caption{Homogeneous ternary free-energy densities on the Gibbs simplex in the case where phases $1$ and $2$ are identical, and $\varepsilon_0=1$, $\bar{W}=1$, $\chi=2.0722$. (a) Mixture-aware bulk free energy $\Psi_0$. (b) Entropic contribution. (c) Enthalpic contribution.}
\label{fig:3 phase same}
\end{figure}

\begin{figure}
\captionsetup[subfigure]{justification=centering}
\begin{subfigure}{0.33\textwidth}
\centering
\includegraphics[scale=0.38]{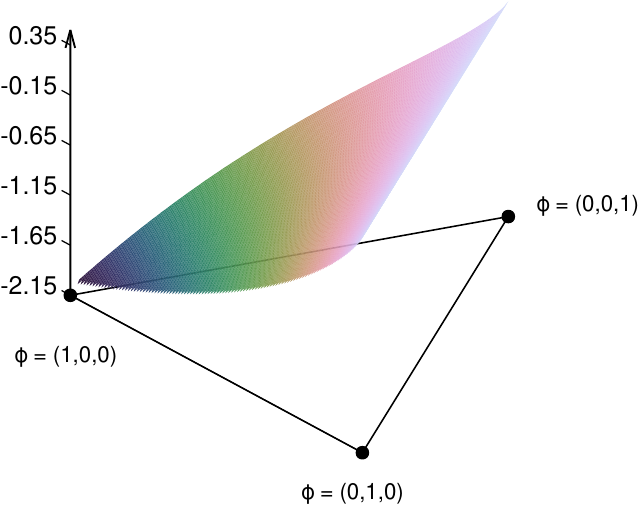}
\caption{Free energy $\Psi_{0,1}$}
\end{subfigure}
\begin{subfigure}{0.33\textwidth}
\centering
\includegraphics[scale=0.38]{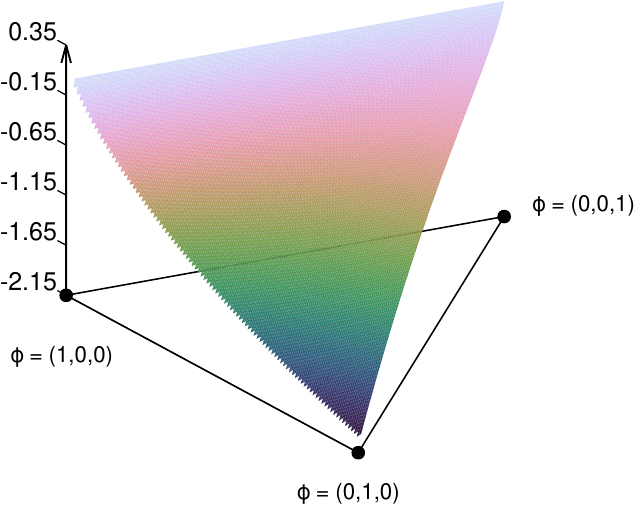}
\caption{Free energy $\Psi_{0,2}$}
\end{subfigure}
\begin{subfigure}{0.33\textwidth}
\centering
\includegraphics[scale=0.38]{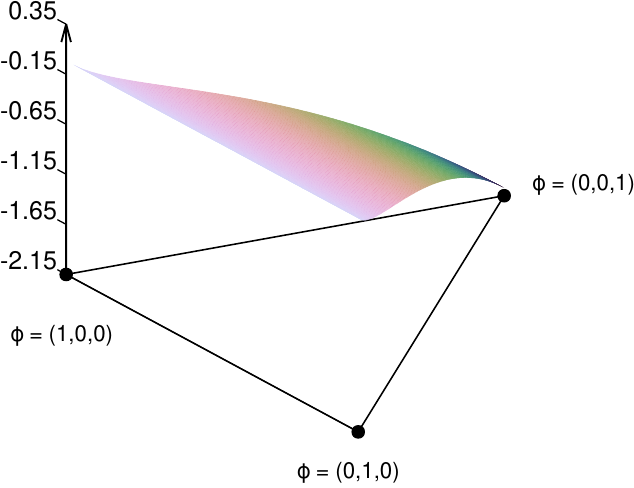}
\caption{Free energy $\Psi_{0,3}$}
\end{subfigure}
\caption{Homogeneous ternary free-energy densities on the Gibbs simplex in the case where phases $1$ and $2$ are identical, and $\varepsilon_0=1$, $\bar{W}=1$, $\chi=2.0722$.  (a)--(c) Partial free energies $\Psi_{0,1}$, $\Psi_{0,2}$, and $\Psi_{0,3}$.}
\label{fig: partial free energies 3 phase same}
\end{figure}

\begin{figure}
\captionsetup[subfigure]{justification=centering}
\begin{subfigure}{0.49\textwidth}
\centering
\includegraphics[scale=0.35]{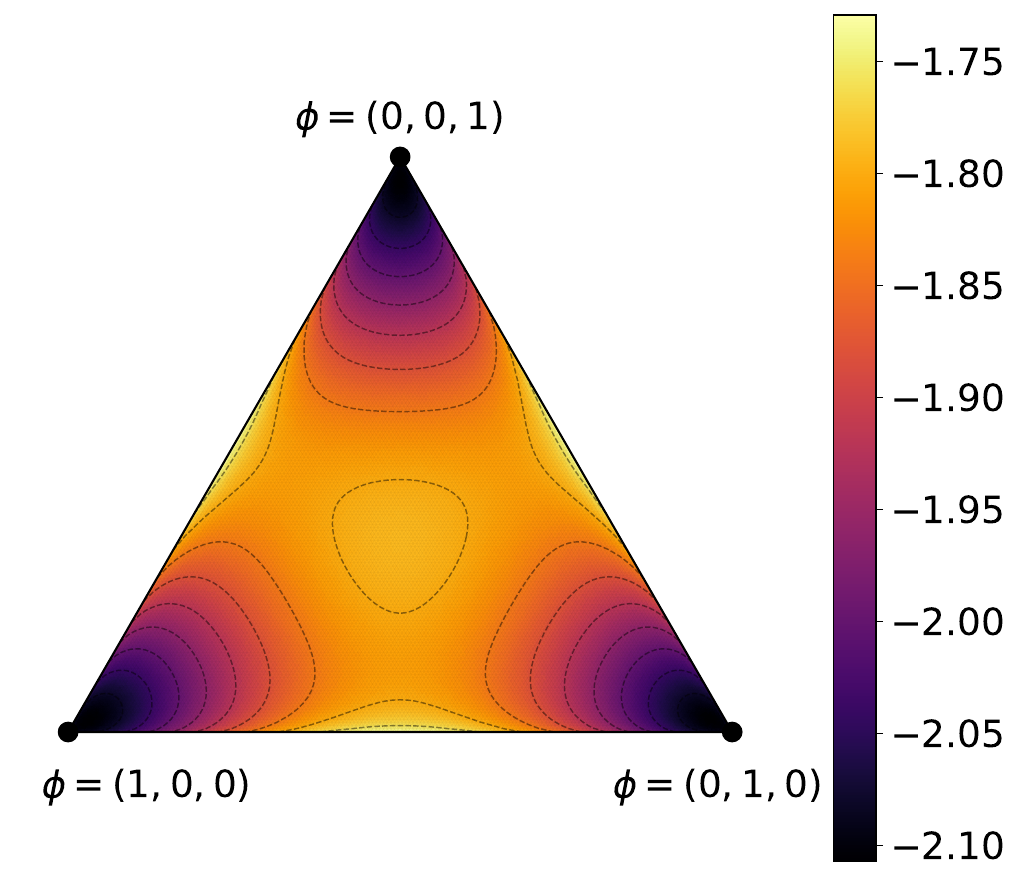}
\caption{Different phases}
\label{fig: heat different}
\end{subfigure}
\begin{subfigure}{0.49\textwidth}
\centering
\includegraphics[scale=0.35]{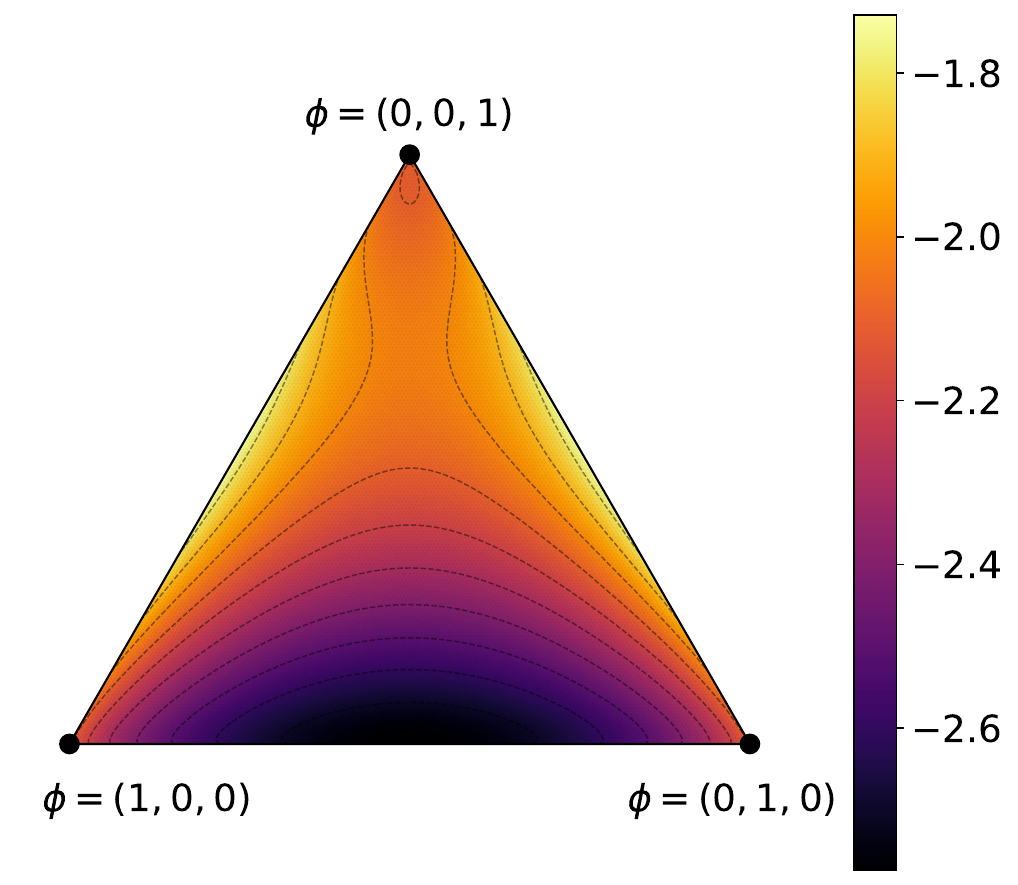}
\caption{Equal phases}
\label{fig: heat equal}
\end{subfigure}
\caption{Heat maps of the homogeneous ternary free energies:(a) different phases, (b) identical phases.}
\label{fig: heatmap}
\end{figure}

\subsection{Regularization of the free energy}

The Flory--Huggins (FH) bulk free-energy density contains entropic contributions of the form
$\phi_\alpha\log\phi_\alpha$. For admissible compositions $\phi_\alpha\in(0,1)$ this expression is smooth;
however, as $\phi_\alpha\to 0^+$ the associated chemical potential diverges, which makes the practical computation of minima or coexistence points very challenging.
To obtain a globally smooth bulk potential with bounded derivatives on $[0,1]$, we replace $\phi\log\phi$ by
a polynomial surrogate and calibrate the interaction parameters such that the pure phases become local
minima of the resulting bulk potential.

The symmetric interaction matrix
$\chi_{\alpha\beta}$ is, in general, fully populated. However, in the current calibration we retain only
its diagonal entries for convenience,
\begin{align}
\chi_{\alpha\beta}=\chi_\alpha\,\delta_{\alpha\beta},
\qquad \alpha,\beta\in\{1,\ldots,N\},
\label{eq:chi_diag}
\end{align}
although this restriction is not required by the formulation.

To remove the logarithmic singularity we approximate the function $x\log x$ on $(0,1]$ by a degree-six
polynomial
\begin{align}
P(x)=\sum_{k=0}^{6} a_k x^k \approx x\log x,
\qquad x\in(0,1],
\label{eq:polyfit}
\end{align}
where the coefficients $\{a_k\}_{k=0}^6$ are obtained by a least-squares fit over a dense sampling of
$(0,1]$. The polynomial bulk energy used in the computations is then
\begin{align}
\widehat{\Psi}_0(\boldsymbol{\phi})
= \frac{\bar W}{\varepsilon_0}\left(\sum_{\alpha=1}^{N} P(\phi_\alpha)
- \sum_{\alpha=1}^{N}\chi_{\alpha}\phi_\alpha^2\right).
\label{eq:Psi0_poly}
\end{align}
Replacing $x\log x$ by $P(x)$ yields a smooth bulk potential on $[0,1]$, but does not by itself enforce that
the pure phases are minima. We therefore calibrate the interaction parameter(s) using the binary case so
that the pure states are stationary points and, in fact, local minima of the polynomial bulk potential.

On the binary manifold $(\phi,1-\phi)$, with the remaining phases set to zero,~\eqref{eq:Psi0_poly}
reduces to
\begin{align}
\widehat{\Psi}_0(\phi,1-\phi)
=
\frac{\bar W}{\varepsilon_0}\Big(
P(\phi)+P(1-\phi)-\chi\big(\phi^2+(1-\phi)^2\big)
\Big),
\label{eq:binary_poly}
\end{align}
where $\chi$ denotes the effective binary interaction parameter for this calibration. Imposing stationarity
at the pure states,
$\mathrm{d}\widehat{\Psi}_0(\phi,1-\phi)/\mathrm{d}\phi=0$ at $\phi=0$ (and by symmetry at $\phi=1$), yields
\begin{align}
\chi
=
\frac12\Big(2a_2+3a_3+4a_4+5a_5+6a_6\Big)
=
2.0753225165.
\label{eq:chi_calib}
\end{align}
With~\eqref{eq:chi_calib}, the pure phases $\phi\in\{0,1\}$ are stationary points of
$\widehat{\Psi}_0(\phi,1-\phi)$. Moreover, for the calibrated value~\eqref{eq:chi_calib} one has
$\mathrm{d}^2\widehat{\Psi}_0/\mathrm{d}\phi^2\big|_{\phi=0,1} > 0$, so that $\phi=0$ and $\phi=1$ are local
minima of the polynomial bulk potential. In the diagonal choice~\eqref{eq:chi_diag}, we take
$\chi_\alpha=\chi$ for all phases in the present implementation.

To match the barrier height of the standard Ginzburg--Landau (GL) double well, we choose the bulk prefactor
$\bar{W}$ such that the barrier of the polynomial surrogate equals $\max_{\varphi \in [-1,1]}W_\mathrm{GL}(\varphi)=1/4$,
where $\bar{W}_\mathrm{GL}(\varphi)=(1-\varphi^2)^2/4$ for $\varphi\in[-1,1]$. Using the affine mapping
$\varphi=2\phi-1$ on the binary manifold, the GL potential becomes $\bar{W}_\mathrm{GL}(\phi)=4\phi^2(1-\phi)^2$,
which attains its maximum $1/4$ at $\phi=1/2$.

Let $\widehat{\psi}(\phi)$ denote the binary polynomial bulk energy with unit prefactor ($\bar{W}=1$),
\begin{align}
\widehat{\psi}(\phi)
=
P(\phi)+P(1-\phi)-\chi\big(\phi^2+(1-\phi)^2\big).
\label{eq:psi_hat_unit}
\end{align}
We then set
\begin{align}
\bar{W}
=
\frac{1/4}{\widehat{\psi}(1/2)-\widehat{\psi}(1)}
=
0.6923807595.
\label{eq:W_calib}
\end{align}
With~\eqref{eq:chi_calib}--\eqref{eq:W_calib}, the surrogate~\eqref{eq:Psi0_poly} is smooth on $[0,1]$,
renders the pure phases as local minima, and reproduces the GL barrier height by construction. The
resulting partial bulk-energy contributions are shown in Figure~\ref{fig: calibration partial free energies}, and a comparison
between GL, classical FH, and the polynomial surrogate is provided in Figure~\ref{fig: calibration free energies}.

\begin{figure}
\captionsetup[subfigure]{justification=centering}
\begin{subfigure}{0.48\textwidth}
\centering
\includegraphics[scale=0.41]{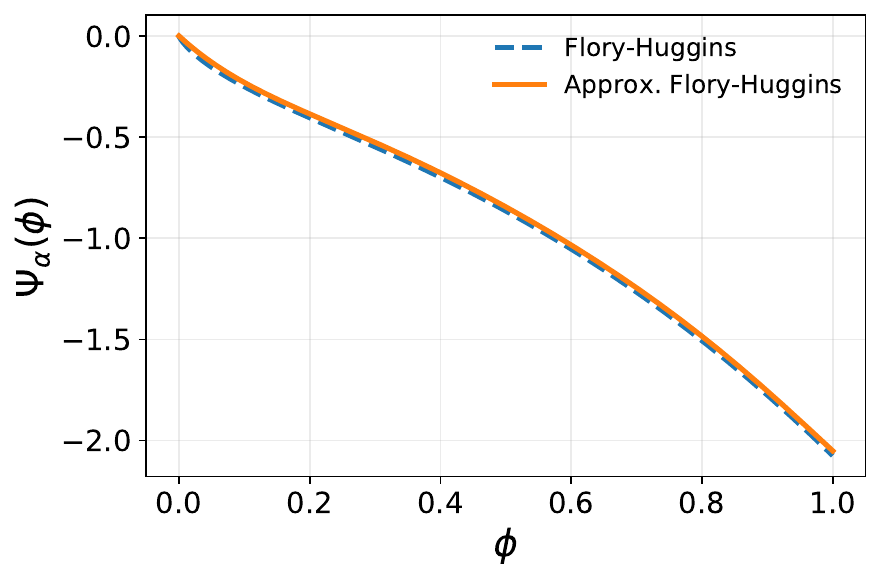}
\caption{Partial free energies}
\label{fig: calibration partial free energies}
\end{subfigure}
\begin{subfigure}{0.49\textwidth}
\centering
\includegraphics[scale=0.41]{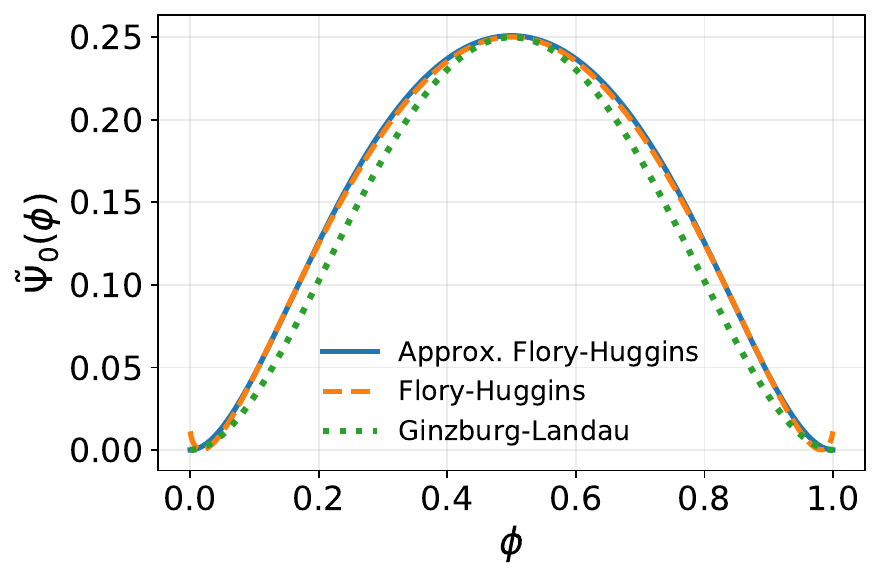}
\caption{Free energies (shifted so that minima are zero)}
\label{fig: calibration free energies}
\end{subfigure}
\caption{Calibration of the free energy.}
\label{fig: calibration free energy}
\end{figure}

\subsection{Mobility tensor}

In \cref{sec:mobility_derivation} we derived a reduction-consistent mobility class and, for practical use, we restrict here to an isotropic mobility
tensor of the form
\begin{align}\label{eq:isotropic_M_def}
  \mathbf{M}_{\mA\mB}(\boldsymbol{\phi}) = M_{\mA\mB}(\boldsymbol{\phi})\mathbf{I},
  \qquad \mA,\mB=1,\ldots,N,
\end{align}
with a scalar mobility matrix $M_{\mA\mB}$ that is symmetric, positive semidefinite, and satisfies
$M(\boldsymbol{\phi})\mathbf{1}=\mathbf{0}$.

Using the constraint compatibility $M\mathbf{1}=\mathbf{0}$, we have
$M_{\mA\mA}=-\sum_{\mB\neq \mA}M_{\mA\mB}$ and therefore, for each $\mA$,
\begin{align}\label{eq:exchange_rewrite}
  \bJ_\mA
  = -\sum_{\mB\neq \mA} M_{\mA\mB}\nabla (g_\mB-g_\mA)%  = \sum_{\mB\neq \mA} M_{\mA\mB}\nabla(g_\mA-g_\mB).
\end{align}
Thus diffusion is driven only by differences of generalized potentials. 

We choose the mobility as:
\begin{align}\label{eq: mob choice}
    M_{\mA\mB} = -4 m_0 \varepsilon_0 (\rho_\mA^{-1}+\rho_\mB^{-1})^{-2} \phi_\mA \phi_\mB,
\end{align}
with $m_0$ a constant for which we choose $m_0 = 10^{-3}$. This provides the standard mobility form in the two-phase case. 
We provide the details in \cref{sec: connection with binary model}. In numerical computations we require the matrix $M_{\mA\mB}$ to remain negative semi-definite, and therefore we choose the regularization:
\begin{align}
    M_{\mA\mB} = -4 m_0 \varepsilon_0 (\rho_\mA^{-1}+\rho_\mB^{-1})^{-2} (\phi_\mA)_+ (\phi_\mB)_+,
\end{align}
and $M_{\mA\mA} = -\sum_{\mB\neq\mA} M_{\mA\mB}$ which complies with the compatibility constraint $M\mathbf{1}=\mathbf{0}$.

\begin{remark}[Maxwell--Stefan mobility]
   Selecting $M_{\mA\mB}=m(\rho)\rho_\mA\rho_\mB/\rho$ yields the Maxwell--Stefan matrix
\begin{align}\label{eq:MS_offdiag_simplified}
  M_{\mA\mB}(\boldsymbol{\phi})
  = m(\rho)\Big(\delta_{\mA\mB}\tilde\rho_\mA-\frac{\tilde\rho_\mA\tilde\rho_\mB}{\rho}\Big),
\end{align}
so that for $\mA\neq \mB$,
\begin{align}\label{eq:MS_offdiag}
  M_{\mA\mB} = -m(\rho)\frac{\tilde\rho_\mA\tilde\rho_\mB}{\rho}
  = -m(\rho)\frac{\rho_\mA\rho_\mB}{\rho}\phi_\mA\phi_\mB.
\end{align}
Substituting \eqref{eq:MS_offdiag} into the form \eqref{eq:exchange_rewrite} shows explicitly that each pairwise
exchange contribution is proportional to $\phi_\mA\phi_\mB$, and hence vanishes when either phase is absent.  
\end{remark}

\begin{remark}[Elliott--Garcke mobility]
The mobility structure used in \cite{elliott1997diffusional} can be written in the form
\begin{align}
L_{\alpha\beta}
=
l_\alpha(\phi_\alpha)
\left(
\delta_{\alpha\beta}
-
\frac{l_\beta(\phi_\beta)}{\sum_{\gamma} l_\gamma(\phi_\gamma)}
\right).
\end{align}
Choosing $
l_\alpha(\phi_\alpha)=m(\rho)\tilde\rho_\alpha
= m(\rho)\rho_\alpha\phi_\alpha$, 
provides
\begin{align}
L_{\alpha\beta}
=
m(\rho)\tilde\rho_\alpha
\left(
\delta_{\alpha\beta}
-
\frac{\tilde\rho_\beta}{\rho}
\right)
=
m(\rho)
\left(
\delta_{\alpha\beta}\tilde\rho_\alpha
-
\frac{\tilde\rho_\alpha\tilde\rho_\beta}{\rho}
\right).
\end{align}
This coincides exactly with the Maxwell--Stefan mobility matrix \eqref{eq:MS_offdiag_simplified}.
\end{remark}

\section{Numerical computations}\label{sec: Numerical results}

In this section we provide supporting numerical computations. First, in \cref{subsec: num red} we provide numerical evidence for the mixture-aware properties established in \cref{sec: mixture-aware closure}. Then in \cref{sec:results:rising_bubble} and \cref{subsec: Quaternary rising bubble simulation} we demonstrate the capabilities of the framework for $N=3$ and $N=4$.

All computations are performed with the symmetric structure-preserving finite element method proposed in \cite{structureNphase2026}. In addition, we choose the capillary tensor to match the prescribed surface tensions, see \cite{surfacetension2026} for details. We report the calibrated capillarity matrices, interface widths, and scale factors in Appendix \ref{sec appendix parameters}.

\subsection{Mixture-aware properties}\label{subsec: num red}

Here, we numerically demonstrate the mixture-aware properties of the framework. For this purpose we perform ternary ($N=3$) simulations of an effectively binary ($N=2$) problem. We choose the standard rising bubble benchmark problem of \cite{hysing2009quantitative}. In this problem, a circular bubble of phase 2 (the lighter phase) with an initial diameter of $D_0 = 2R_0 = 0.5$ is positioned at $(0.5, 0.5)$ within a rectangular domain $[0,1] \times [0,2]$, surrounded by phase 1 (the heavier phase). The boundary conditions are no-penetration ($\vv \cdot \mathbf{n} = 0$) on the vertical boundaries (left and right) and no-slip ($\vv = 0$) on the horizontal boundaries (top and bottom). We visualize a schematic representation of the problem in \cref{fig:sketch 2D rising bubble problem}. The benchmark problem consists of two cases, each defined by different parameter values, as listed in \cref{table: parameters 2D RB cases}.

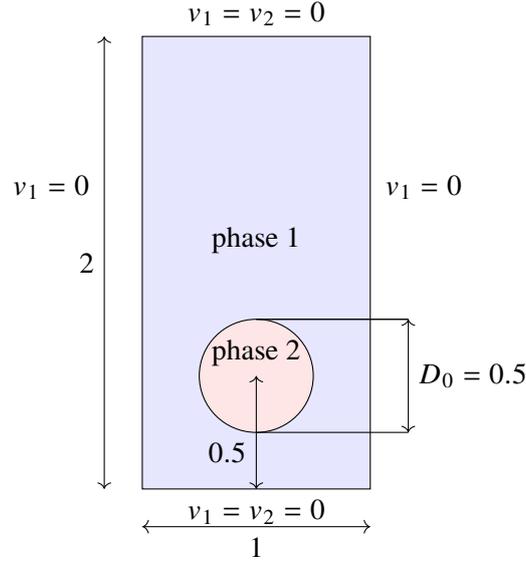
\begin{figure}[h]
\begin{center}
\begin{tikzpicture}
    % Draw the gray background rectangle
    \fill[fill=blue!10,draw=black] (0, 0) rectangle (3, 6);
    
    % Draw the white circle with a black border
    \draw[fill=red!10,draw=black] (1.5, 1.5) circle (0.75);
    
    % Add arrows to indicate dimensions
    \draw[<->] (0, -0.5) -- (3, -0.5) node[midway,below] {$1$};
    \draw[<->] (-0.5, 0) -- (-0.5, 6) node[midway,left] {$2$};
    \draw[<->] (3.5, 0.75) -- (3.5, 2.25) node[midway,right] {$D_0 = 0.5$};
    \draw[-] (1.5, 0.75) -- (3.5, 0.75);
    \draw[-] (1.5, 2.25) -- (3.5, 2.25);
    \draw[<->] (1.5, 0.0) -- (1.5, 1.5) node[midway,below left] {$0.5$};
    \draw[-] (1.5, 2.25) -- (3.5, 2.25);
    \node at (1.5, 6.3) {$v_1 = v_2 = 0$}; 
    \node at (1.5, -0.3) {$v_1 = v_2 = 0$};
    \node at (-1.2, 4.0) {$v_1 = 0$}; 
    \node at (3.7, 4.0) {$v_1 = 0$}; 
    \node at (1.5, 1.8) {phase 2};
    \node at (1.5, 3.3) {phase 1}; 
\end{tikzpicture}
    \caption{Schematic representation of the \cite{hysing2009quantitative} rising bubble problem}
    \label{fig:sketch 2D rising bubble problem}
\end{center}
\end{figure}

\begin{table}
\centering
\begin{tabularx}{\textwidth}{XXXXXXX}
%\hline\\[-6pt]
Case & \hspace{0.1cm} $\rho_1$ & \hspace{0.1cm} $\rho_2$ & $\nu_1$ & $\nu_2$ & \hspace{0.1cm} $\gamma_{12}$ & \hspace{0.1cm} $g$  \\[4pt]
\hline\\[-6pt]
\hspace{0.5cm}1 & $1000$ & $100$ & $10$ & $1$   & $24.5$ & $0.98$    \\[6pt]
\hspace{0.5cm}2 & $1000$ & \hspace{0.1cm} $1$   & \hspace{0.1cm} $1$  & $0.1$ &$1.96$  & $0.98$  \\[6pt]
\hline
\end{tabularx}
\caption{Parameters for the two-dimensional rising bubble cases.}
\label{table: parameters 2D RB cases}
\end{table}
We conduct simulations on a uniform rectangular mesh with element size $h = 1/128 $. The time step size is $\Delta t_n = 0.128 h $ and the interface width parameter is $\varepsilon = h$.

We consider three different scenarios.

\subsubsection*{Absent phase}
Here we test the scenario of an absent phase. We choose the initial phase fields as:
\begin{subequations}
  \begin{align}
    \phi^h_{1,0}(\mathbf{x}) =&~ \frac{1}{2}\left(1+\tanh{\dfrac{\sqrt{(x-0.5)^2+(y-0.5)^2}-R_0}{\varepsilon\sqrt{2}}}\right),\\
    \phi^h_{2,0}(\mathbf{x}) =&~ \frac{1}{2}\left(1-\tanh{\dfrac{\sqrt{(x-0.5)^2+(y-0.5)^2}-R_0}{\varepsilon\sqrt{2}}}\right),\\
    \phi^h_{3,0}(\mathbf{x}) =&~ 0.
\end{align}
\end{subequations}
In Figure \ref{fig: reduction case 1 2 - phi3=0} we visualize the phase fields for Case 1 and 2. We observe that the absent phase $\phi_3$ remains absent. In addition, we have verified that the profiles of $\phi_1$ and $\phi_2$ are in perfect agreement with reference results, such as \cite{brunk2026simple}. These observations are in agreement with axioms A6 (Reduction for absent phase) and A7 (Absent phase remains absent).

\begin{figure}
\captionsetup[subfigure]{justification=centering}
\begin{subfigure}{0.078\textwidth}
\centering
\includegraphics[width=1\textwidth]{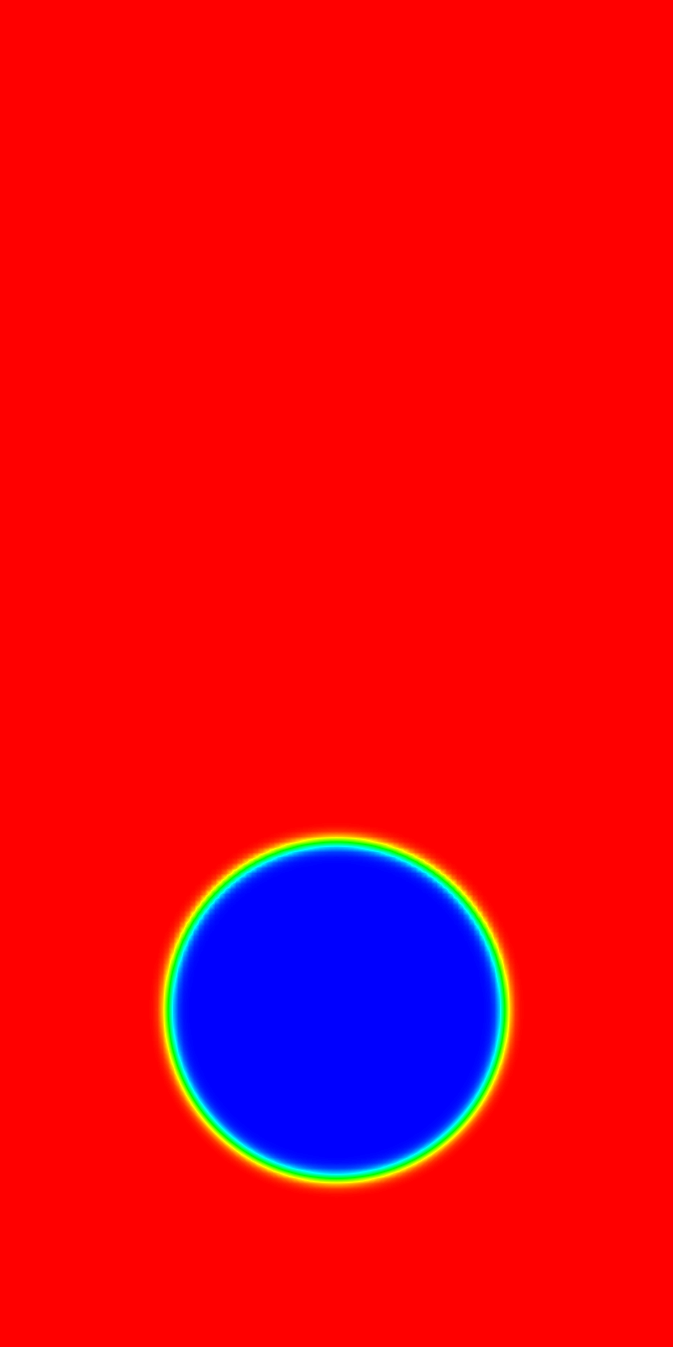}
%\caption{$t=0.0$}
\end{subfigure}
\begin{subfigure}{0.078\textwidth}
\centering
\includegraphics[width=1\textwidth]{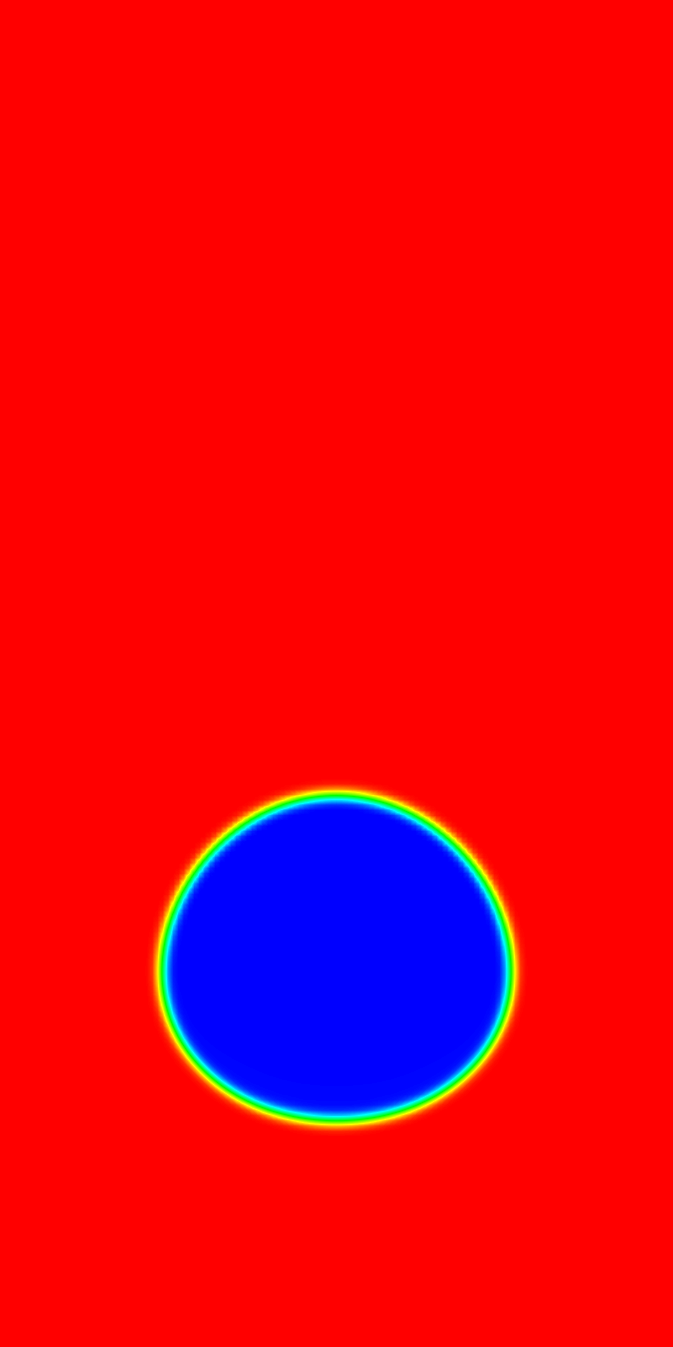}
%\caption{$t=0.6$}
\end{subfigure}
\begin{subfigure}{0.078\textwidth}
\centering
\includegraphics[width=1\textwidth]{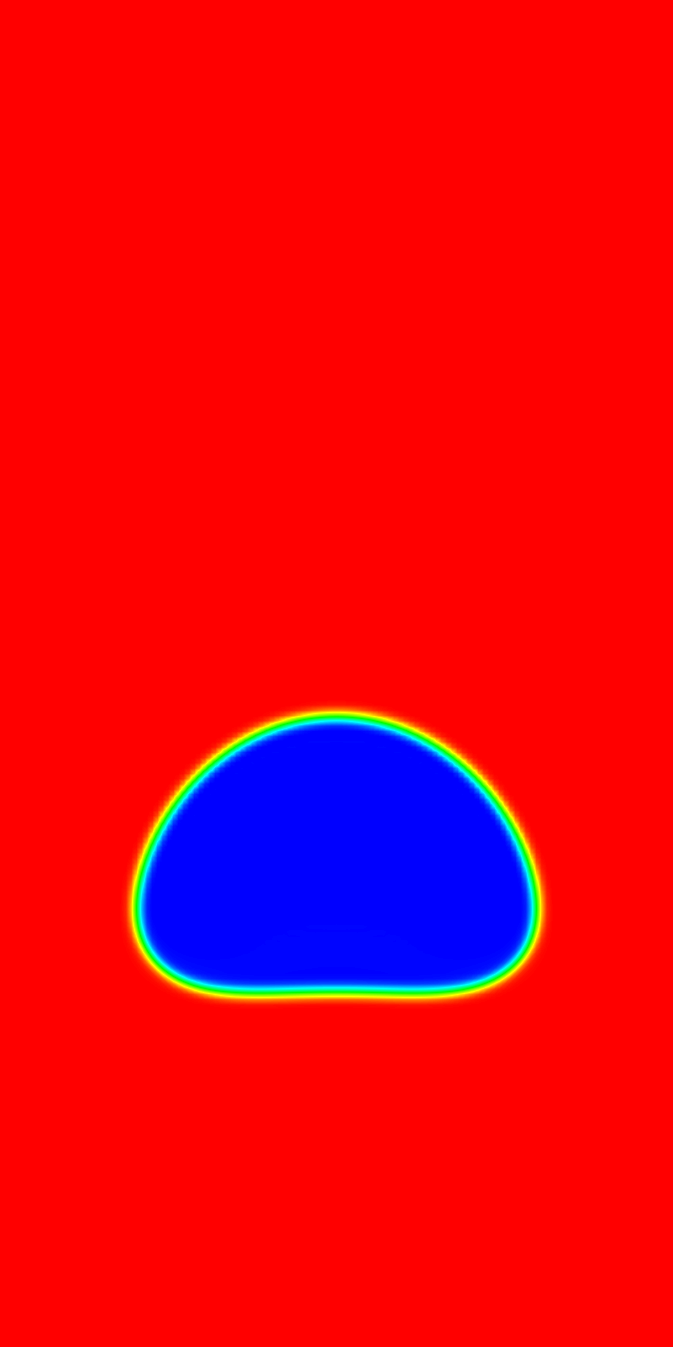}
%\caption{$t=1.2$}
\end{subfigure}
\begin{subfigure}{0.078\textwidth}
\centering
\includegraphics[width=1\textwidth]{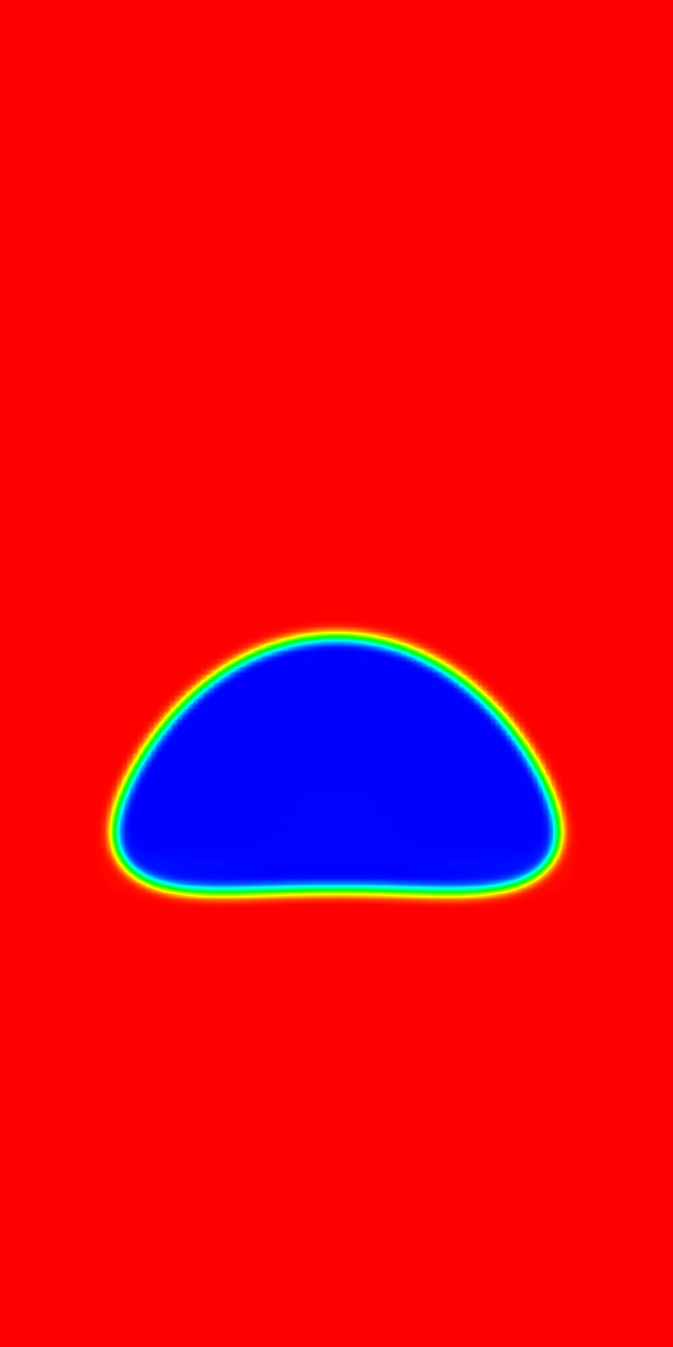}
%\caption{$t=1.8$}
\end{subfigure}
\begin{subfigure}{0.078\textwidth}
\centering
\includegraphics[width=1\textwidth]{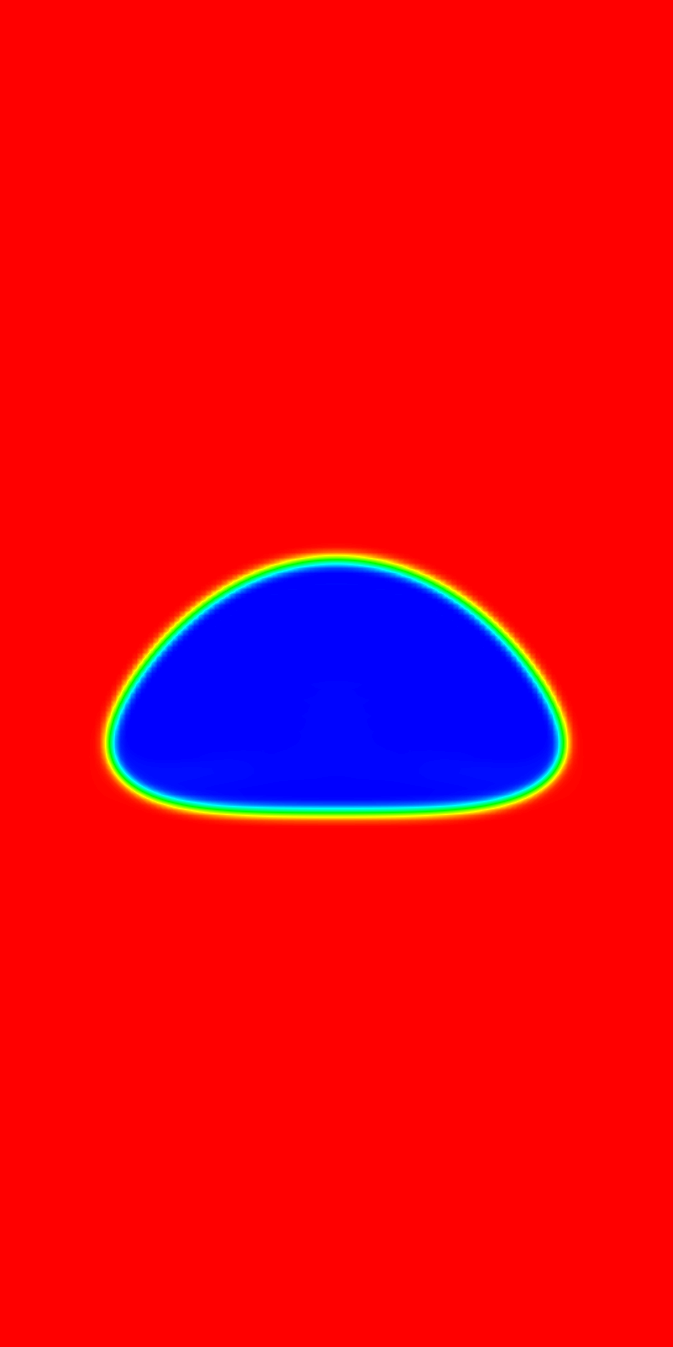}
%\caption{$t=2.4$}
\end{subfigure}
\begin{subfigure}{0.078\textwidth}
\centering
\includegraphics[width=1\textwidth]{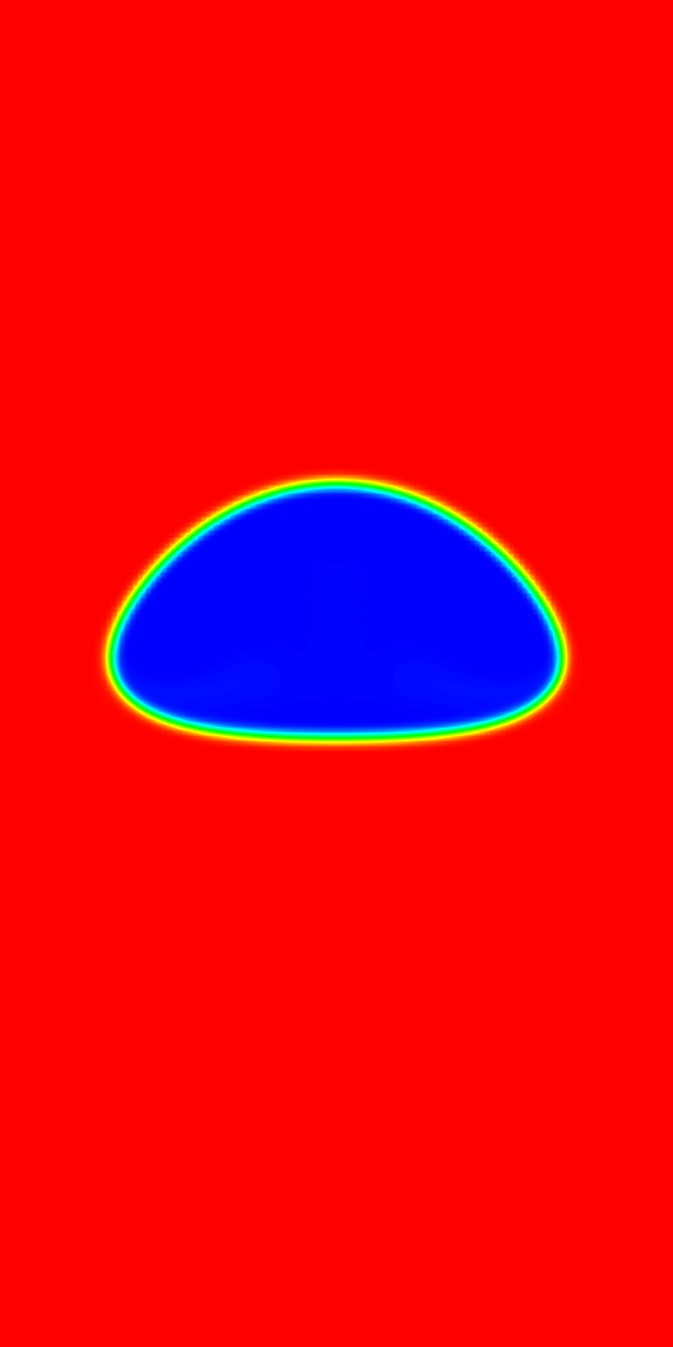}
%\caption{$t=3.0$}
\end{subfigure}
\begin{subfigure}{0.078\textwidth}
\centering
\includegraphics[width=1\textwidth]{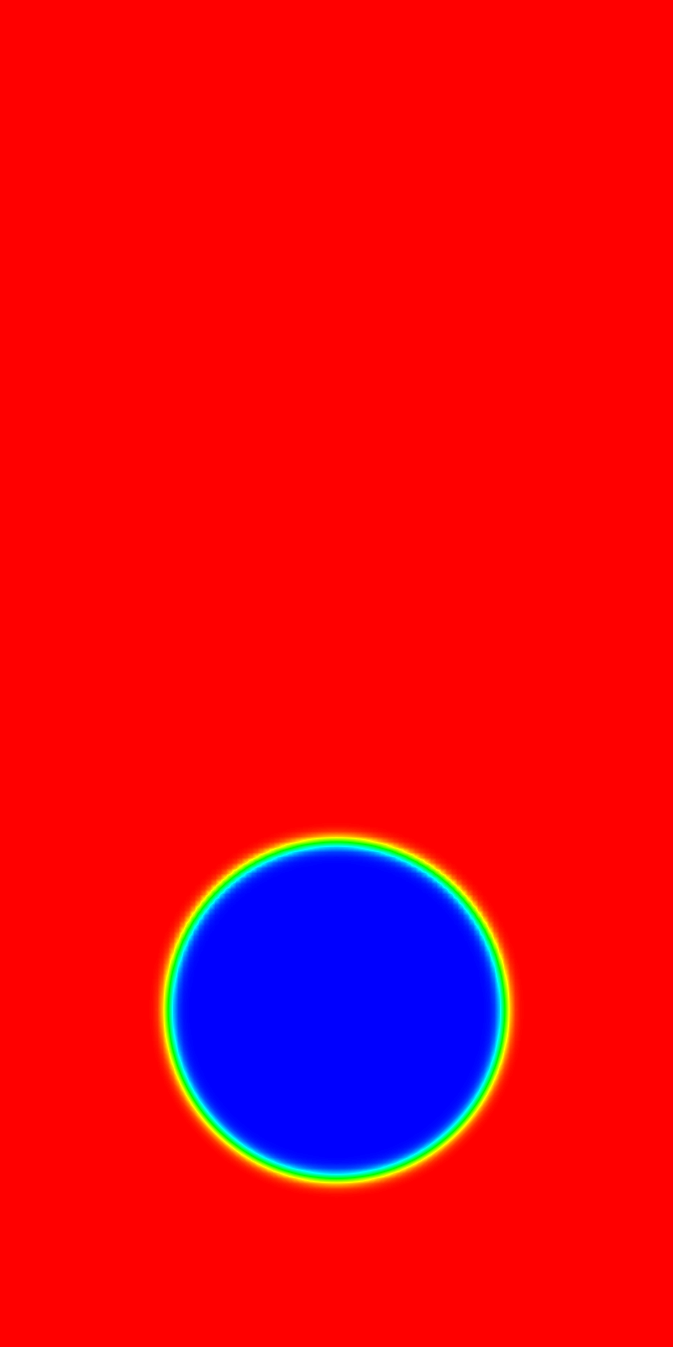}
%\caption{$t=0.0$}
\end{subfigure}
\begin{subfigure}{0.078\textwidth}
\centering
\includegraphics[width=1\textwidth]{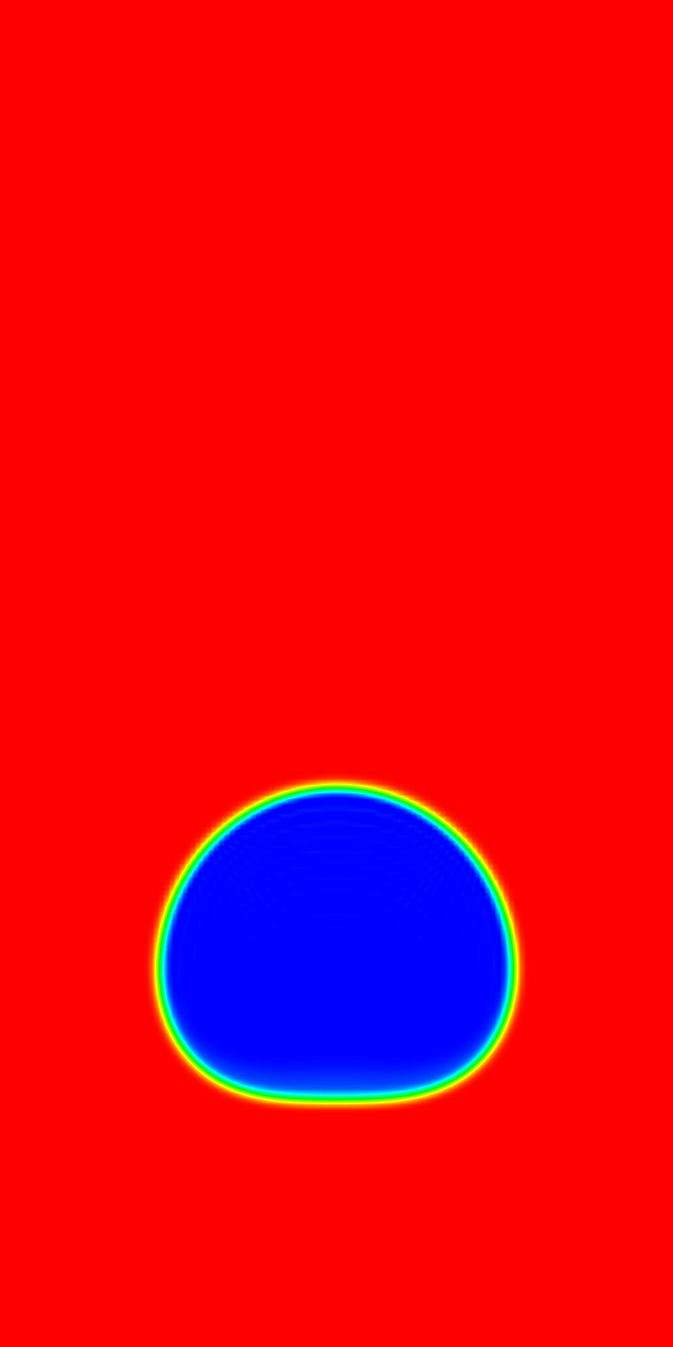}
%\caption{$t=0.6$}
\end{subfigure}
\begin{subfigure}{0.078\textwidth}
\centering
\includegraphics[width=1\textwidth]{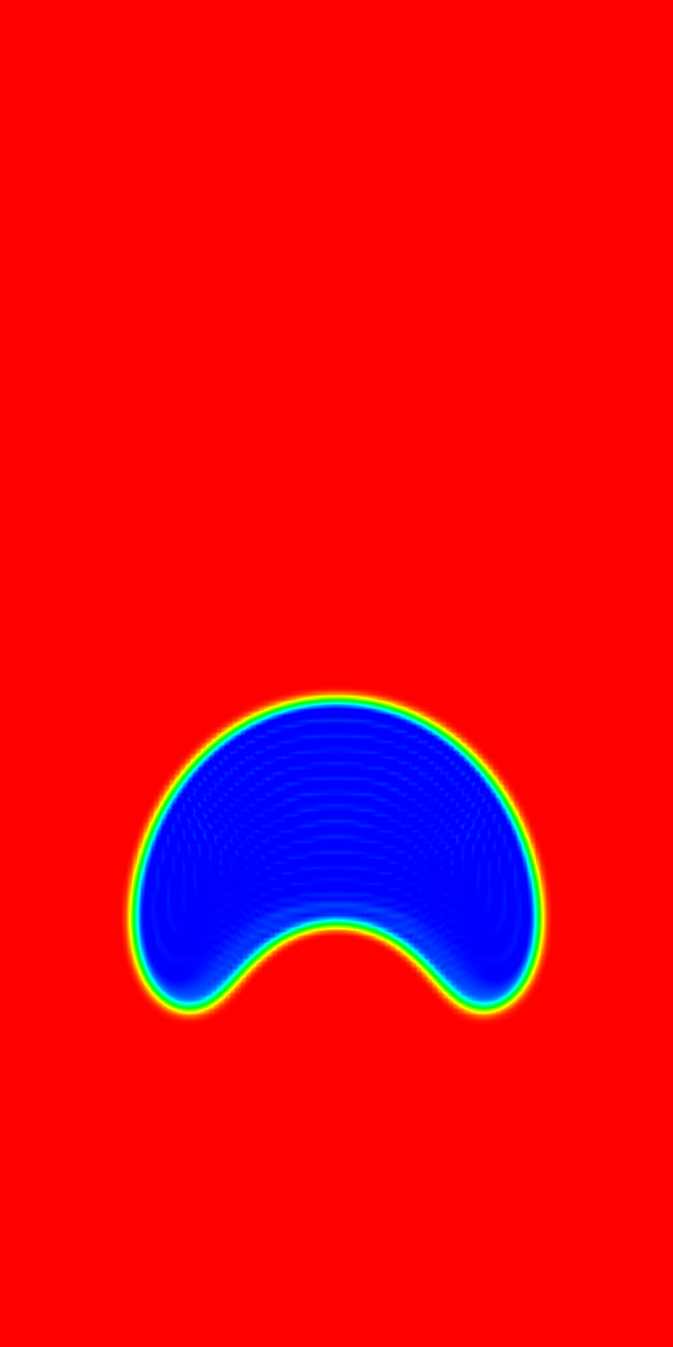}
%\caption{$t=1.2$}
\end{subfigure}
\begin{subfigure}{0.078\textwidth}
\centering
\includegraphics[width=1\textwidth]{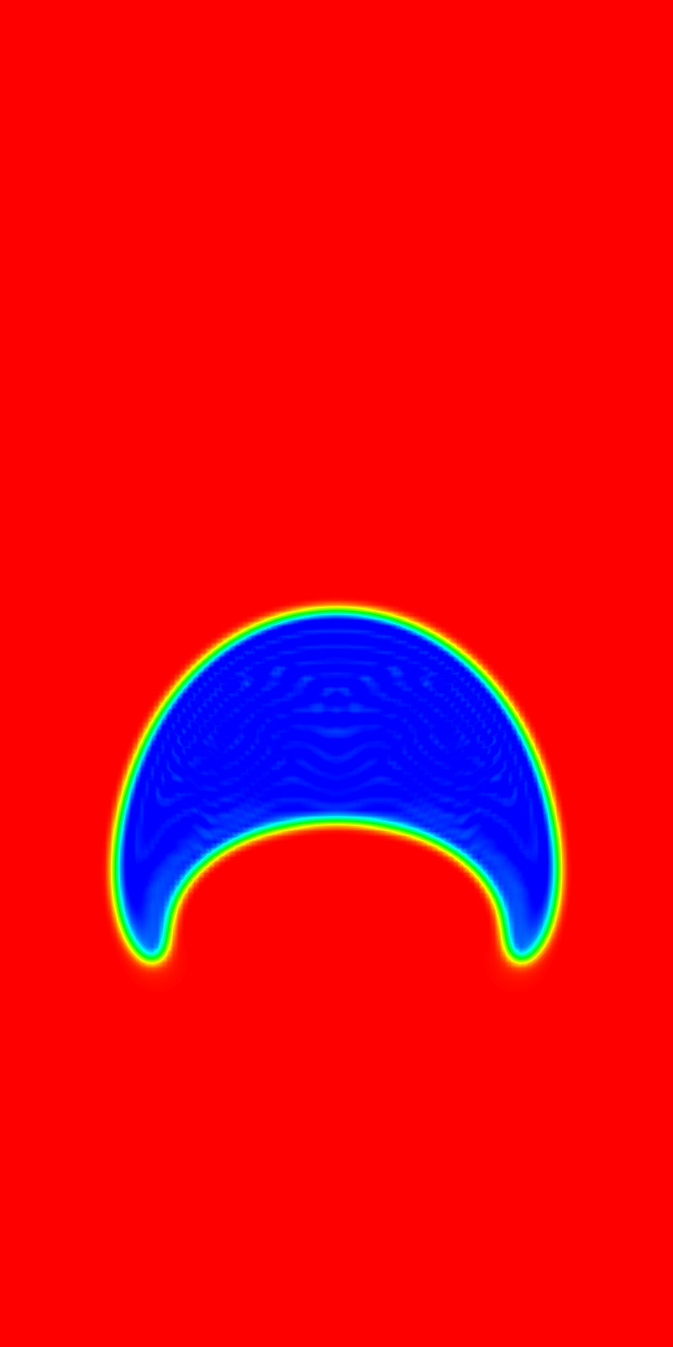}
%\caption{$t=1.8$}
\end{subfigure}
\begin{subfigure}{0.078\textwidth}
\centering
\includegraphics[width=1\textwidth]{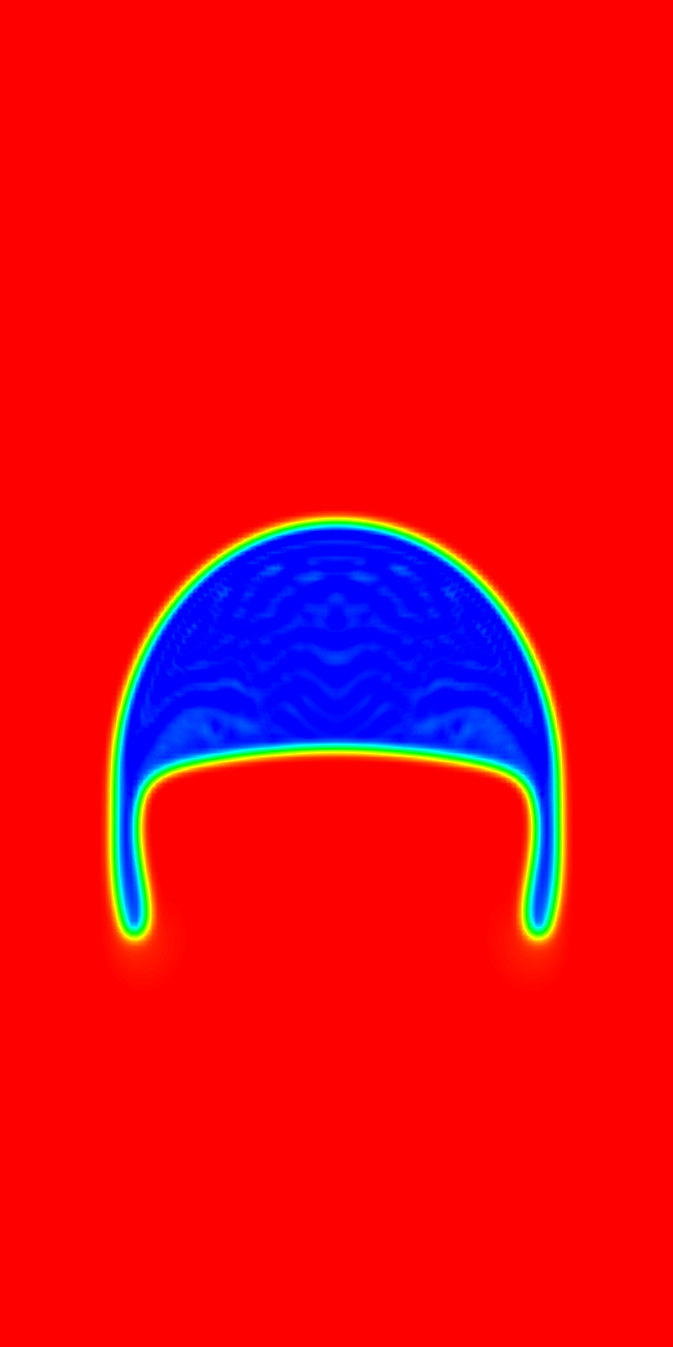}
%\caption{$t=2.4$}
\end{subfigure}
\begin{subfigure}{0.078\textwidth}
\centering
\includegraphics[width=1\textwidth]{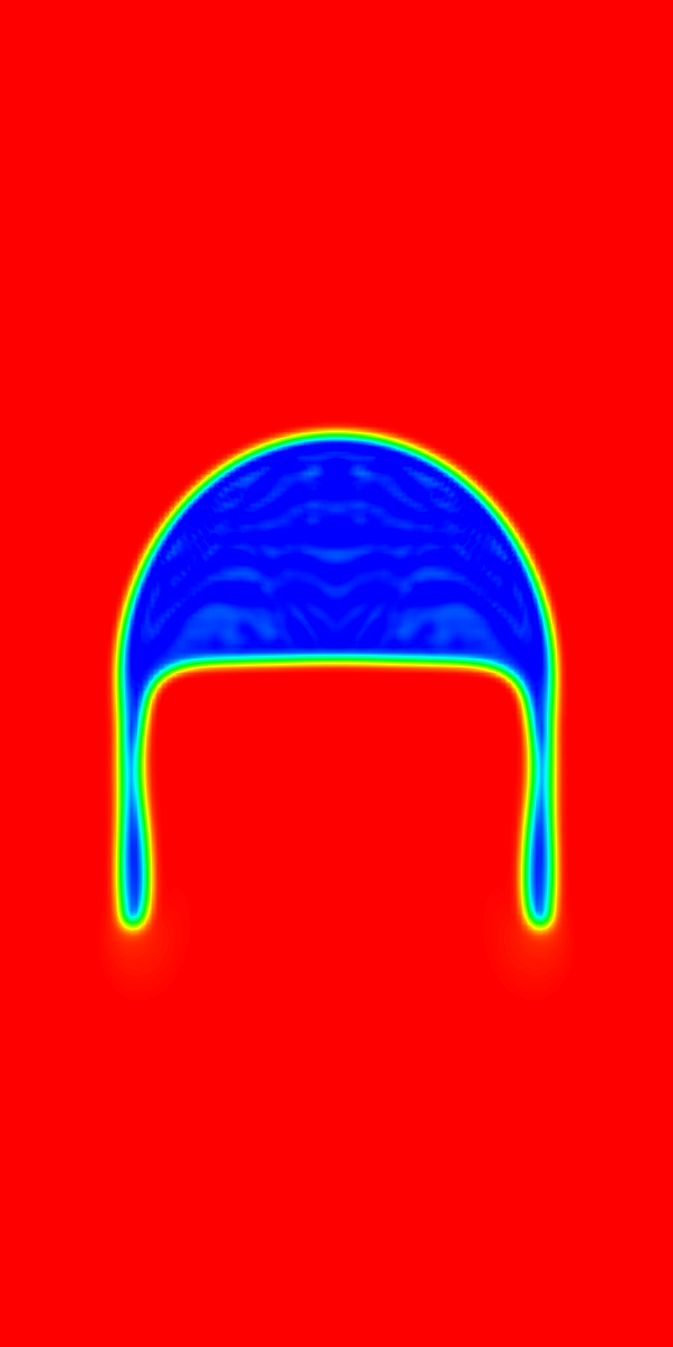}
%\caption{$t=3.0$}
\end{subfigure}
\begin{subfigure}{0.078\textwidth}
\centering
\includegraphics[width=1\textwidth]{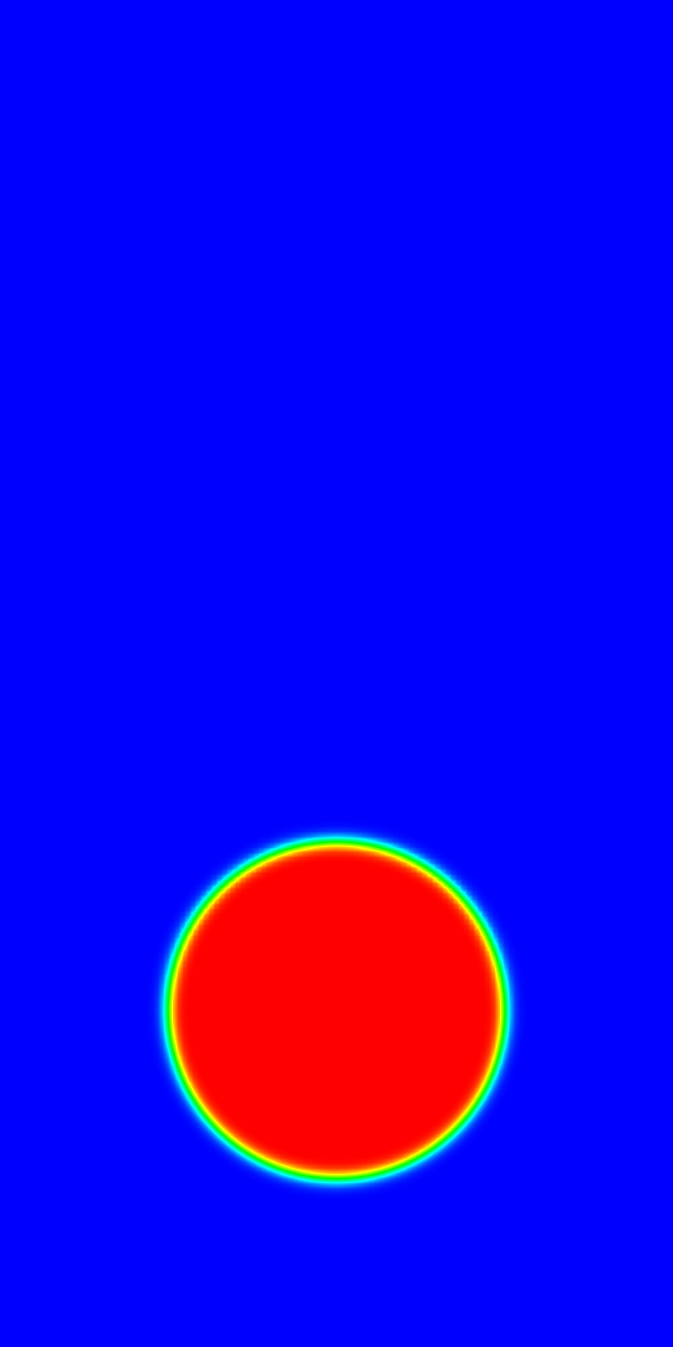}
%\caption{$t=0.0$}
\end{subfigure}
\begin{subfigure}{0.078\textwidth}
\centering
\includegraphics[width=1\textwidth]{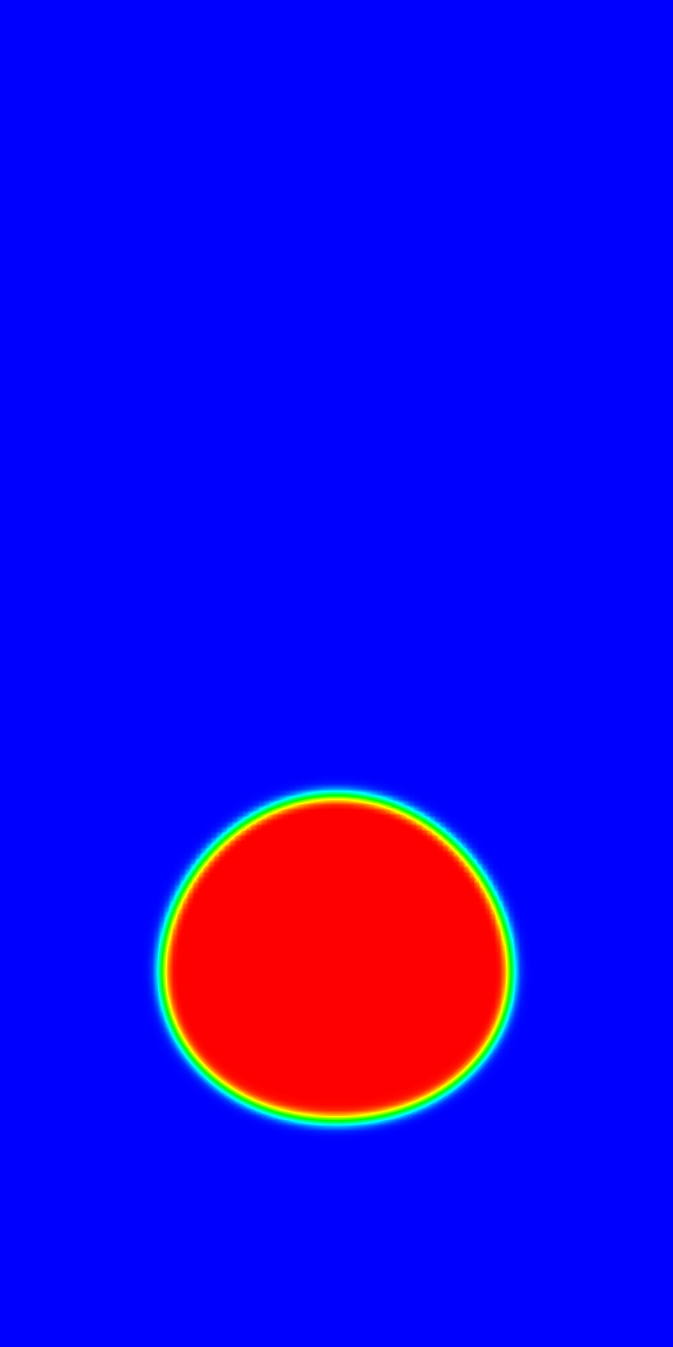}
%\caption{$t=0.6$}
\end{subfigure}
\begin{subfigure}{0.078\textwidth}
\centering
\includegraphics[width=1\textwidth]{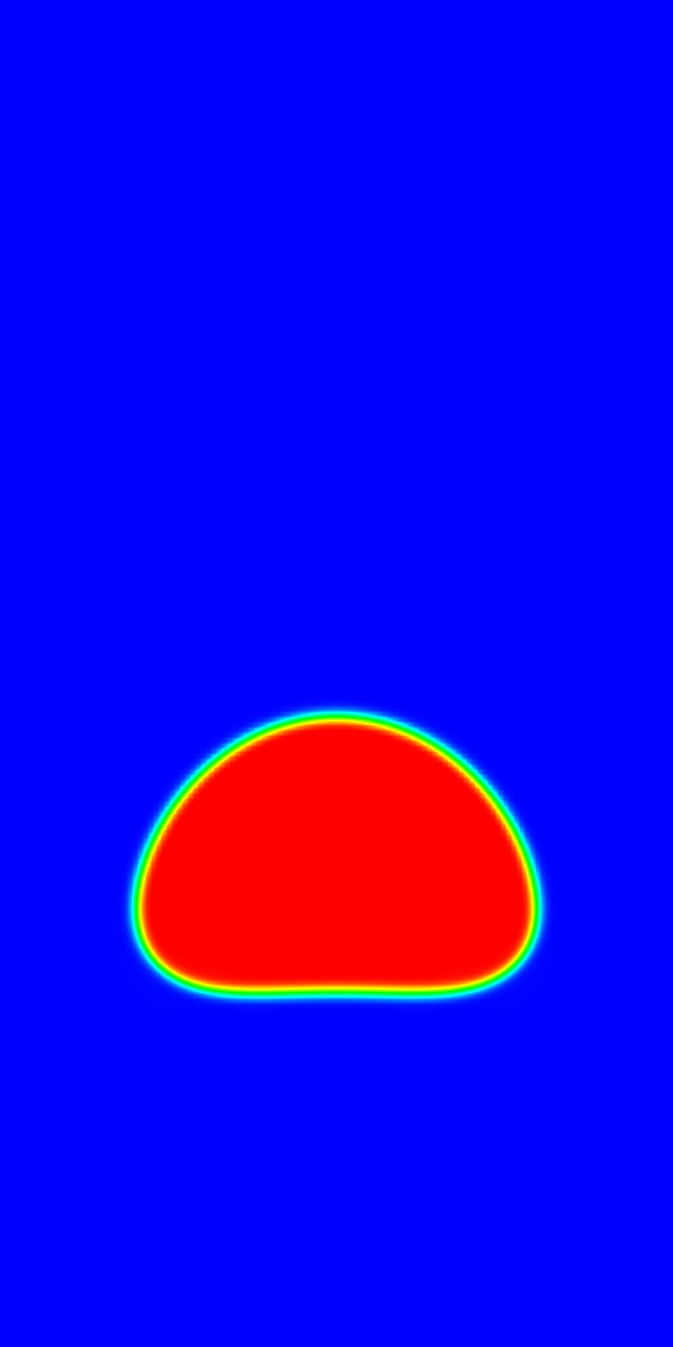}
%\caption{$t=1.2$}
\end{subfigure}
\begin{subfigure}{0.078\textwidth}
\centering
\includegraphics[width=1\textwidth]{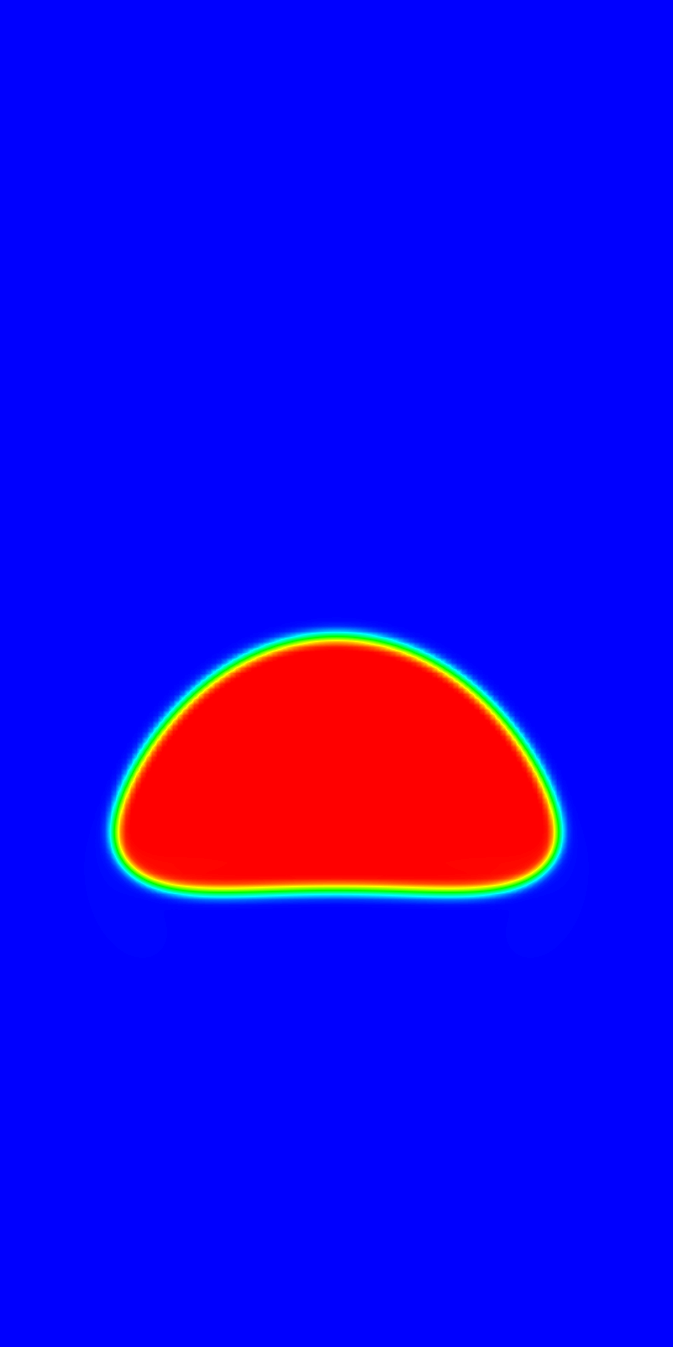}
%\caption{$t=1.8$}
\end{subfigure}
\begin{subfigure}{0.078\textwidth}
\centering
\includegraphics[width=1\textwidth]{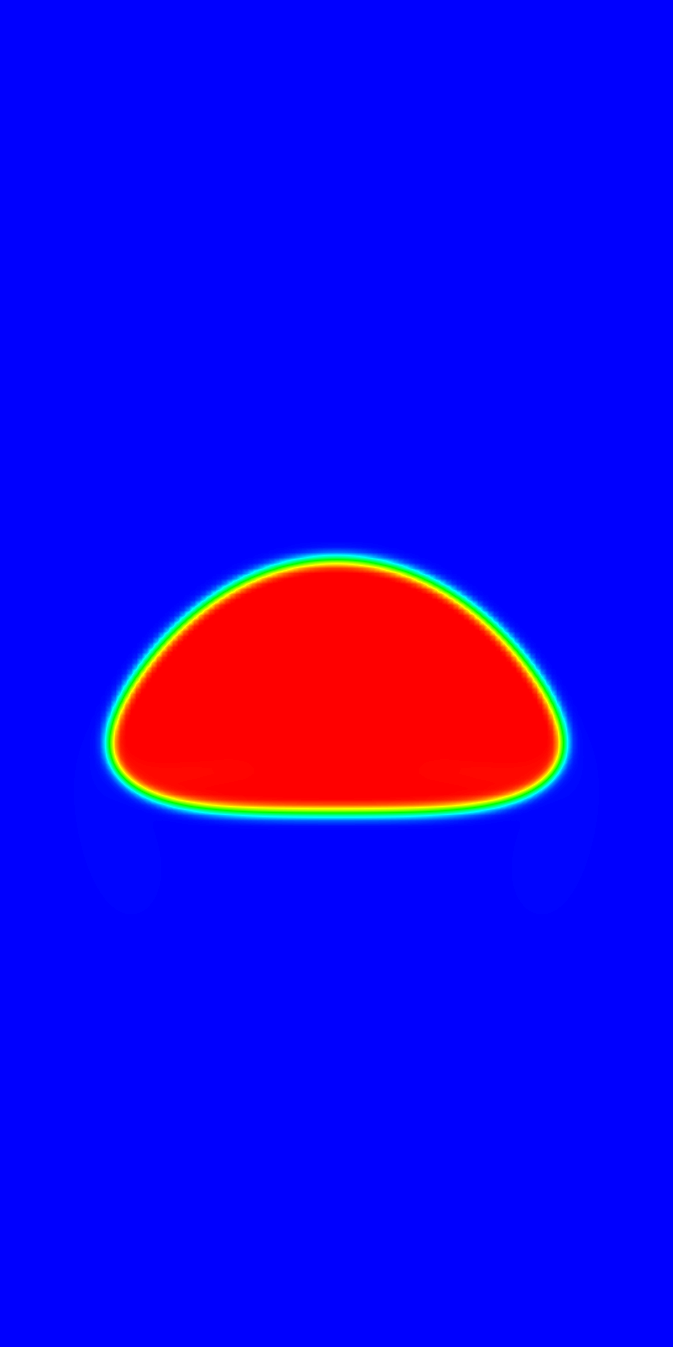}
%\caption{$t=2.4$}
\end{subfigure}
\begin{subfigure}{0.078\textwidth}
\centering
\includegraphics[width=1\textwidth]{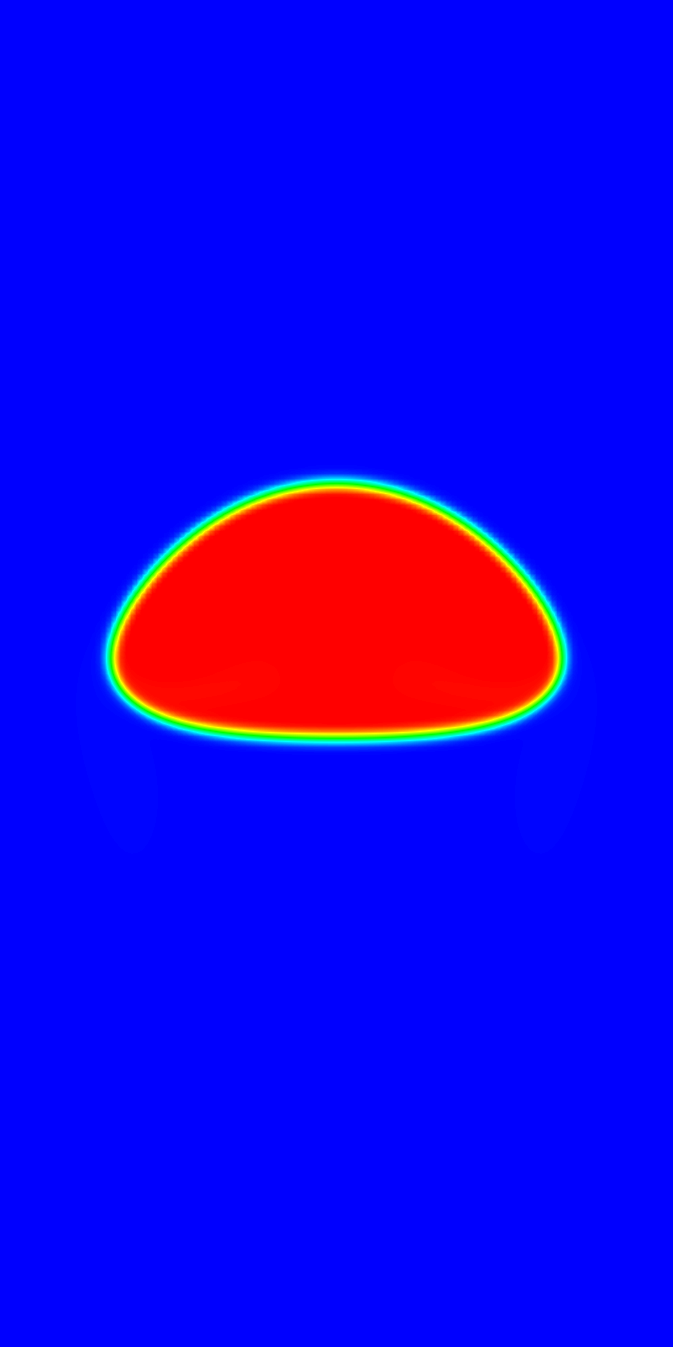}
%\caption{$t=3.0$}
\end{subfigure}
\begin{subfigure}{0.078\textwidth}
\centering
\includegraphics[width=1\textwidth]{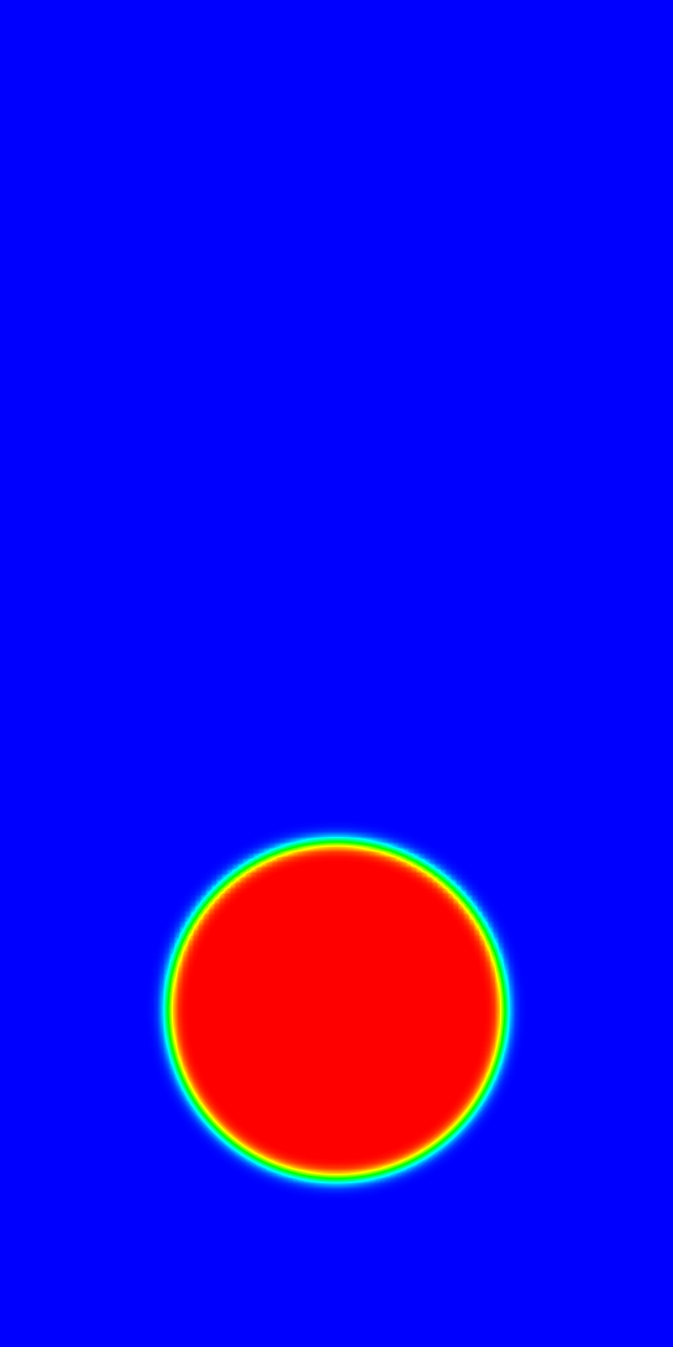}
%\caption{$t=0.0$}
\end{subfigure}
\begin{subfigure}{0.078\textwidth}
\centering
\includegraphics[width=1\textwidth]{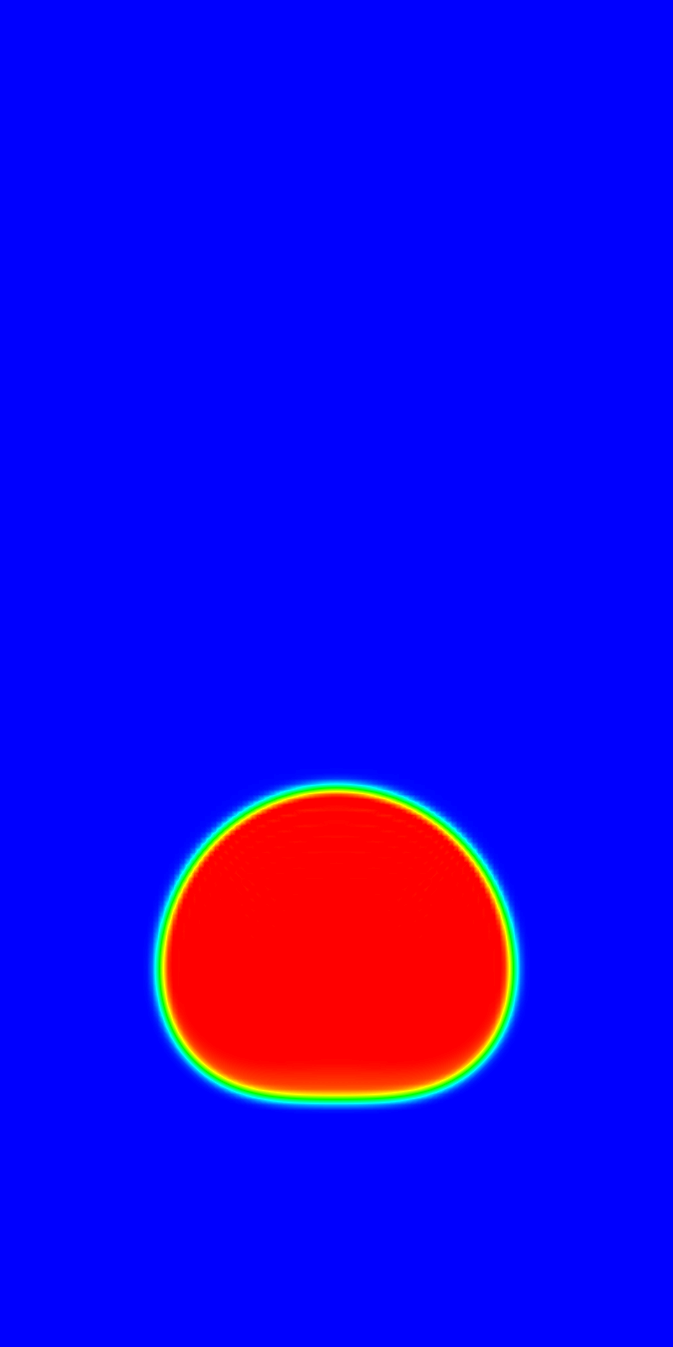}
%\caption{$t=0.6$}
\end{subfigure}
\begin{subfigure}{0.078\textwidth}
\centering
\includegraphics[width=1\textwidth]{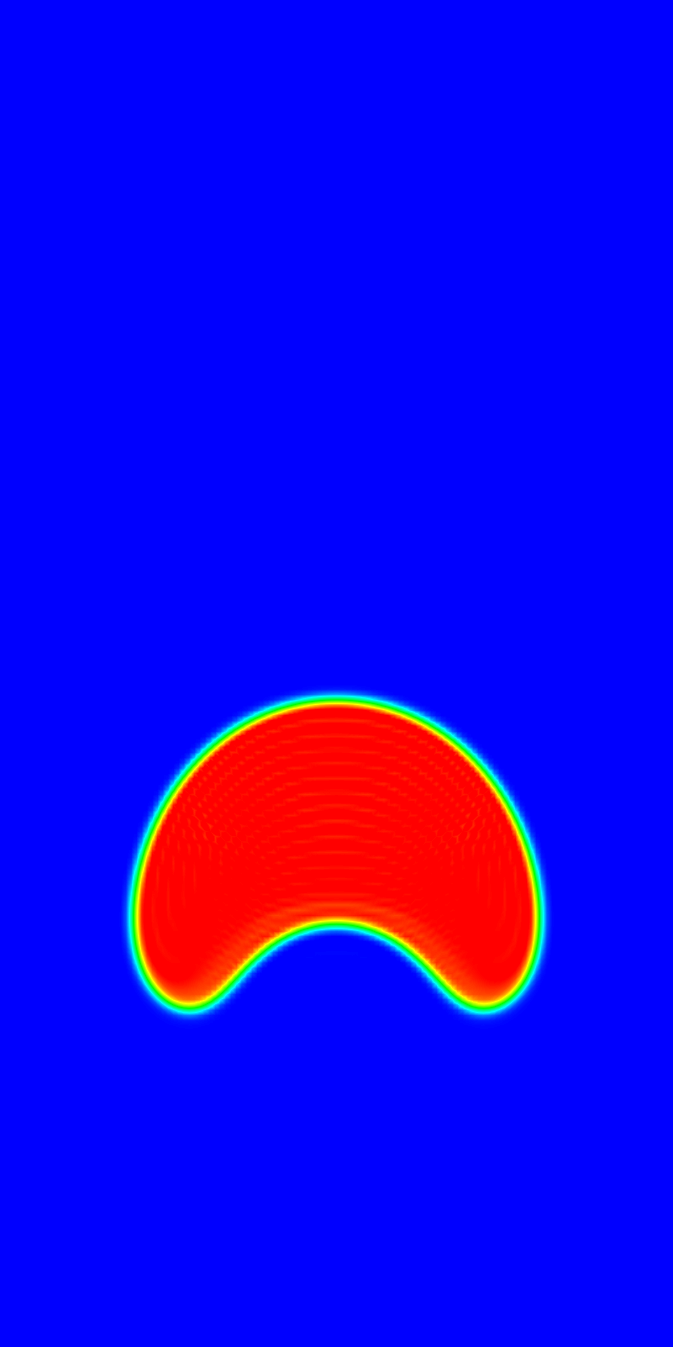}
%\caption{$t=1.2$}
\end{subfigure}
\begin{subfigure}{0.078\textwidth}
\centering
\includegraphics[width=1\textwidth]{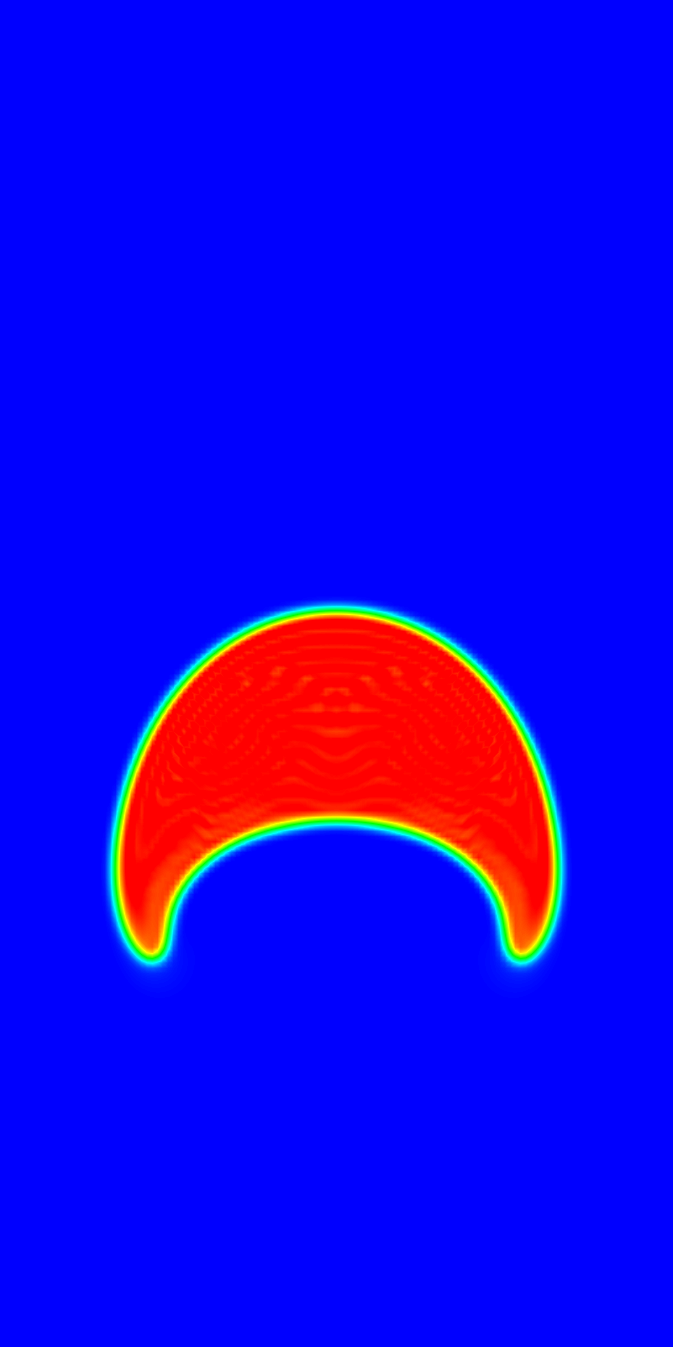}
%\caption{$t=1.8$}
\end{subfigure}
\begin{subfigure}{0.078\textwidth}
\centering
\includegraphics[width=1\textwidth]{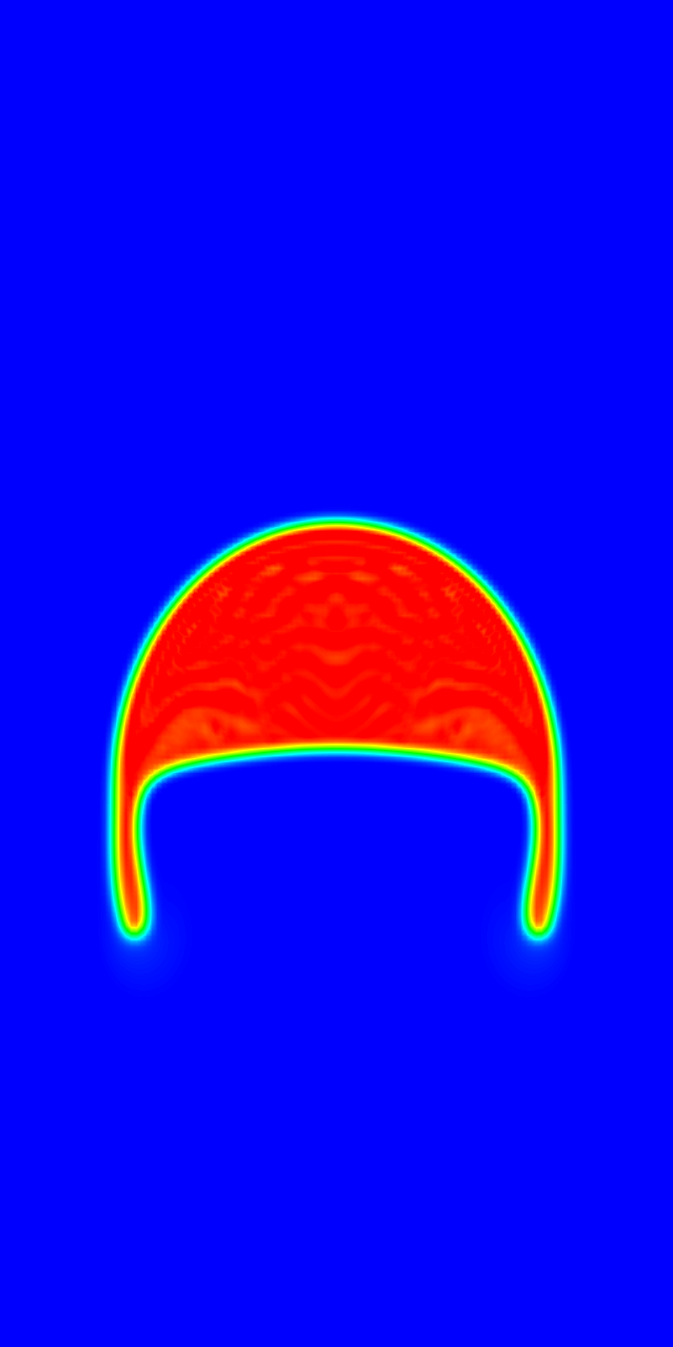}
%\caption{$t=2.4$}
\end{subfigure}
\begin{subfigure}{0.078\textwidth}
\centering
\includegraphics[width=1\textwidth]{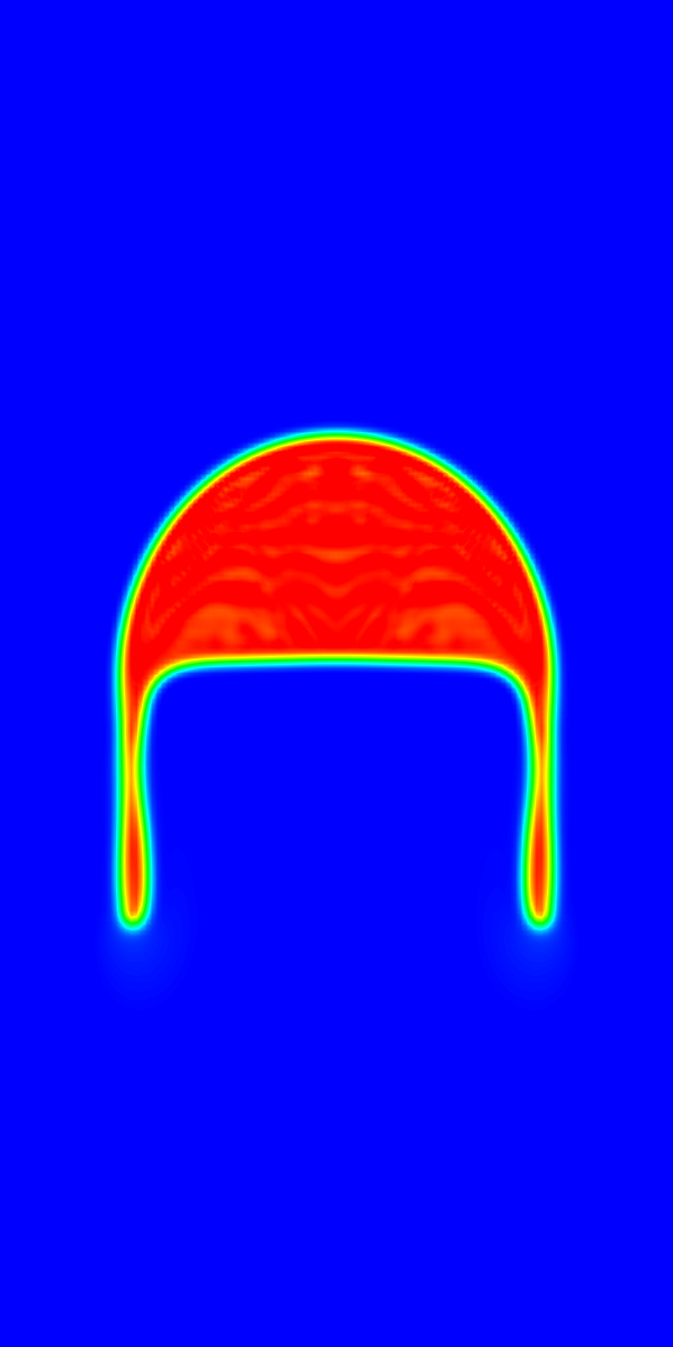}
%\caption{$t=3.0$}
\end{subfigure}
\begin{subfigure}{0.078\textwidth}
\centering
\includegraphics[width=1\textwidth]{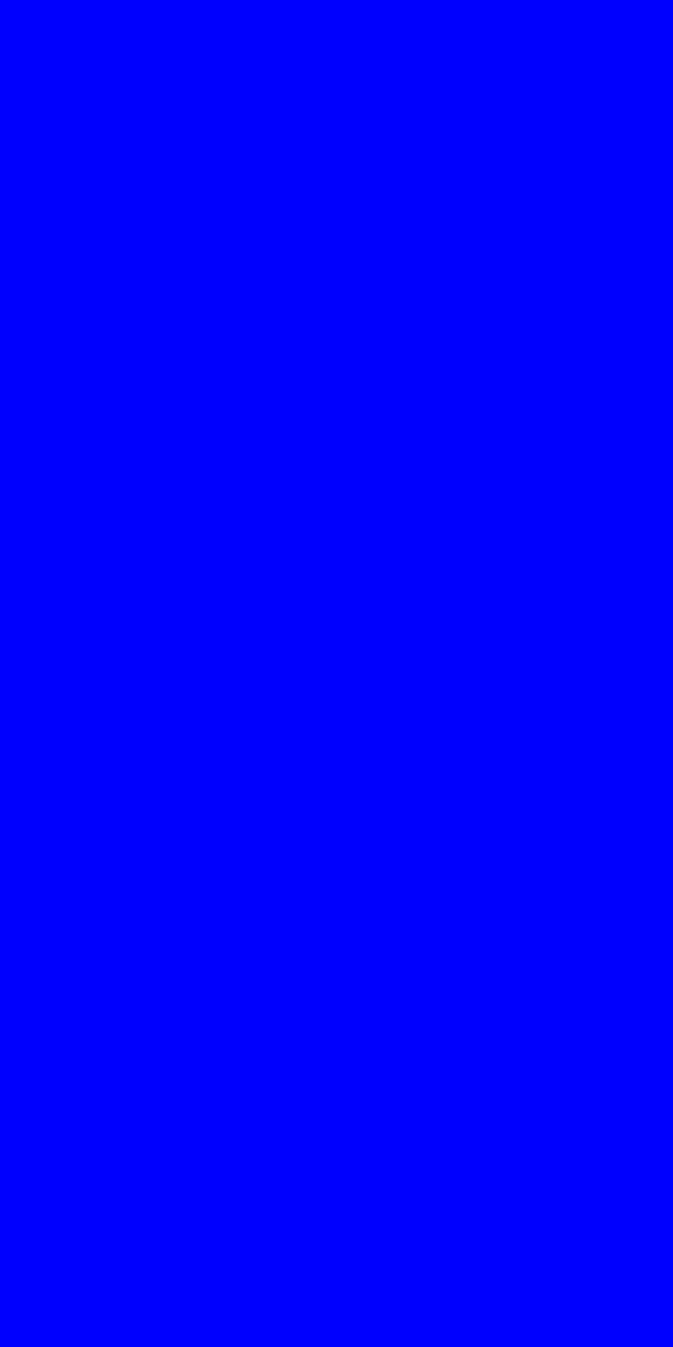}
%\caption{$t=0.0$}
\end{subfigure}
\begin{subfigure}{0.078\textwidth}
\centering
\includegraphics[width=1\textwidth]{figures/2D_N3_reduction/case1_absent/case1_absent_phi3.png}
%\caption{$t=0.6$}
\end{subfigure}
\begin{subfigure}{0.078\textwidth}
\centering
\includegraphics[width=1\textwidth]{figures/2D_N3_reduction/case1_absent/case1_absent_phi3.png}
%\caption{$t=1.2$}
\end{subfigure}
\begin{subfigure}{0.078\textwidth}
\centering
\includegraphics[width=1\textwidth]{figures/2D_N3_reduction/case1_absent/case1_absent_phi3.png}
%\caption{$t=1.8$}
\end{subfigure}
\begin{subfigure}{0.078\textwidth}
\centering
\includegraphics[width=1\textwidth]{figures/2D_N3_reduction/case1_absent/case1_absent_phi3.png}
%\caption{$t=2.4$}
\end{subfigure}
\begin{subfigure}{0.078\textwidth}
\centering
\includegraphics[width=1\textwidth]{figures/2D_N3_reduction/case1_absent/case1_absent_phi3.png}
%\caption{$t=3.0$}
\end{subfigure}
\begin{subfigure}{0.078\textwidth}
\centering
\includegraphics[width=1\textwidth]{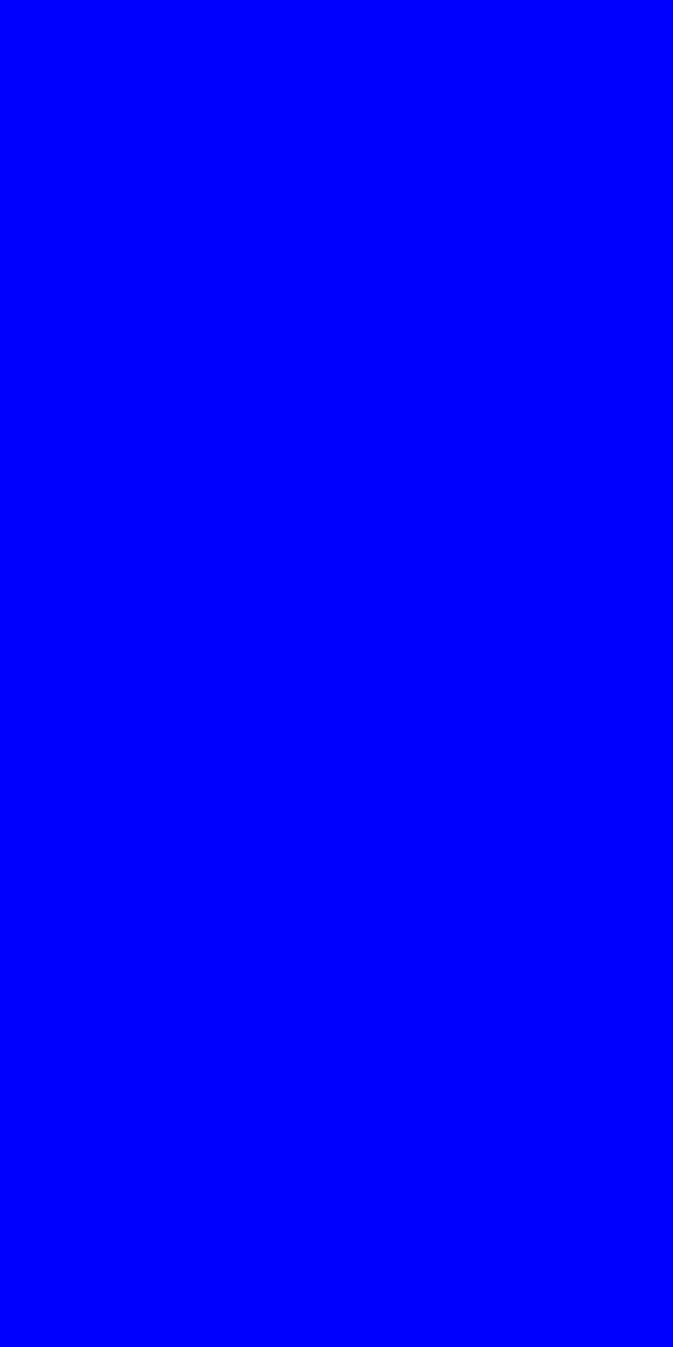}
%\caption{$t=0.0$}
\end{subfigure}
\begin{subfigure}{0.078\textwidth}
\centering
\includegraphics[width=1\textwidth]{figures/2D_N3_reduction/case2_absent/case2_absent_phi3.png}
%\caption{$t=0.6$}
\end{subfigure}
\begin{subfigure}{0.078\textwidth}
\centering
\includegraphics[width=1\textwidth]{figures/2D_N3_reduction/case2_absent/case2_absent_phi3.png}
%\caption{$t=1.2$}
\end{subfigure}
\begin{subfigure}{0.078\textwidth}
\centering
\includegraphics[width=1\textwidth]{figures/2D_N3_reduction/case2_absent/case2_absent_phi3.png}
%\caption{$t=1.8$}
\end{subfigure}
\begin{subfigure}{0.078\textwidth}
\centering
\includegraphics[width=1\textwidth]{figures/2D_N3_reduction/case2_absent/case2_absent_phi3.png}
%\caption{$t=2.4$}
\end{subfigure}
\begin{subfigure}{0.078\textwidth}
\centering
\includegraphics[width=1\textwidth]{figures/2D_N3_reduction/case2_absent/case2_absent_phi3.png}
%\caption{$t=3.0$}
\end{subfigure}
\caption{Mixture-aware simulations -- absent phase. Case 1 (left) and Case 2 (right). Visualization of the phase fields (top to below) $\phi_1, \phi_2, \phi_3$ at times $t=0.0, 0.6, 1.2, 1.8, 2.4, 3.0$ (left to right).}
\label{fig: reduction case 1 2 - phi3=0}
\end{figure}

\subsubsection*{Equal phases}
Here we test the merging of identical phases. We choose the initial phase fields as:
\begin{subequations}
  \begin{align}
    \phi^h_{1,0}(\mathbf{x}) =&~ \frac{1}{4}\left(1+\tanh{\dfrac{\sqrt{(x-0.5)^2+(y-0.5)^2}-R_0}{\varepsilon\sqrt{2}}}\right),\\
    \phi^h_{2,0}(\mathbf{x}) =&~ \frac{1}{2}\left(1-\tanh{\dfrac{\sqrt{(x-0.5)^2+(y-0.5)^2}-R_0}{\varepsilon\sqrt{2}}}\right),\\
    \phi^h_{3,0}(\mathbf{x}) =&~ \phi^h_{1,0}(\mathbf{x}).
\end{align}
\end{subequations}
In Figure \ref{fig: reduction case 1 2 - phi1=phi3} we visualize the phase fields for Case 1 and 2. We observe that the phase $\phi_2$ shows the exact same evolution as in the previous scenario, and phases $\phi_1$ and $\phi_3$ remain identical. These observations are in agreement with axiom A5 (Merging identical phases).

\begin{figure}
\captionsetup[subfigure]{justification=centering}
\begin{subfigure}{0.078\textwidth}
\centering
\includegraphics[width=1\textwidth]{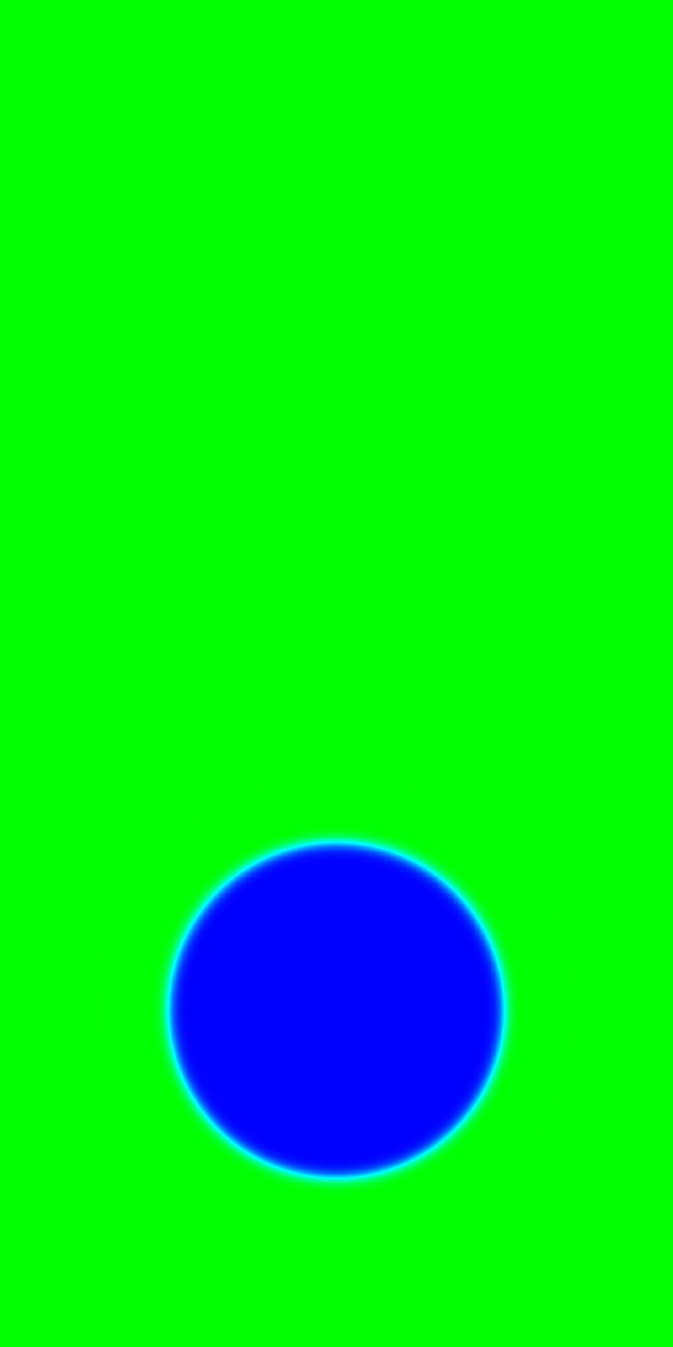}
%\caption{$t=0.0$}
\end{subfigure}
\begin{subfigure}{0.078\textwidth}
\centering
\includegraphics[width=1\textwidth]{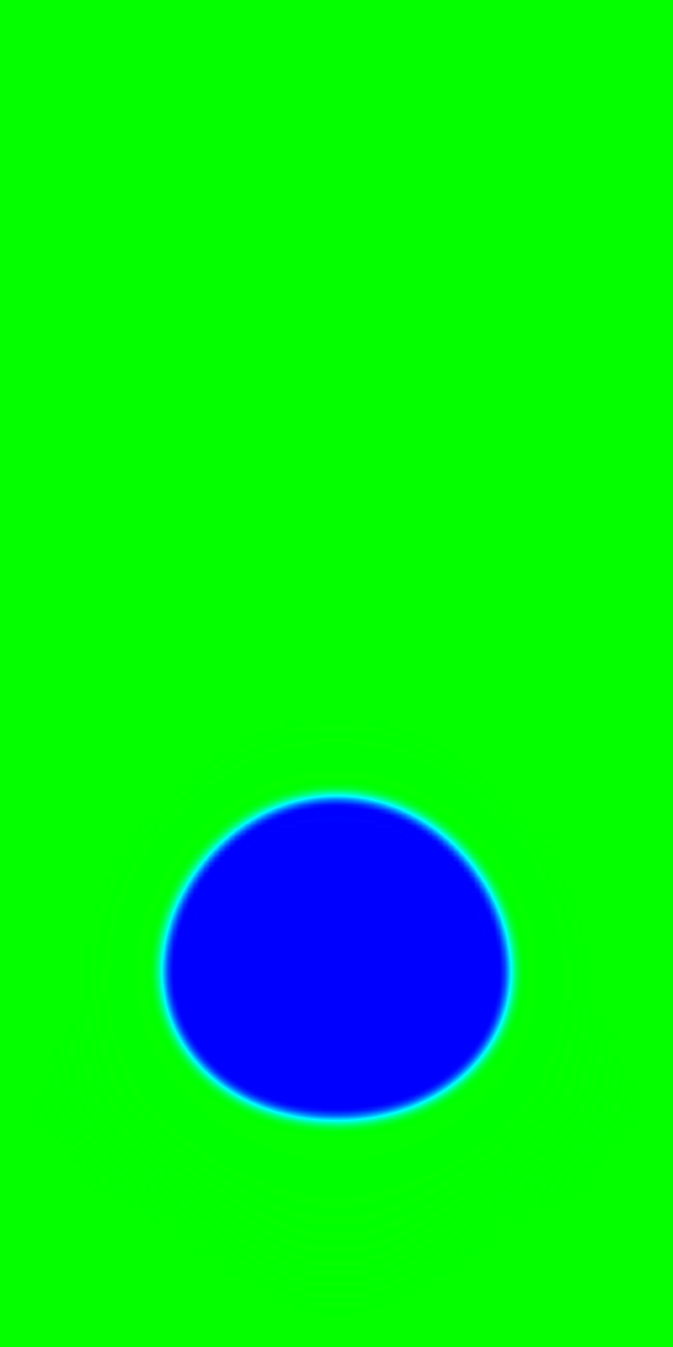}
%\caption{$t=0.6$}
\end{subfigure}
\begin{subfigure}{0.078\textwidth}
\centering
\includegraphics[width=1\textwidth]{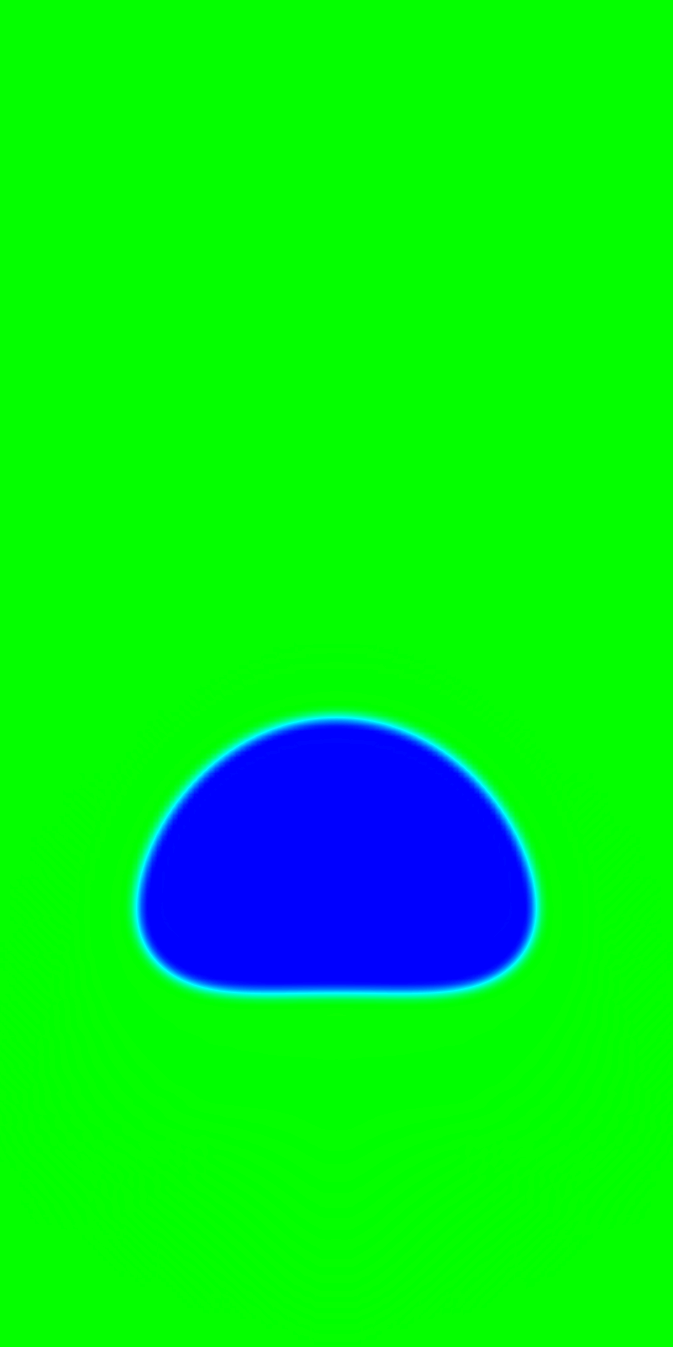}
%\caption{$t=1.2$}
\end{subfigure}
\begin{subfigure}{0.078\textwidth}
\centering
\includegraphics[width=1\textwidth]{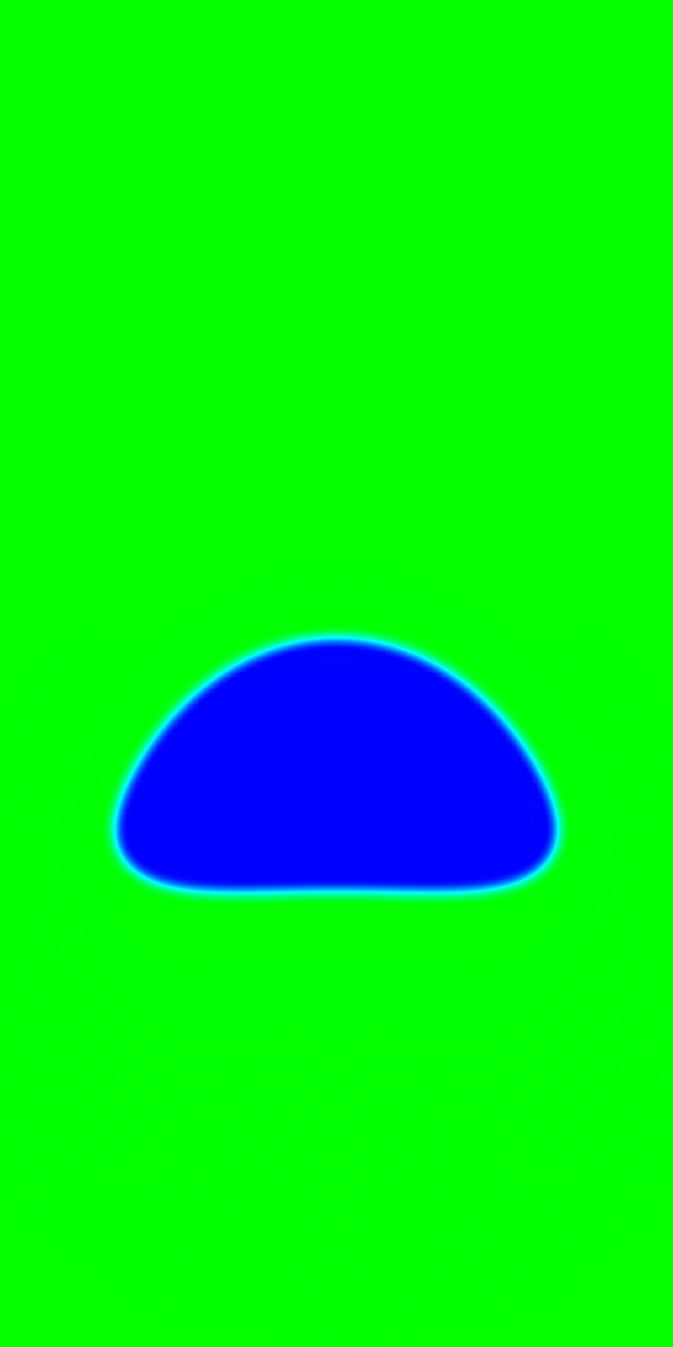}
%\caption{$t=1.8$}
\end{subfigure}
\begin{subfigure}{0.078\textwidth}
\centering
\includegraphics[width=1\textwidth]{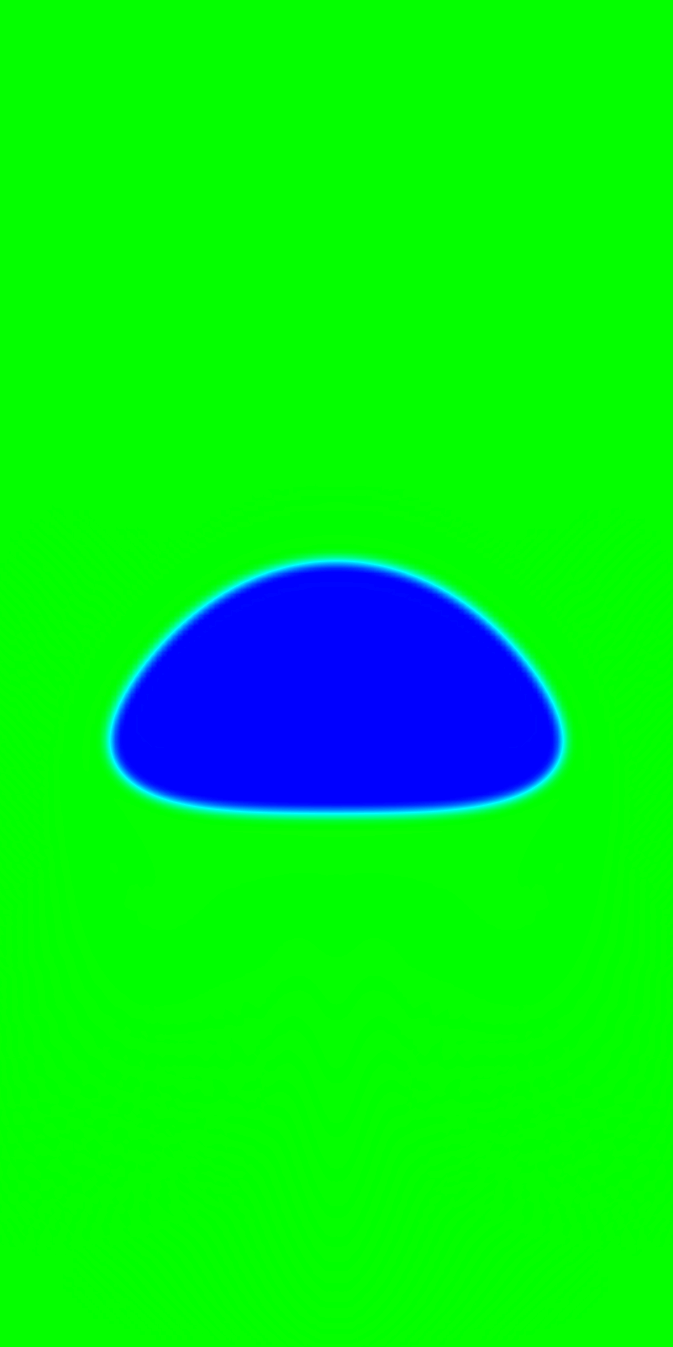}
%\caption{$t=2.4$}
\end{subfigure}
\begin{subfigure}{0.078\textwidth}
\centering
\includegraphics[width=1\textwidth]{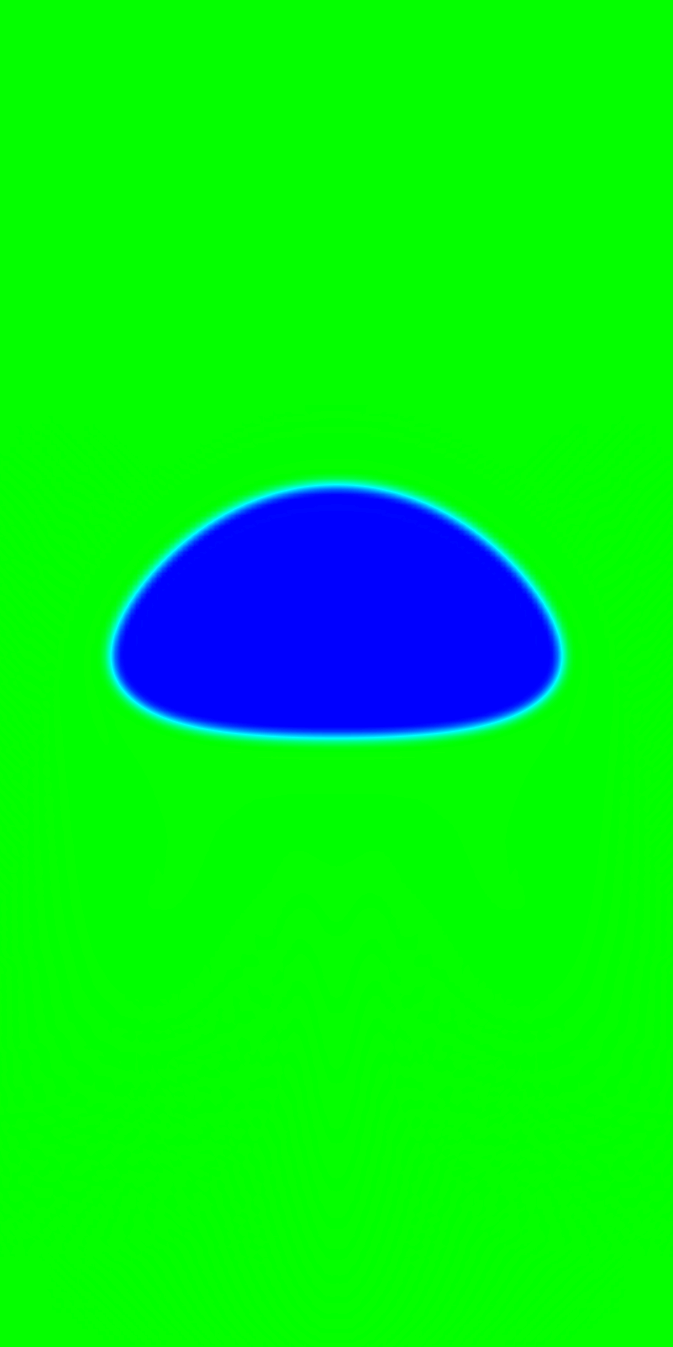}
%\caption{$t=3.0$}
\end{subfigure}
\begin{subfigure}{0.078\textwidth}
\centering
\includegraphics[width=1\textwidth]{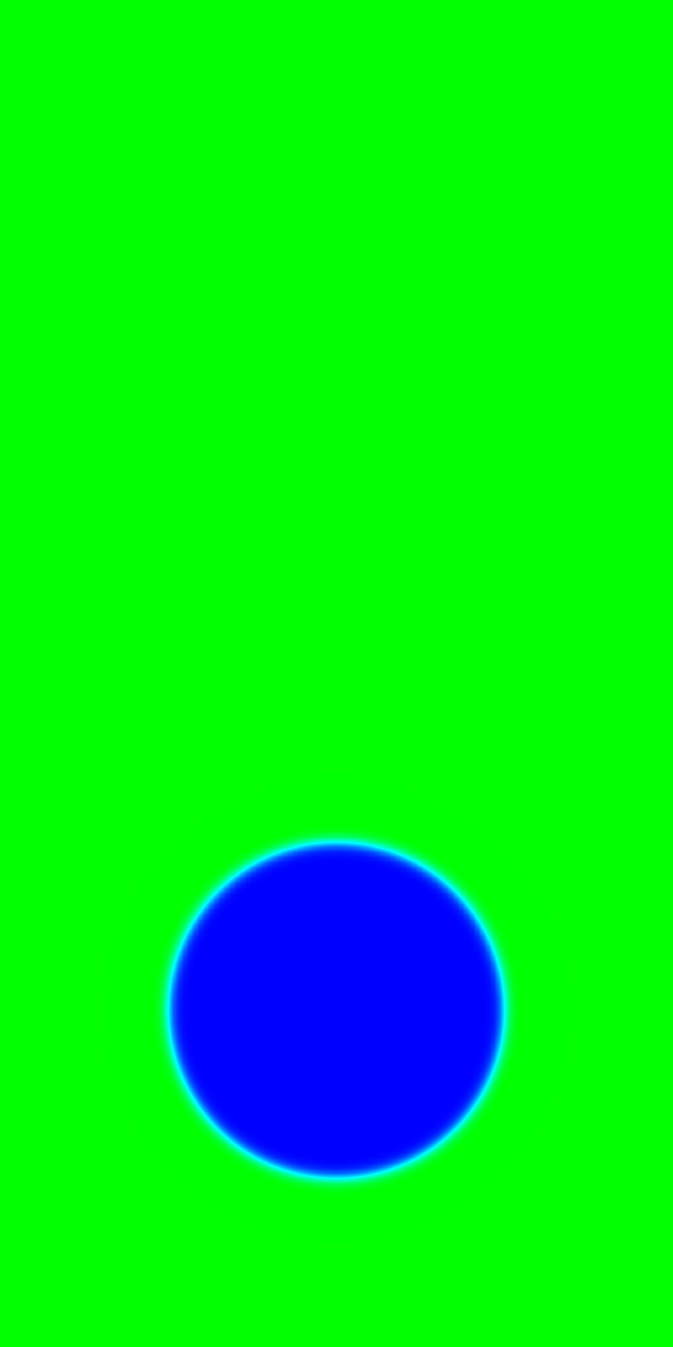}
%\caption{$t=0.0$}
\end{subfigure}
\begin{subfigure}{0.078\textwidth}
\centering
\includegraphics[width=1\textwidth]{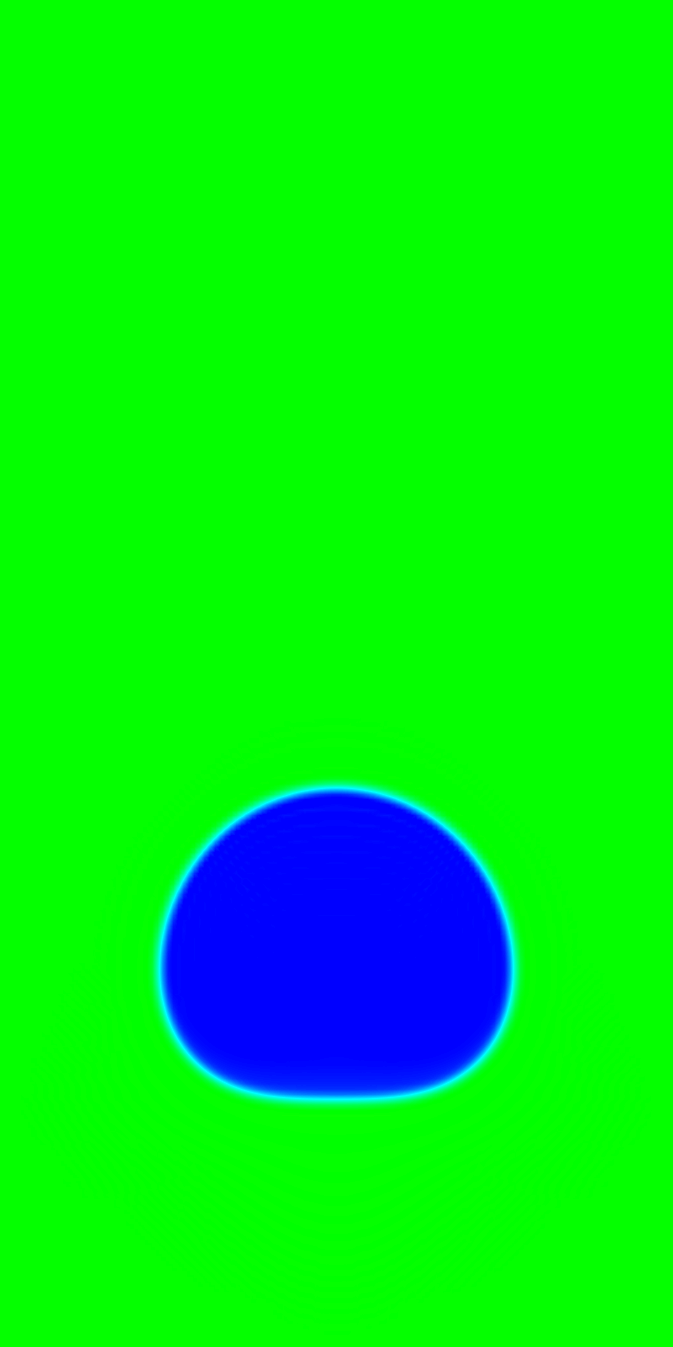}
%\caption{$t=0.6$}
\end{subfigure}
\begin{subfigure}{0.078\textwidth}
\centering
\includegraphics[width=1\textwidth]{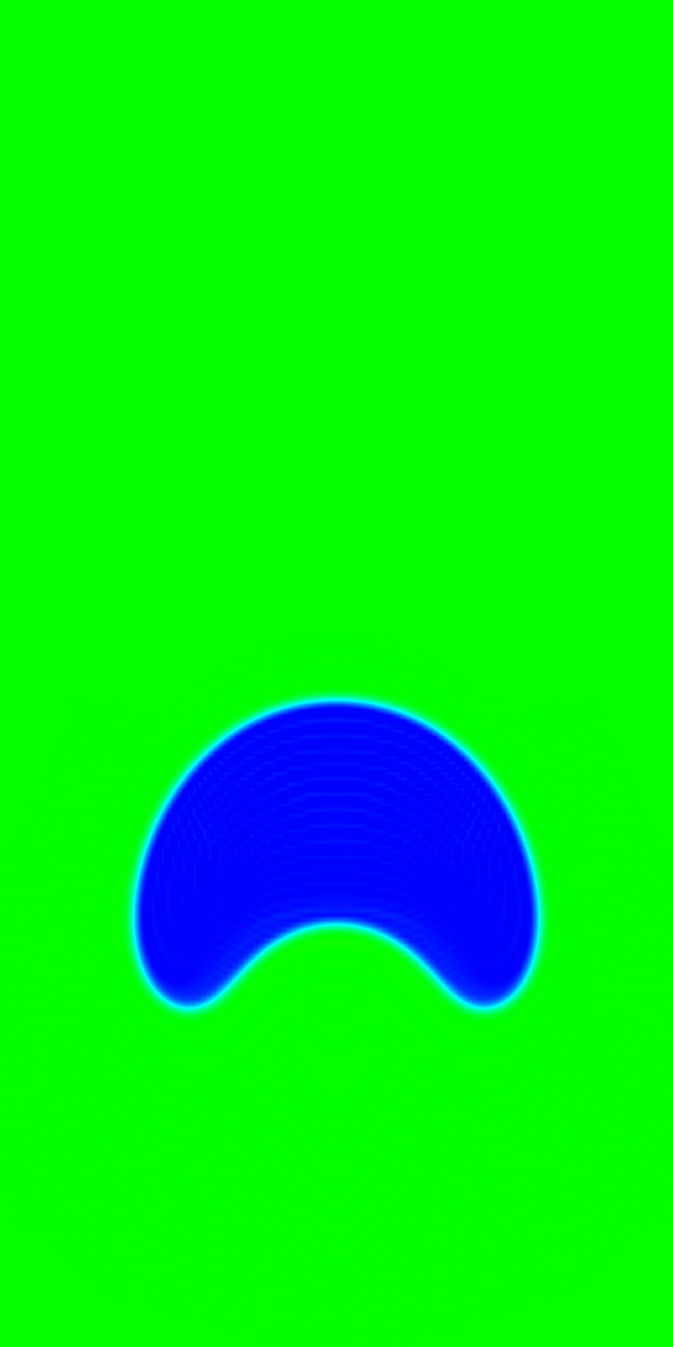}
%\caption{$t=1.2$}
\end{subfigure}
\begin{subfigure}{0.078\textwidth}
\centering
\includegraphics[width=1\textwidth]{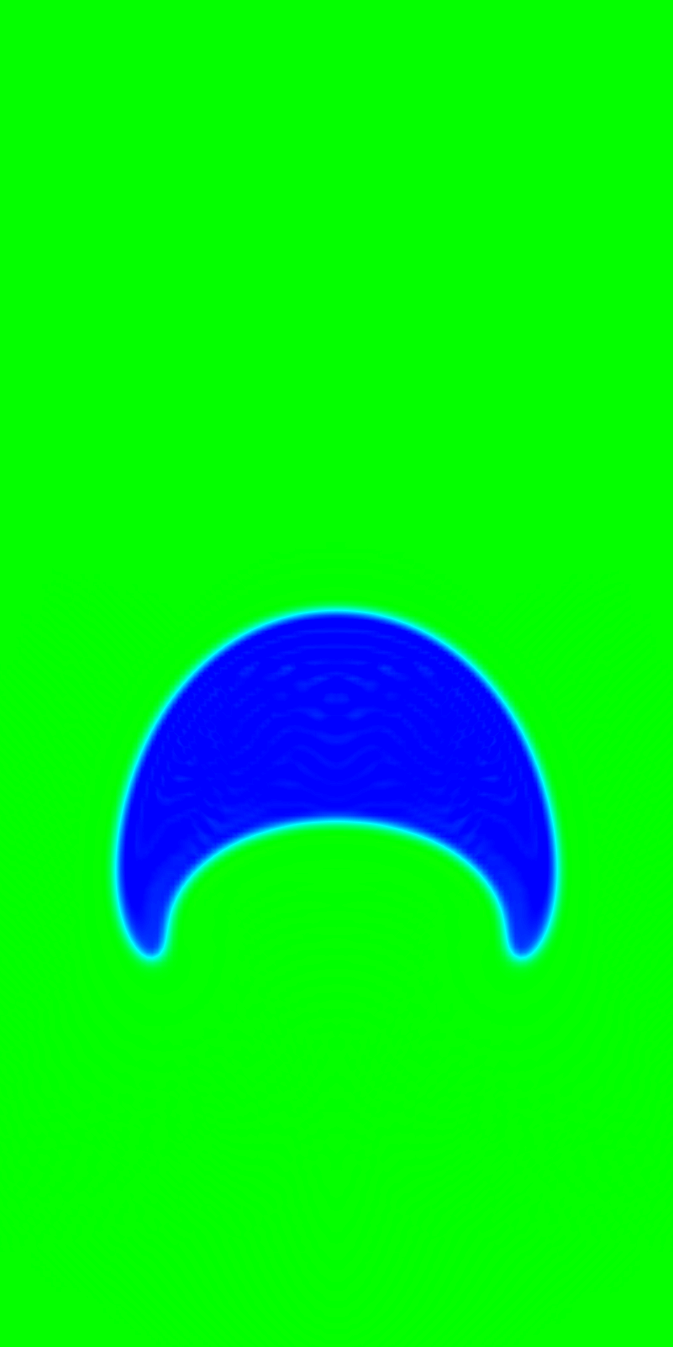}
%\caption{$t=1.8$}
\end{subfigure}
\begin{subfigure}{0.078\textwidth}
\centering
\includegraphics[width=1\textwidth]{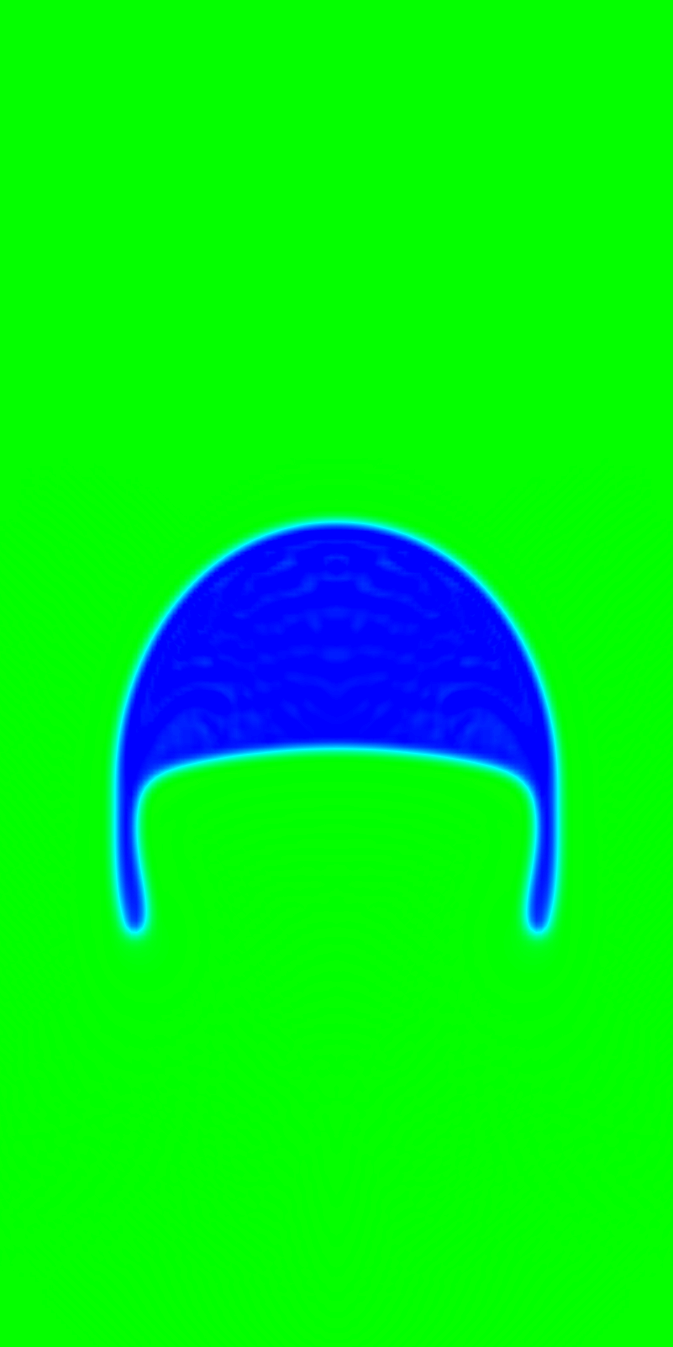}
%\caption{$t=2.4$}
\end{subfigure}
\begin{subfigure}{0.078\textwidth}
\centering
\includegraphics[width=1\textwidth]{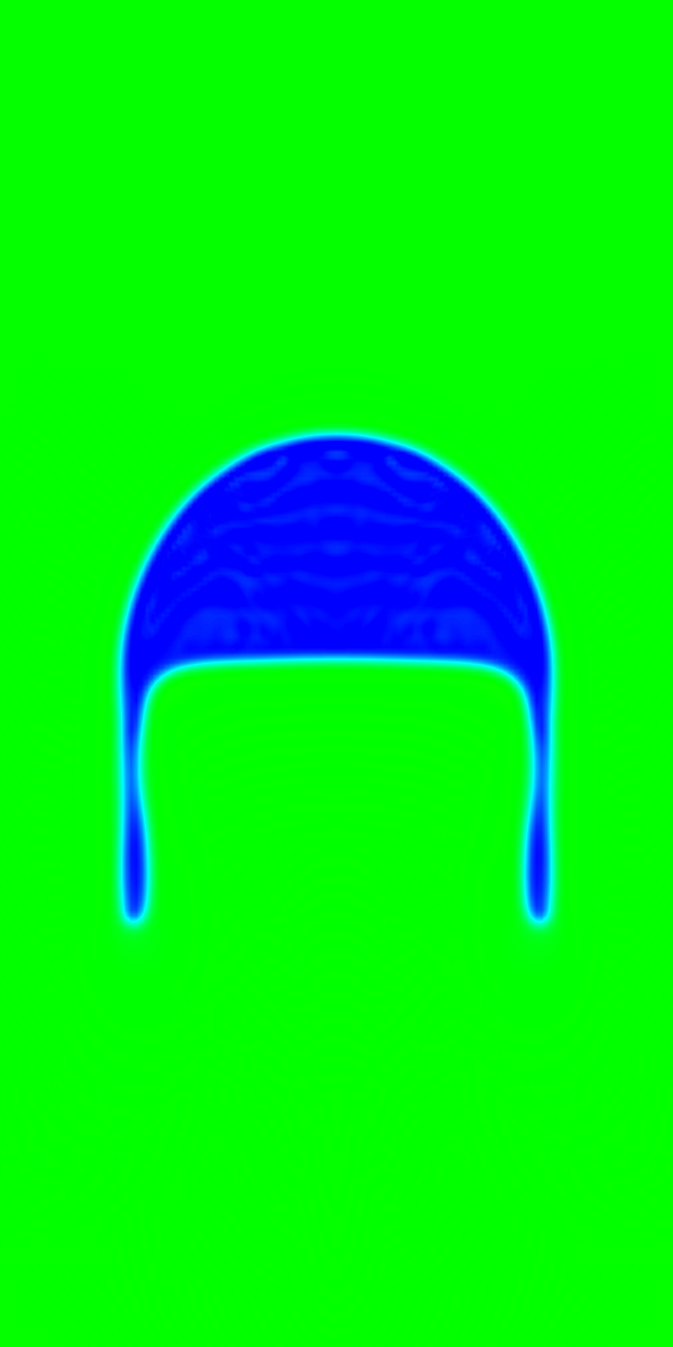}
%\caption{$t=3.0$}
\end{subfigure}
\begin{subfigure}{0.078\textwidth}
\centering
\includegraphics[width=1\textwidth]{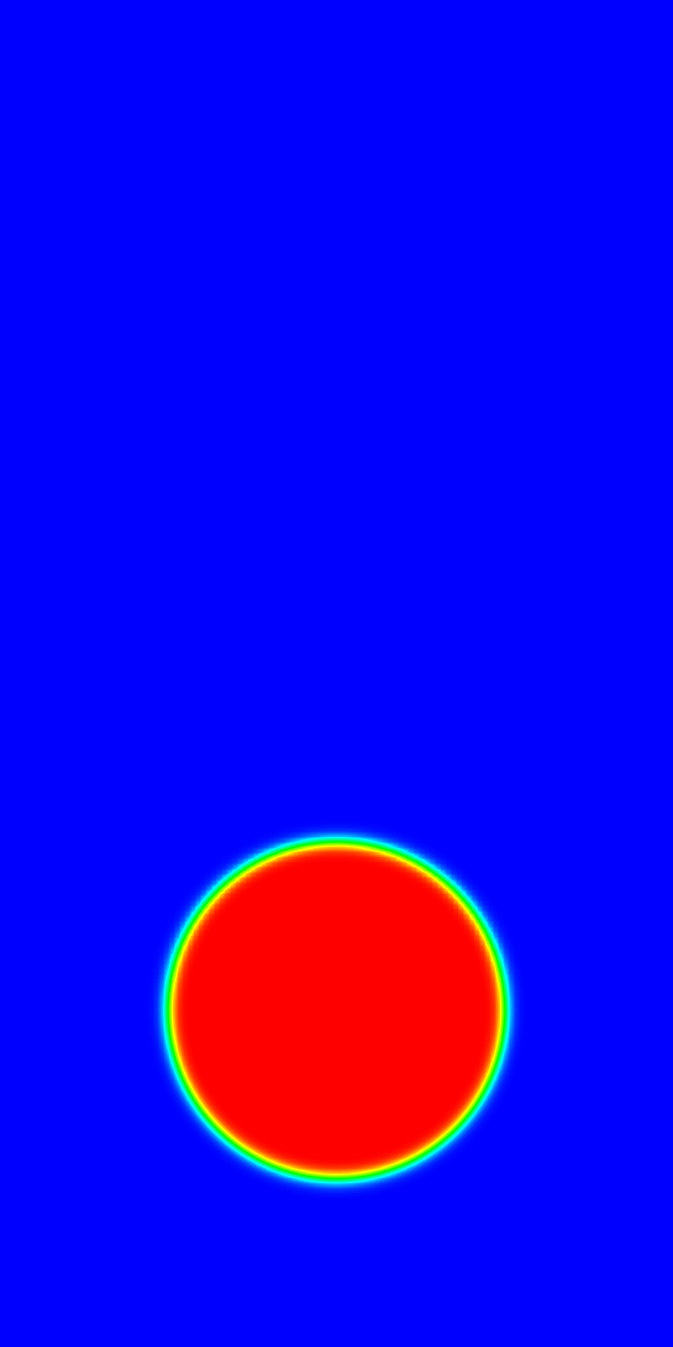}
%\caption{$t=0.0$}
\end{subfigure}
\begin{subfigure}{0.078\textwidth}
\centering
\includegraphics[width=1\textwidth]{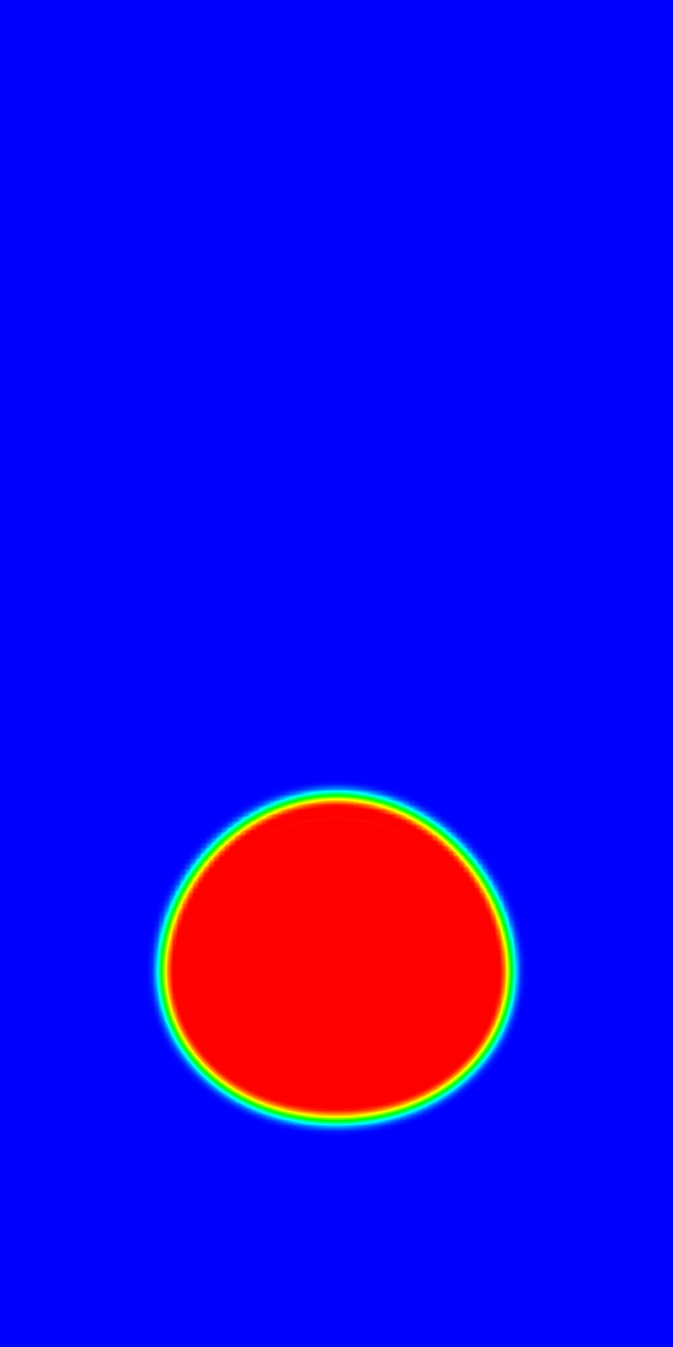}
%\caption{$t=0.6$}
\end{subfigure}
\begin{subfigure}{0.078\textwidth}
\centering
\includegraphics[width=1\textwidth]{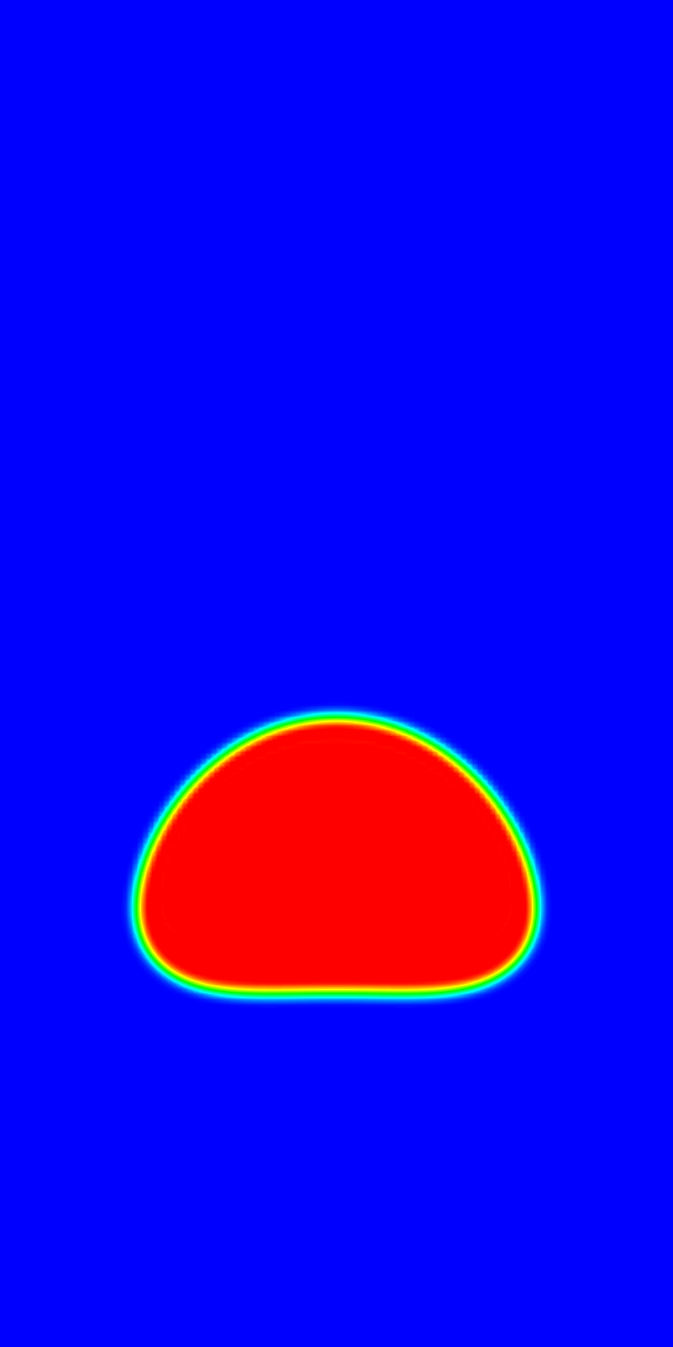}
%\caption{$t=1.2$}
\end{subfigure}
\begin{subfigure}{0.078\textwidth}
\centering
\includegraphics[width=1\textwidth]{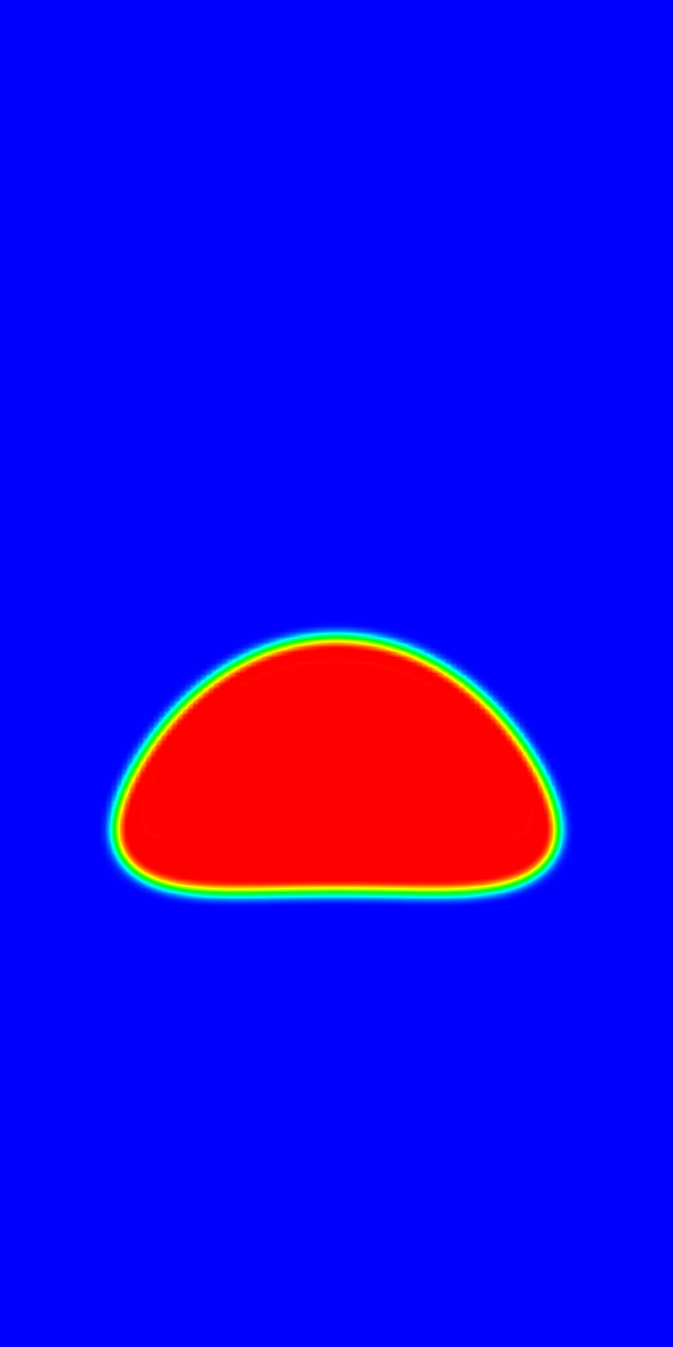}
%\caption{$t=1.8$}
\end{subfigure}
\begin{subfigure}{0.078\textwidth}
\centering
\includegraphics[width=1\textwidth]{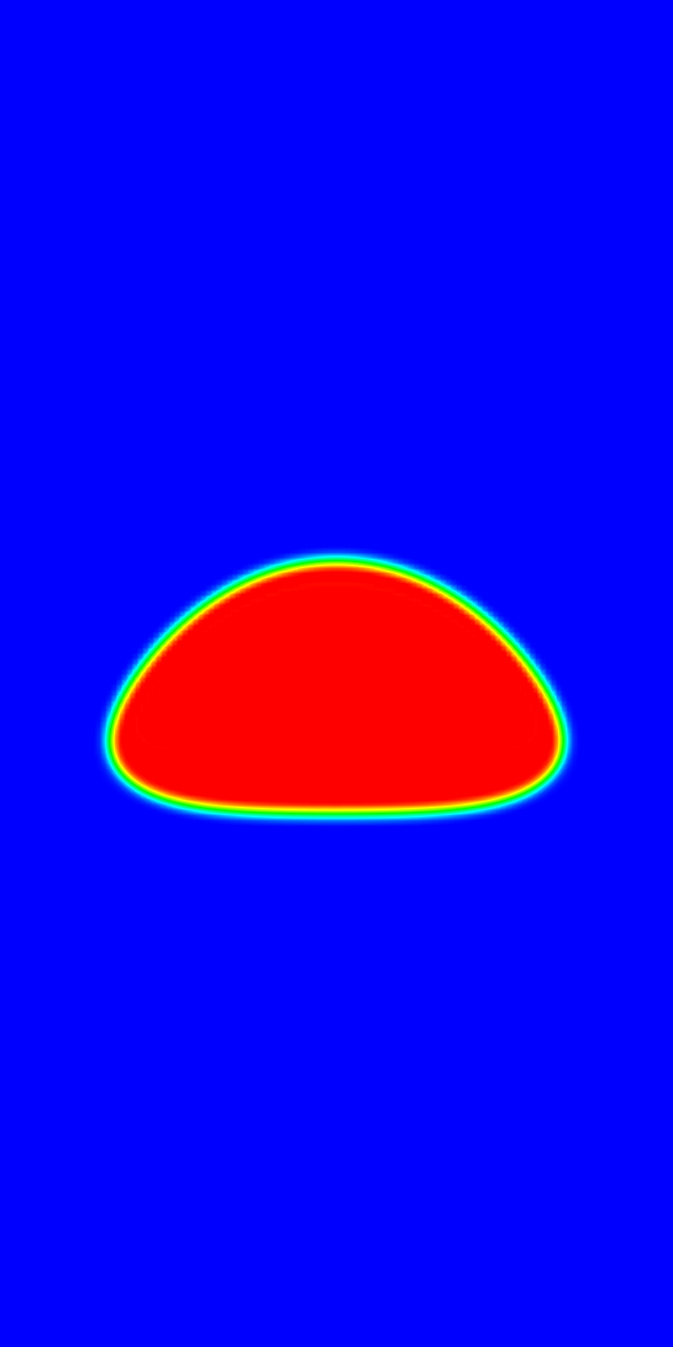}
%\caption{$t=2.4$}
\end{subfigure}
\begin{subfigure}{0.078\textwidth}
\centering
\includegraphics[width=1\textwidth]{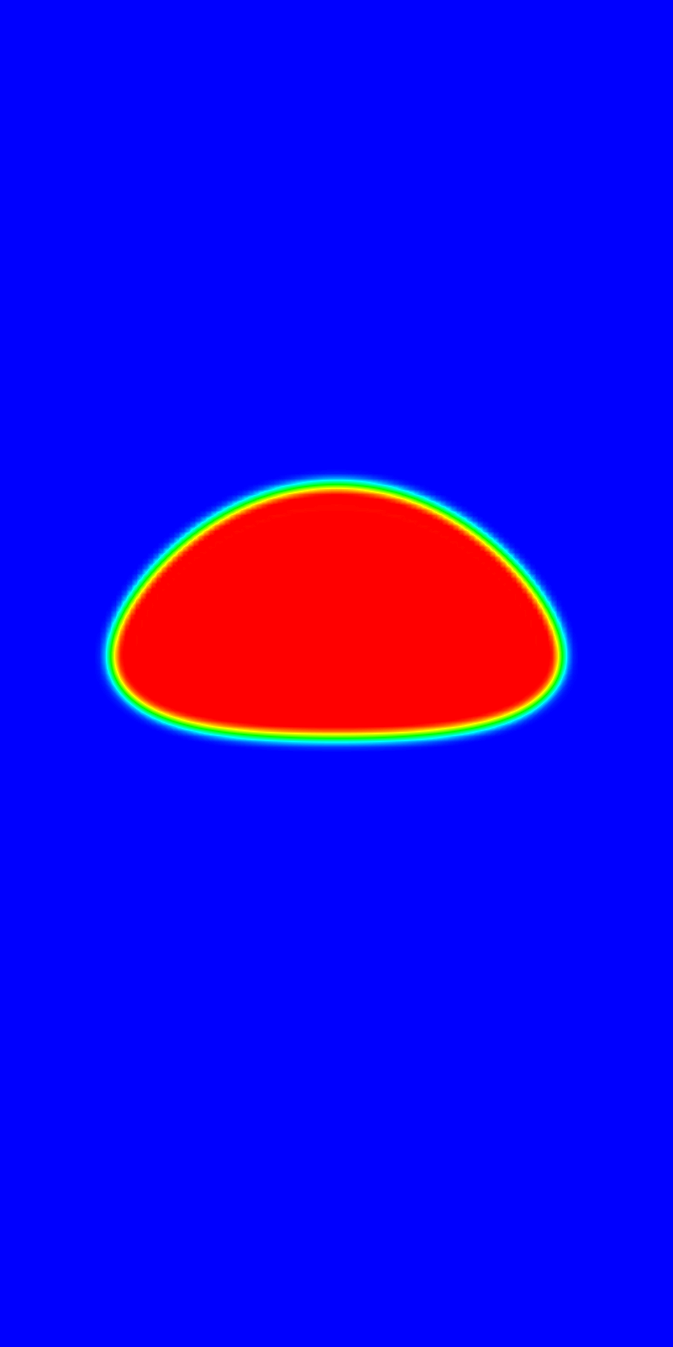}
%\caption{$t=3.0$}
\end{subfigure}
\begin{subfigure}{0.078\textwidth}
\centering
\includegraphics[width=1\textwidth]{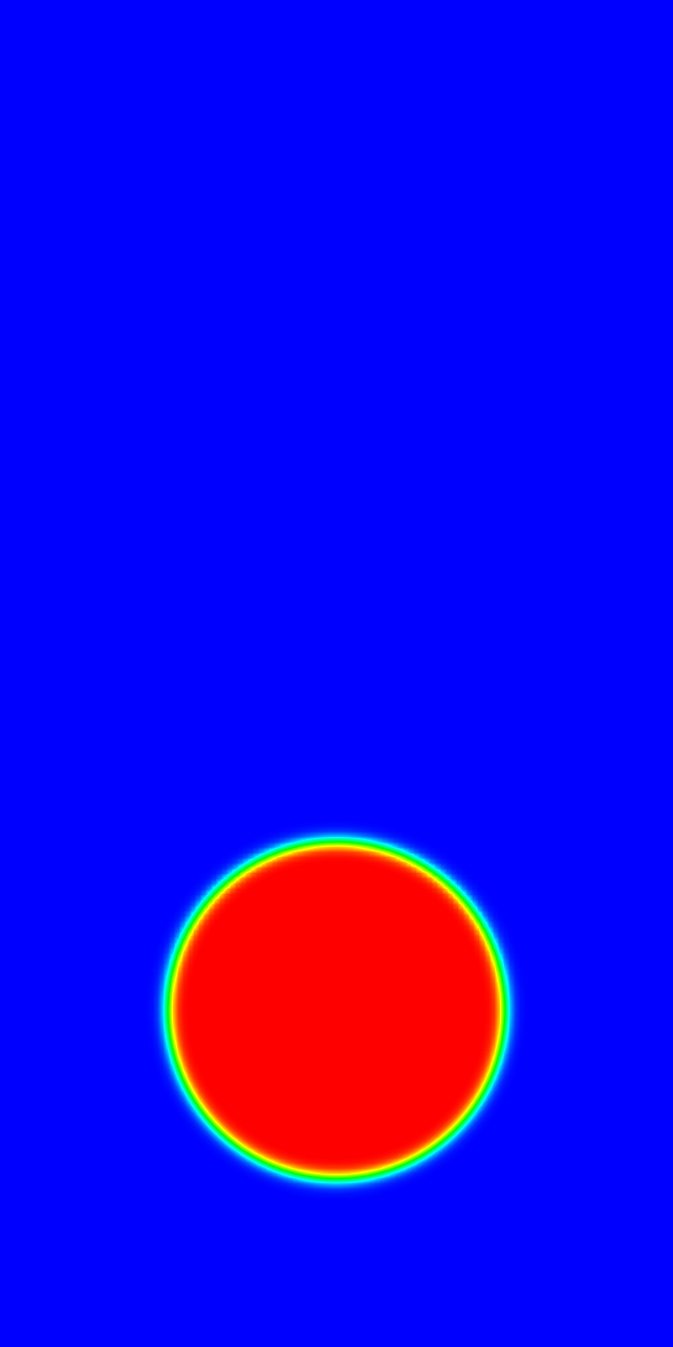}
%\caption{$t=0.0$}
\end{subfigure}
\begin{subfigure}{0.078\textwidth}
\centering
\includegraphics[width=1\textwidth]{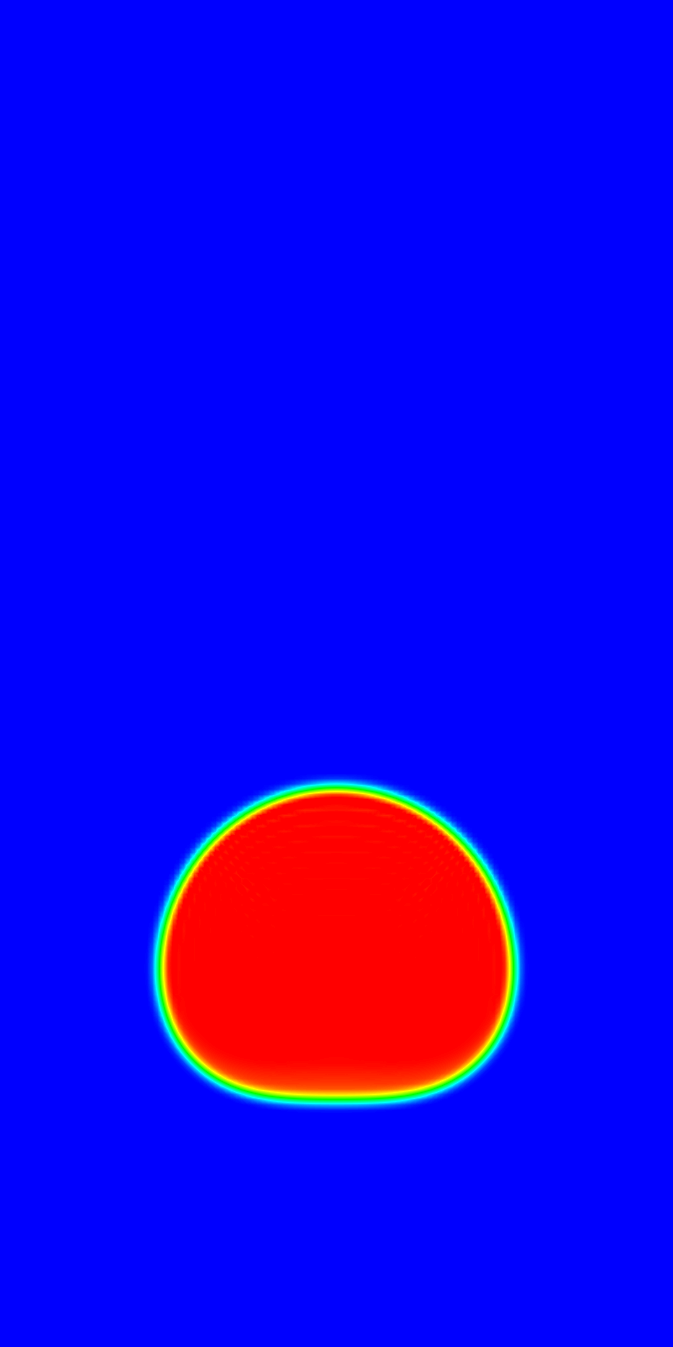}
%\caption{$t=0.6$}
\end{subfigure}
\begin{subfigure}{0.078\textwidth}
\centering
\includegraphics[width=1\textwidth]{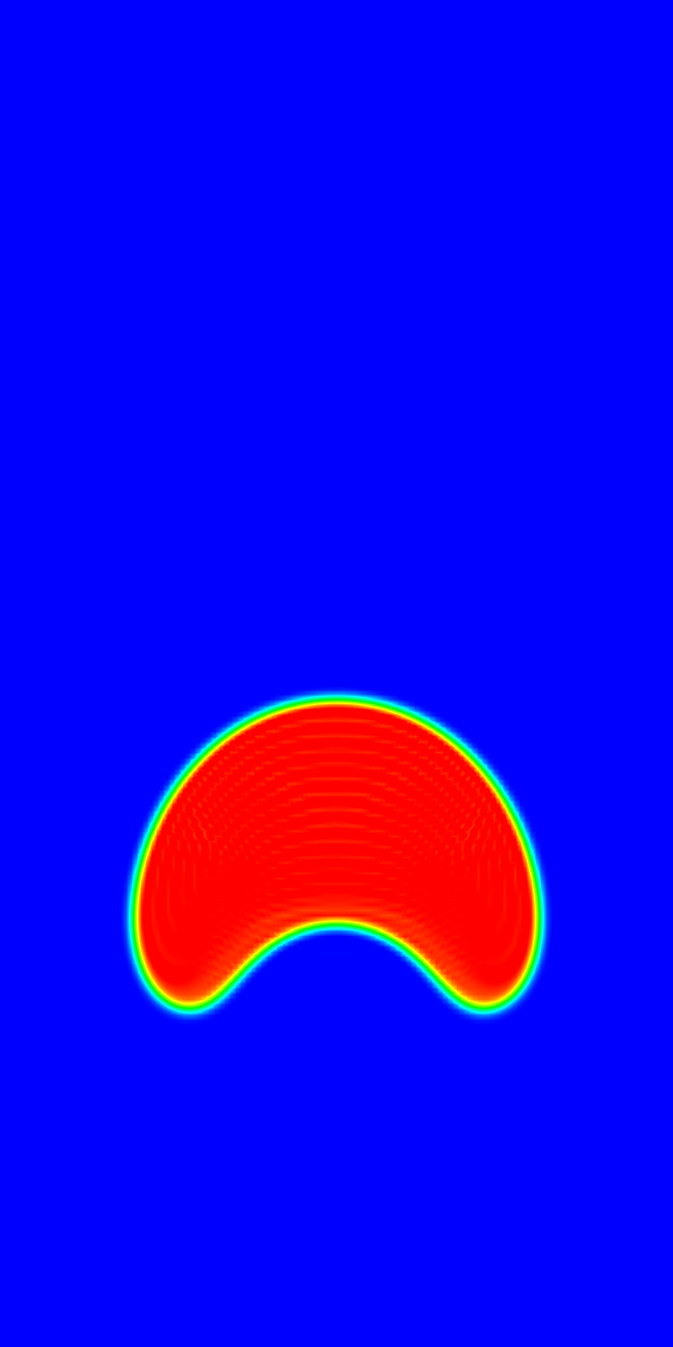}
%\caption{$t=1.2$}
\end{subfigure}
\begin{subfigure}{0.078\textwidth}
\centering
\includegraphics[width=1\textwidth]{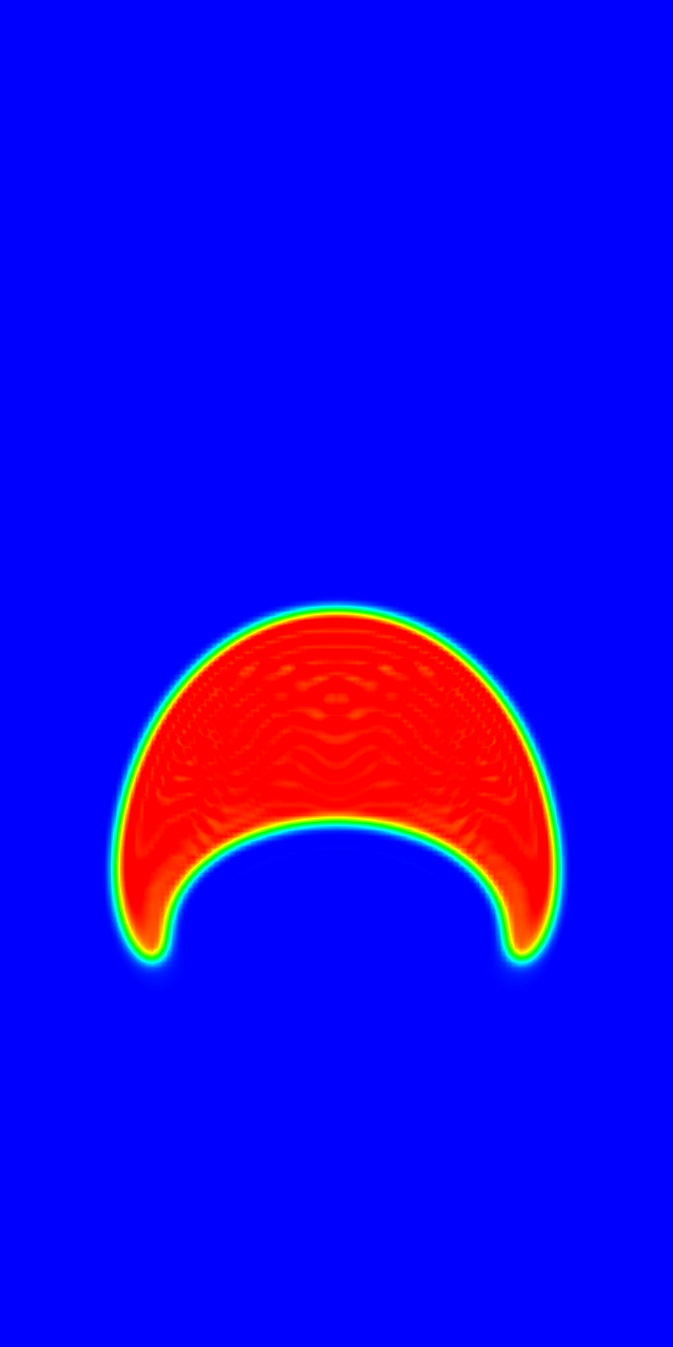}
%\caption{$t=1.8$}
\end{subfigure}
\begin{subfigure}{0.078\textwidth}
\centering
\includegraphics[width=1\textwidth]{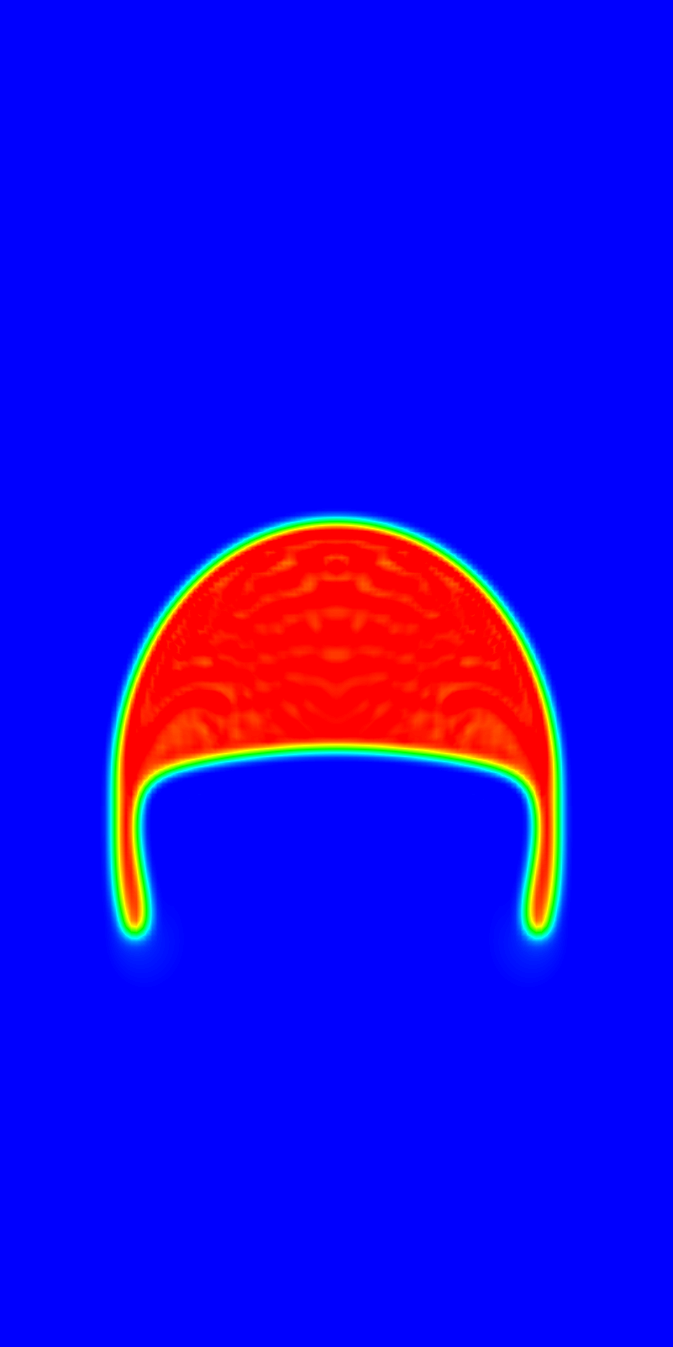}
%\caption{$t=2.4$}
\end{subfigure}
\begin{subfigure}{0.078\textwidth}
\centering
\includegraphics[width=1\textwidth]{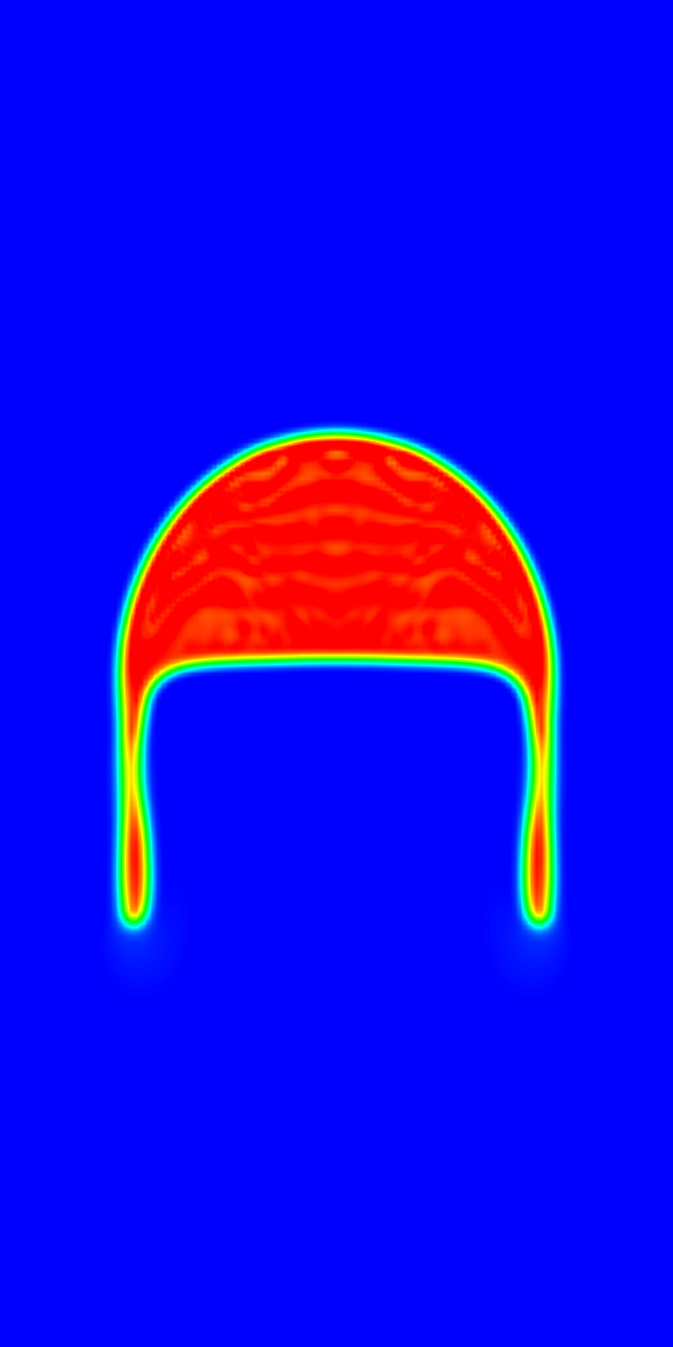}
%\caption{$t=3.0$}
\end{subfigure}
\begin{subfigure}{0.078\textwidth}
\centering
\includegraphics[width=1\textwidth]{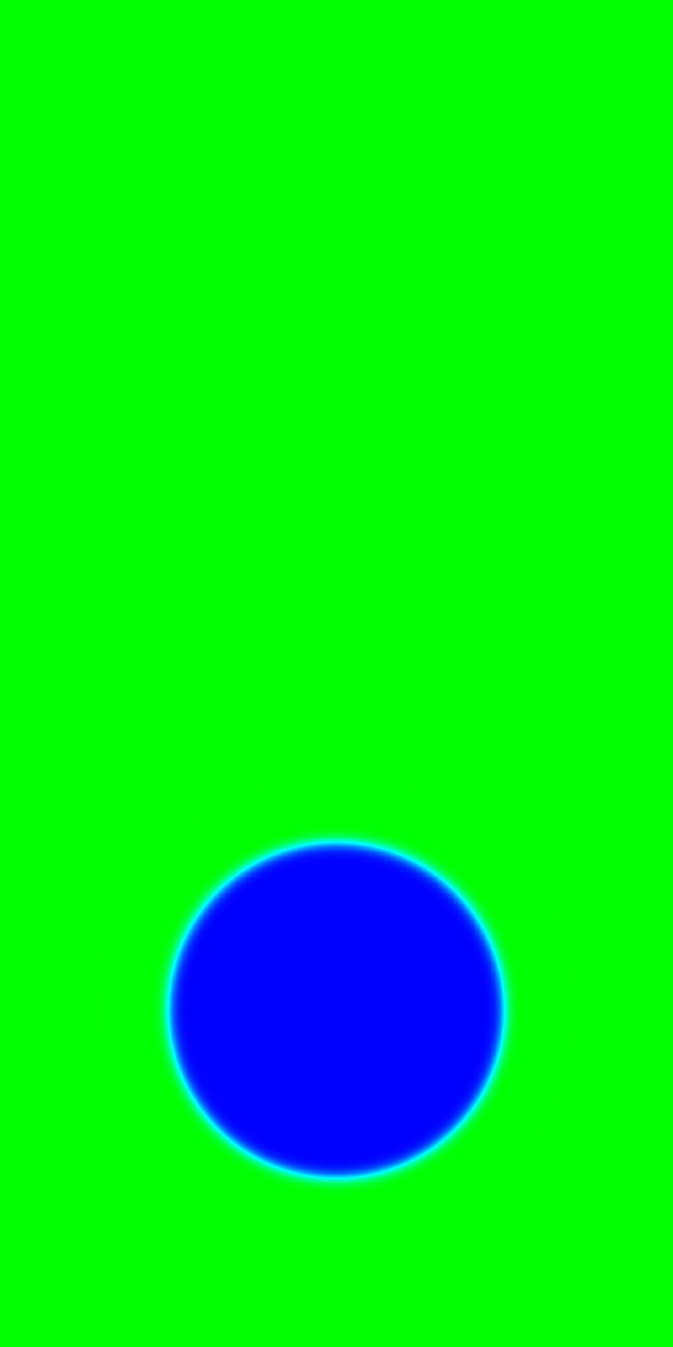}
%\caption{$t=0.0$}
\end{subfigure}
\begin{subfigure}{0.078\textwidth}
\centering
\includegraphics[width=1\textwidth]{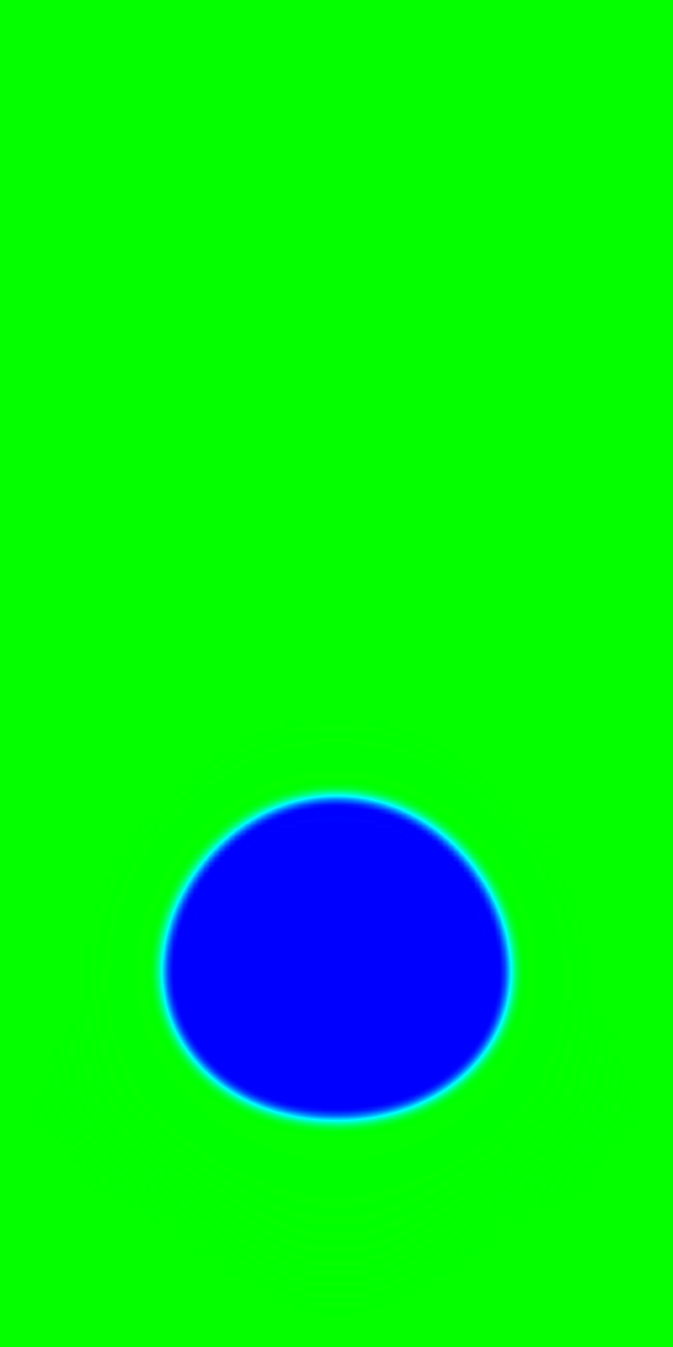}
%\caption{$t=0.6$}
\end{subfigure}
\begin{subfigure}{0.078\textwidth}
\centering
\includegraphics[width=1\textwidth]{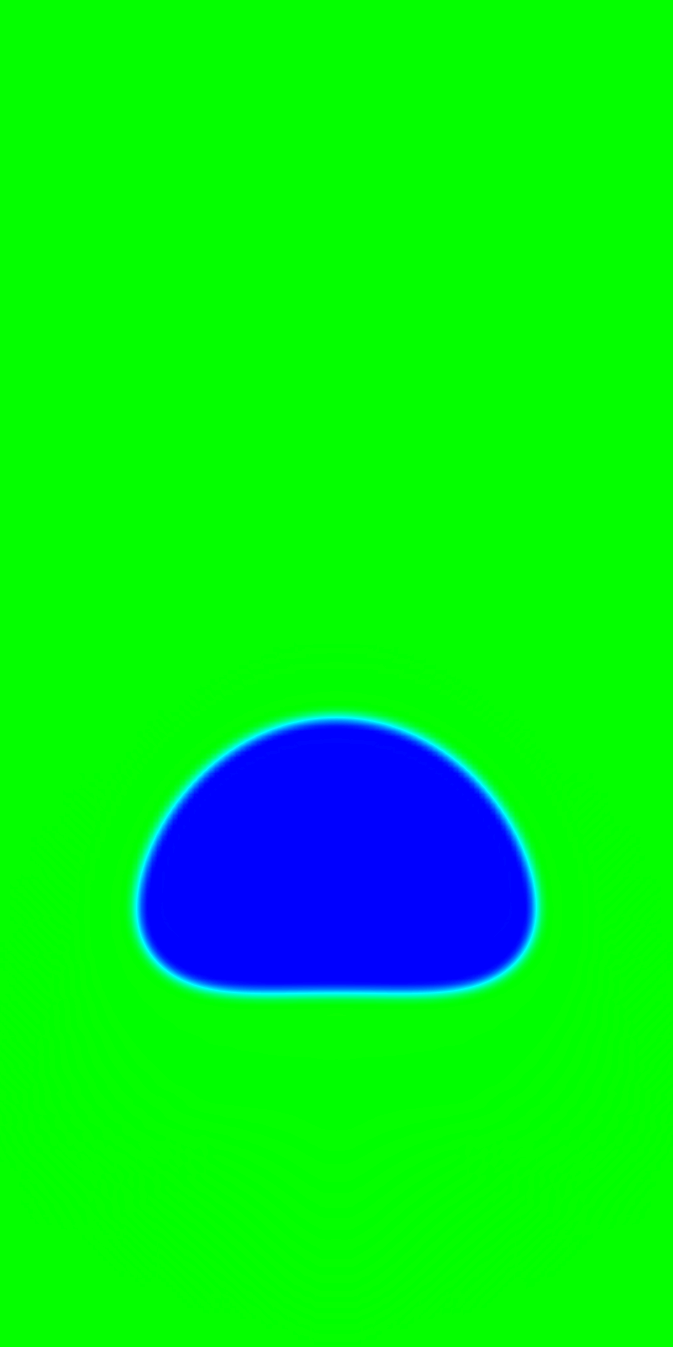}
%\caption{$t=1.2$}
\end{subfigure}
\begin{subfigure}{0.078\textwidth}
\centering
\includegraphics[width=1\textwidth]{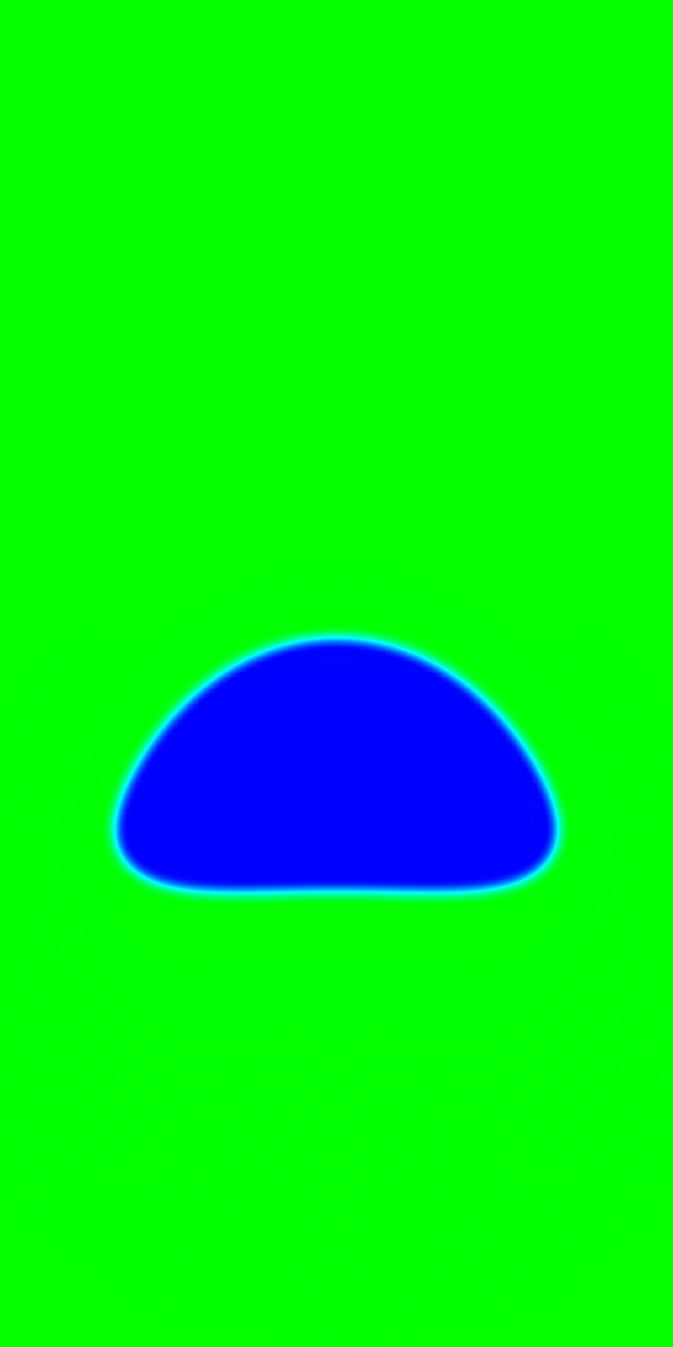}
%\caption{$t=1.8$}
\end{subfigure}
\begin{subfigure}{0.078\textwidth}
\centering
\includegraphics[width=1\textwidth]{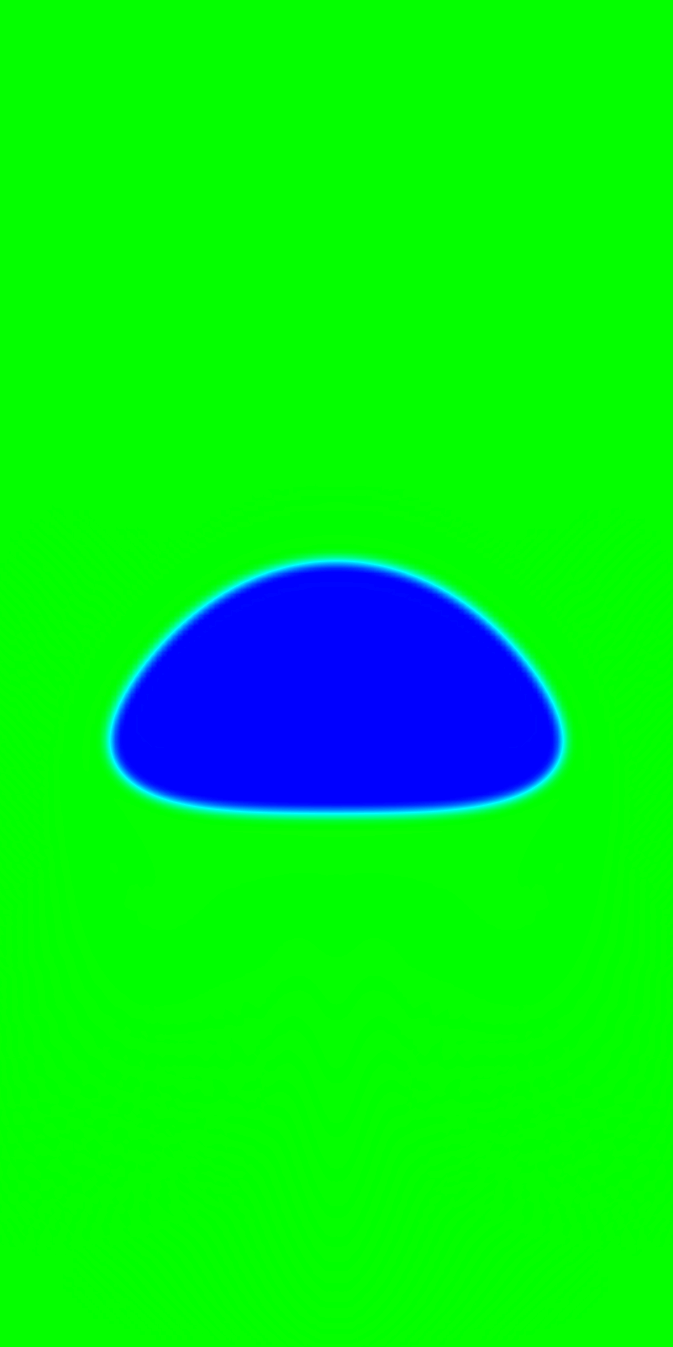}
%\caption{$t=2.4$}
\end{subfigure}
\begin{subfigure}{0.078\textwidth}
\centering
\includegraphics[width=1\textwidth]{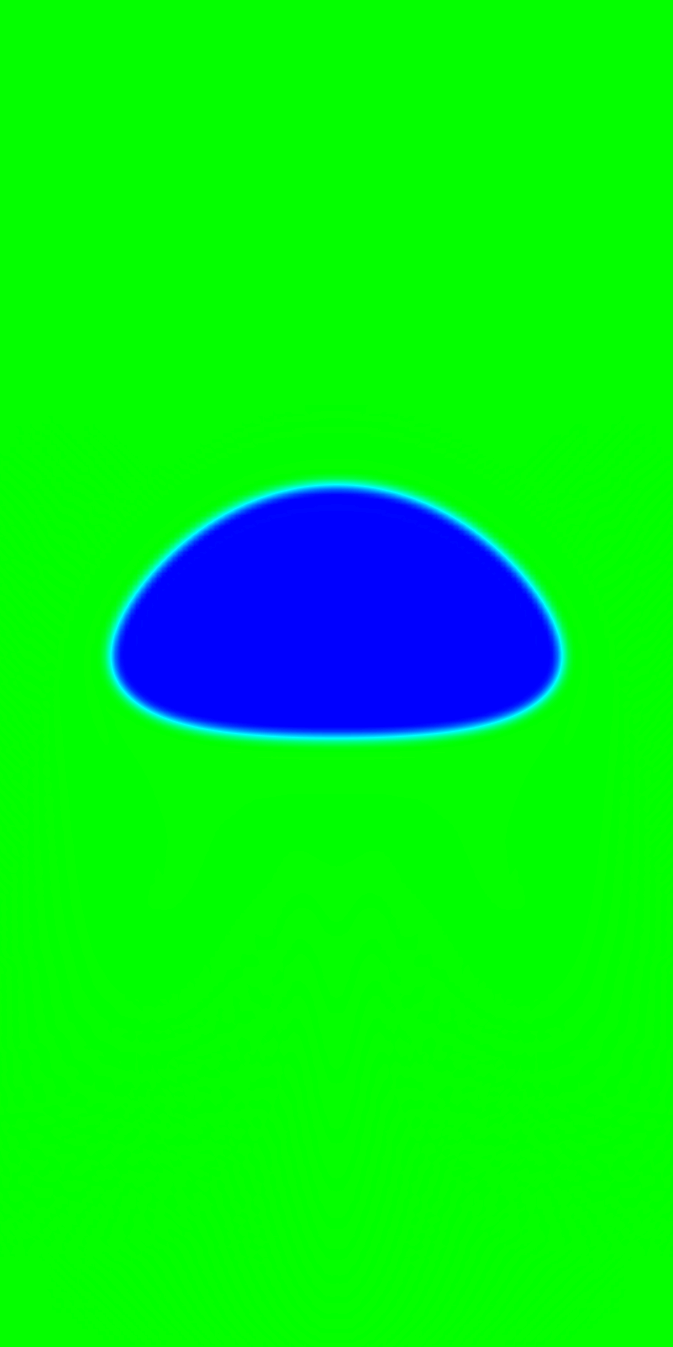}
%\caption{$t=3.0$}
\end{subfigure}
\begin{subfigure}{0.078\textwidth}
\centering
\includegraphics[width=1\textwidth]{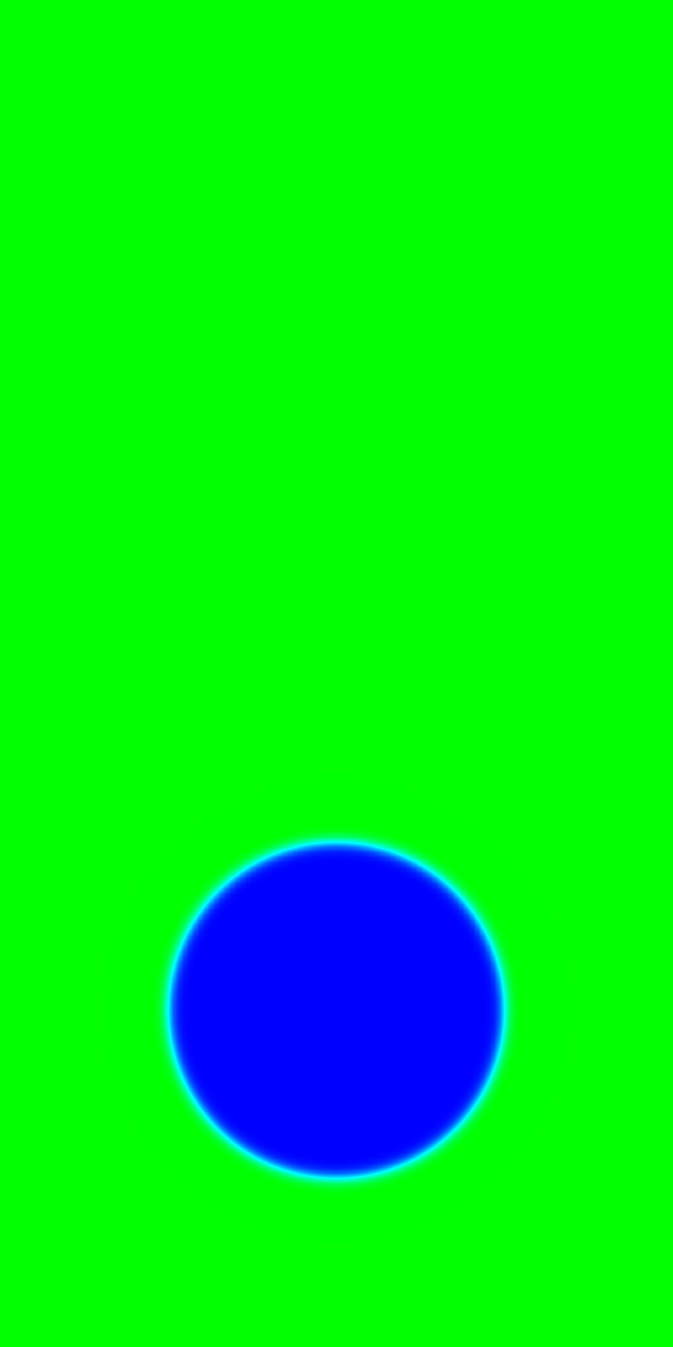}
%\caption{$t=0.0$}
\end{subfigure}
\begin{subfigure}{0.078\textwidth}
\centering
\includegraphics[width=1\textwidth]{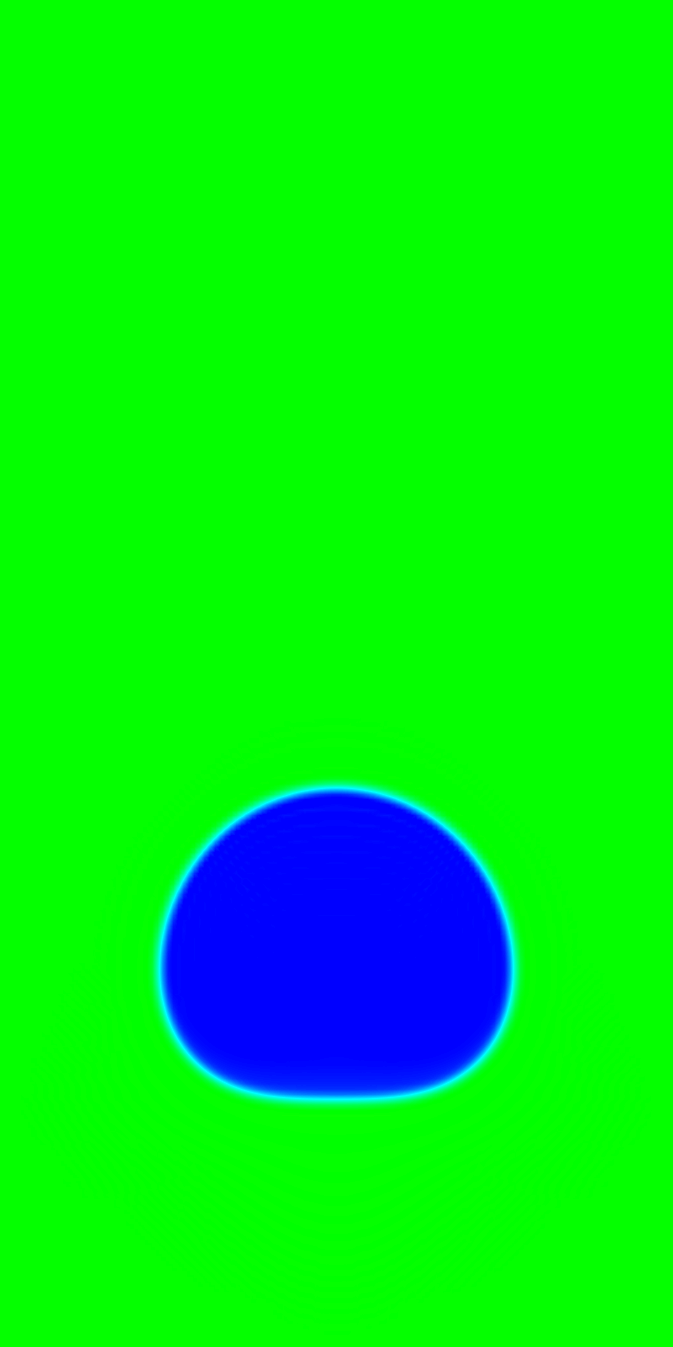}
%\caption{$t=0.6$}
\end{subfigure}
\begin{subfigure}{0.078\textwidth}
\centering
\includegraphics[width=1\textwidth]{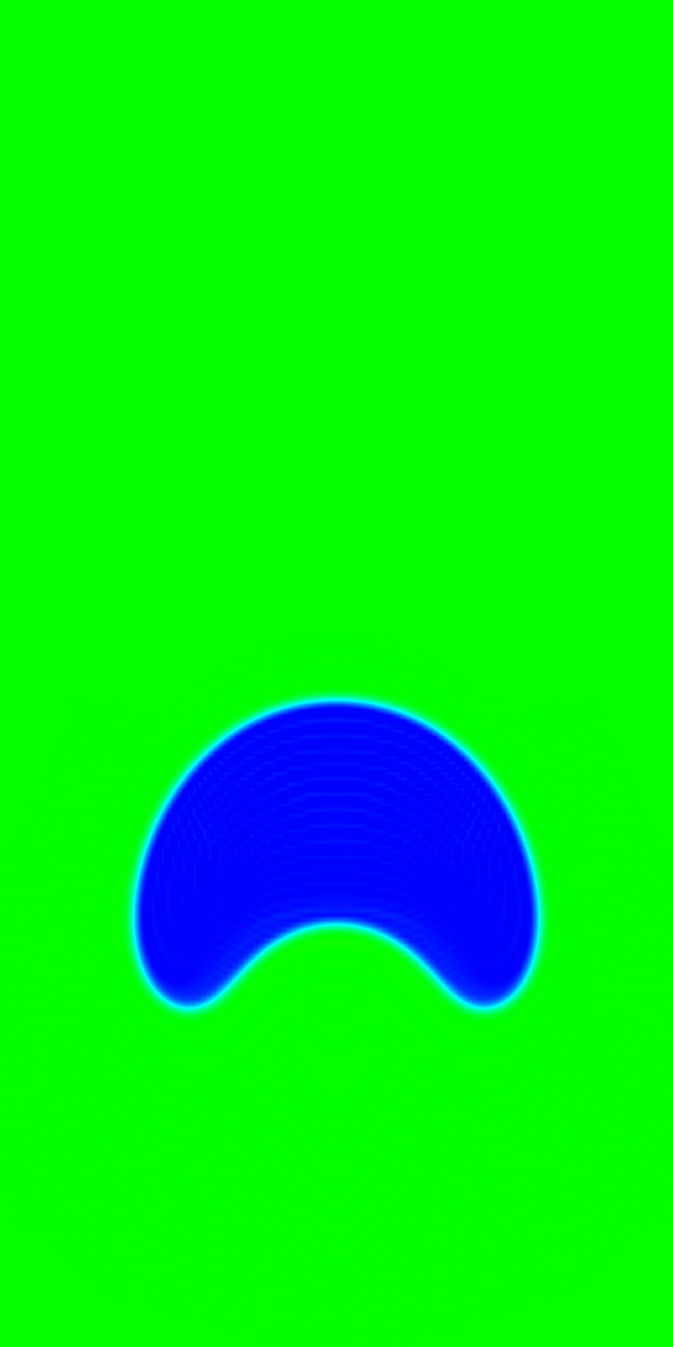}
%\caption{$t=1.2$}
\end{subfigure}
\begin{subfigure}{0.078\textwidth}
\centering
\includegraphics[width=1\textwidth]{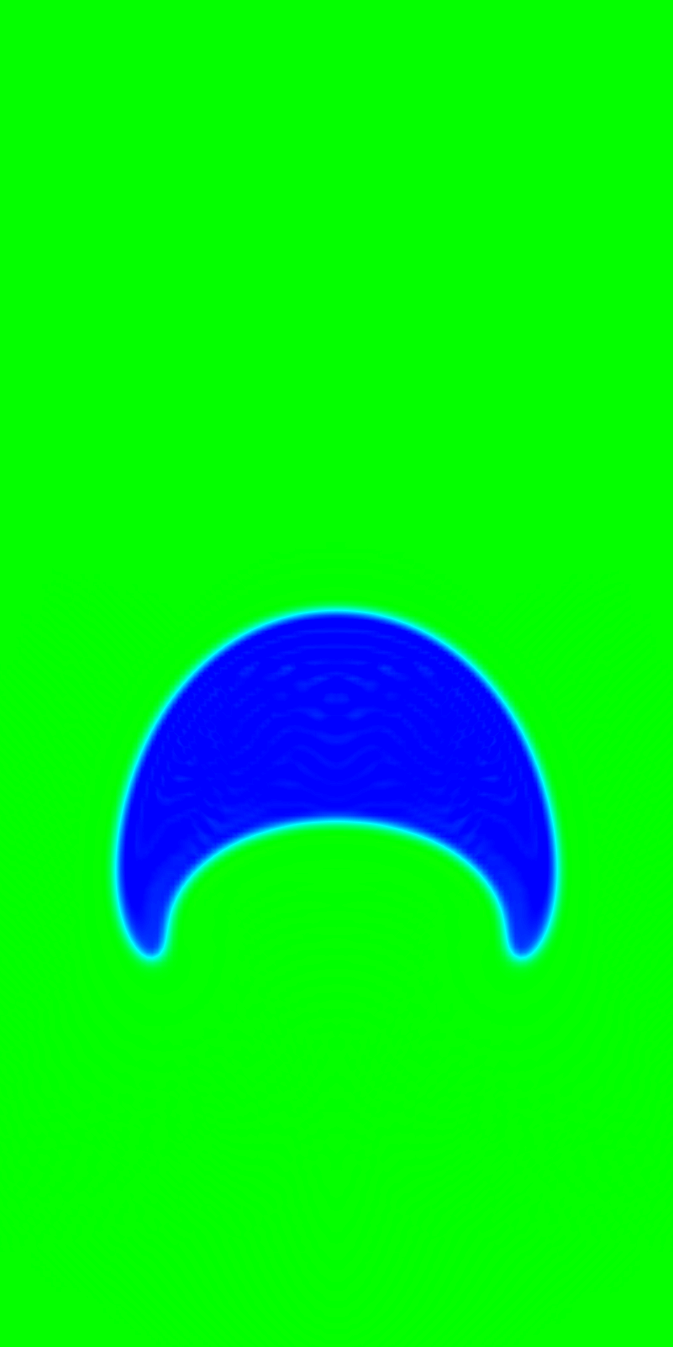}
%\caption{$t=1.8$}
\end{subfigure}
\begin{subfigure}{0.078\textwidth}
\centering
\includegraphics[width=1\textwidth]{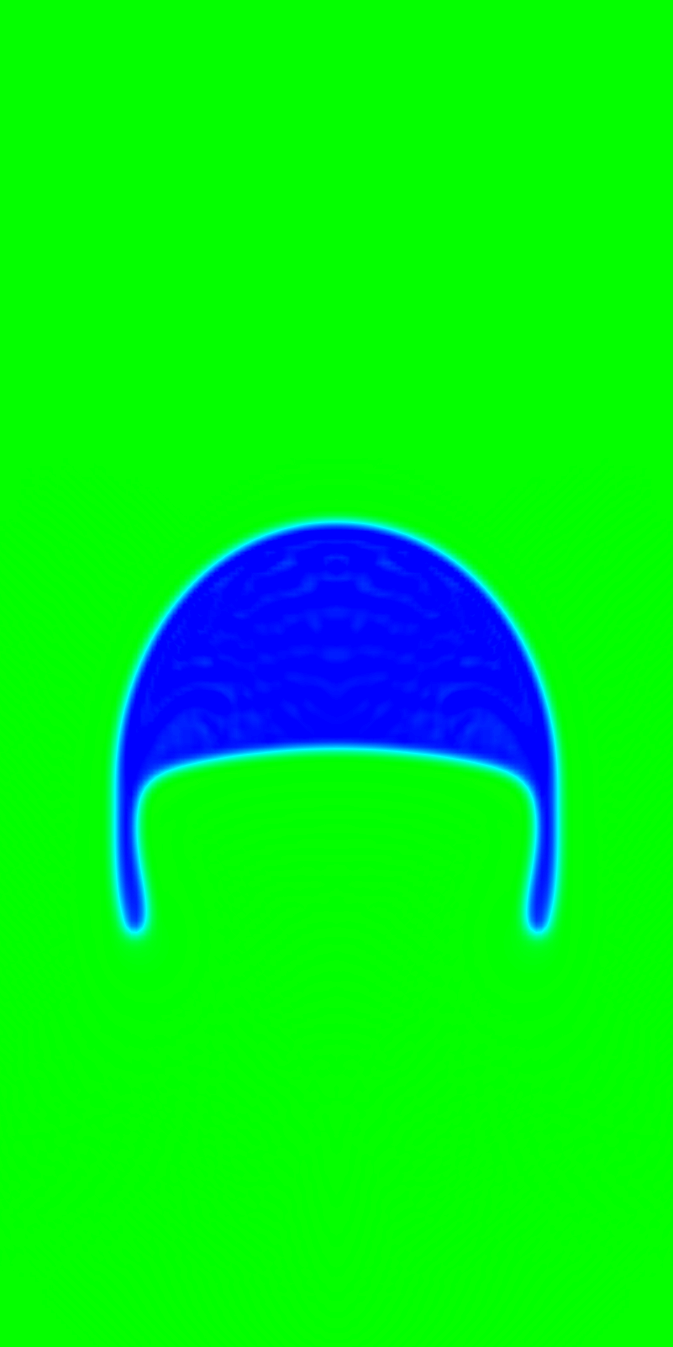}
%\caption{$t=2.4$}
\end{subfigure}
\begin{subfigure}{0.078\textwidth}
\centering
\includegraphics[width=1\textwidth]{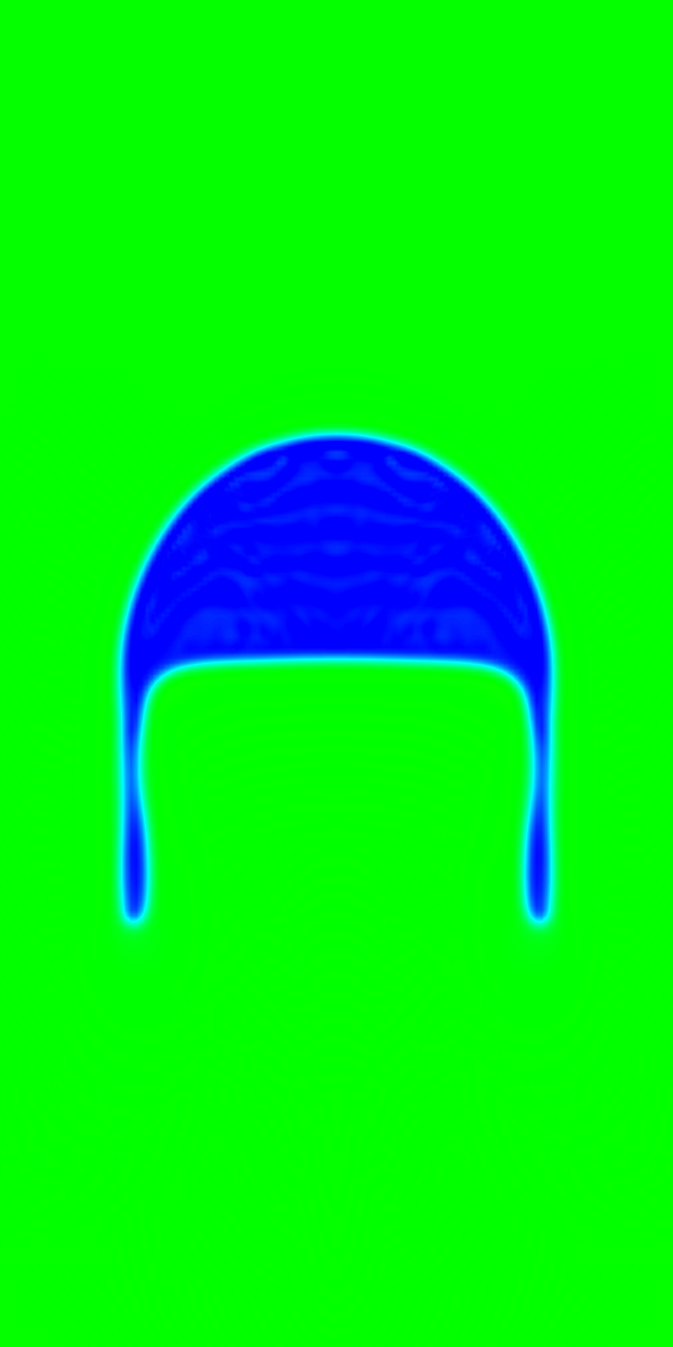}
%\caption{$t=3.0$}
\end{subfigure}
\caption{Mixture-aware simulations -- equal phases. Case 1 (left) and Case 2 (right). Visualization of the phase fields (top to below) $\phi_1, \phi_2, \phi_3$ at times $t=0.0, 0.6, 1.2, 1.8, 2.4, 3.0$ (left to right).}
\label{fig: reduction case 1 2 - phi1=phi3}
\end{figure}

\subsubsection*{Vertically separated identical phases}
Here we again test the merging of identical phases. We choose the initial phase fields as:
\begin{subequations}
  \begin{align}
    \phi^h_{1,0}(\mathbf{x}) =&~ \frac{1}{2}\left(1+\tanh{\dfrac{\sqrt{(x-0.5)^2+(y-0.5)^2}-R_0}{\varepsilon\sqrt{2}}}\right)-\phi^h_{3,0}(\mathbf{x}),\\
    \phi^h_{2,0}(\mathbf{x}) =&~ \frac{1}{2}\left(1-\tanh{\dfrac{\sqrt{(x-0.5)^2+(y-0.5)^2}-R_0}{\varepsilon\sqrt{2}}}\right),\\
    \phi^h_{3,0}(\mathbf{x}) =&~ \frac{1}{2}\left(1+\tanh{\dfrac{y-3.5 R_0}{\varepsilon\sqrt{2}}}\right).    
\end{align}
\end{subequations}
In Figure \ref{fig: reduction case 1 2 - last} we visualize the phase fields for Case 1 and 2. Compared to the previous scenarios, phases $\phi_1$ and $\phi_3$ now show significantly different behavior. However, we observe that, for both cases, the evolution of phase $\phi_2$ remains identical. These observations again confirm axioms A5 (Merging identical phases).

\begin{figure}
\captionsetup[subfigure]{justification=centering}
\begin{subfigure}{0.078\textwidth}
\centering
\includegraphics[width=1\textwidth]{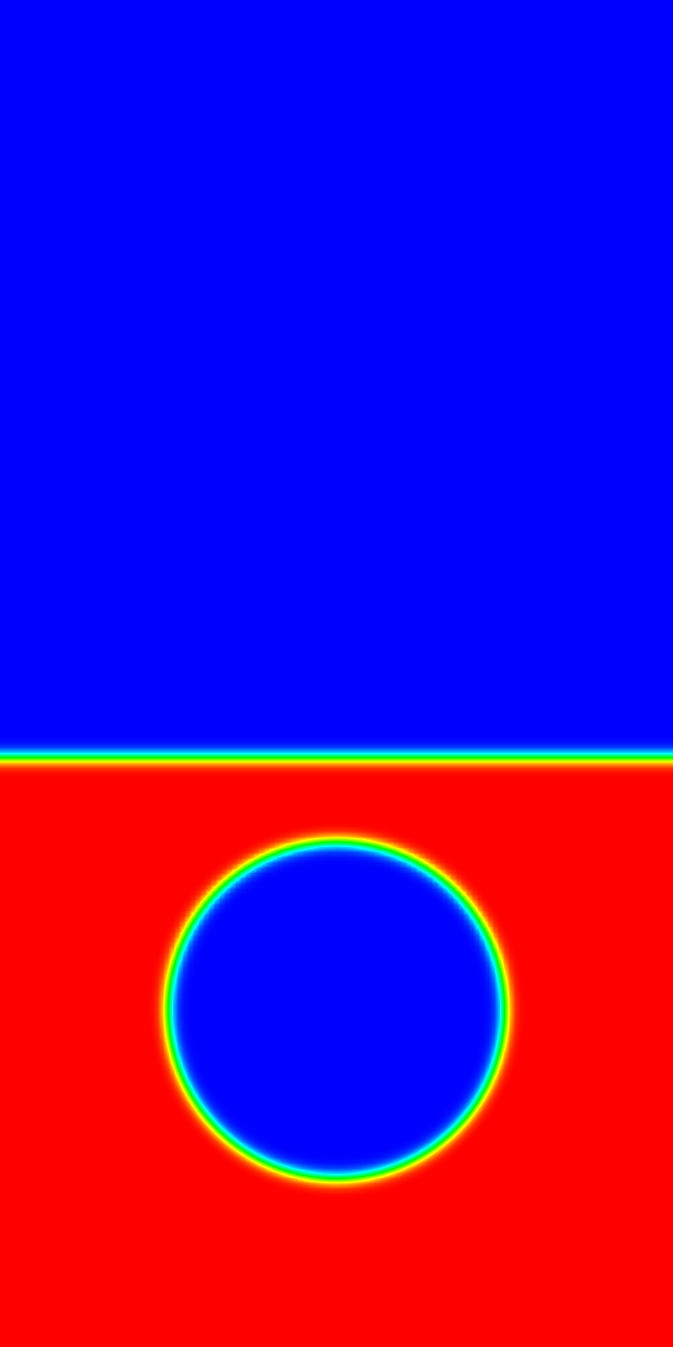}
%\caption{$t=0.0$}
\end{subfigure}
\begin{subfigure}{0.078\textwidth}
\centering
\includegraphics[width=1\textwidth]{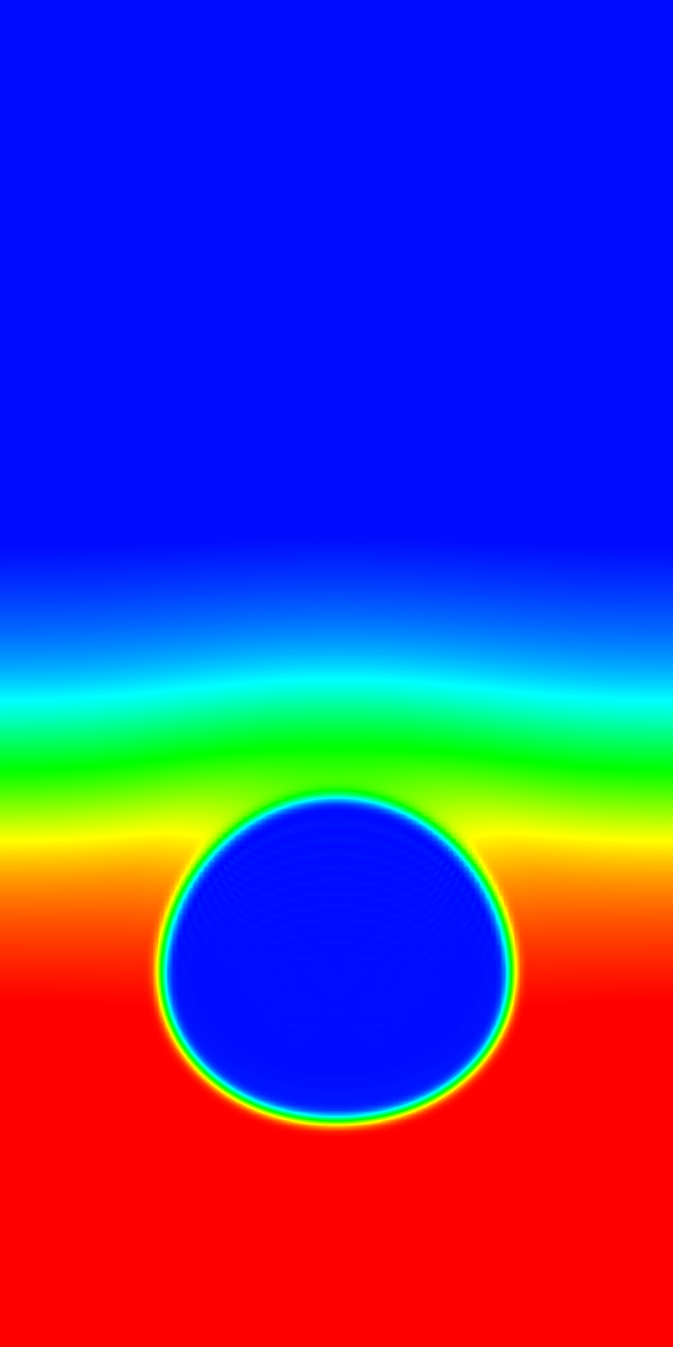}
%\caption{$t=0.6$}
\end{subfigure}
\begin{subfigure}{0.078\textwidth}
\centering
\includegraphics[width=1\textwidth]{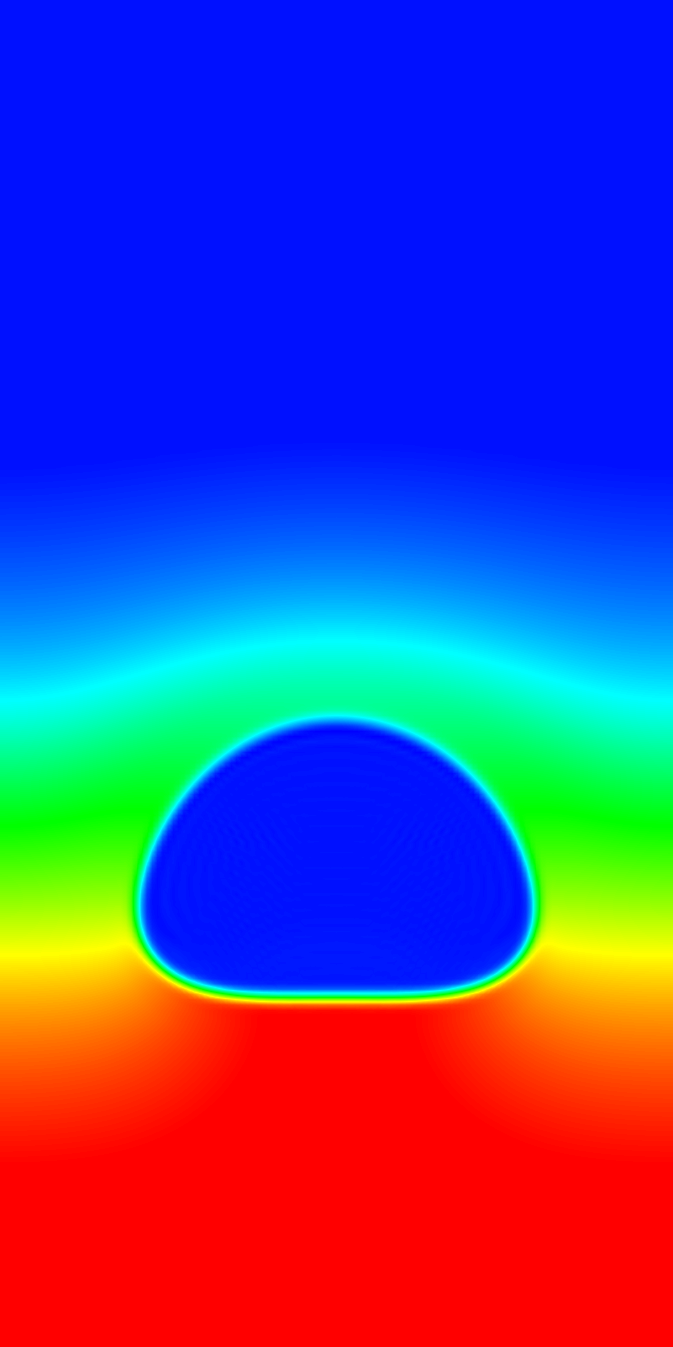}
%\caption{$t=1.2$}
\end{subfigure}
\begin{subfigure}{0.078\textwidth}
\centering
\includegraphics[width=1\textwidth]{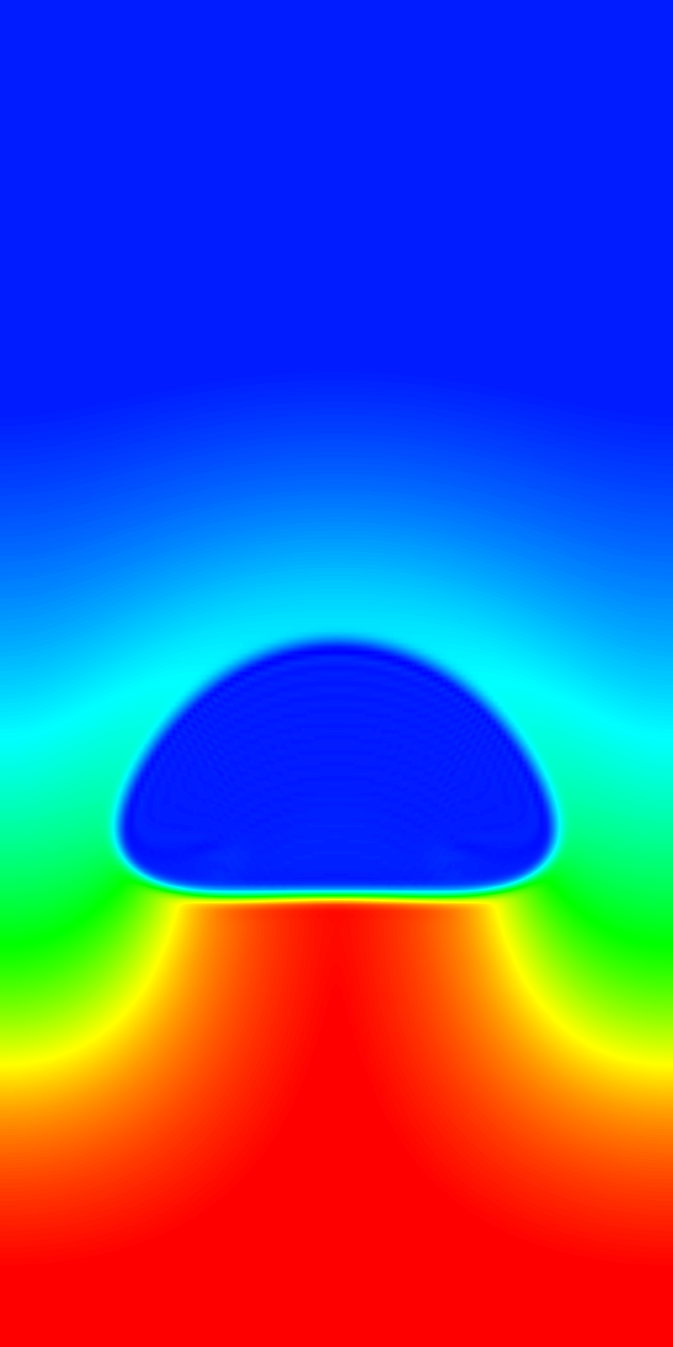}
%\caption{$t=1.8$}
\end{subfigure}
\begin{subfigure}{0.078\textwidth}
\centering
\includegraphics[width=1\textwidth]{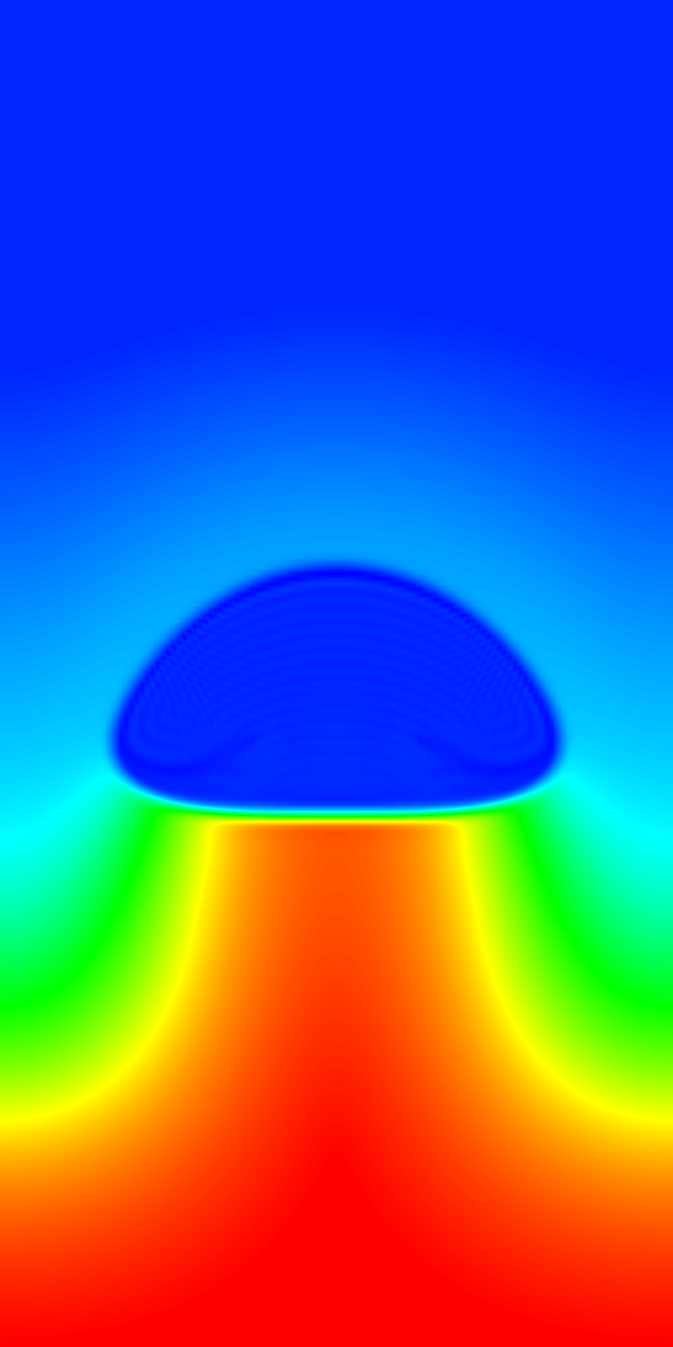}
%\caption{$t=2.4$}
\end{subfigure}
\begin{subfigure}{0.078\textwidth}
\centering
\includegraphics[width=1\textwidth]{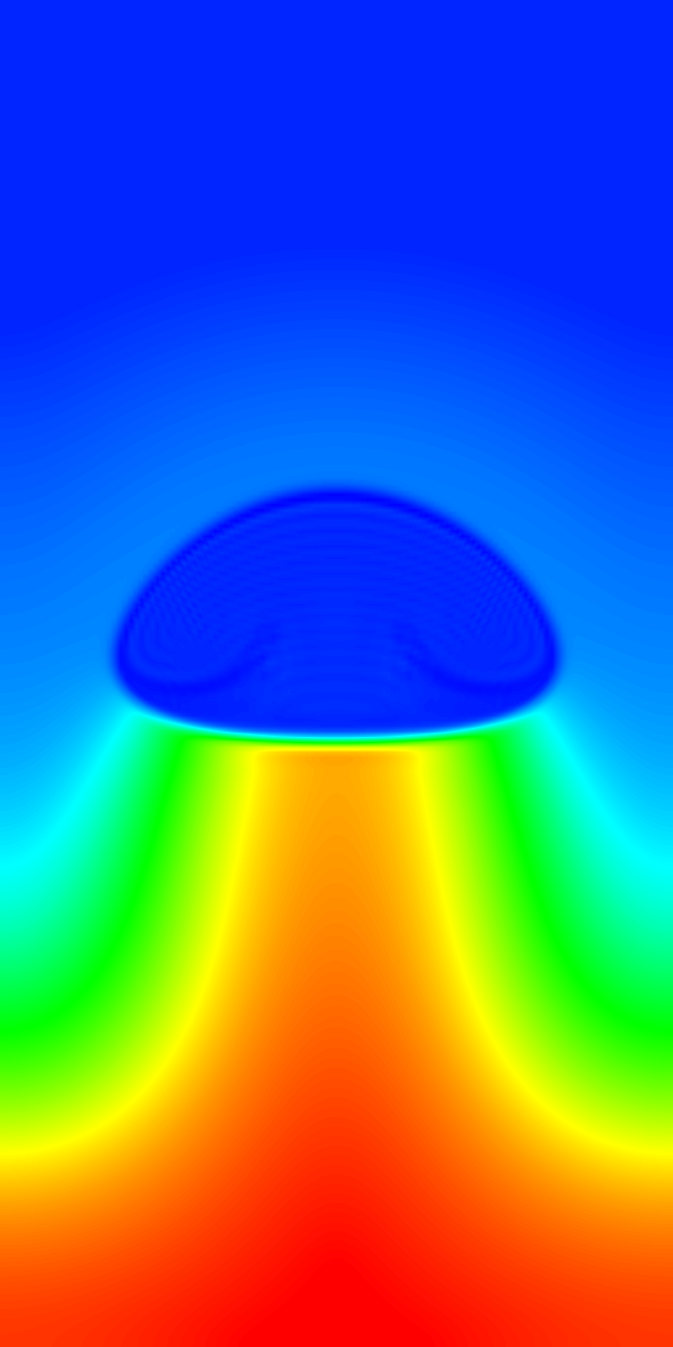}
%\caption{$t=3.0$}
\end{subfigure}
\begin{subfigure}{0.078\textwidth}
\centering
\includegraphics[width=1\textwidth]{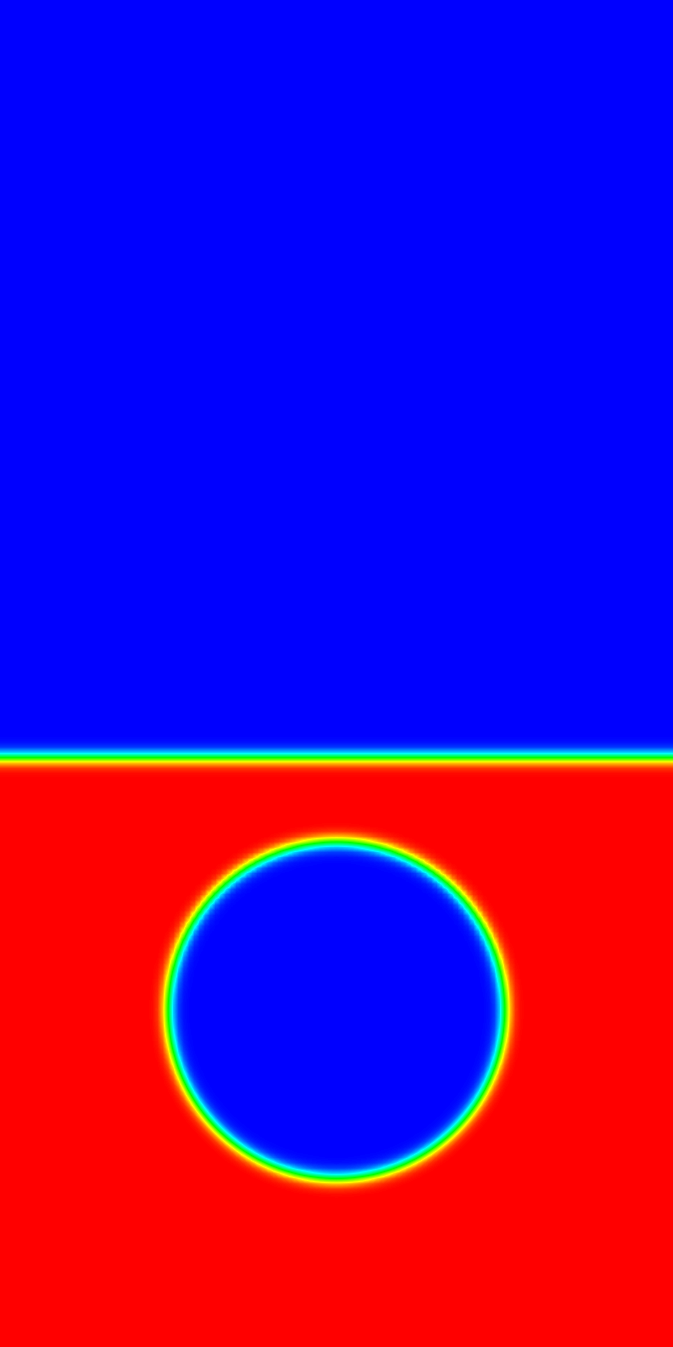}
%\caption{$t=0.0$}
\end{subfigure}
\begin{subfigure}{0.078\textwidth}
\centering
\includegraphics[width=1\textwidth]{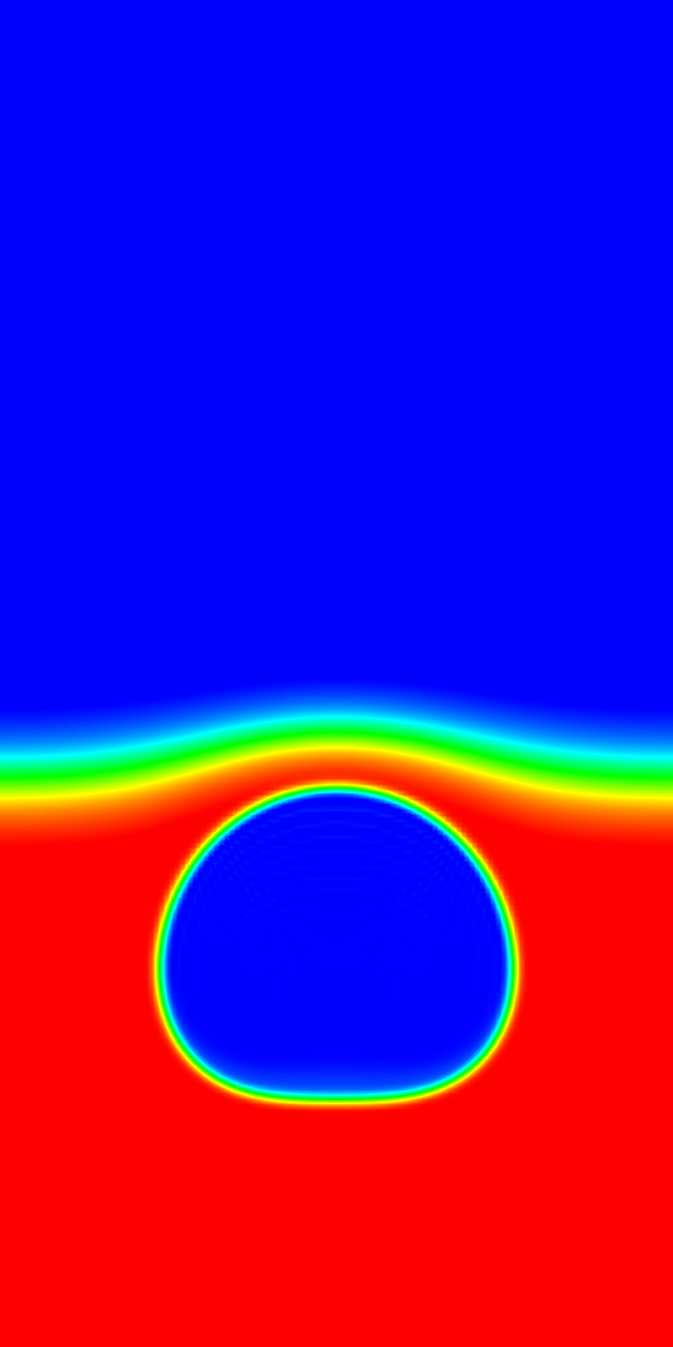}
%\caption{$t=0.6$}
\end{subfigure}
\begin{subfigure}{0.078\textwidth}
\centering
\includegraphics[width=1\textwidth]{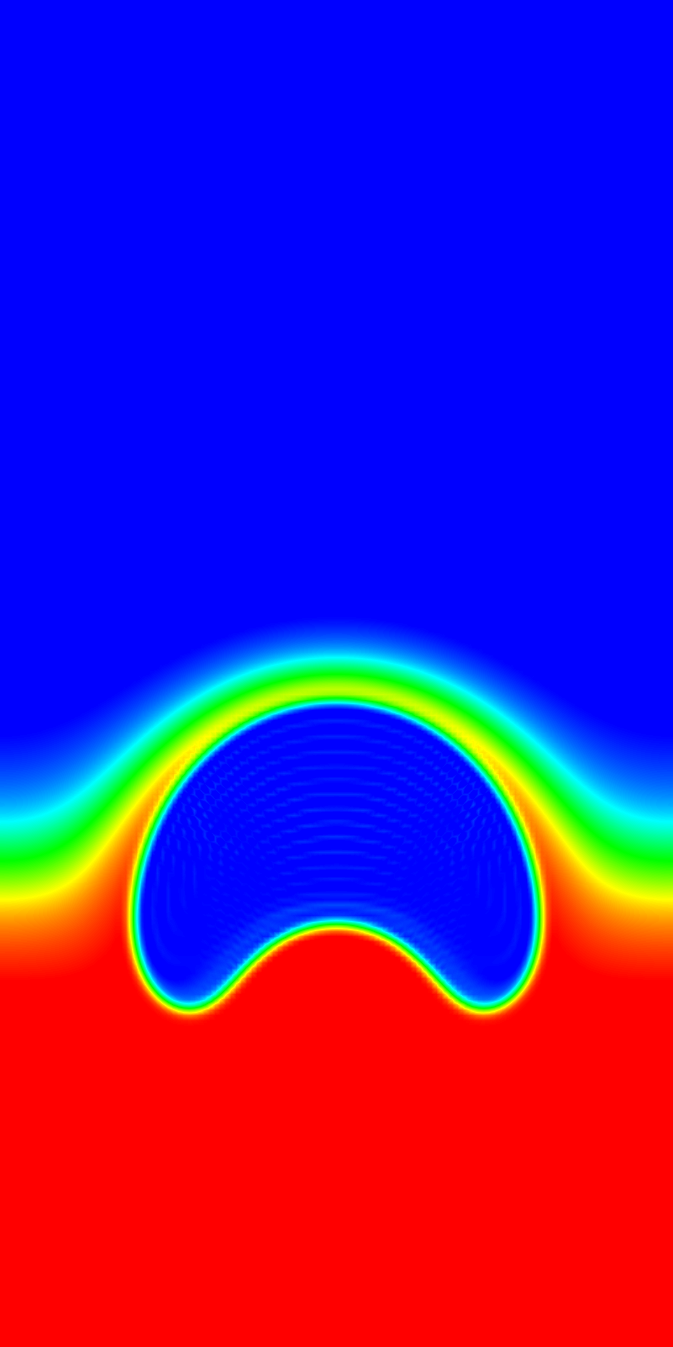}
%\caption{$t=1.2$}
\end{subfigure}
\begin{subfigure}{0.078\textwidth}
\centering
\includegraphics[width=1\textwidth]{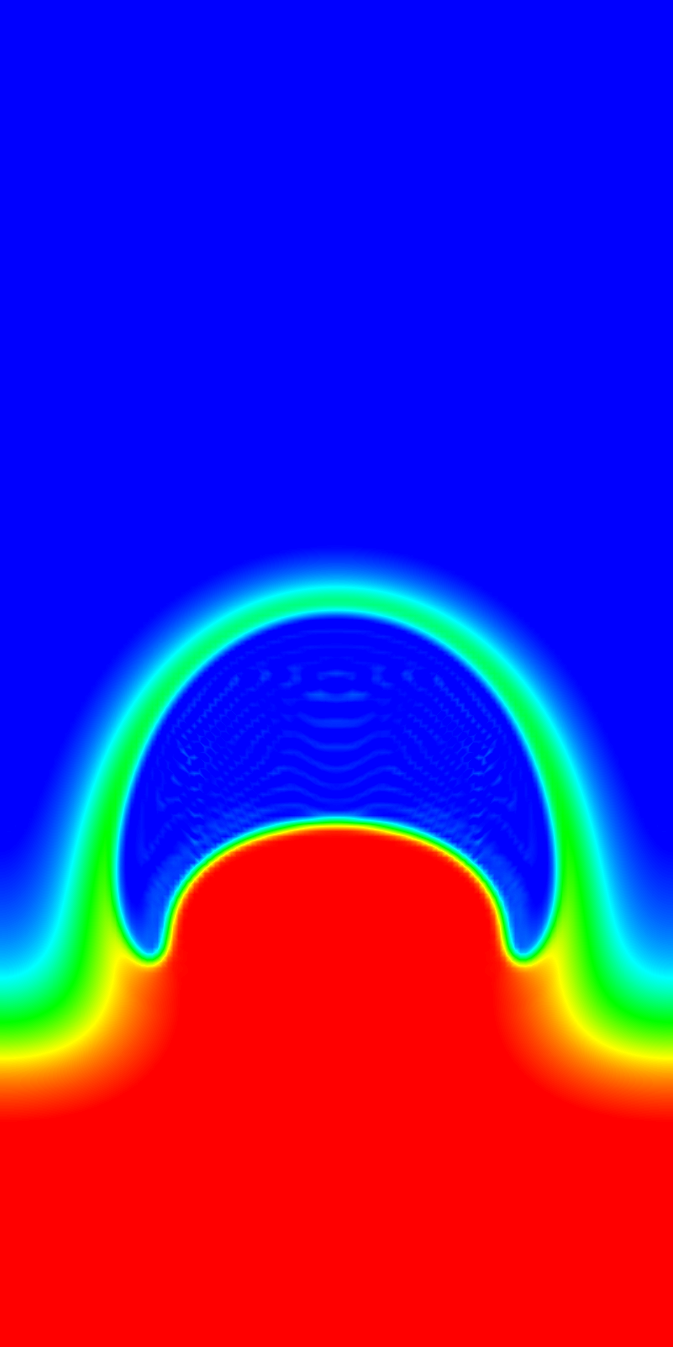}
%\caption{$t=1.8$}
\end{subfigure}
\begin{subfigure}{0.078\textwidth}
\centering
\includegraphics[width=1\textwidth]{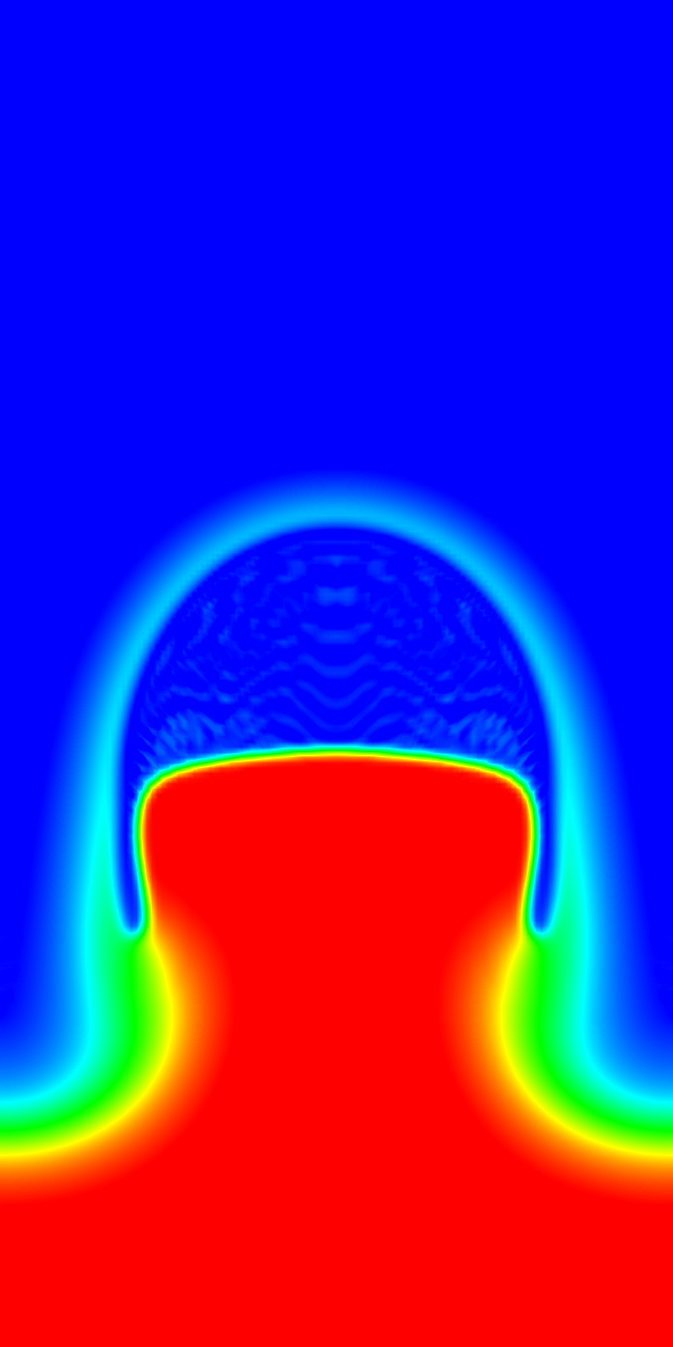}
%\caption{$t=2.4$}
\end{subfigure}
\begin{subfigure}{0.078\textwidth}
\centering
\includegraphics[width=1\textwidth]{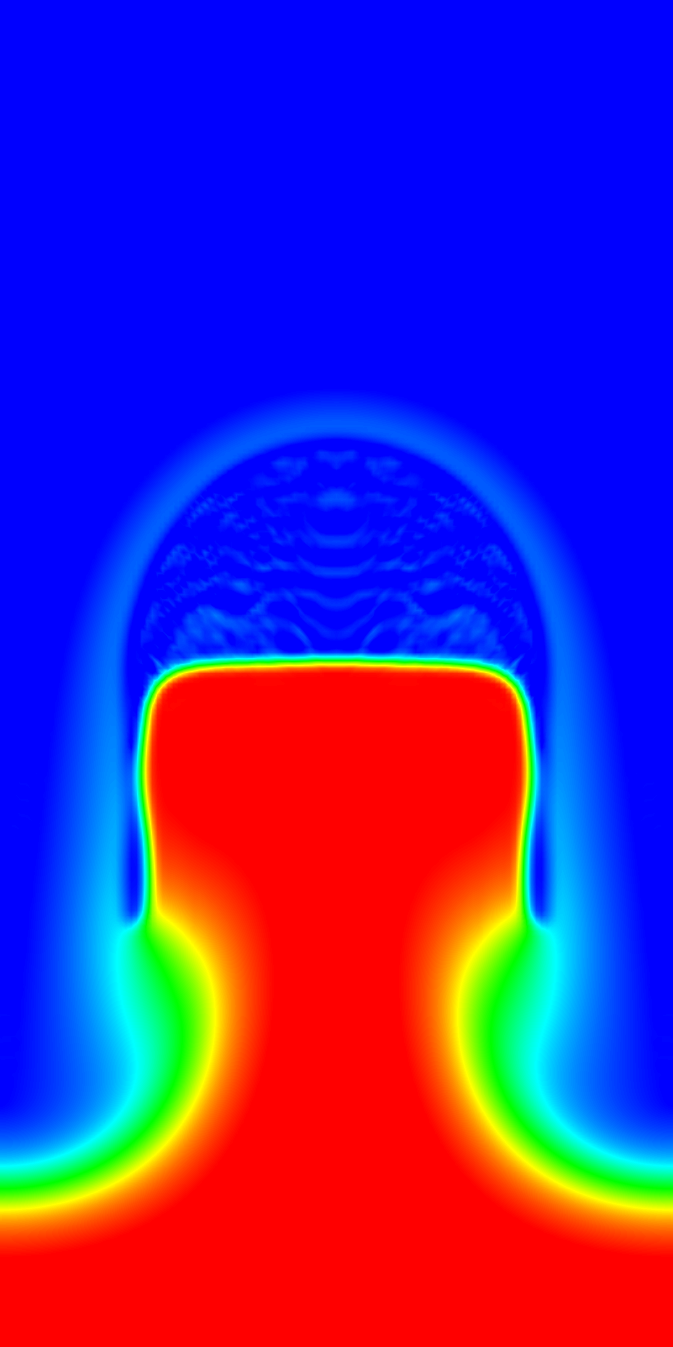}
%\caption{$t=3.0$}
\end{subfigure}
\begin{subfigure}{0.078\textwidth}
\centering
\includegraphics[width=1\textwidth]{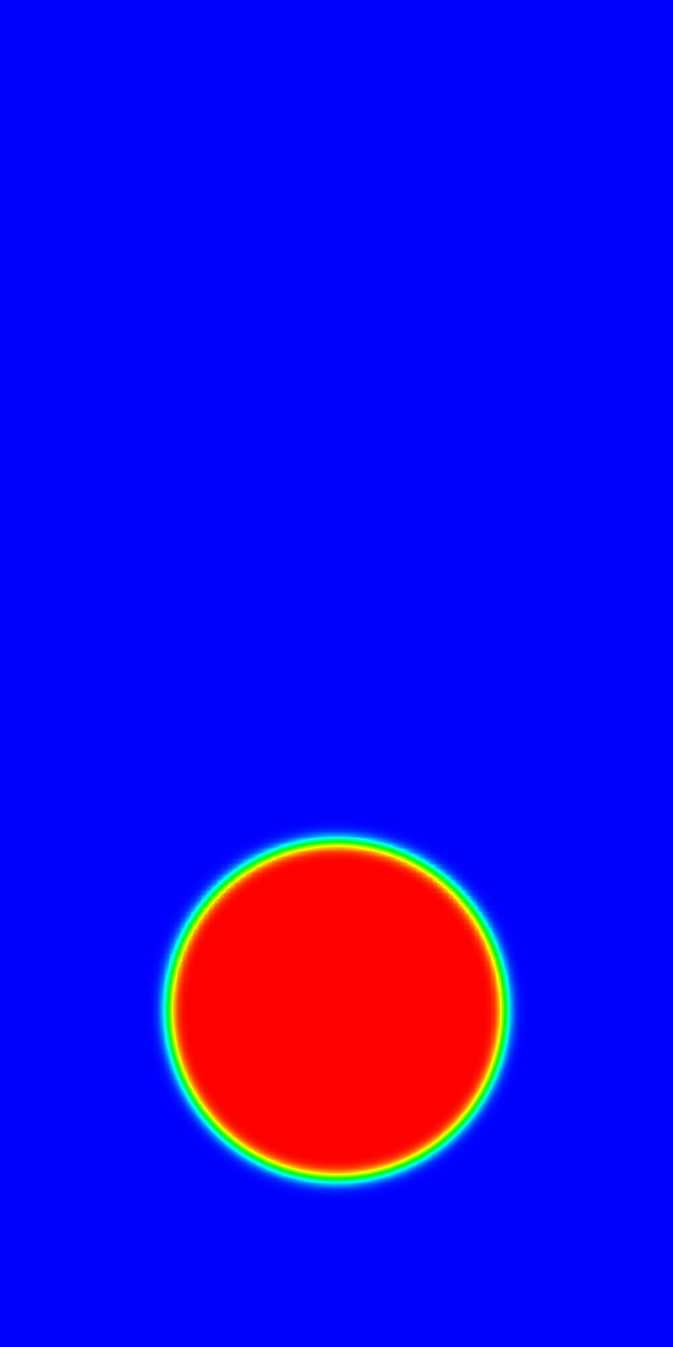}
%\caption{$t=0.0$}
\end{subfigure}
\begin{subfigure}{0.078\textwidth}
\centering
\includegraphics[width=1\textwidth]{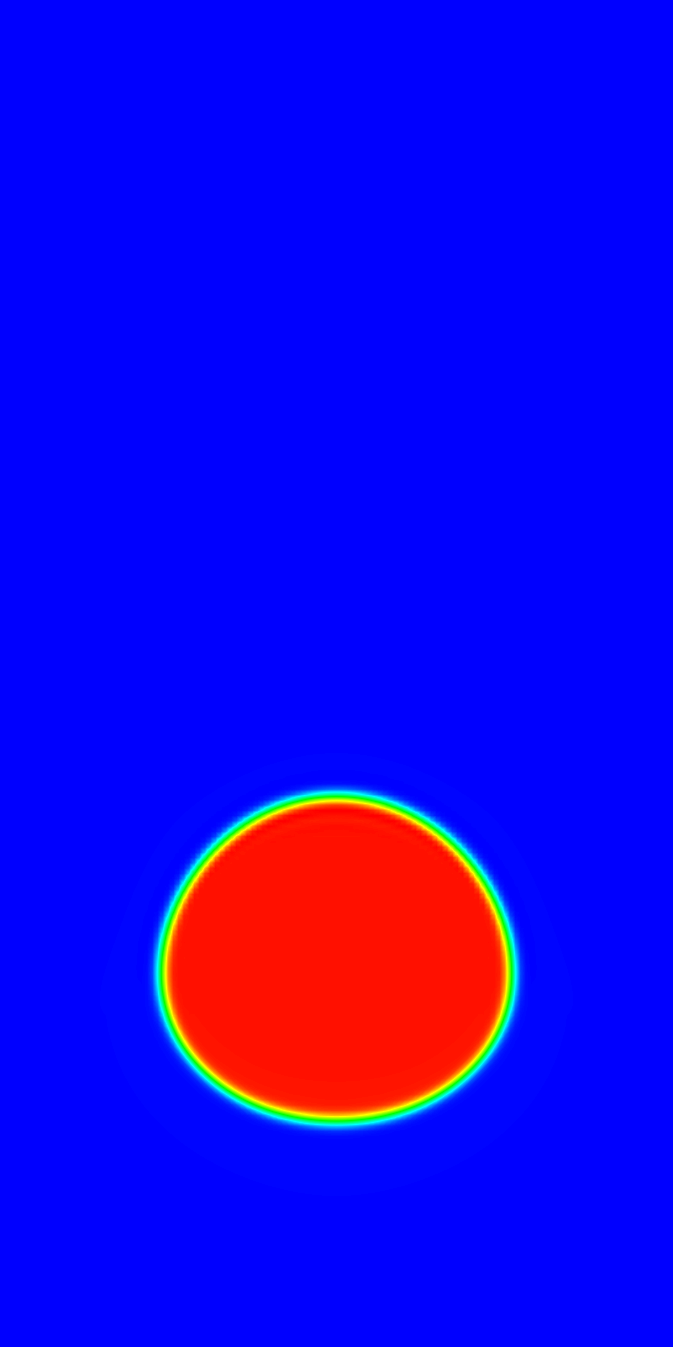}
%\caption{$t=0.6$}
\end{subfigure}
\begin{subfigure}{0.078\textwidth}
\centering
\includegraphics[width=1\textwidth]{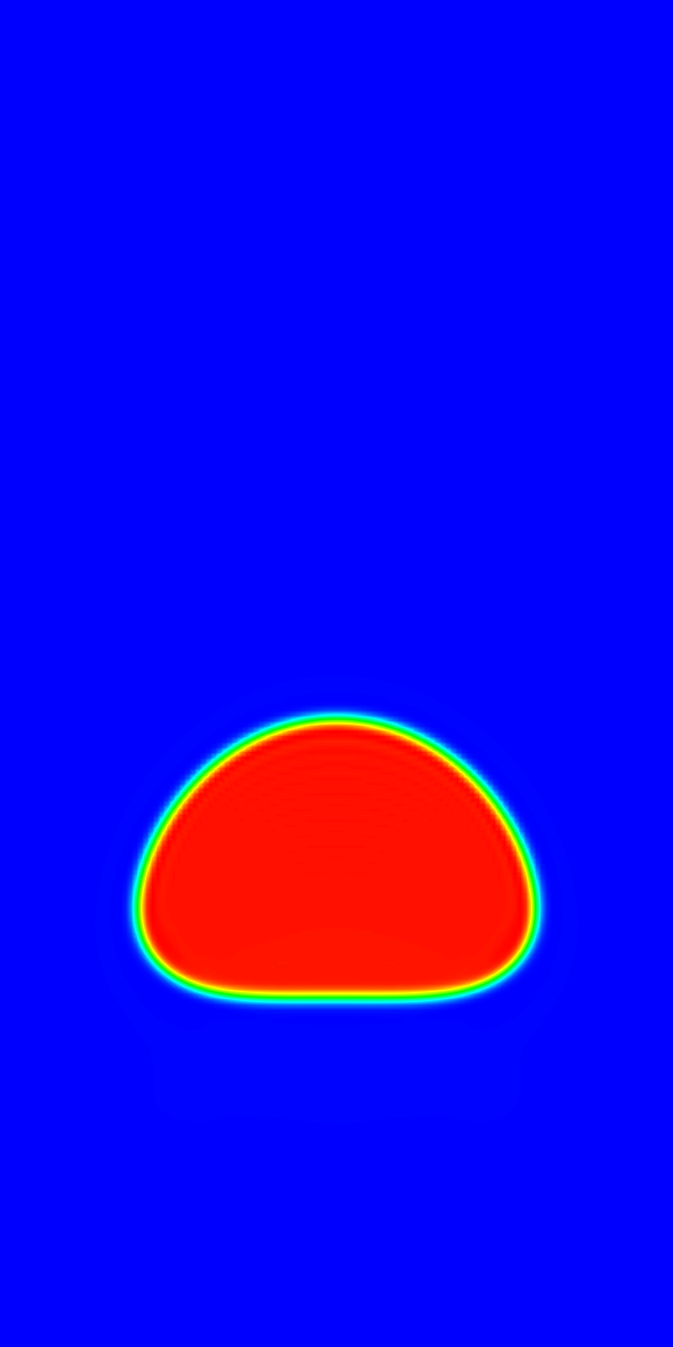}
%\caption{$t=1.2$}
\end{subfigure}
\begin{subfigure}{0.078\textwidth}
\centering
\includegraphics[width=1\textwidth]{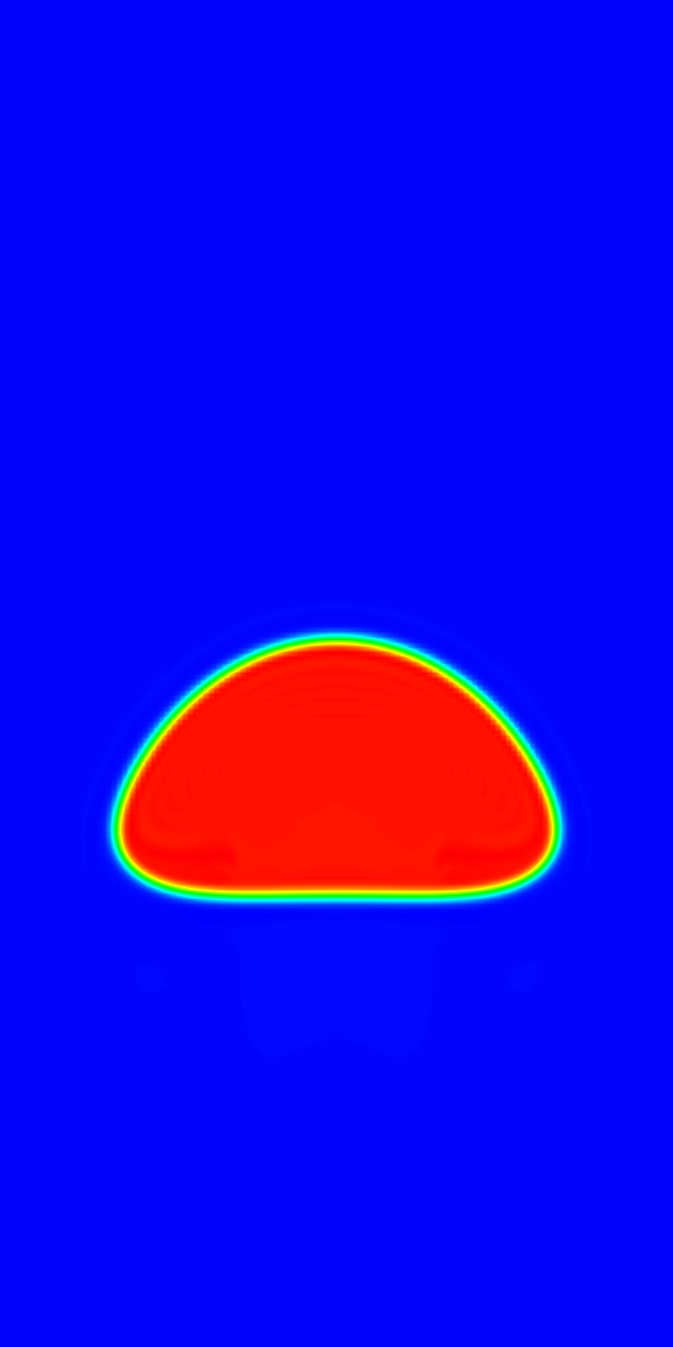}
%\caption{$t=1.8$}
\end{subfigure}
\begin{subfigure}{0.078\textwidth}
\centering
\includegraphics[width=1\textwidth]{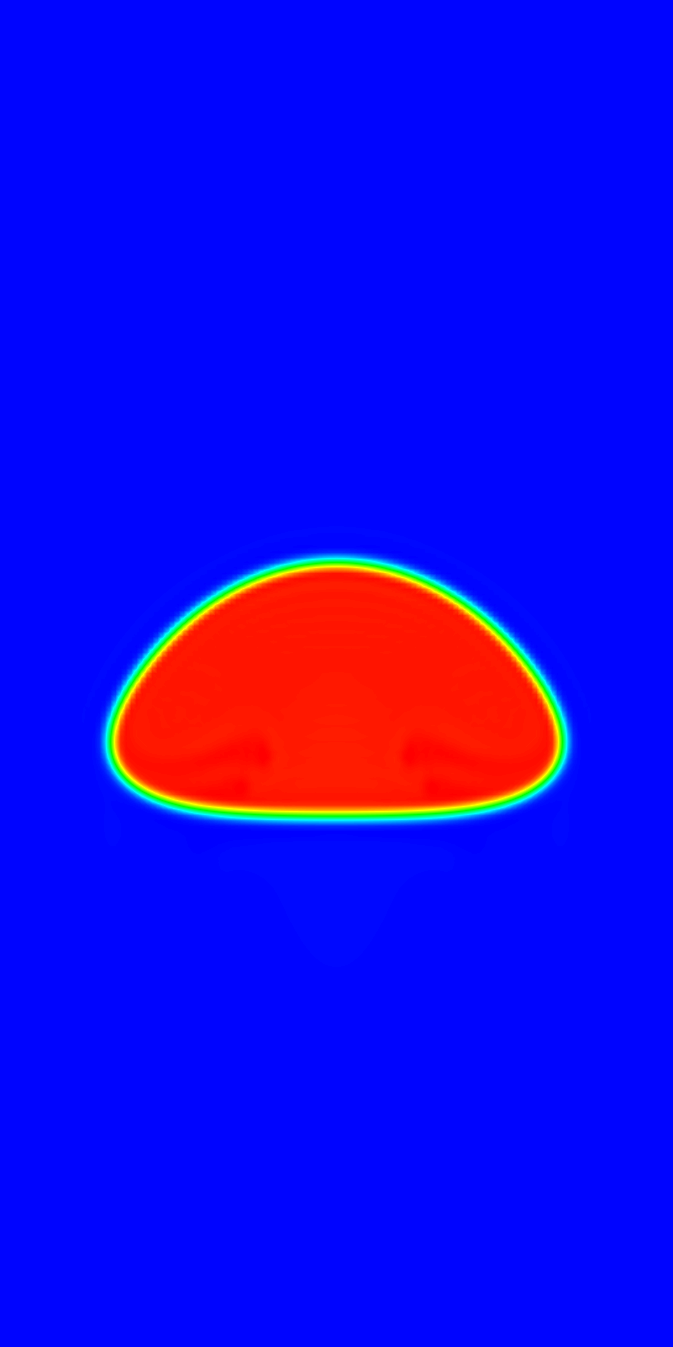}
%\caption{$t=2.4$}
\end{subfigure}
\begin{subfigure}{0.078\textwidth}
\centering
\includegraphics[width=1\textwidth]{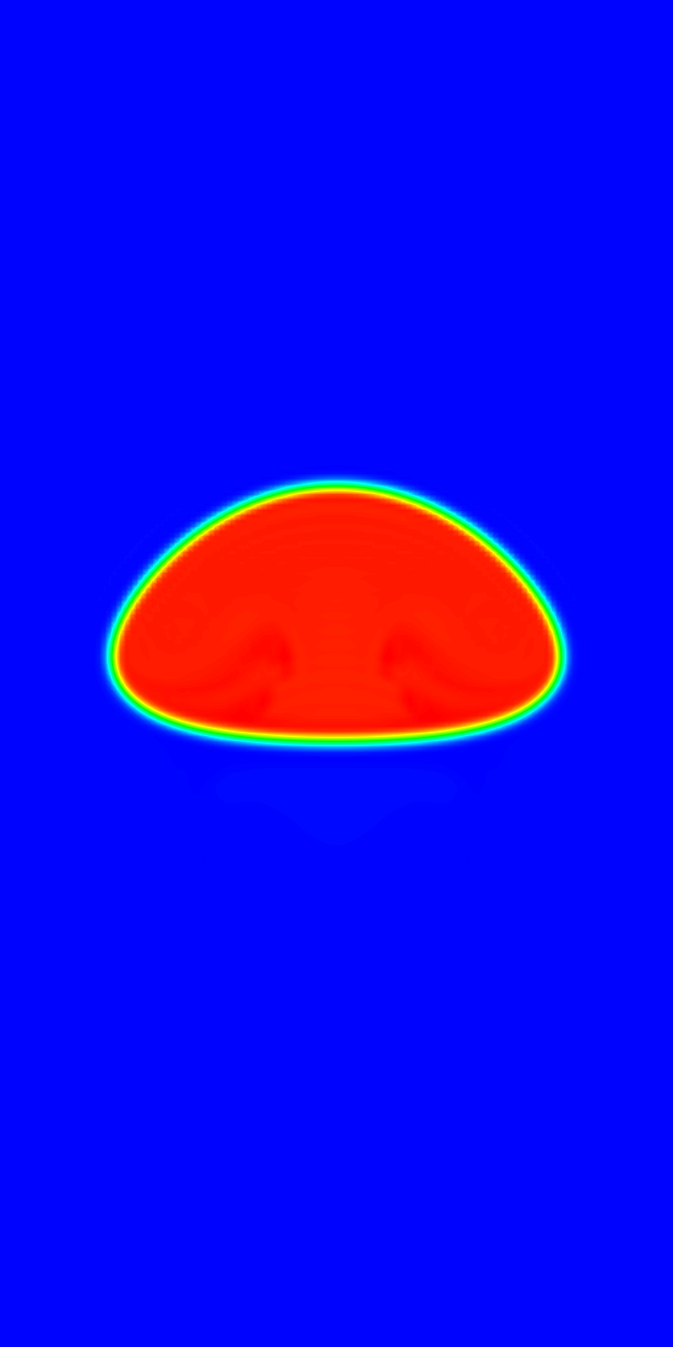}
%\caption{$t=3.0$}
\end{subfigure}
\begin{subfigure}{0.078\textwidth}
\centering
\includegraphics[width=1\textwidth]{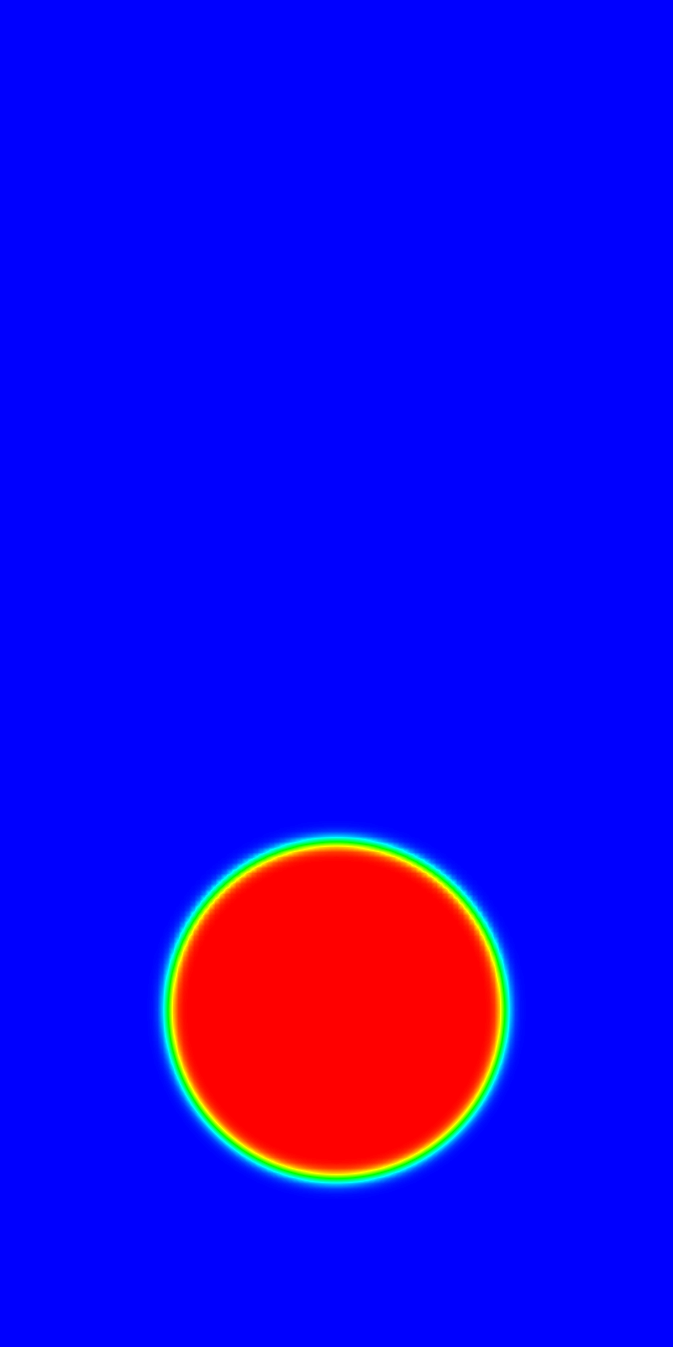}
%\caption{$t=0.0$}
\end{subfigure}
\begin{subfigure}{0.078\textwidth}
\centering
\includegraphics[width=1\textwidth]{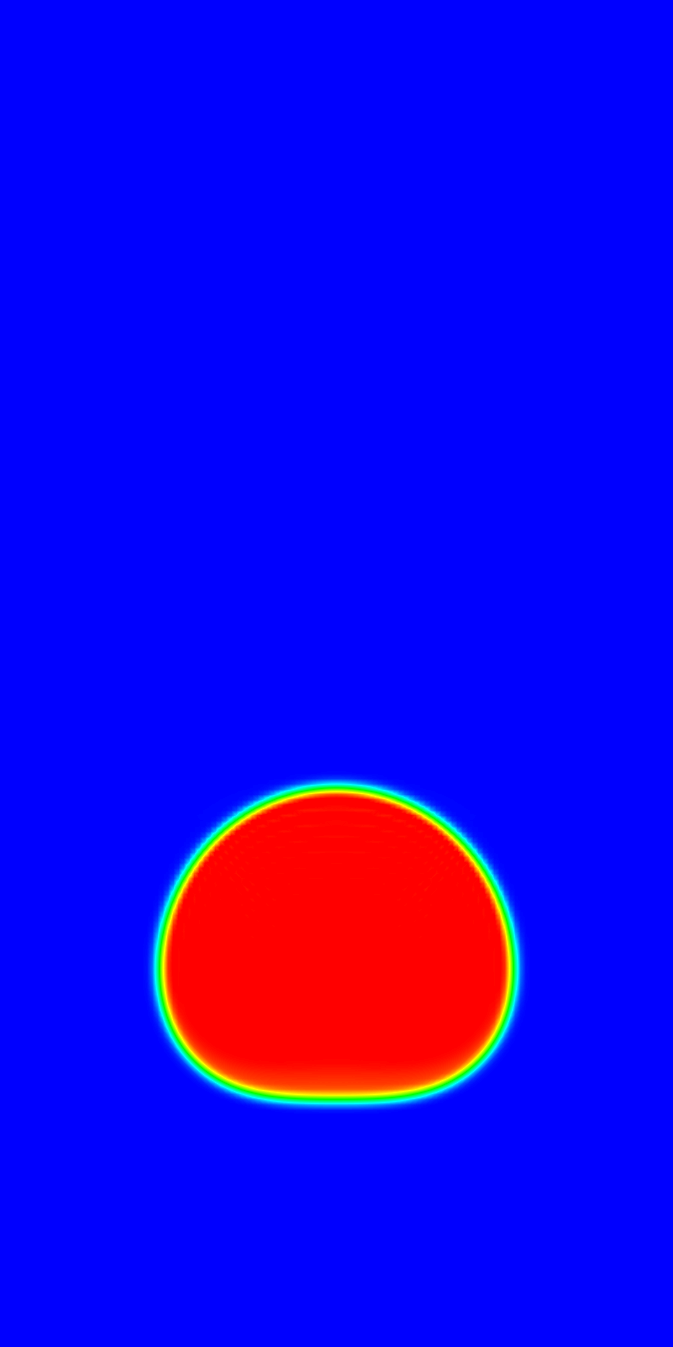}
%\caption{$t=0.6$}
\end{subfigure}
\begin{subfigure}{0.078\textwidth}
\centering
\includegraphics[width=1\textwidth]{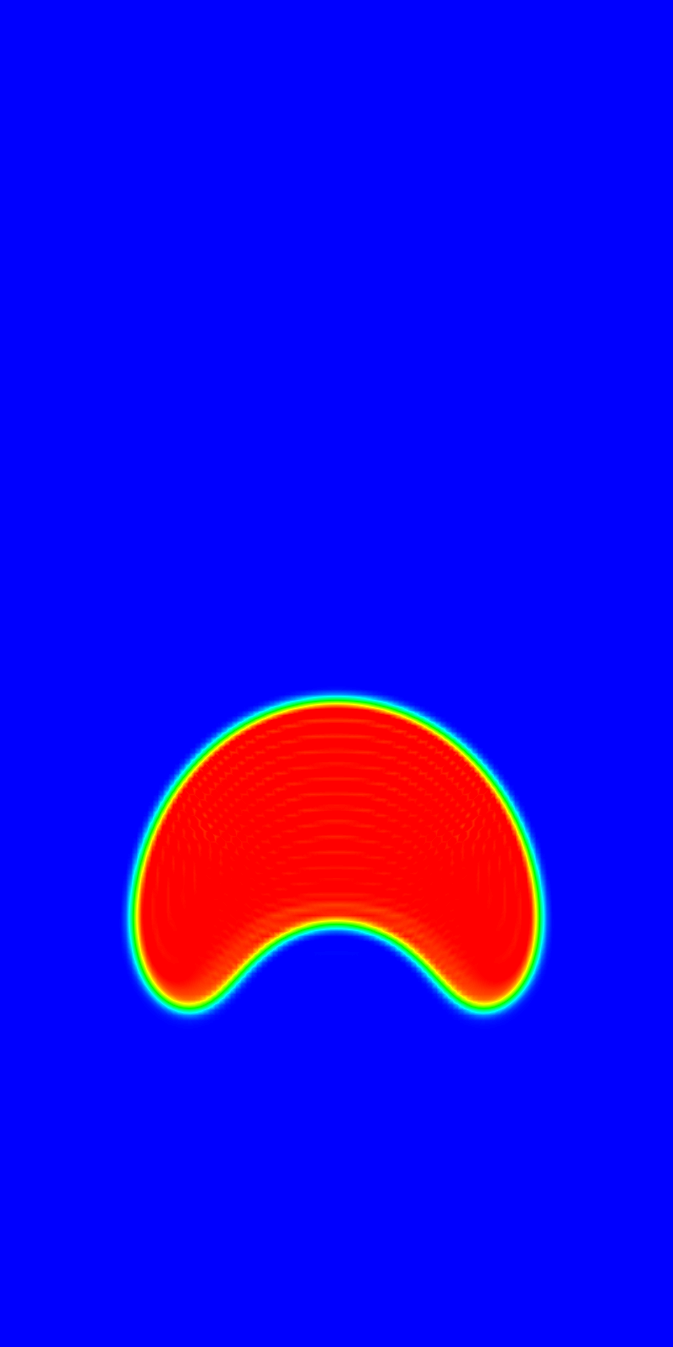}
%\caption{$t=1.2$}
\end{subfigure}
\begin{subfigure}{0.078\textwidth}
\centering
\includegraphics[width=1\textwidth]{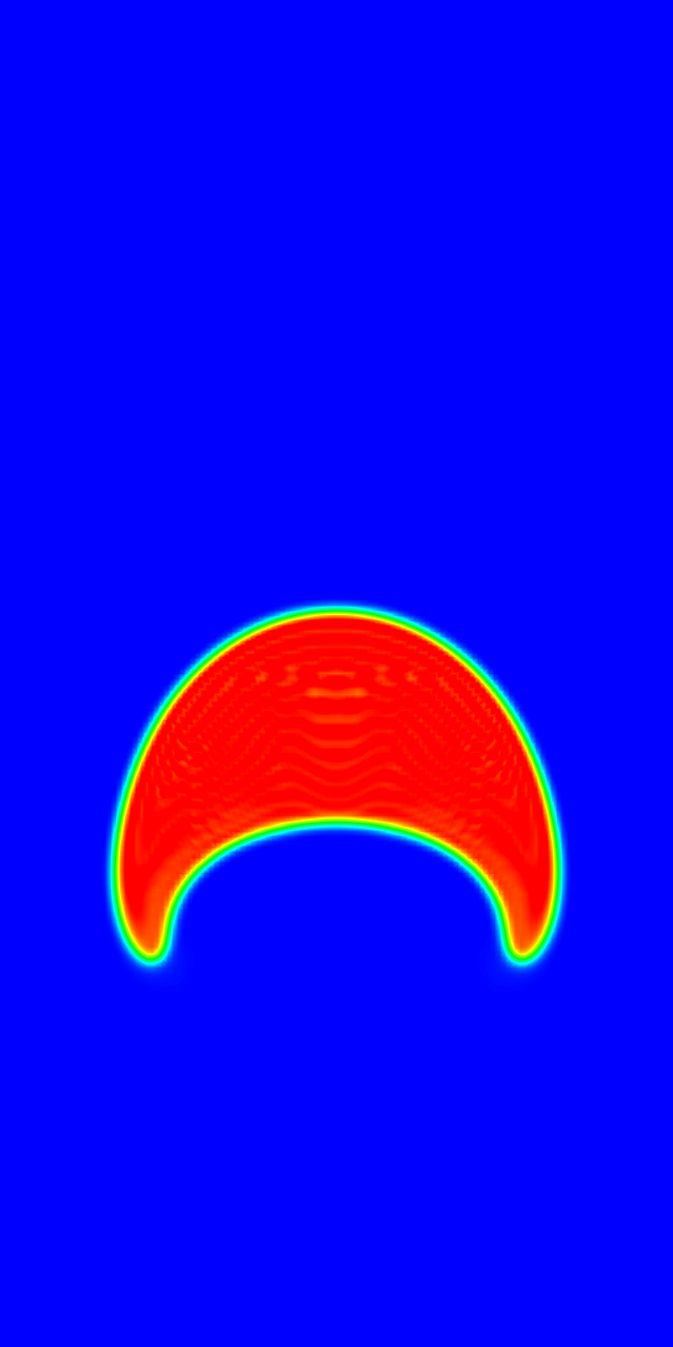}
%\caption{$t=1.8$}
\end{subfigure}
\begin{subfigure}{0.078\textwidth}
\centering
\includegraphics[width=1\textwidth]{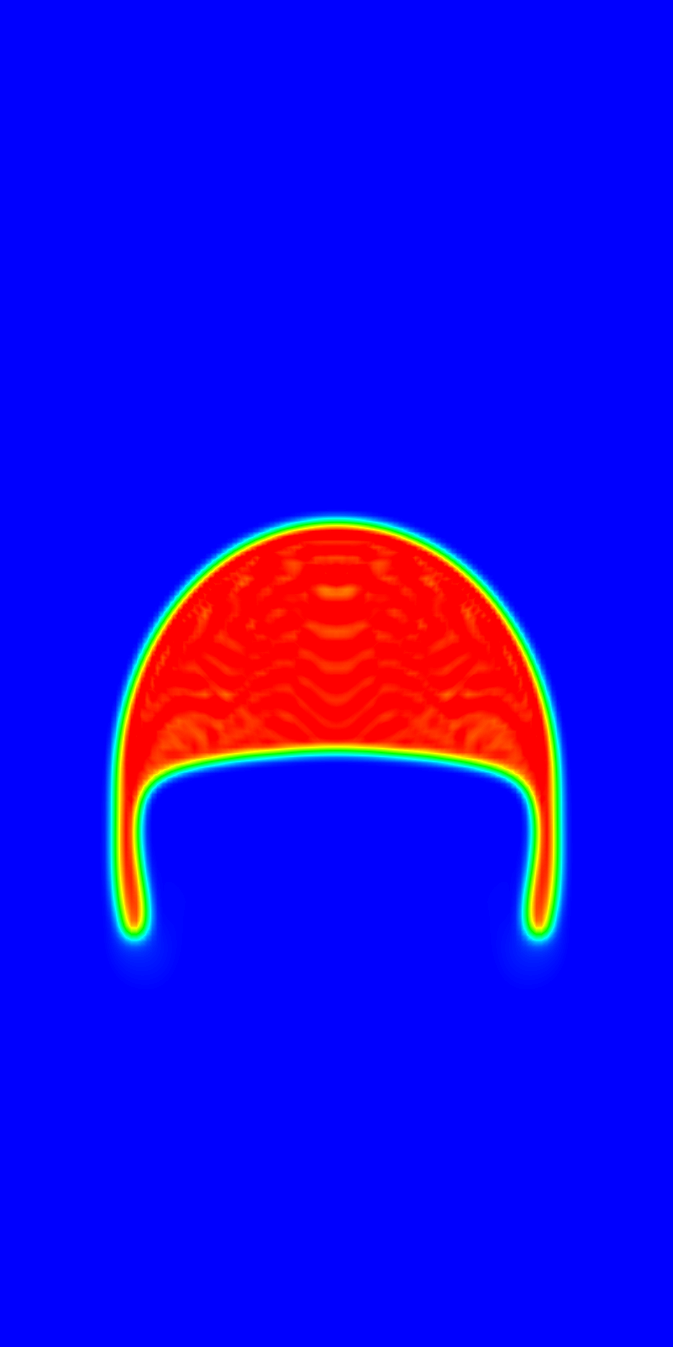}
%\caption{$t=2.4$}
\end{subfigure}
\begin{subfigure}{0.078\textwidth}
\centering
\includegraphics[width=1\textwidth]{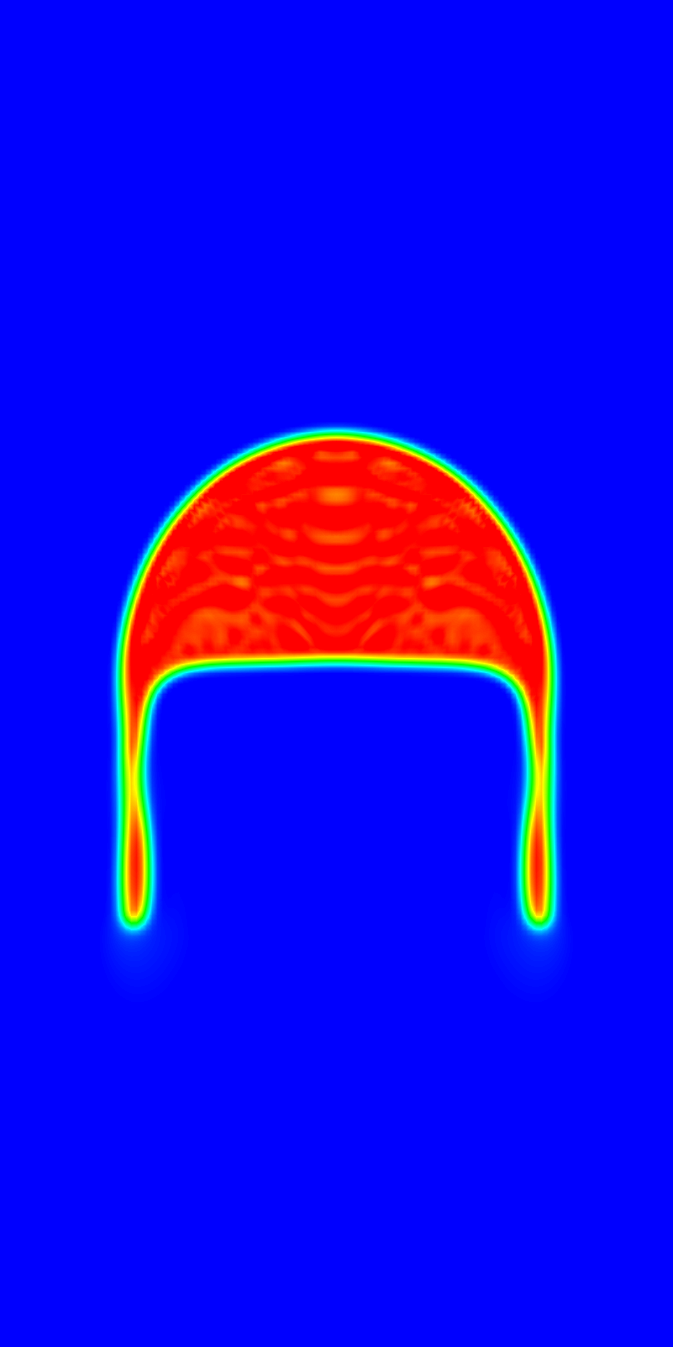}
%\caption{$t=3.0$}
\end{subfigure}
\begin{subfigure}{0.078\textwidth}
\centering
\includegraphics[width=1\textwidth]{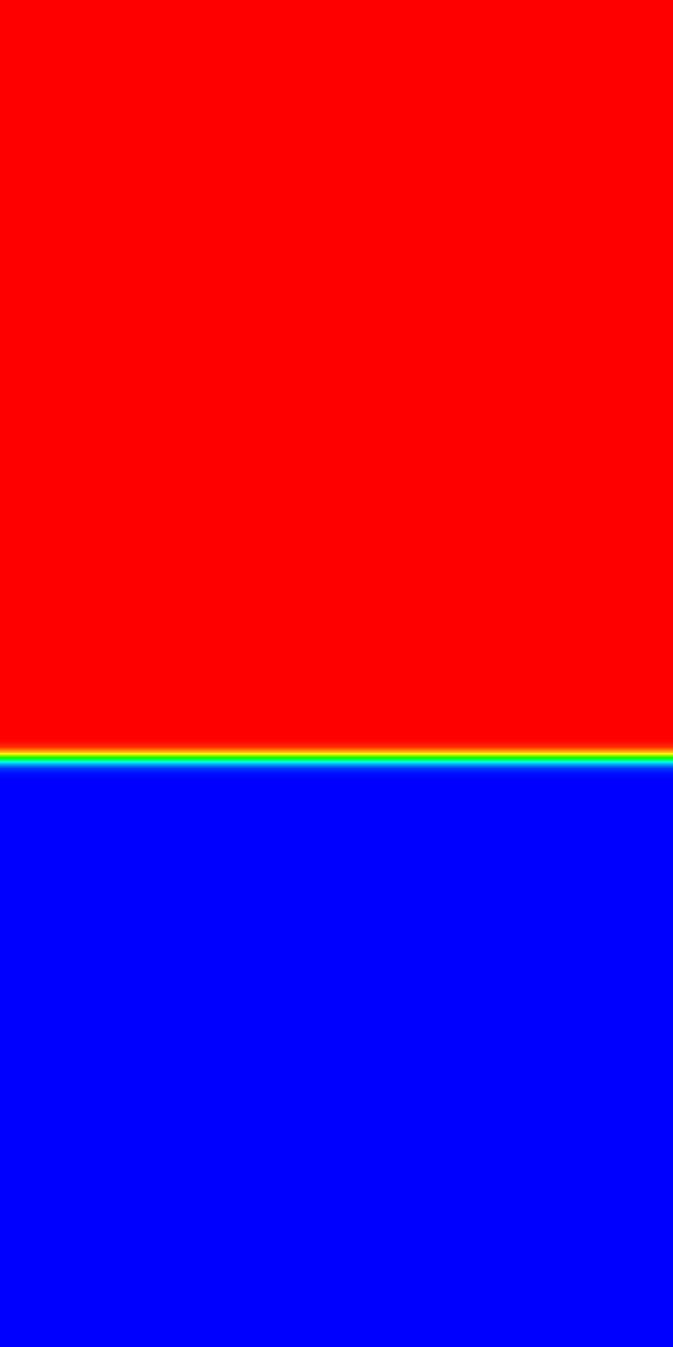}
%\caption{$t=0.0$}
\end{subfigure}
\begin{subfigure}{0.078\textwidth}
\centering
\includegraphics[width=1\textwidth]{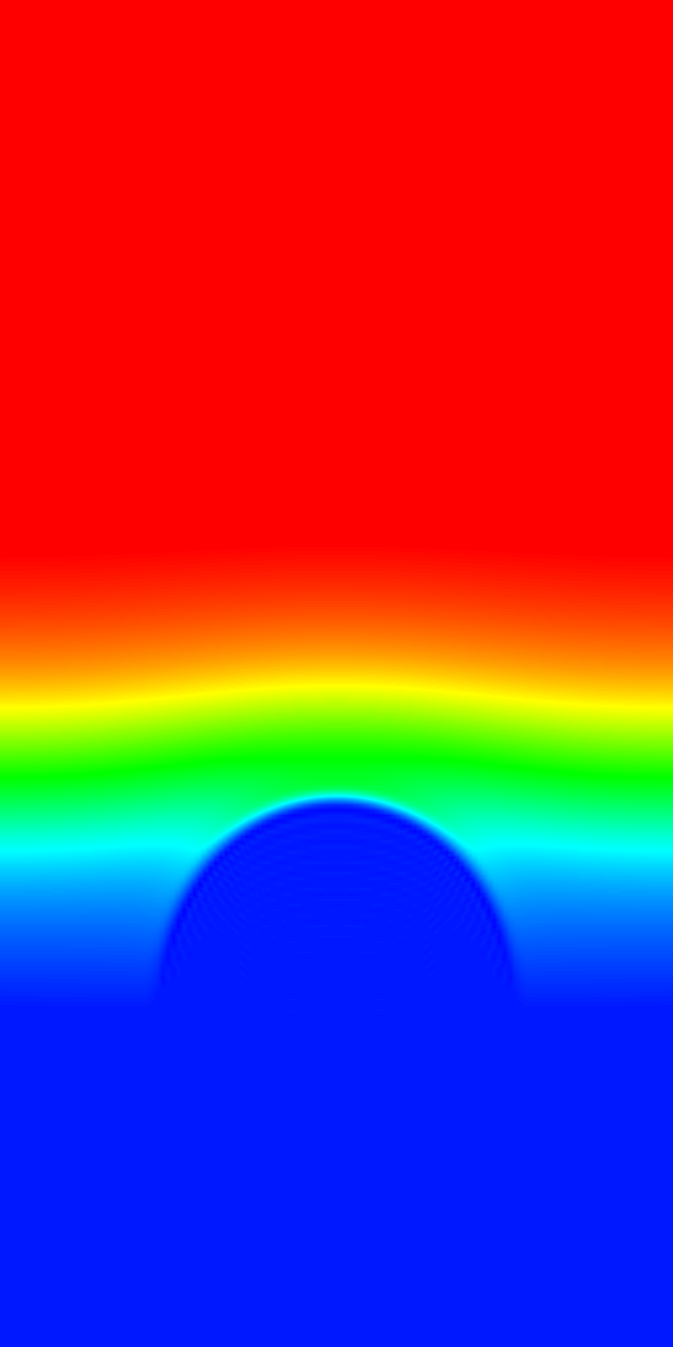}
%\caption{$t=0.6$}
\end{subfigure}
\begin{subfigure}{0.078\textwidth}
\centering
\includegraphics[width=1\textwidth]{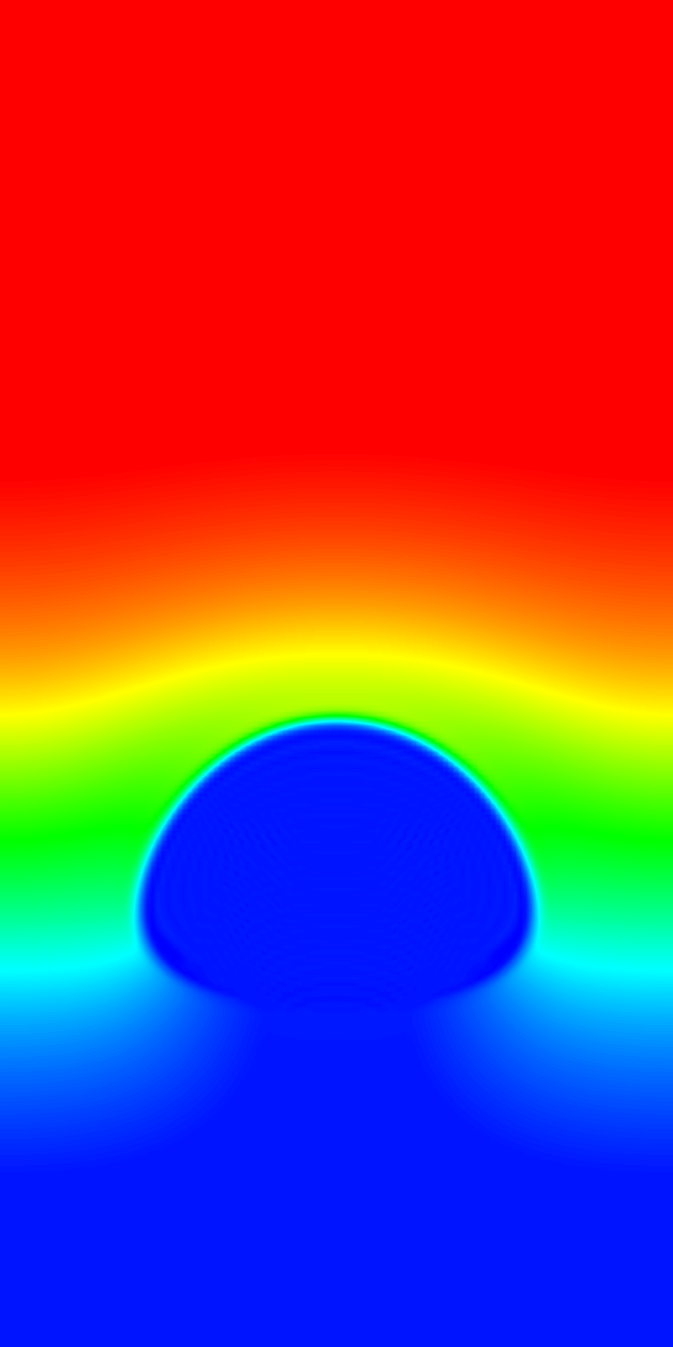}
%\caption{$t=1.2$}
\end{subfigure}
\begin{subfigure}{0.078\textwidth}
\centering
\includegraphics[width=1\textwidth]{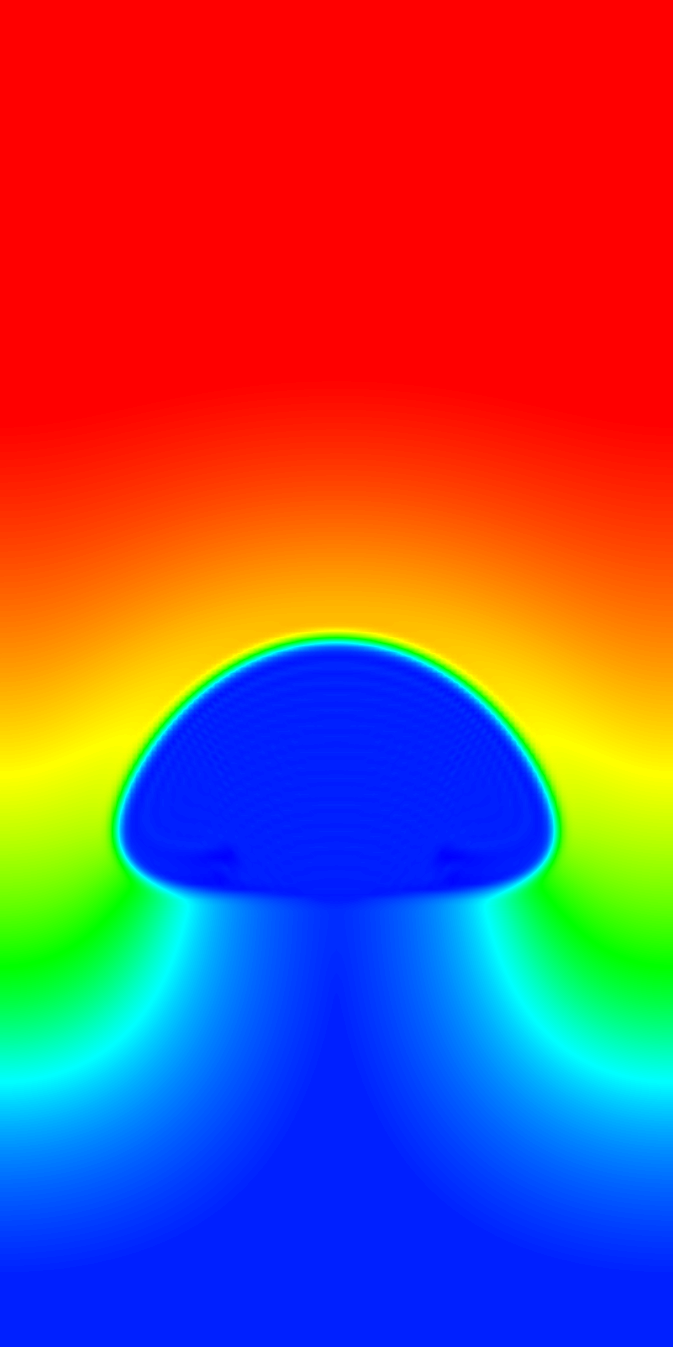}
%\caption{$t=1.8$}
\end{subfigure}
\begin{subfigure}{0.078\textwidth}
\centering
\includegraphics[width=1\textwidth]{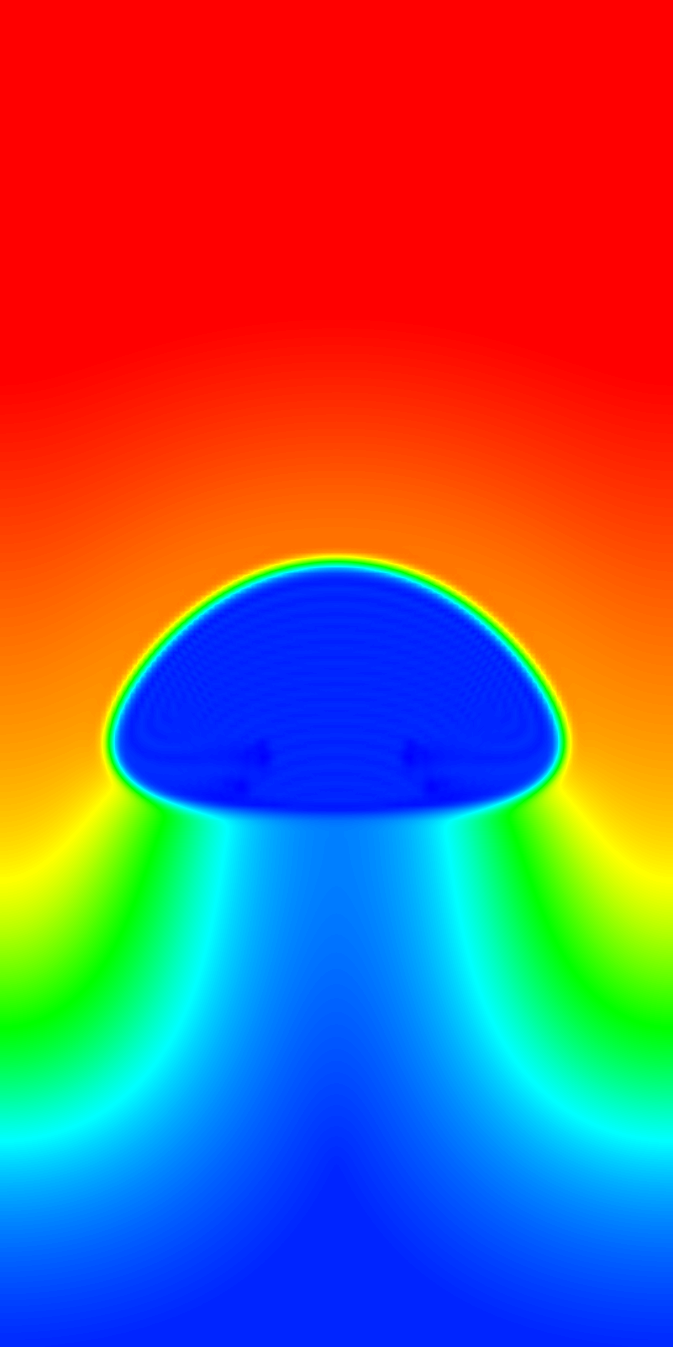}
%\caption{$t=2.4$}
\end{subfigure}
\begin{subfigure}{0.078\textwidth}
\centering
\includegraphics[width=1\textwidth]{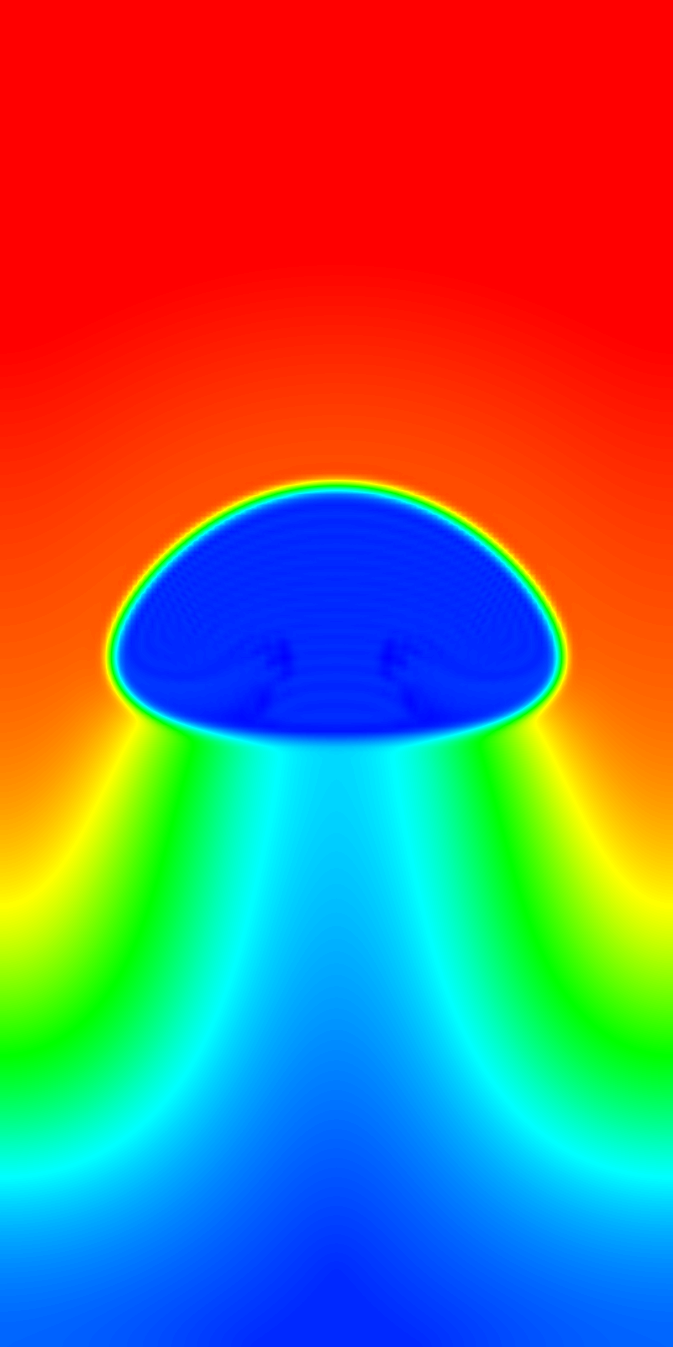}
%\caption{$t=3.0$}
\end{subfigure}
\begin{subfigure}{0.078\textwidth}
\centering
\includegraphics[width=1\textwidth]{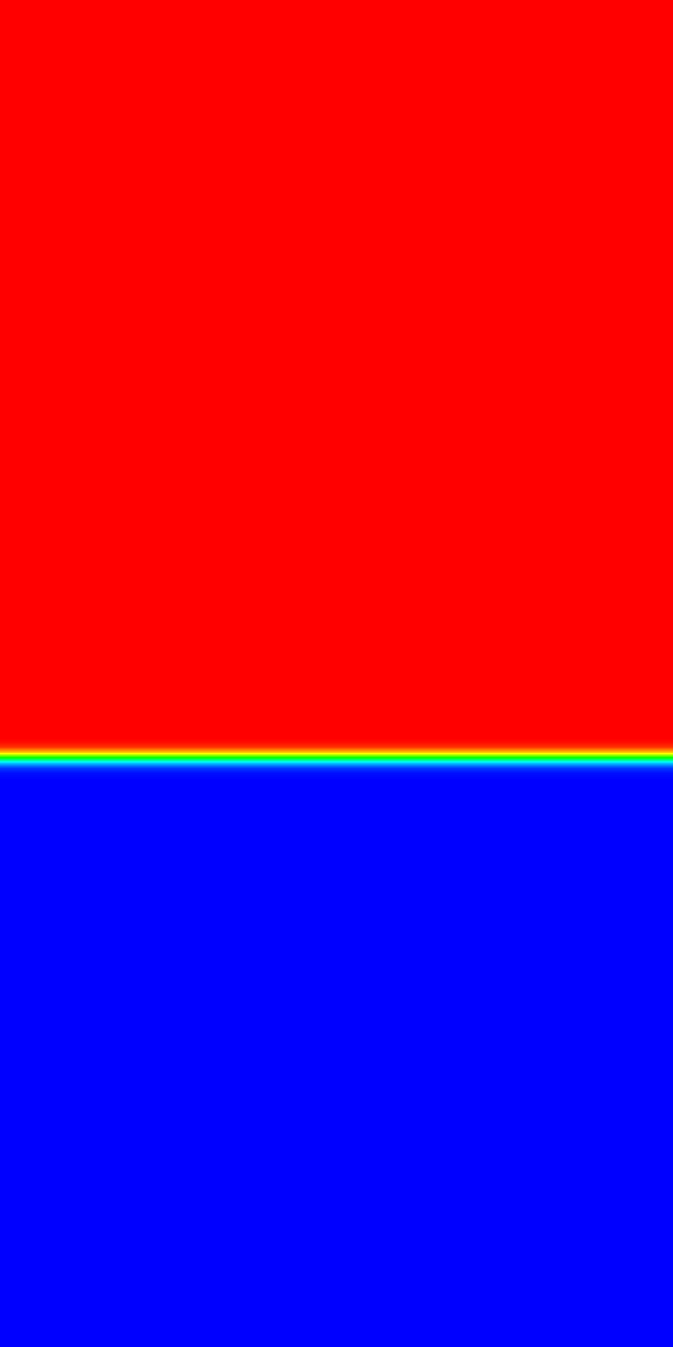}
%\caption{$t=0.0$}
\end{subfigure}
\begin{subfigure}{0.078\textwidth}
\centering
\includegraphics[width=1\textwidth]{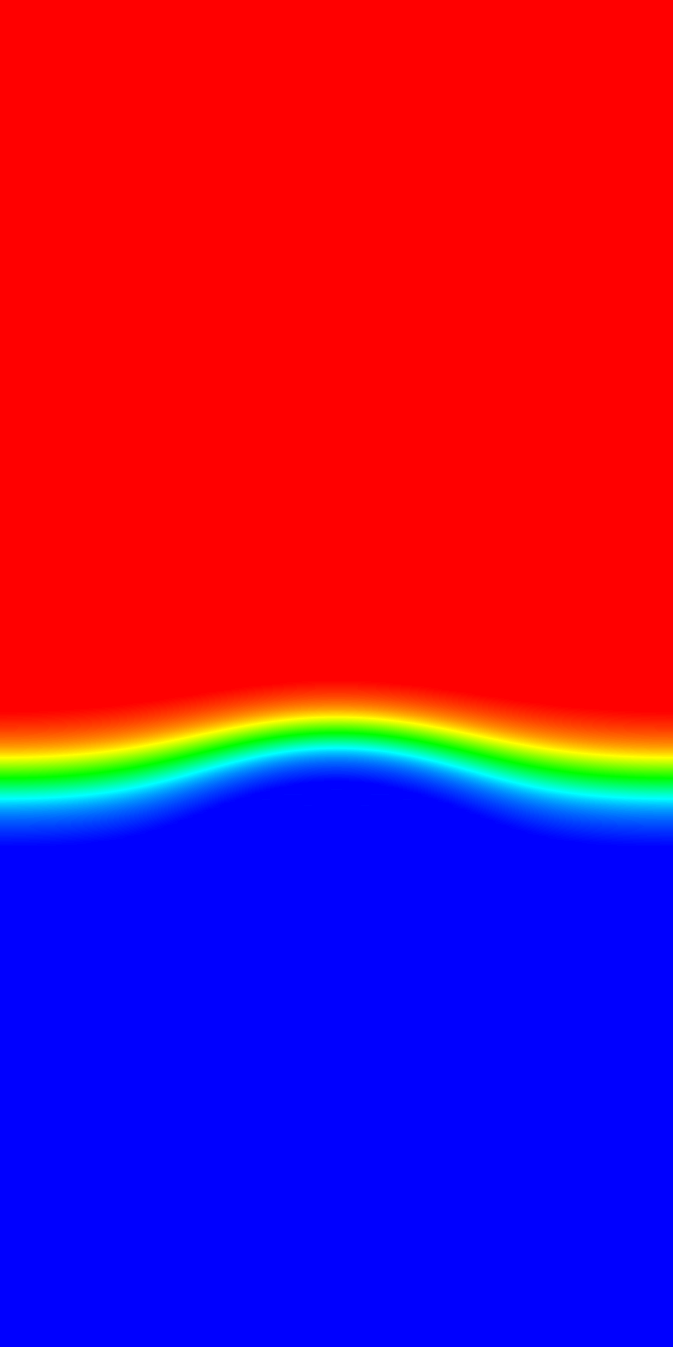}
%\caption{$t=0.6$}
\end{subfigure}
\begin{subfigure}{0.078\textwidth}
\centering
\includegraphics[width=1\textwidth]{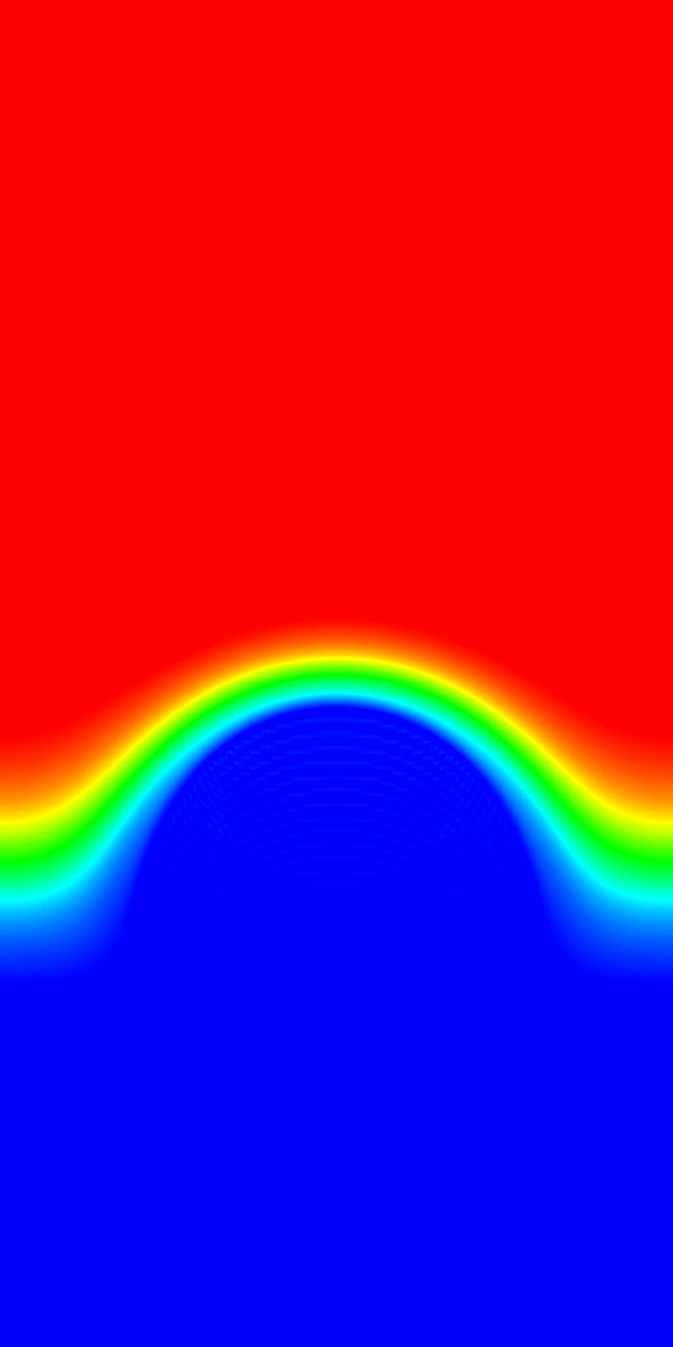}
%\caption{$t=1.2$}
\end{subfigure}
\begin{subfigure}{0.078\textwidth}
\centering
\includegraphics[width=1\textwidth]{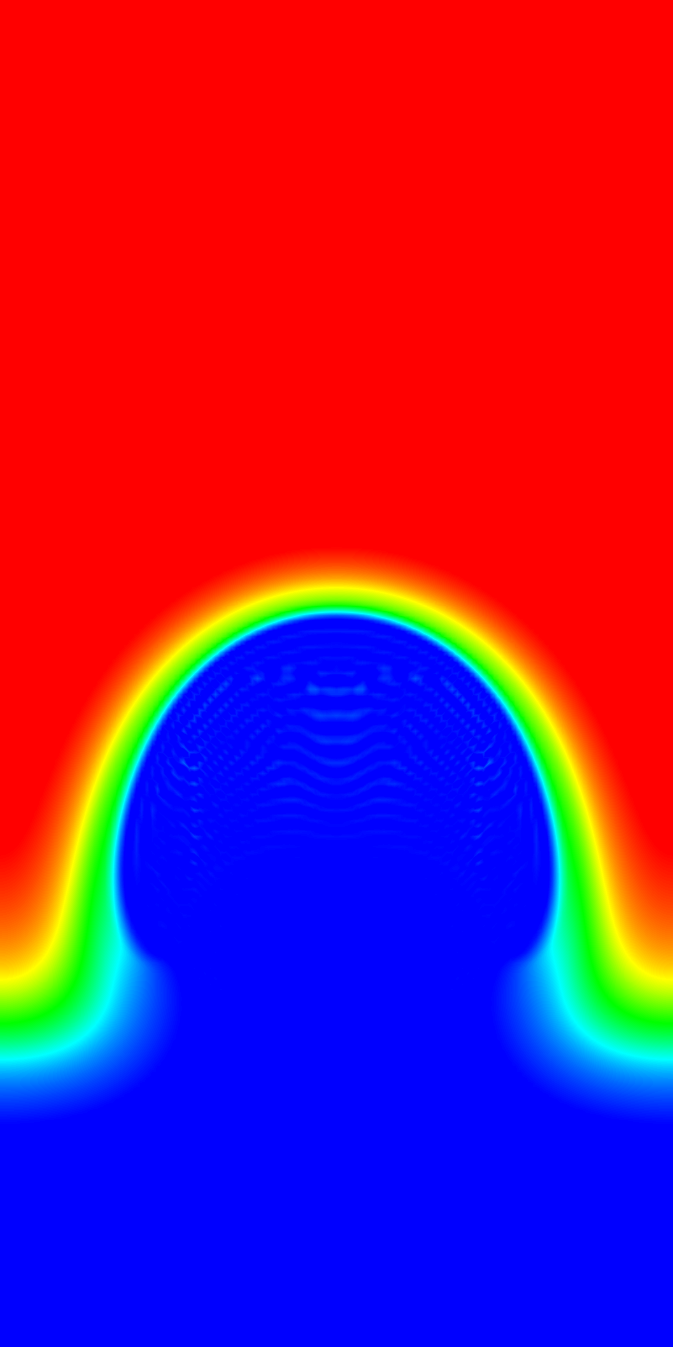}
%\caption{$t=1.8$}
\end{subfigure}
\begin{subfigure}{0.078\textwidth}
\centering
\includegraphics[width=1\textwidth]{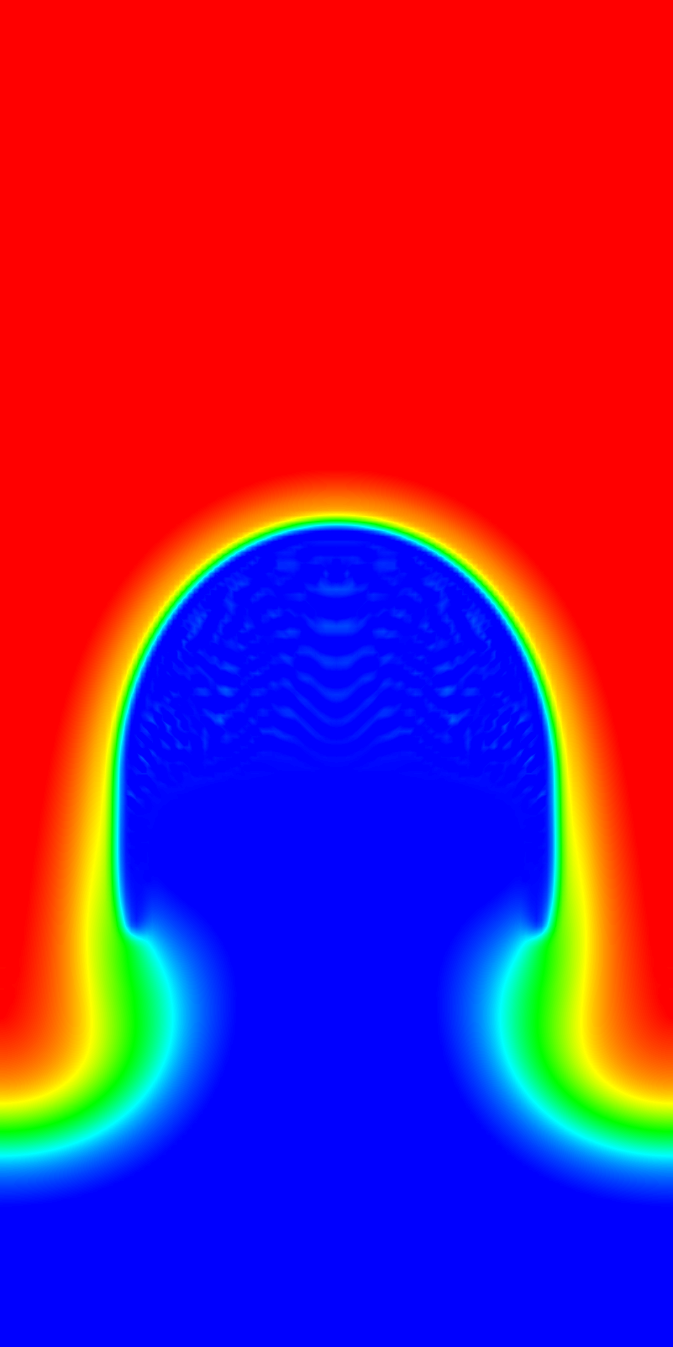}
%\caption{$t=2.4$}
\end{subfigure}
\begin{subfigure}{0.078\textwidth}
\centering
\includegraphics[width=1\textwidth]{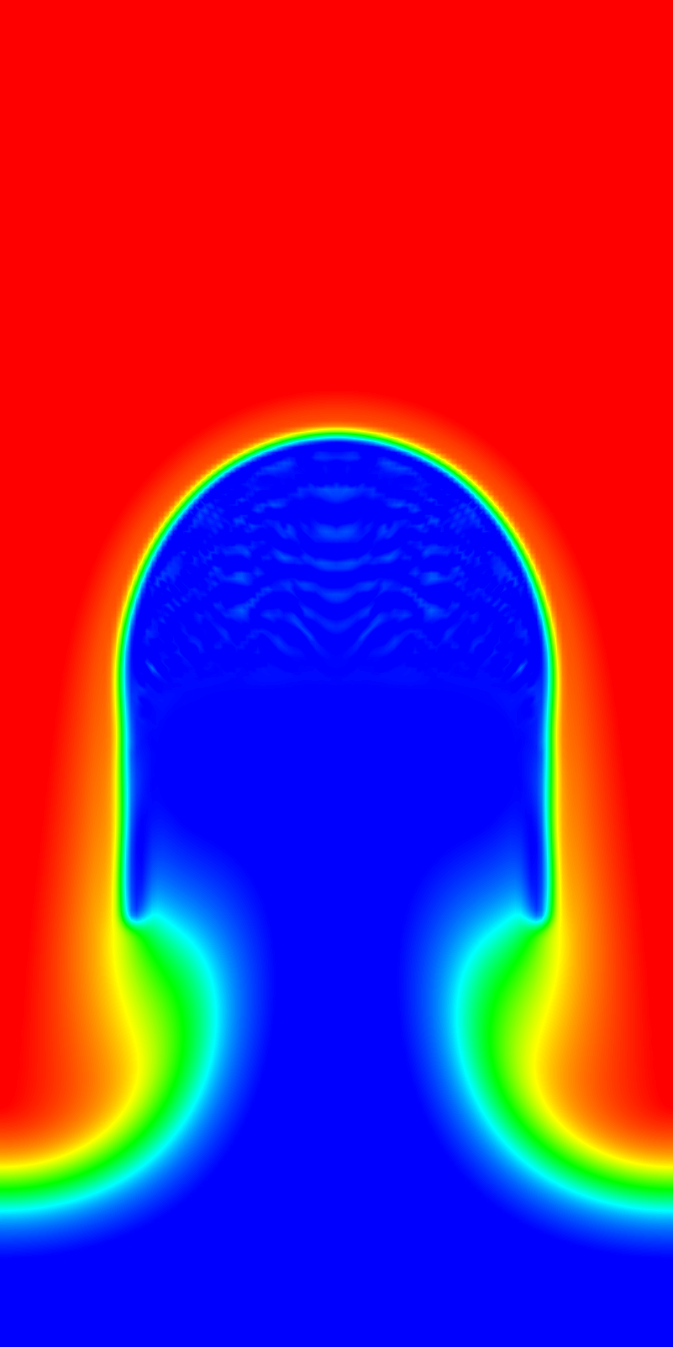}
%\caption{$t=3.0$}
\end{subfigure}
\caption{Mixture-aware simulations -- phases above each other. Case 1 (left) and Case 2 (right). Visualization of the phase fields (top to below) $\phi_1, \phi_2, \phi_3$ at times $t=0.0, 0.6, 1.2, 1.8, 2.4, 3.0$ (left to right).}
\label{fig: reduction case 1 2 - last}
\end{figure}

\subsection{Rising bubble through two stratified liquid layers}\label{sec:results:rising_bubble}

We consider the three-phase benchmark proposed by \cite{boyer2010cahn} for a gas bubble rising in a vertically
stratified configuration of two immiscible liquids. 
This benchmark combines (i) three distinct surface tensions, (ii) large density and viscosity ratios, and (iii) strong topological changes (interface approach, possible penetration, and entrainment). Initially, the heavy liquid (phase $2$) occupies the lower part of the domain and the light liquid (phase $3$) the upper part, separated by a planar interface. A spherical (axisymmetric) bubble of phase $1$ is placed in the heavy liquid below the interface and rises under gravity. 

We adopt the material parameters reported in \cite{boyer2010cahn}. The (target) surface tensions are
\begin{align}
  \gamma_{12}=\gamma_{13}=0.07~\mathrm{N\,m^{-1}},\qquad
  \gamma_{23}=0.05~\mathrm{N\,m^{-1}},
\end{align}
and the densities and viscosities are given in \cref{table: N=3}.
\begin{table}
\begin{center}
\begin{tabular}{lcc}
\hline
phase & density $\rho$ [$\mathrm{kg\,m^{-3}}$] & viscosity $\eta$ [$\mathrm{Pa\,s}$] \\
\hline
bubble ($\phi_1$) & $1$ & $10^{-4}$ \\
heavy liquid ($\phi_2$) & $1200$ & $0.15$ \\
light liquid ($\phi_3$) & $1000$ & $0.10$ \\
\hline
\end{tabular}
\end{center}
\caption{Rising bubble through two stratified liquid layers. Density and viscosity values.}
\label{table: N=3}
\end{table}
We use gravitational acceleration $\mathbf{g}=-g\,\mathbf{e}_y$ with $g=9.81~\mathrm{m\,s^{-2}}$.

The computations are performed in an axisymmetric mesh with $64 \times 480$ elements with $\epsilon = h$. The initial condition consists of (i) a diffuse but nearly planar liquid--liquid interface (transition between
phases $2$ and $3$) and (ii) a diffuse spherical bubble of radius $r$ (phase $1$) embedded in phase $2$ below the interface.
No-penetration boundary conditions are imposed on the velocity on the outer boundary, and homogeneous Neumann (no-flux) boundary conditions are imposed for the phase fields and chemical potentials (i.e. no diffusive mass flux through the boundary).

\begin{figure}
\captionsetup[subfigure]{justification=centering}
\begin{subfigure}{0.245\textwidth}
\centering
\includegraphics[width=\textwidth]{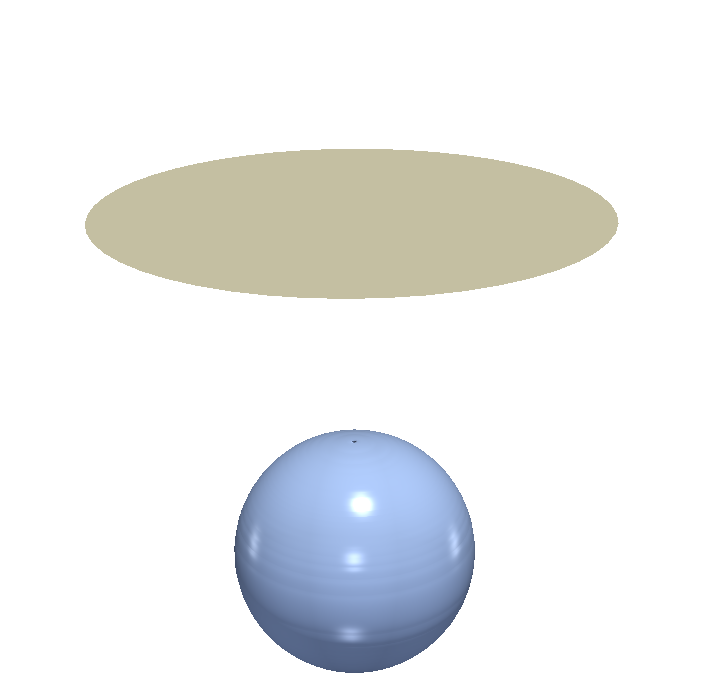}
\caption{$t=0.0$}
\end{subfigure}
\begin{subfigure}{0.245\textwidth}
\centering
\includegraphics[width=\textwidth]{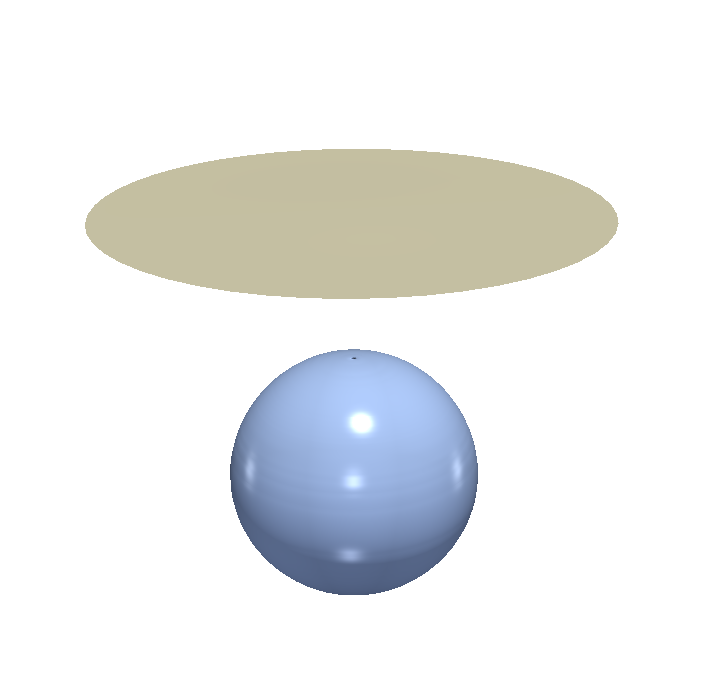}
\caption{$t=0.05$}
\end{subfigure}
\begin{subfigure}{0.245\textwidth}
\centering
\includegraphics[width=\textwidth]{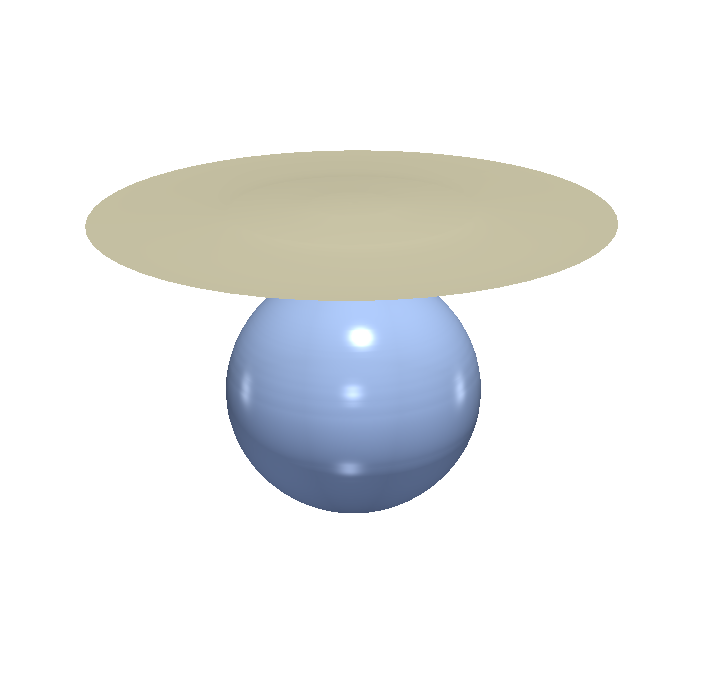}
\caption{$t=0.10$}
\end{subfigure}
\begin{subfigure}{0.245\textwidth}
\centering
\includegraphics[width=\textwidth]{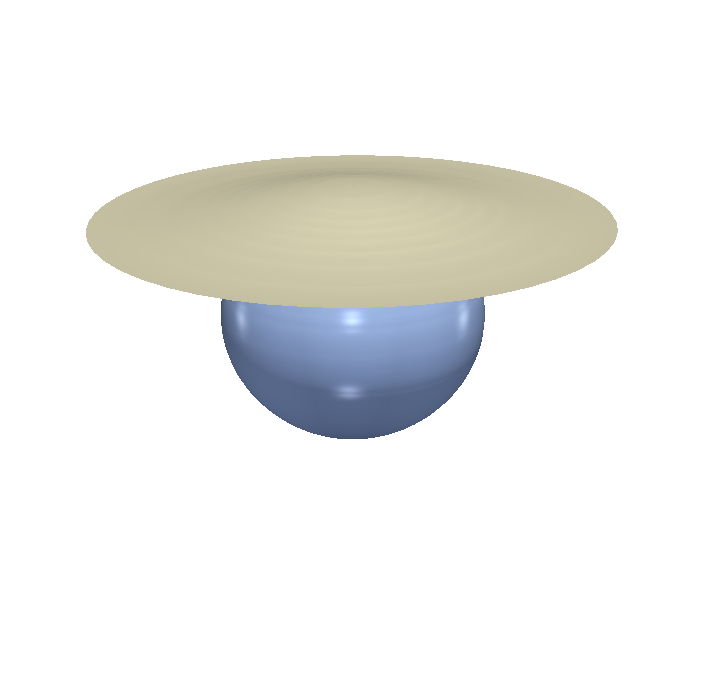}
\caption{$t=0.15$}
\end{subfigure}
\begin{subfigure}{0.245\textwidth}
\centering
\includegraphics[width=\textwidth]{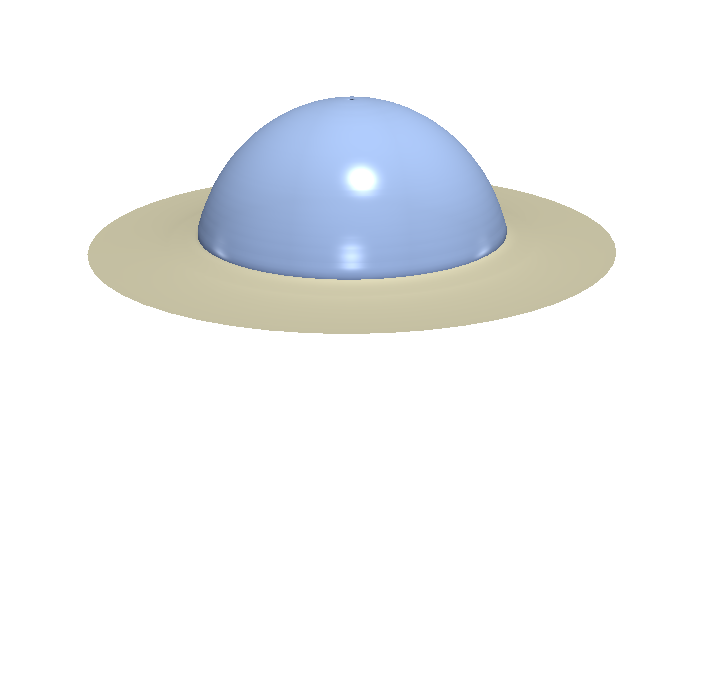}
\caption{$t=0.20$}
\end{subfigure}
\begin{subfigure}{0.245\textwidth}
\centering
\includegraphics[width=\textwidth]{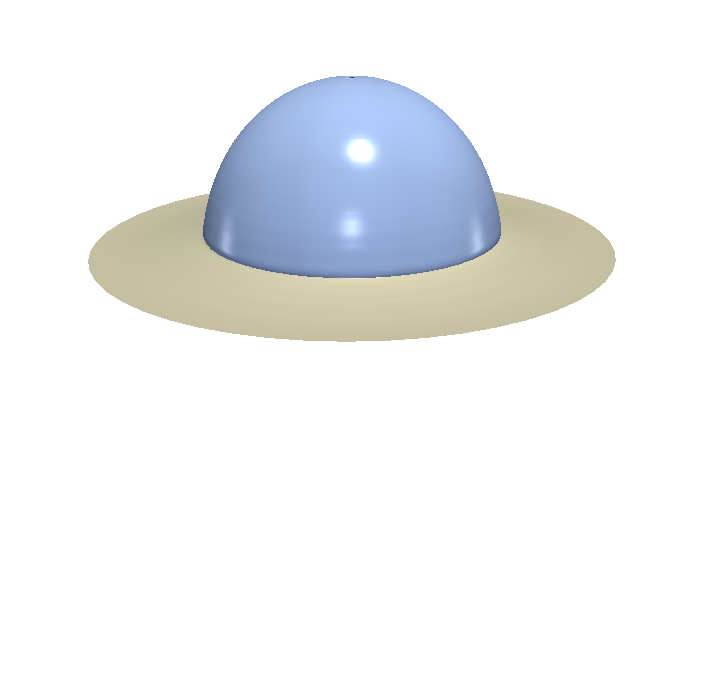}
\caption{$t=0.25$}
\end{subfigure}
\begin{subfigure}{0.245\textwidth}
\centering
\includegraphics[width=\textwidth]{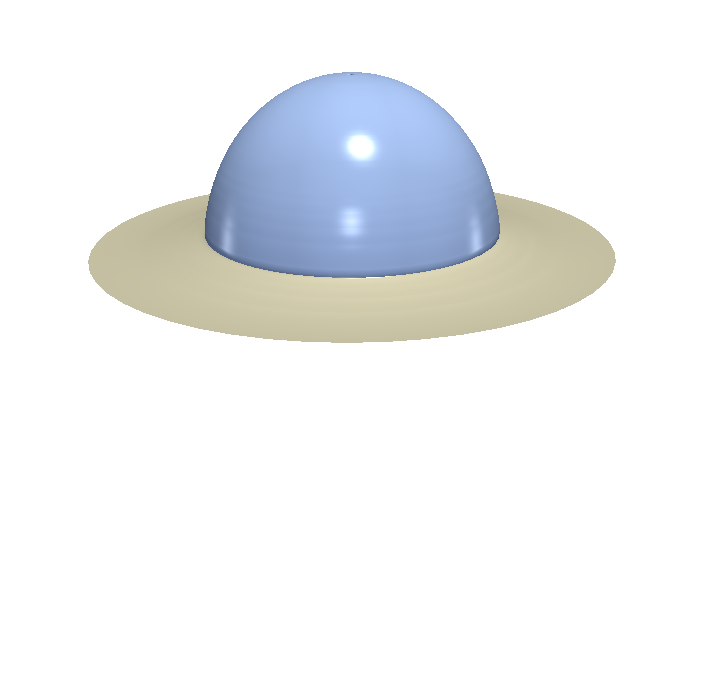}
\caption{$t=0.30$}
\end{subfigure}
\begin{subfigure}{0.245\textwidth}
\centering
\includegraphics[width=\textwidth]{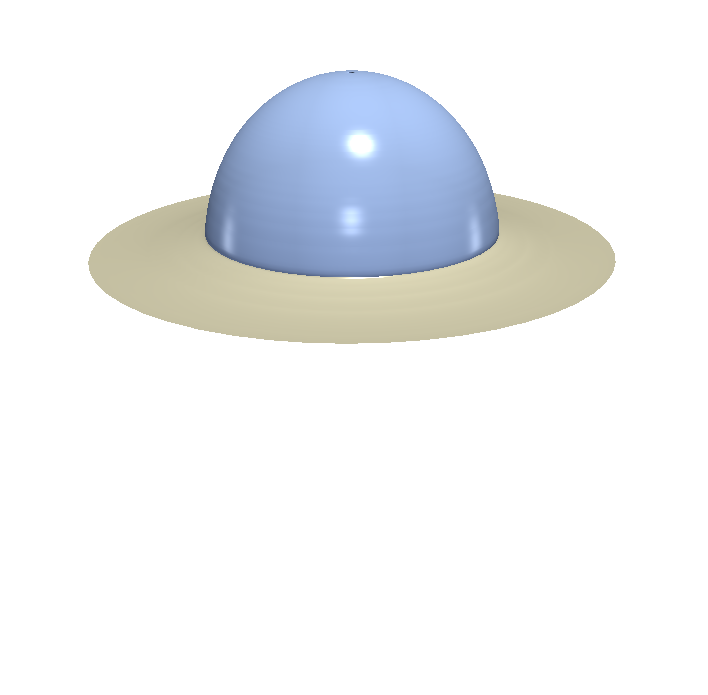}
\caption{$t=0.35$}
\end{subfigure}
\caption{Rising bubble in two stratified liquid layers. Time snapshots of the phase fields (iso-contours) for $r=r_1>r_p$.}
\label{fig: no penetration case}
\end{figure}

Depending on its size, the bubble either remains captured in the liquid--liquid interface or penetrates into the upper (lighter) liquid layer. The criterion of 
\cite{greene1991bubble,greene1988onset} states that the bubble penetrates the interface if its volume $V$ exceeds a critical value $V_p$, given by 
\begin{align}\label{eq:Vp_criterion}
  V_p =
  \left(
    2\pi\left(\frac{3}{4\pi}\right)^{1/3}
    \frac{\sigma_{23}}{(\rho_3-\rho_1)g}
  \right)^{3/2}
  \approx 8.87\times 10^{-8}\ \mathrm{m^3},
\end{align}
corresponding to a critical radius $r_p\approx 2.76\times 10^{-3}~\mathrm{m}$. This criterion is experimentally validated in \cite{greene1991bubble,greene1988onset}. In the present work we reproduce this behavior by considering two representative bubble radii,
\begin{align}
  r_1 = 2.0\times 10^{-3}~\mathrm{m}\quad(<r_p), \qquad
  r_2 = 2.9\times 10^{-3}~\mathrm{m}\quad(>r_p),
\end{align}
matching the reference study.

Figures~\ref{fig: no penetration case}-\ref{fig: penetration case} show time snapshots of the phase distribution for the two radii.
For the smaller bubble ($r=r_1$), the bubble rises until it reaches the liquid--liquid interface, where it deforms the
interface but remains trapped (no passage into the upper layer). For the larger bubble ($r=r_2$), the capillary barrier is overcome and the bubble penetrates into the light liquid, possibly
pulling a thin column (or filament) of heavy liquid upward during the penetration stage. This qualitative transition is
consistent with the penetration criterion \eqref{eq:Vp_criterion} and the numerical observations reported in \cite{boyer2010cahn}.

\begin{figure}
\captionsetup[subfigure]{justification=centering}
\begin{subfigure}{0.245\textwidth}
\centering
\includegraphics[width=\textwidth]{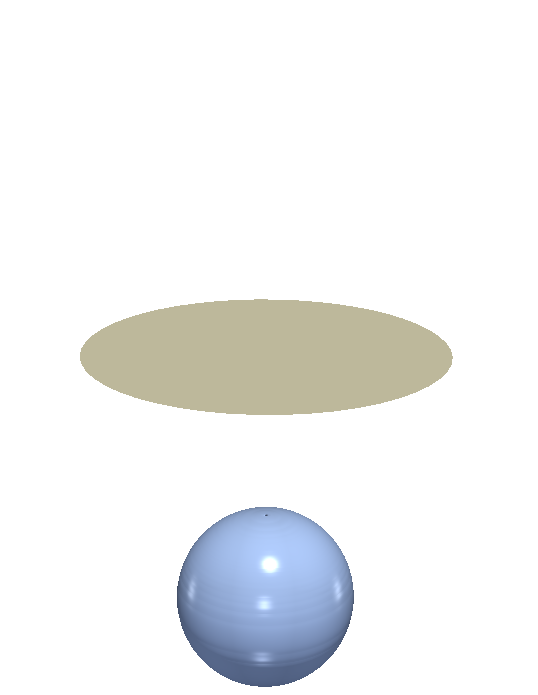}
\caption{$t=0.0$}
\end{subfigure}
\begin{subfigure}{0.245\textwidth}
\centering
\includegraphics[width=\textwidth]{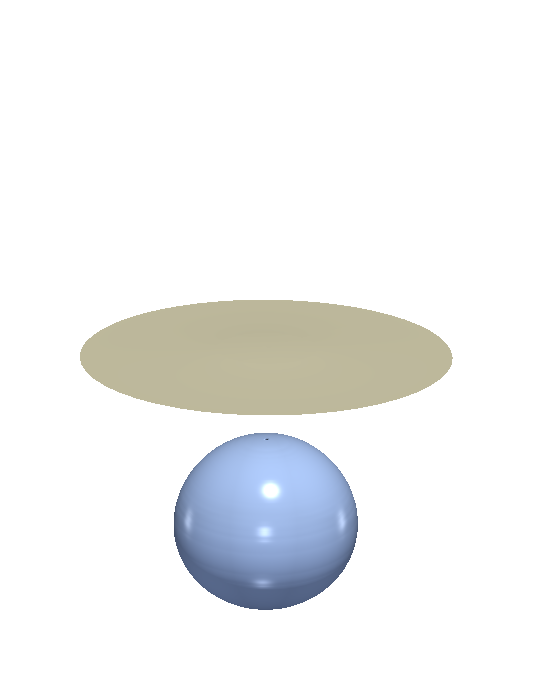}
\caption{$t=0.05$}
\end{subfigure}
\begin{subfigure}{0.245\textwidth}
\centering
\includegraphics[width=\textwidth]{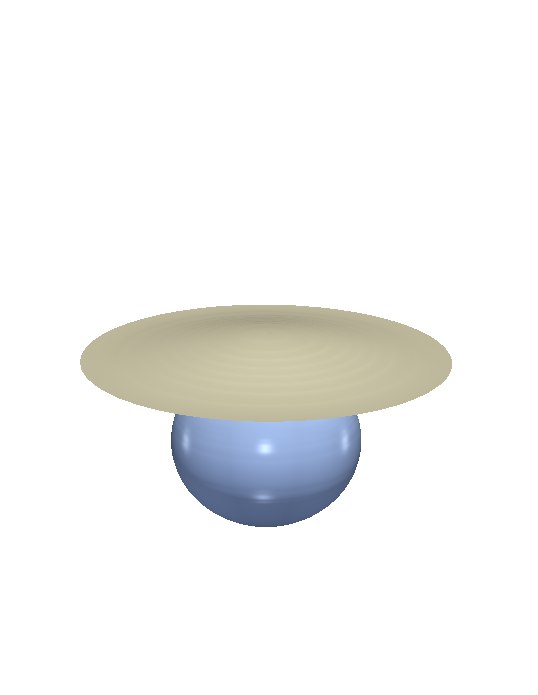}
\caption{$t=0.10$}
\end{subfigure}
\begin{subfigure}{0.245\textwidth}
\centering
\includegraphics[width=\textwidth]{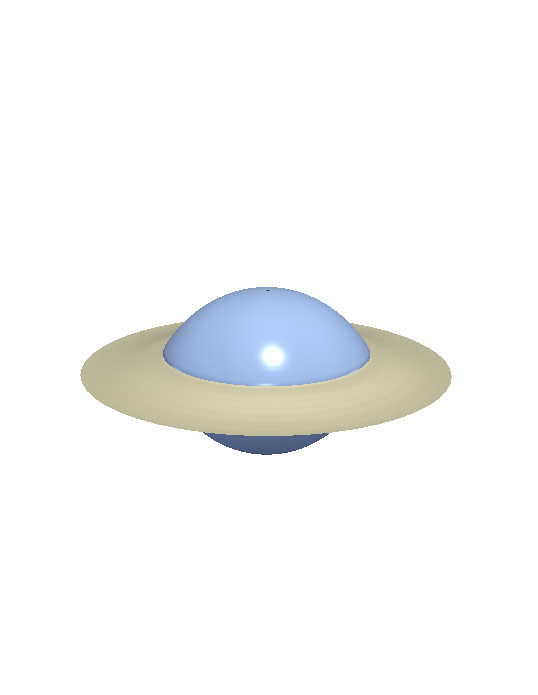}
\caption{$t=0.15$}
\end{subfigure}
\begin{subfigure}{0.245\textwidth}
\centering
\includegraphics[width=\textwidth]{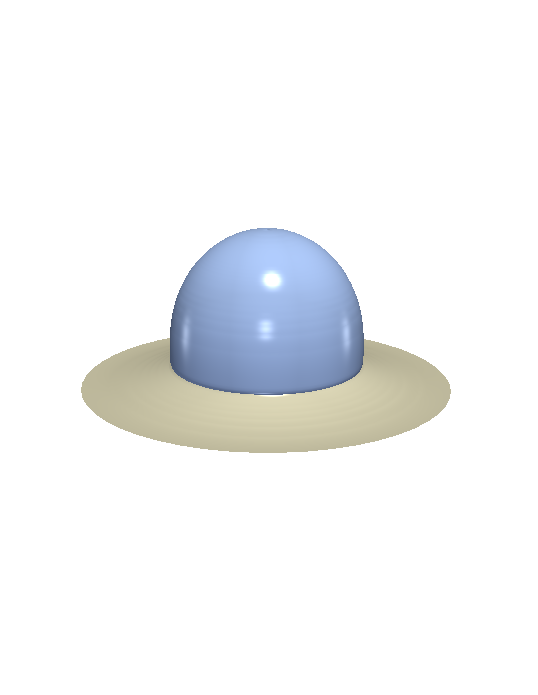}
\caption{$t=0.20$}
\end{subfigure}
\begin{subfigure}{0.245\textwidth}
\centering
\includegraphics[width=\textwidth]{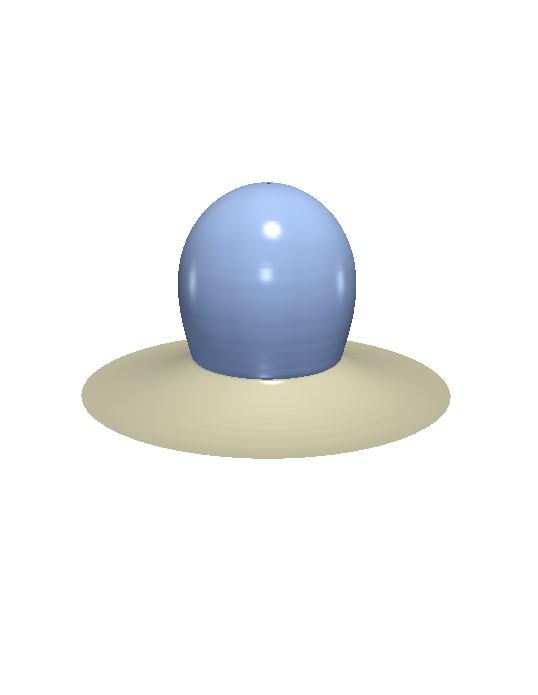}
\caption{$t=0.25$}
\end{subfigure}
\begin{subfigure}{0.245\textwidth}
\centering
\includegraphics[width=\textwidth]{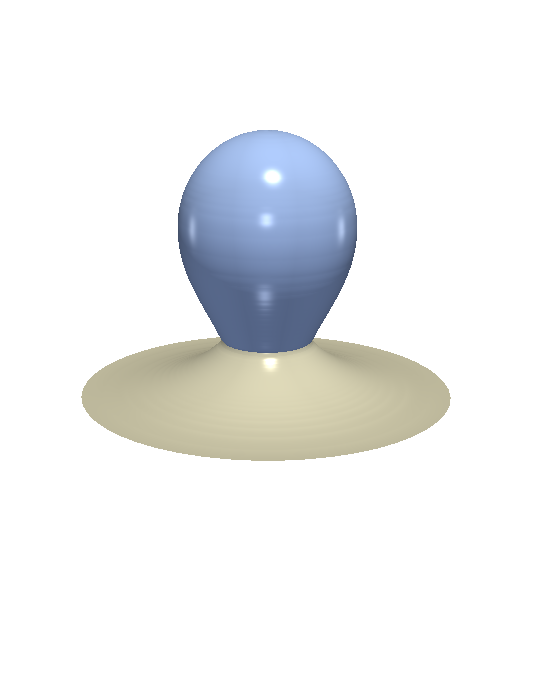}
\caption{$t=0.30$}
\end{subfigure}
\begin{subfigure}{0.245\textwidth}
\centering
\includegraphics[width=\textwidth]{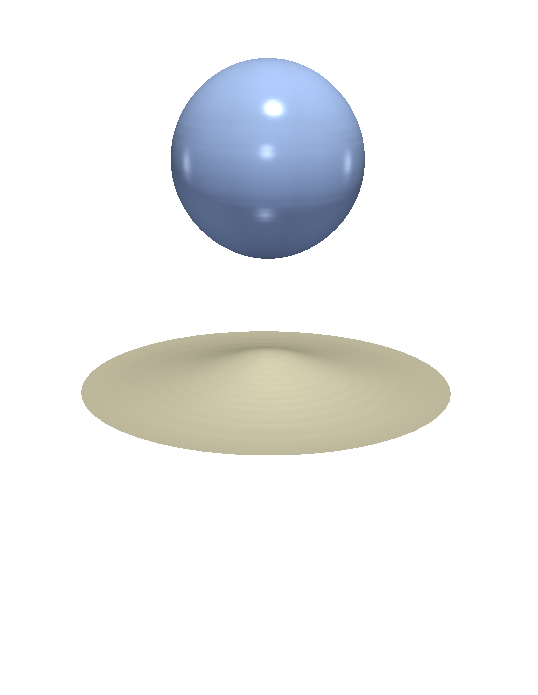}
\caption{$t=0.35$}
\end{subfigure}
\caption{Rising bubble in two stratified liquid layers. Time snapshots of the phase fields (iso-contours) for $r=r_2>r_p$.}
\label{fig: penetration case}
\end{figure}

\subsection{Quaternary droplet-bubble simulation}\label{subsec: Quaternary rising bubble simulation}
Finally, we demonstrate the capabilities of the framework for the $N=4$ case. Now we consider both a falling droplet (phase $4$) and a rising gas bubble (phase $2$) in a vertically
stratified configuration of immiscible liquids (phases $1$ and $3$). Initially, the heavy liquid (phase $1$) occupies the lower part of the domain and the lighter liquid (phase $3$) the upper part, separated by a planar interface. A spherical bubble of phase $2$ is placed in the heavy liquid below the interface and rises, and a spherical droplet of phase $4$ is placed in the lighter liquid above the interface.

We adopt the material parameters similar to those in \cref{sec:results:rising_bubble}. The (target) surface tensions are
\begin{align}
  \gamma_{12}=0.07~\mathrm{N\,m^{-1}}, \gamma_{13}=0.1~\mathrm{N\,m^{-1}},  \gamma_{14}=0.06~\mathrm{N\,m^{-1}},\nn\\
  \gamma_{23}=0.09~\mathrm{N\,m^{-1}},
  \gamma_{24}=0.05~\mathrm{N\,m^{-1}},
  \gamma_{34}=0.09~\mathrm{N\,m^{-1}},
\end{align}
and the densities and viscosities are given in \cref{table N4}.
Again, we use gravitational acceleration $\mathbf{g}=-g\,\mathbf{e}_y$ with $g=9.81~\mathrm{m\,s^{-2}}$.
\begin{table}
\begin{center}
\begin{tabular}{lcc}
\hline
phase & density $\rho$ [$\mathrm{kg\,m^{-3}}$] & viscosity $\eta$ [$\mathrm{Pa\,s}$] \\
\hline
lower heavy phase ($\phi_1$) & $1000$ & $0.1$ \\
bubble in lower heavy phase ($\phi_2$) & $1$ & $10^{-4}$ \\
upper light gas ($\phi_3$) & $1$ & $10^{-4}$ \\
upper heavy droplet ($\phi_4$) & $1200$ & $0.15$ \\
\hline
\end{tabular}
\end{center}
\caption{Quaternary droplet-bubble simulation. Density and viscosity values.}
\label{table N4}
\end{table}
We perform the computations in an axisymmetric geometry of $64\times 256$ elements. We impose no-penetration boundary conditions on the velocity on the outer boundary, and homogeneous Neumann (no-flux) boundary conditions for the phase fields and chemical potentials (i.e. no diffusive mass flux through the boundary). 

Figures~\ref{fig: 4phase 1}-\ref{fig: 4phase 2} show time snapshots of the phase distributions. The bubble ($r=3.48\times 10^{-3}~\mathrm{m}$) rises until it reaches the $1$-$3$ interface, where it deforms the interface but remains trapped (no passage into the upper layer). On the other hand, the droplet ($r=3.19\times 10^{-3}~\mathrm{m}$) falls toward the $1$--$3$ interface, deforms it, and also remains trapped.

\begin{figure}
\captionsetup[subfigure]{justification=centering}
\begin{subfigure}{0.245\textwidth}
\centering
\includegraphics[width=\textwidth]{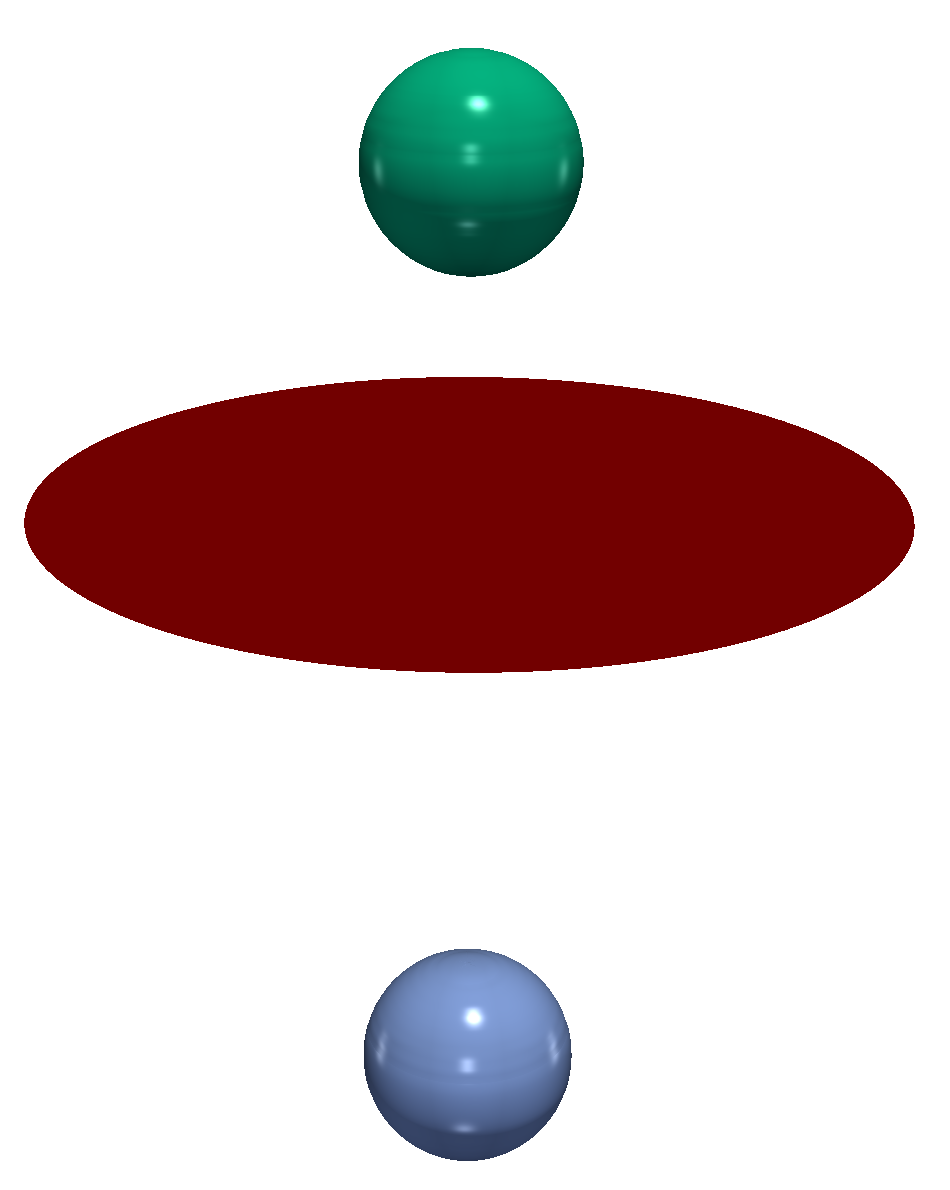}
\caption{$t=0.0$}
\end{subfigure}
\begin{subfigure}{0.245\textwidth}
\centering
\includegraphics[width=\textwidth]{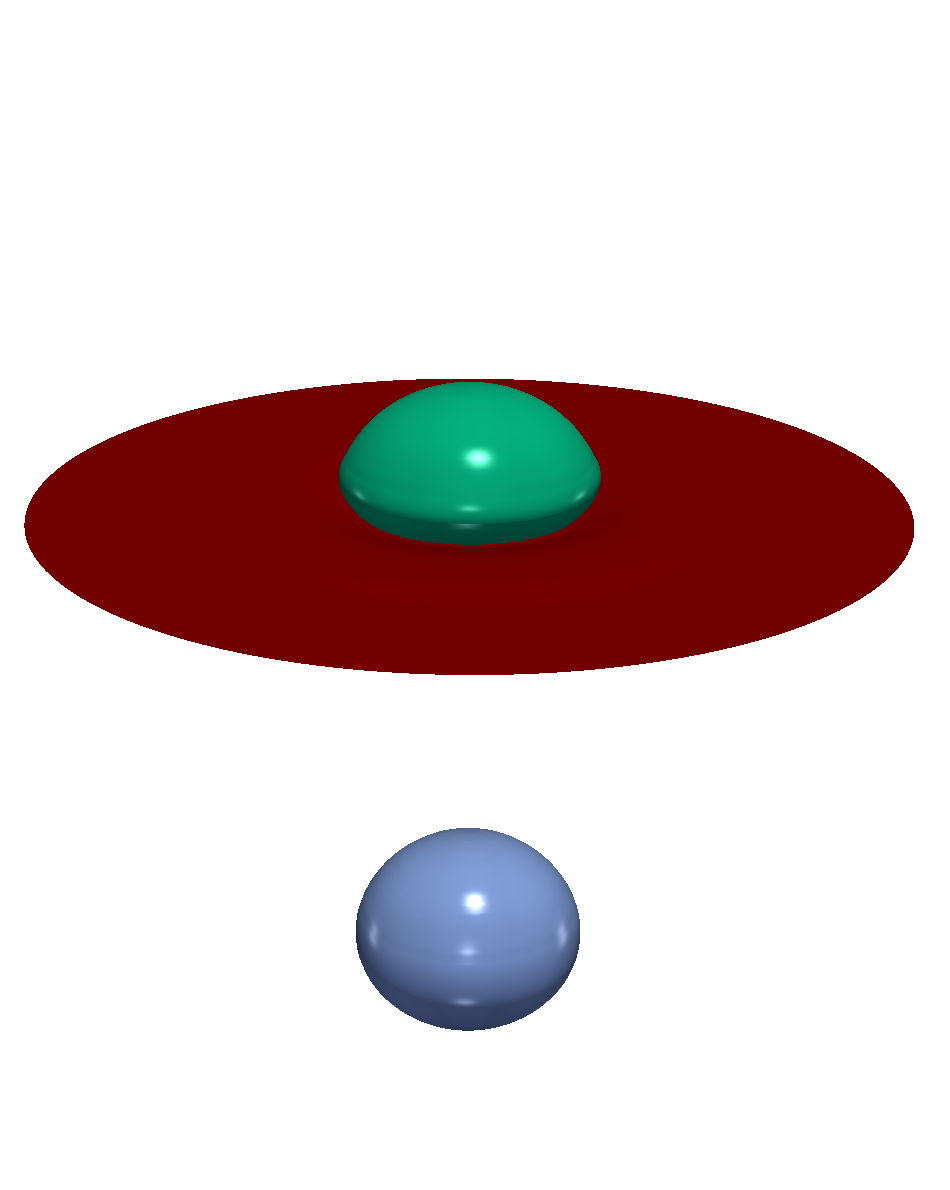}
\caption{$t=0.05$}
\end{subfigure}
\begin{subfigure}{0.245\textwidth}
\centering
\includegraphics[width=\textwidth]{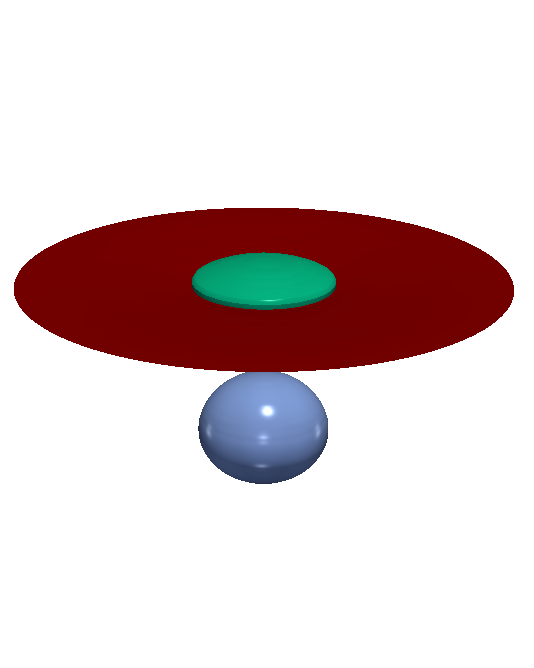}
\caption{$t=0.10$}
\end{subfigure}
\begin{subfigure}{0.245\textwidth}
\centering
\includegraphics[width=\textwidth]{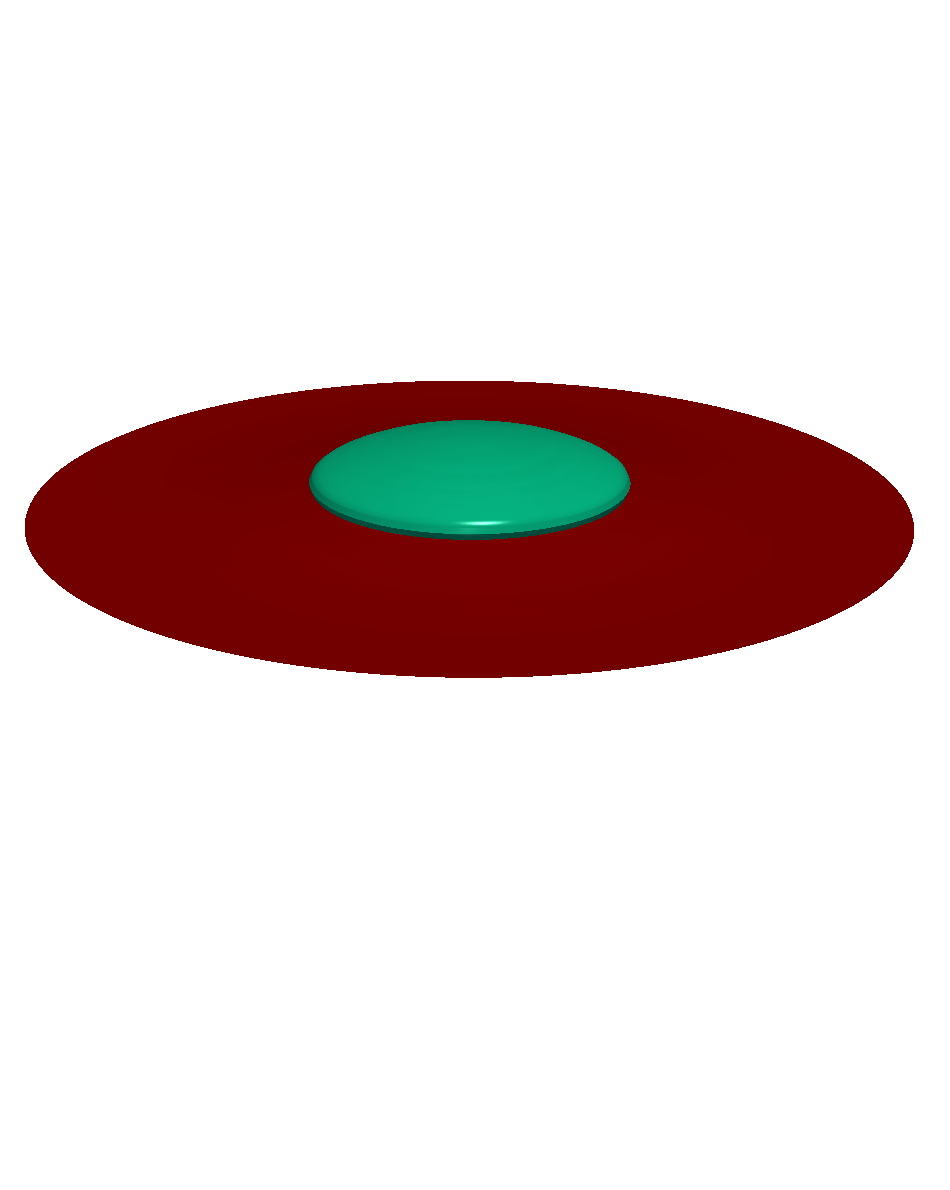}
\caption{$t=0.15$}
\end{subfigure}
\begin{subfigure}{0.245\textwidth}
\centering
\includegraphics[width=\textwidth]{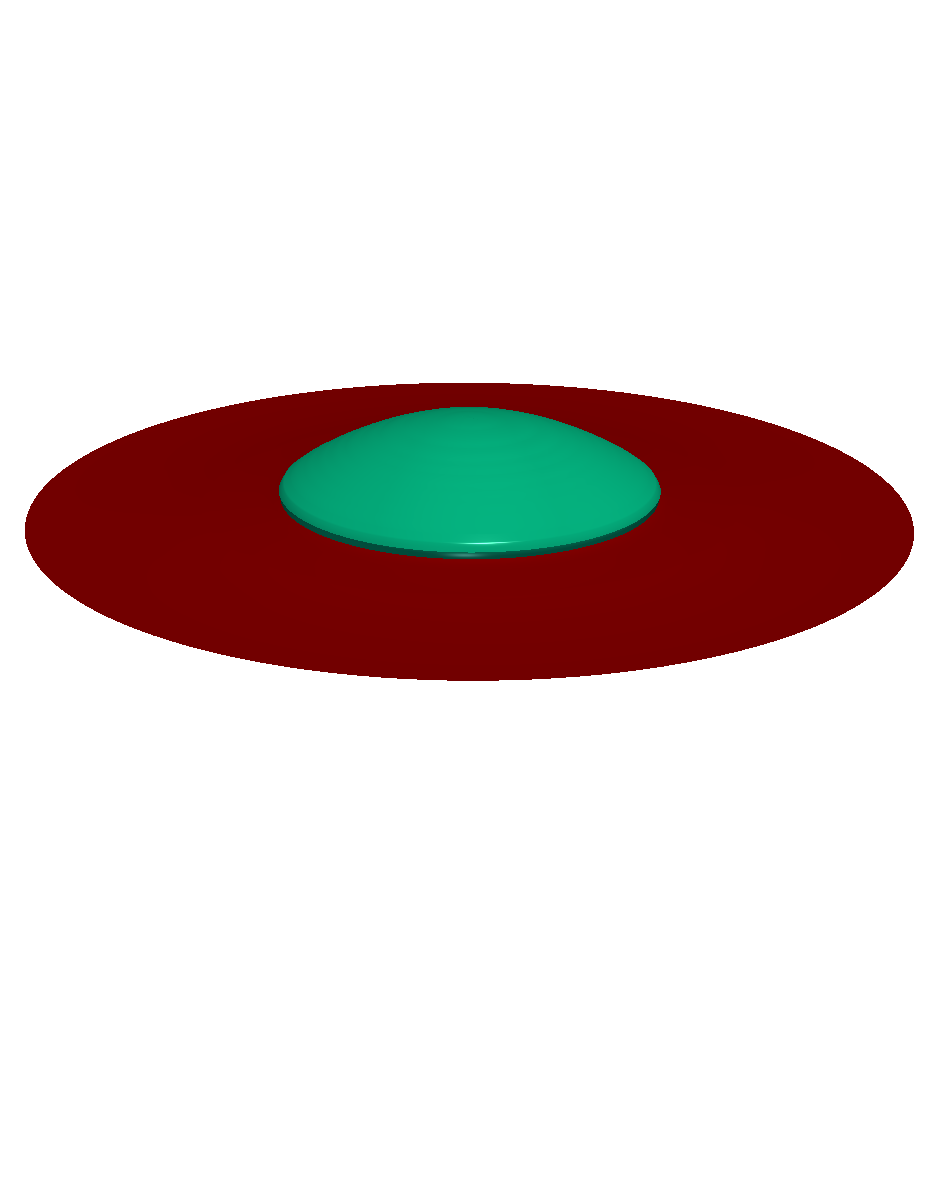}
\caption{$t=0.20$}
\end{subfigure}
\begin{subfigure}{0.245\textwidth}
\centering
\includegraphics[width=\textwidth]{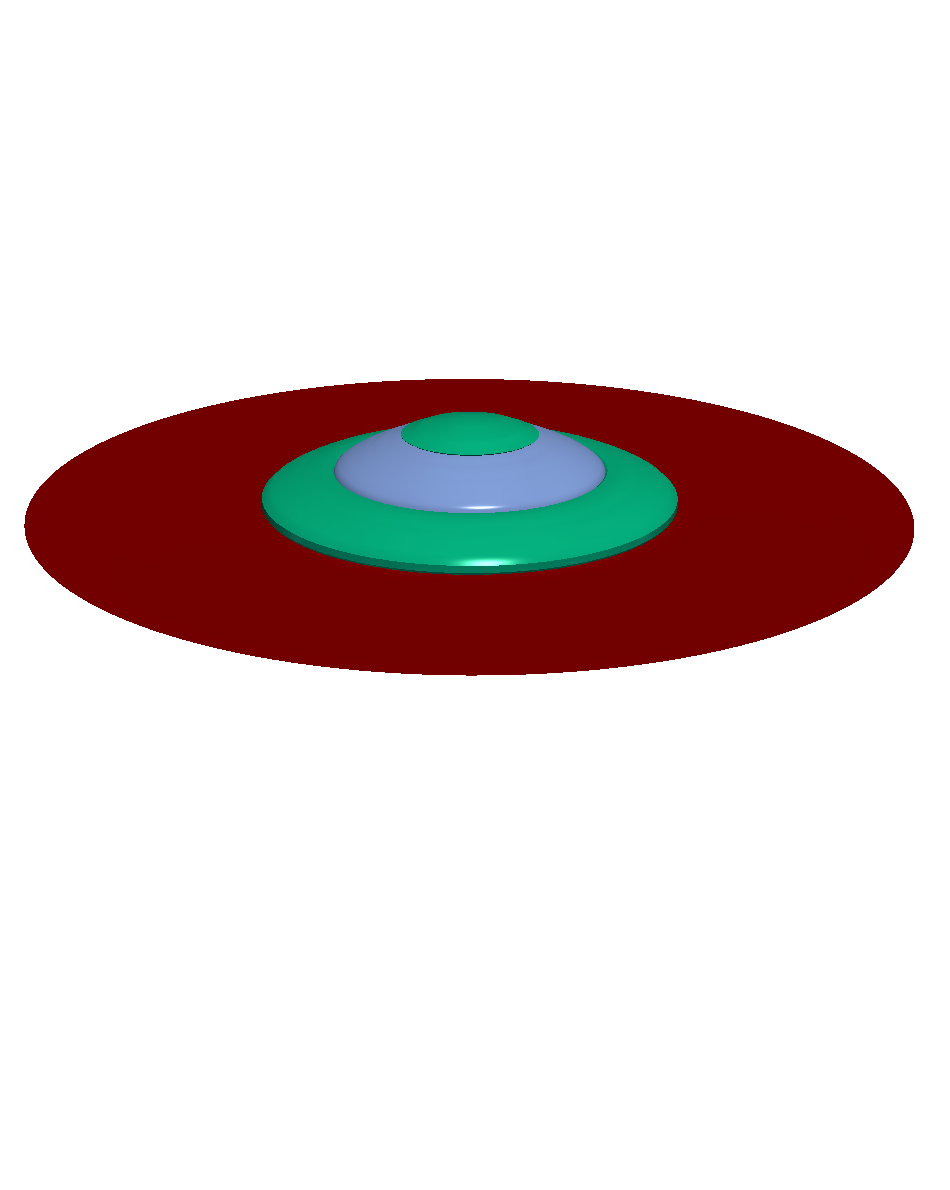}
\caption{$t=0.25$}
\end{subfigure}
\begin{subfigure}{0.245\textwidth}
\centering
\includegraphics[width=\textwidth]{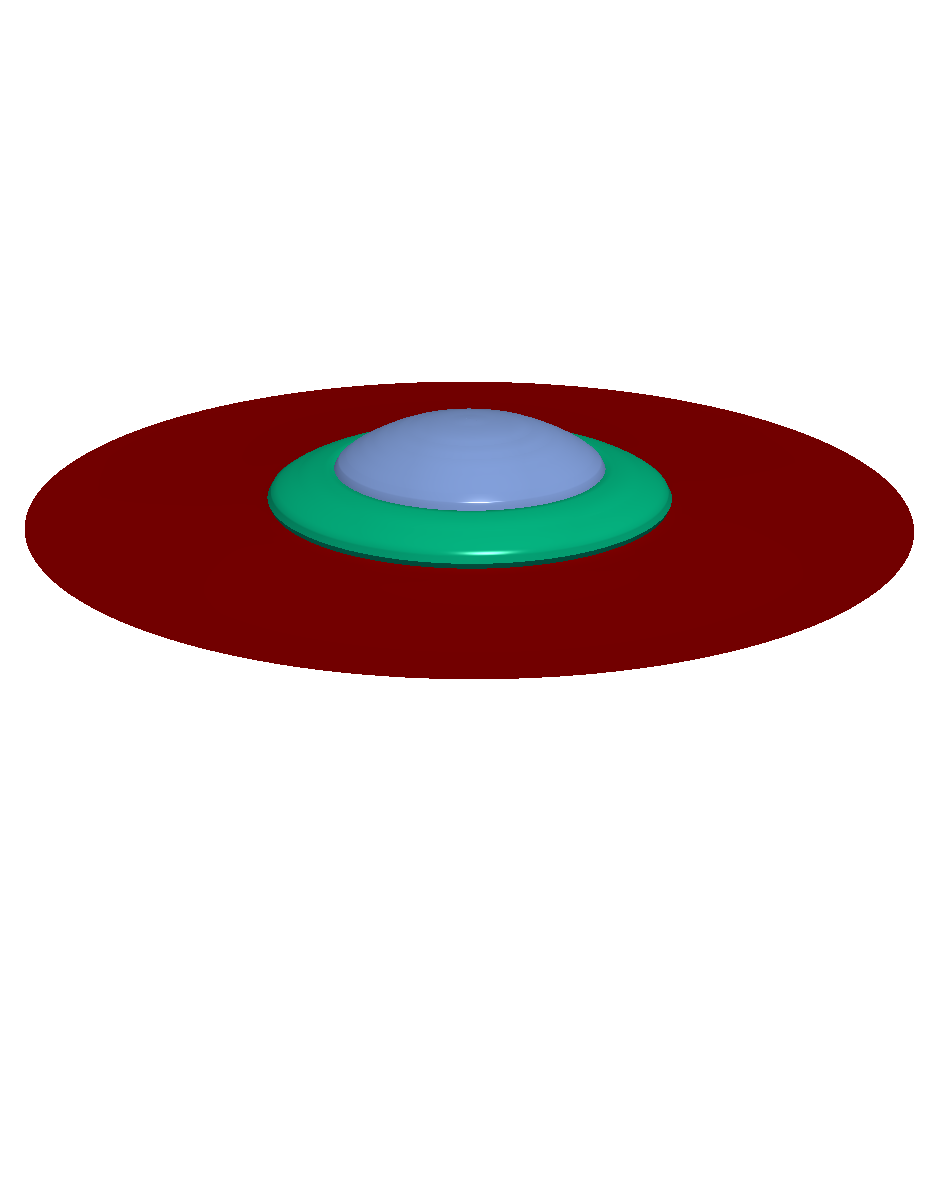}
\caption{$t=0.30$}
\end{subfigure}
\begin{subfigure}{0.245\textwidth}
\centering
\includegraphics[width=\textwidth]{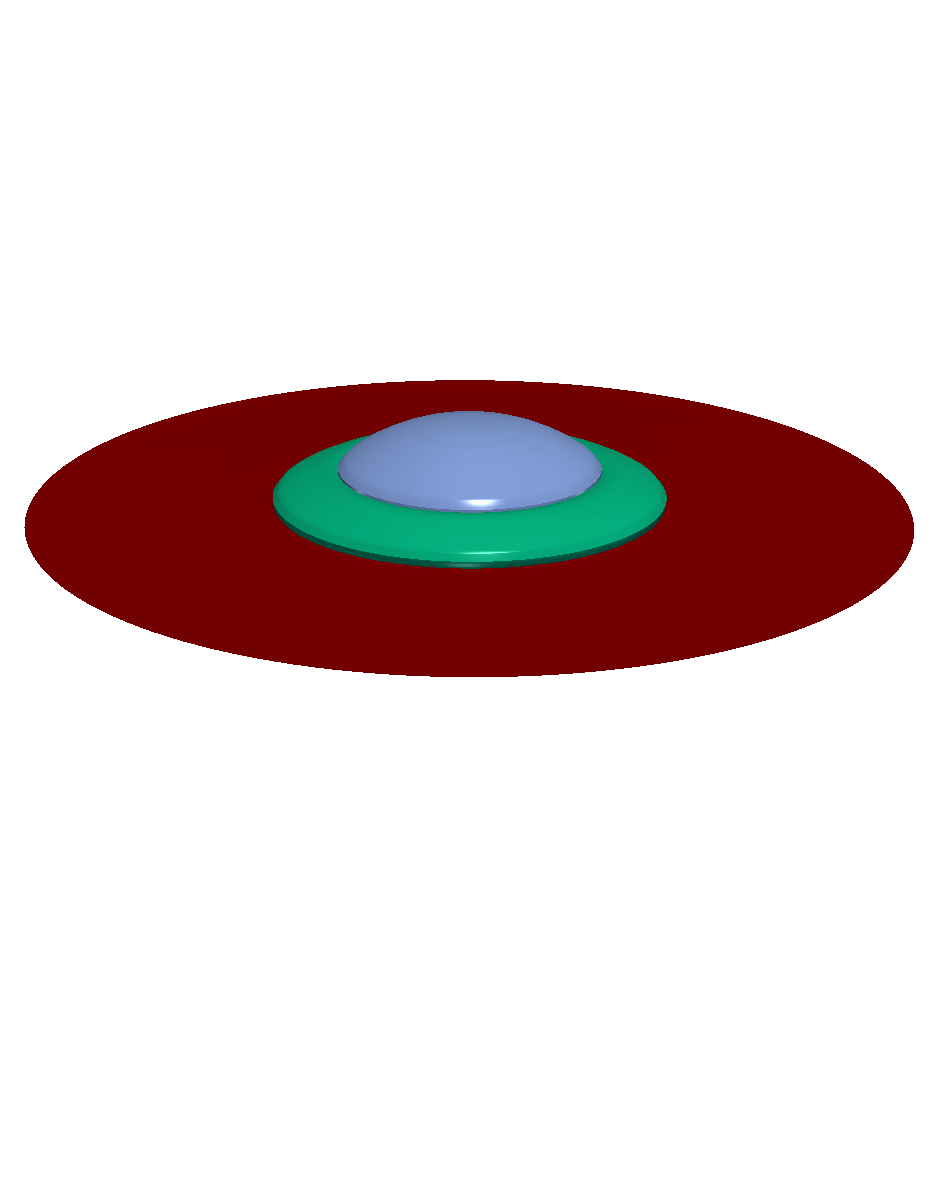}
\caption{$t=0.35$}
\end{subfigure}
\caption{Quaternary droplet-bubble simulation: side view. Time snapshots of the phase fields (iso-contours).}
\label{fig: 4phase 1}
\end{figure}

\begin{figure}
\captionsetup[subfigure]{justification=centering}
\begin{subfigure}{0.245\textwidth}
\centering
\includegraphics[width=\textwidth]{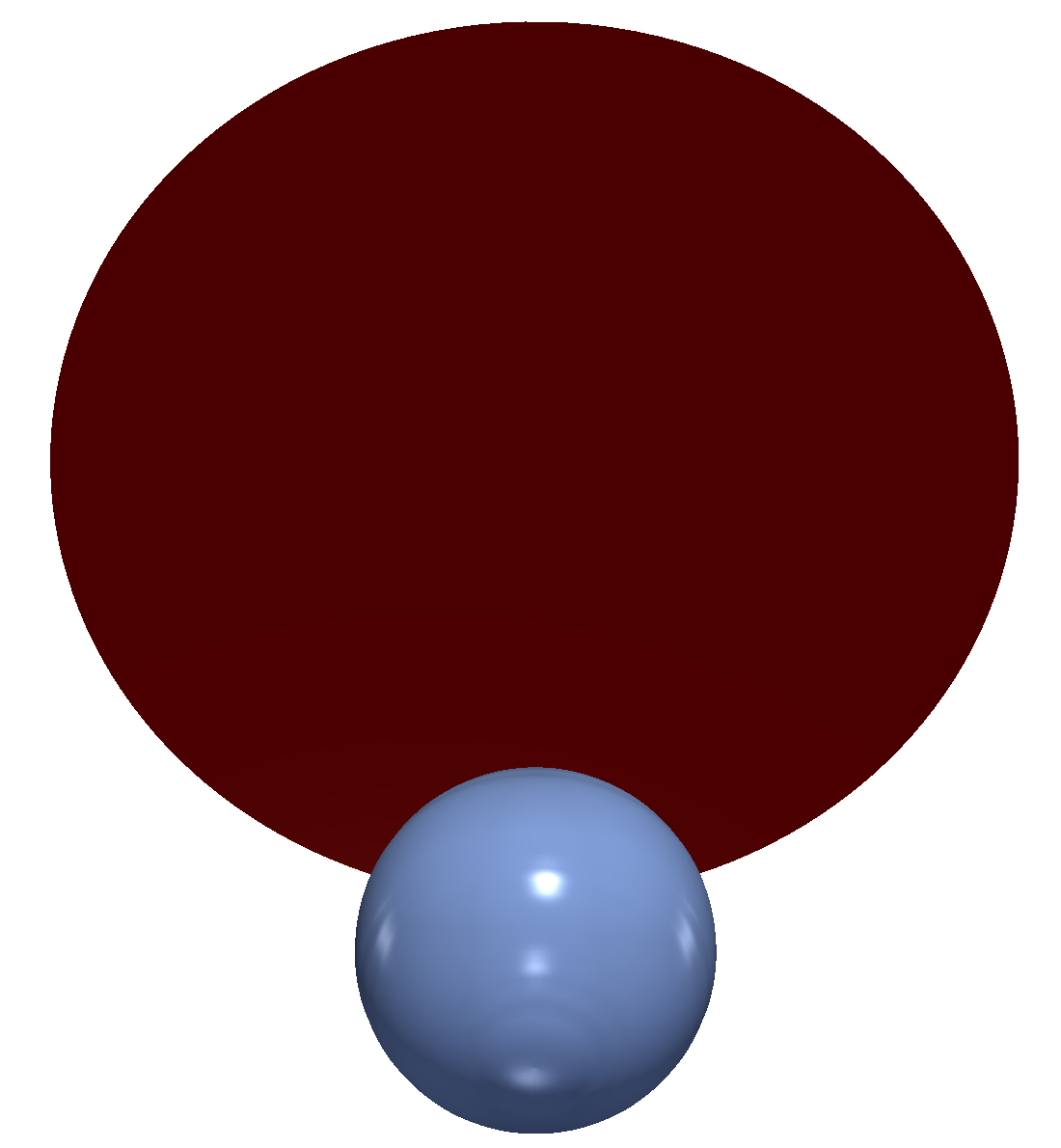}
\caption{$t=0.0$}
\end{subfigure}
\begin{subfigure}{0.245\textwidth}
\centering
\includegraphics[width=\textwidth]{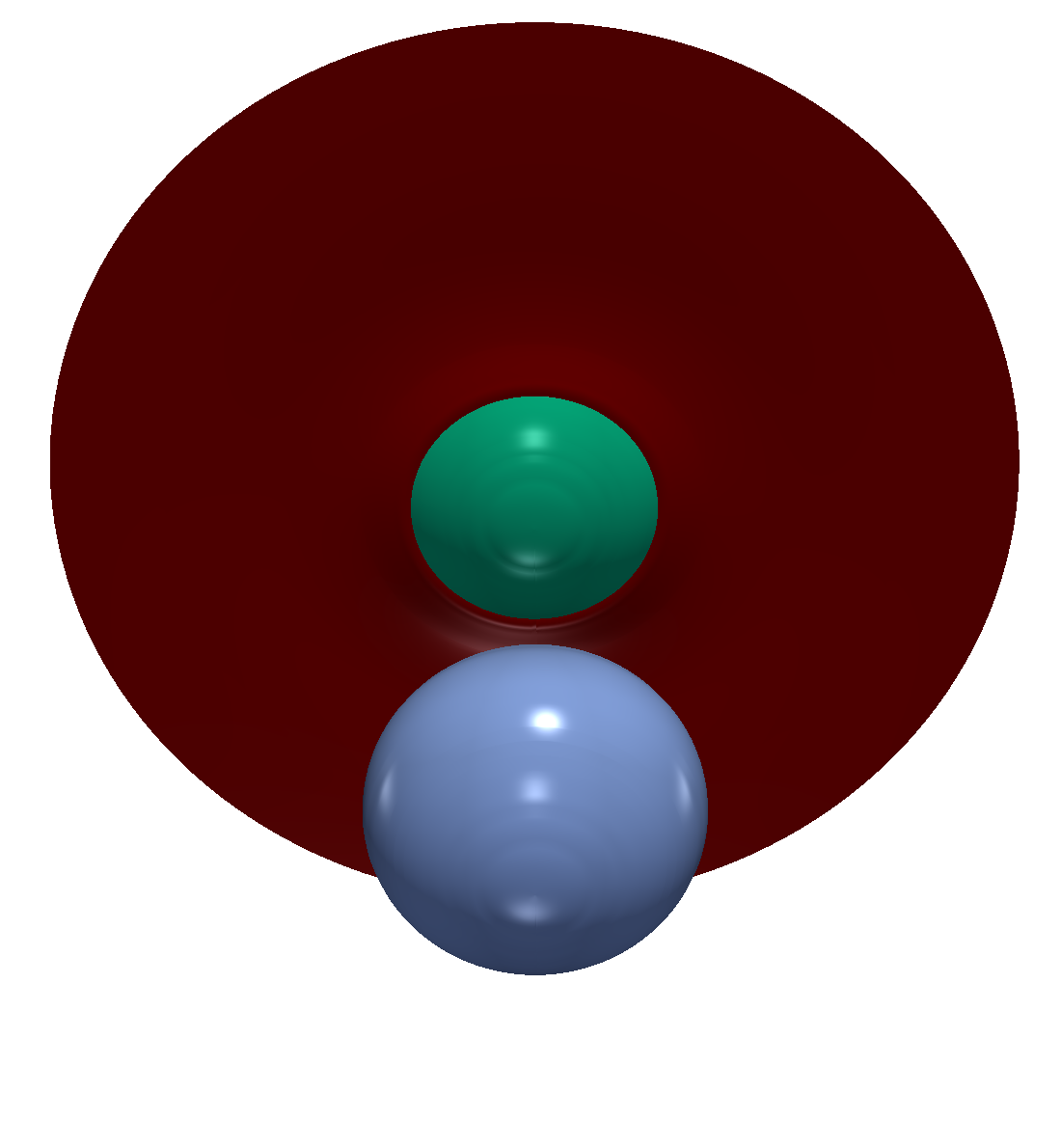}
\caption{$t=0.05$}
\end{subfigure}
\begin{subfigure}{0.245\textwidth}
\centering
\includegraphics[width=\textwidth]{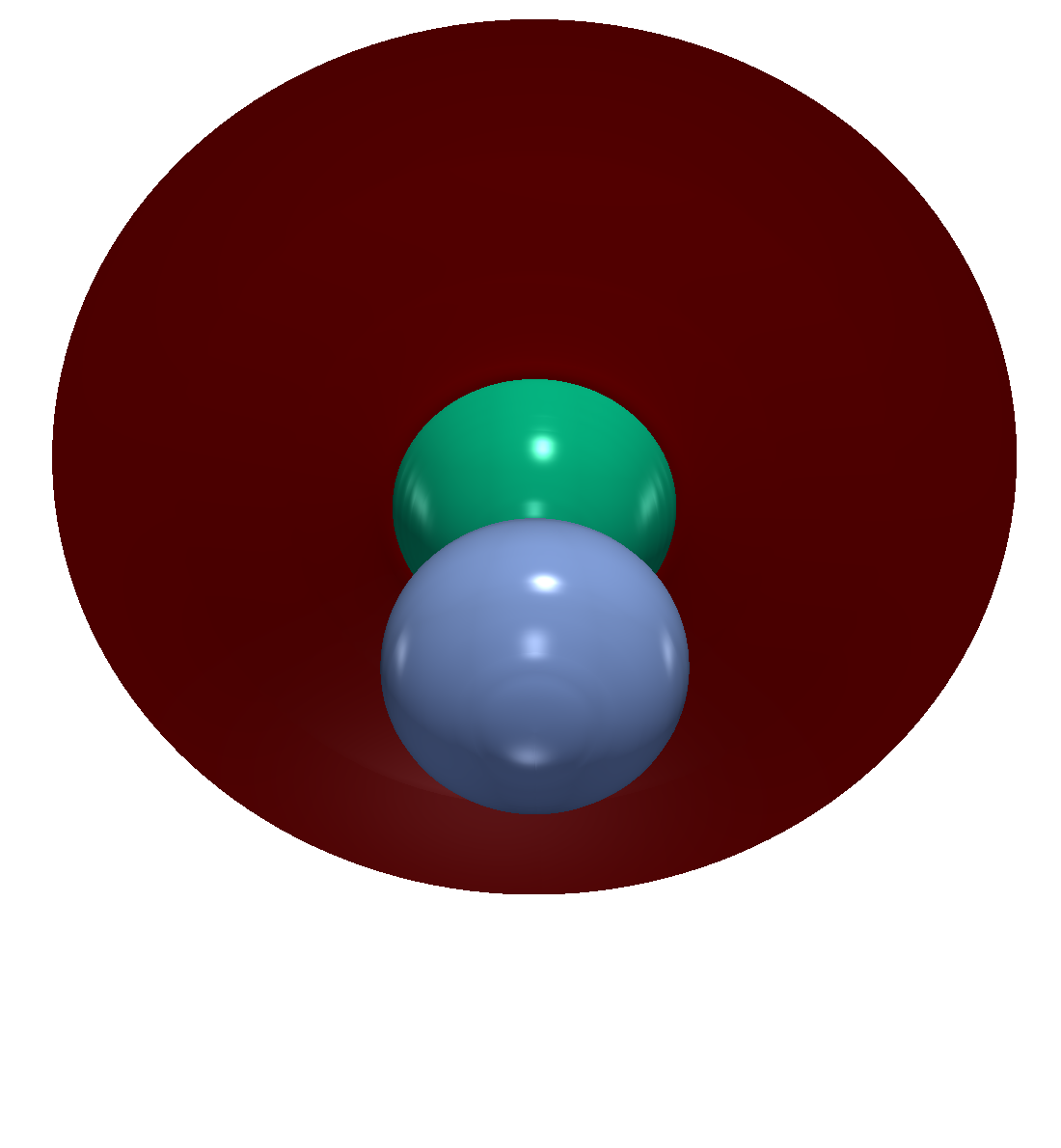}
\caption{$t=0.10$}
\end{subfigure}
\begin{subfigure}{0.245\textwidth}
\centering
\includegraphics[width=\textwidth]{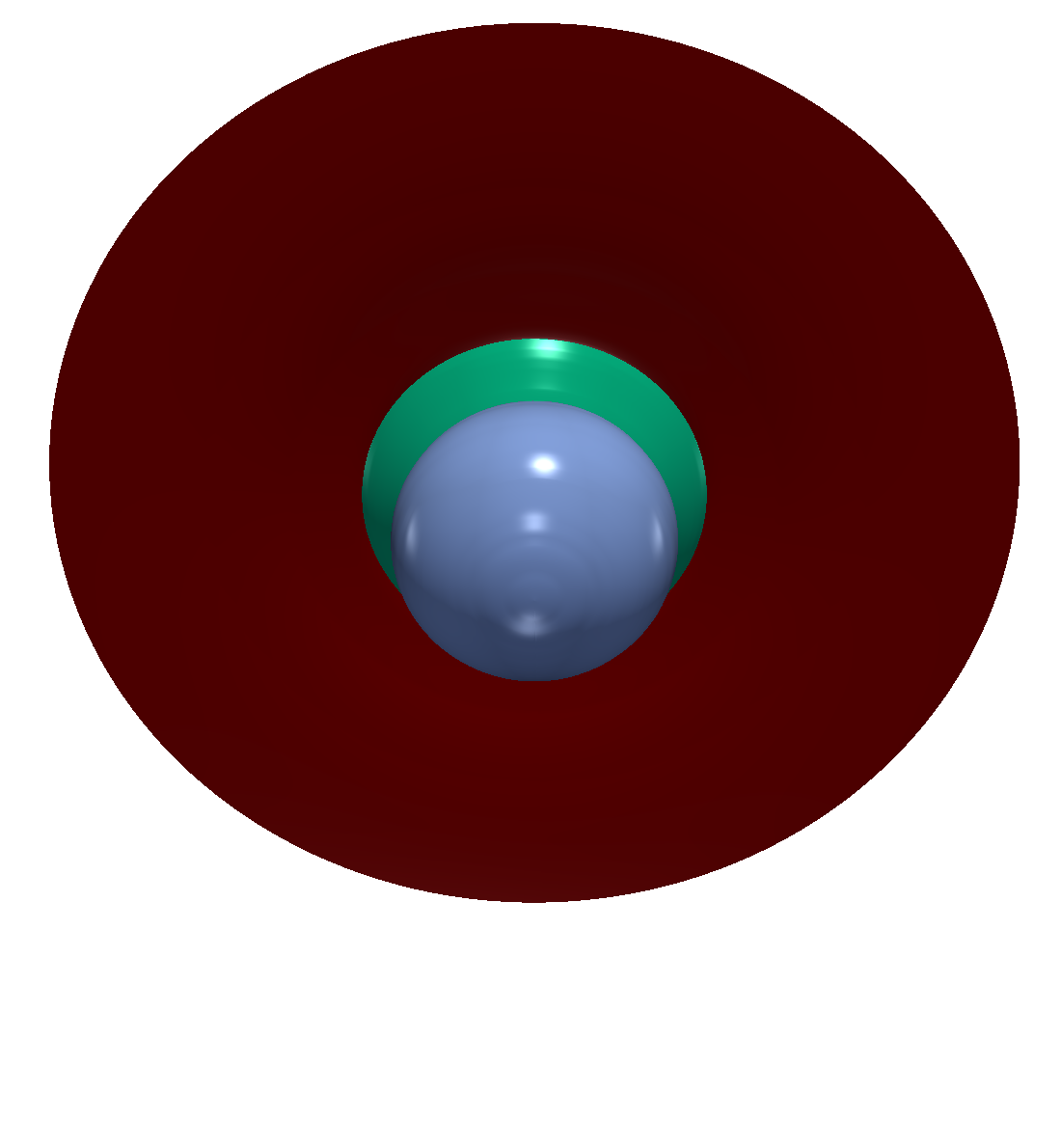}
\caption{$t=0.15$}
\end{subfigure}
\begin{subfigure}{0.245\textwidth}
\centering
\includegraphics[width=\textwidth]{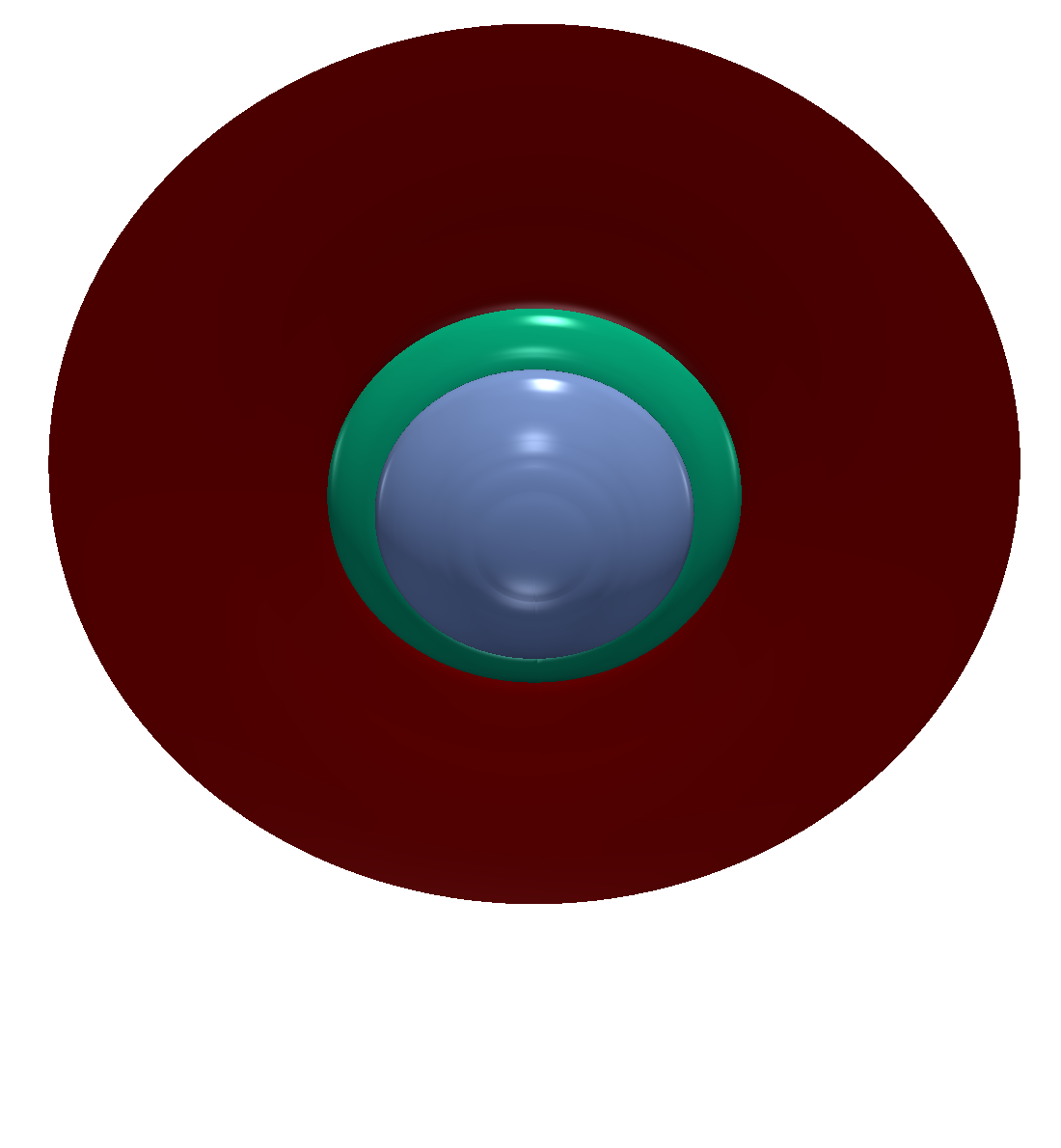}
\caption{$t=0.20$}
\end{subfigure}
\begin{subfigure}{0.245\textwidth}
\centering
\includegraphics[width=\textwidth]{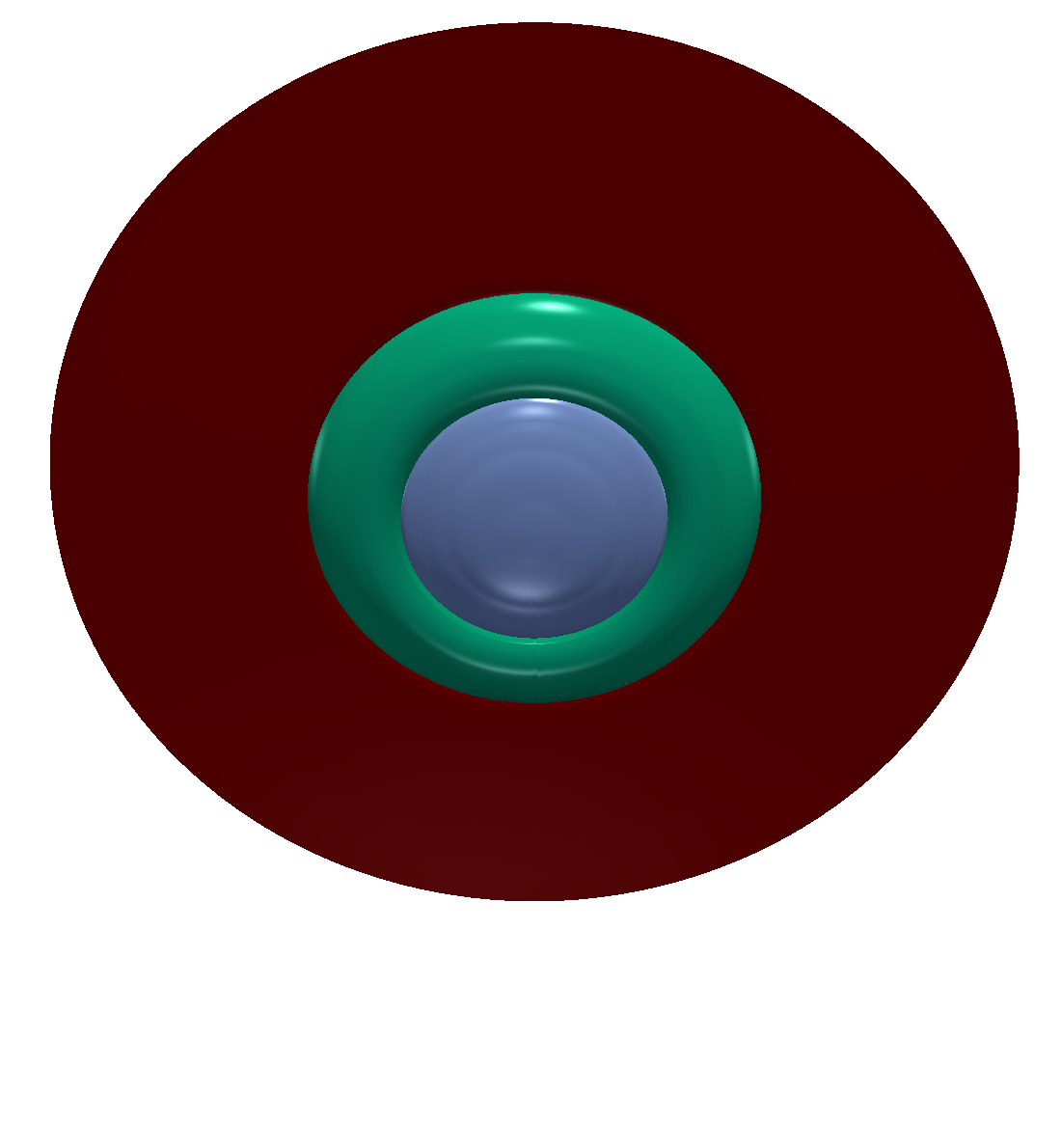}
\caption{$t=0.25$}
\end{subfigure}
\begin{subfigure}{0.245\textwidth}
\centering
\includegraphics[width=\textwidth]{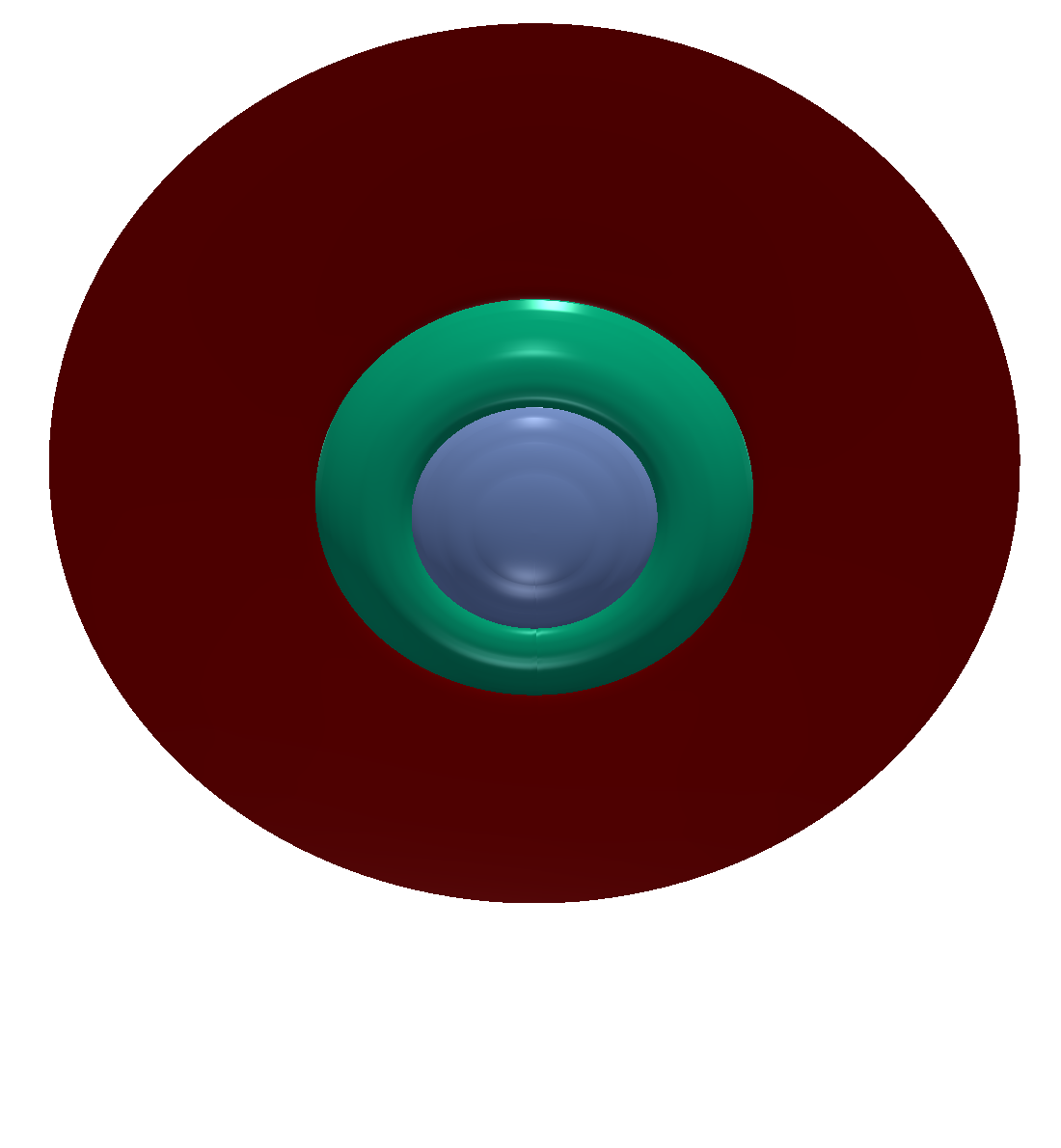}
\caption{$t=0.30$}
\end{subfigure}
\begin{subfigure}{0.245\textwidth}
\centering
\includegraphics[width=\textwidth]{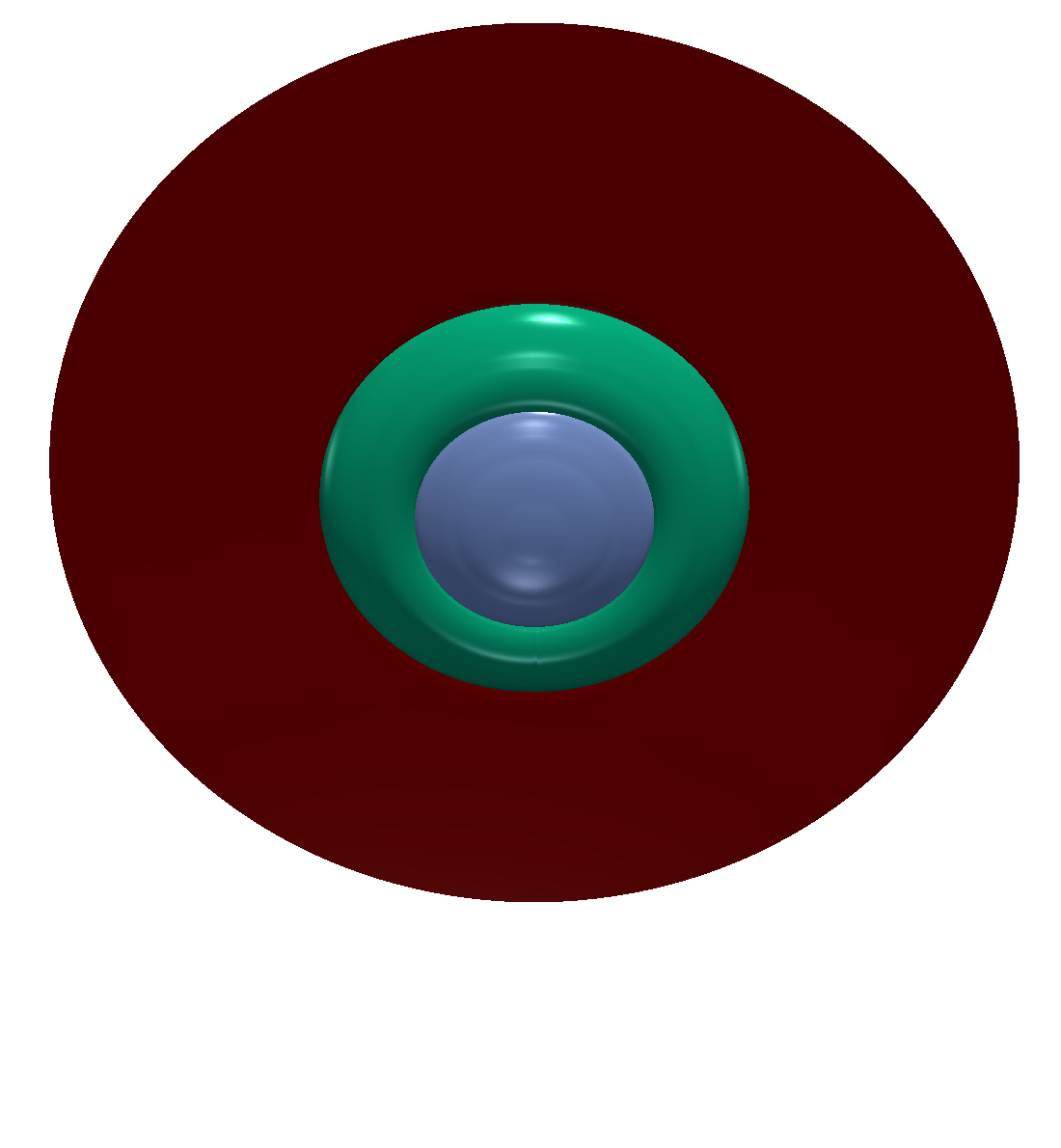}
\caption{$t=0.35$}
\end{subfigure}
\caption{Quaternary droplet-bubble simulation: bottom view. Time snapshots of the phase fields (iso-contours).}
\label{fig: 4phase 2}
\end{figure}

\section{Conclusions}\label{sec: Conclusions}

We have studied constitutive closures for the $N$-phase incompressible Navier--Stokes--Cahn--Hilliard mixture model from the viewpoint of \emph{mixture-awareness}, namely invariance under merging physically identical phases. The central point idea is that reduction should not only be required when a phase is absent, but also when two labels represent the same phase and are merged into a single phase.

Within the model class considered here, this principle has strong consequences. For the free energy, PDE-level reduction consistency under merging identical phases selects an ideal-mixing contribution together with a symmetric pairwise interaction term and a constant quadratic gradient penalty. For the mobility, the same principle yields a pairwise-exchange Onsager structure with bilinear degeneracy in the volume fractions. Thus, rather than proposing one possible closure, the paper identifies a distinguished multiphase closure within the chosen axiomatic setting. We note that the merging of identical phases naturally provides two complementary properties. Namely, the model reduces correctly on faces of the Gibbs simplex when a phase is absent. At the same time, it ensures that an absent phase remains absent.

We also introduced a concrete constitutive parameterization and discussed the resulting energy landscapes for distinct and identical phases. This provides a closed model form suitable for both analysis and computation. We numerically demonstrated the mixture-aware properties, and demonstrated the ability of the framework to perform ternary and quaternary simulations.

Several questions for future research remain open. First, the mixture-aware framework presented here has been formulated for an incompressible flow model with a single momentum equation. A natural next step is to apply it to other mixture-theoretic diffuse-interface formulations, including compressible mixtures \cite{ten2026compressible} and mixture models with $N$-momentum equations \cite{ten2024thermodynamically}. Second, it is well known in the $N$-phase flow literature that quaternary coexistence points are typically unstable. In this respect, the present framework seems well suited for studying higher-order junctions and their stability in multiphase diffuse-interface flows. In particular, it would be interesting to clarify whether nominal quaternary coexistence states can persist within a mixture-aware closure, or whether they necessarily relax toward lower-order interfacial configurations. Closely related to this issue, the sharp‑interface asymptotic analysis of the present mixture‑aware closure, and the identification of the corresponding limiting free‑boundary problem, remain open problems. Finally, since the underlying PDE framework was only recently introduced, a rigorous mathematical analysis of the system is still lacking, including results on the existence of weak solutions, well‑posedness, and long‑time behavior.

\appendix
\section{Gibbs simplex}\label{appendix Gibbs simplex}
To visualize states parameterized by non-negative phases subject to a single linear relation, it is convenient to use the \emph{Gibbs simplex}. Consider a set of $N$ non-negative fractions $\{\phi_{\alpha}\}_{\alpha=1}^{N}$ satisfying
\begin{align}
\phi_{\alpha} &\ge 0, \qquad \alpha=1,\dots,N, \\
\sum_{\alpha=1}^{N}\phi_{\alpha} &= 1.
\end{align}
The admissible set is the $(N-1)$-dimensional simplex. The Gibbs triangle corresponds to the case $N=3$, in which the simplex is represented as an equilateral triangle: each vertex is a pure state ($\phi_{\alpha}=1$ for one phase), each edge corresponds to binary mixtures (one phase identically zero), and the interior represents fully ternary states. Lines of constant $\phi_{\alpha}$ are straight and parallel to the edge opposite the $\alpha$th vertex. An example Gibbs triangle with the associated coordinate axes and simplex boundaries is shown in Figure~\ref{fig:gibbs-triangle}.

For $N=3$, we use barycentric coordinates on an equilateral triangle with vertices associated with $\phi_{1}$, $\phi_{2}$, and $\phi_{3}$. A convenient Cartesian embedding is obtained by placing the vertices at
\begin{align}
\boldsymbol{v}_1 = (0,0), \qquad
\boldsymbol{v}_2 = (1,0), \qquad
\boldsymbol{v}_3 = \left(\tfrac{1}{2},\,\tfrac{\sqrt{3}}{2}\right),
\end{align}
so that any composition maps to
\begin{align}
(x,y) &= \sum_{\alpha=1}^{3}\phi_{\alpha}\,\boldsymbol{v}_{\alpha} = \left(\phi_{2}+\tfrac{1}{2}\phi_{3},\ \tfrac{\sqrt{3}}{2}\phi_{3}\right),
\end{align}
with the constraint $\sum_{\alpha=1}^{3}\phi_{\alpha}=1$ enforced implicitly.

For $N=4$, each vertex of the tetrahedron corresponds to a pure state, each edge to a binary mixture, each face to a ternary subsystem ($\phi_{\alpha}=0$ for one phase), and the interior to fully quaternary states. A convenient realization uses barycentric coordinates with respect to four vertices $\{\boldsymbol{v}_{\alpha}\}_{\alpha=1}^{4}\subset{\R}^3$, so that
\begin{align}
(x,y,z) &= \sum_{\alpha=1}^{4}\phi_{\alpha}\,\boldsymbol{v}_{\alpha}.
\end{align}
One simple choice is the regular tetrahedron with vertices
\begin{align}
\boldsymbol{v}_1 = (0,0,0), \quad
\boldsymbol{v}_2 = (1,0,0), \quad
\boldsymbol{v}_3 = \left(\tfrac{1}{2},\,\tfrac{\sqrt{3}}{2},\,0\right), \quad
\boldsymbol{v}_4 = \left(\tfrac{1}{2},\,\tfrac{1}{2\sqrt{3}},\,\sqrt{\tfrac{2}{3}}\right),
\end{align}
which yields the explicit mapping
\begin{align}
x = \phi_{2}+\tfrac{1}{2}\phi_{3}+\tfrac{1}{2}\phi_{4}, \quad
y = \tfrac{\sqrt{3}}{2}\phi_{3}+\tfrac{1}{2\sqrt{3}}\phi_{4}, \quad
z = \sqrt{\tfrac{2}{3}}\phi_{4}.
\end{align}
This embedding enforces $\sum_{\alpha=1}^{4}\phi_{\alpha}=1$ implicitly and makes the lower-dimensional subsimplices immediately visible: setting $\phi_{4}=0$ recovers the Gibbs triangle on the face spanned by $(\phi_{1},\phi_{2},\phi_{3})$, while analogous triangular faces are obtained by setting any one phase to zero.

\section{Proofs of reduction-consistency properties}\label{section: appendix: Proofs}

\begin{theorem}[Reduction by merging identical phases]\label{thm:merge_reduction appendix}
Assume $N\ge 3$ and consider two phases, without loss of generality phases $1$ and $N$, which are identical in the sense that
\begin{subequations}
\begin{align}
\rho_1&=\rho_N,\\
\chi_{1\mA}&=\chi_{N\mA},\qquad \mA=1,\dots,N,\\
\kappa_{1\mA}&=\kappa_{N\mA},\qquad \mA=1,\dots,N,\\
m_{1\mA}&=m_{N\mA},\qquad \mA=2,\dots,N-1.
\end{align}
\end{subequations}
Define the merged variables by
\begin{subequations}
\begin{align}
\widehat\phi_1&:=\phi_1+\phi_N,\\
\widehat\phi_\mA&:=\phi_\mA,\qquad \mA=2,\dots,N-1.
\end{align}
\end{subequations}
Then the $N$-phase NSCH system reduces to the $(N-1)$-phase NSCH system for
\begin{align}
\widehat{\boldsymbol\phi}=(\widehat\phi_1,\widehat\phi_2,\dots,\widehat\phi_{N-1}),
\end{align}
with the reduced coefficients:
\begin{subequations}
\begin{align}
\hat\rho_1&:=\rho_1,\qquad \hat\rho_\mA:=\rho_\mA,\qquad \mA=2,\dots,N-1,\\
\hat\chi_{11}&:=\chi_{11}=\chi_{1N}=\chi_{NN},\\
\hat\chi_{1\mA}&:=\chi_{1\mA}=\chi_{N\mA},\qquad \mA=2,\dots,N-1,\\
\hat\chi_{\mB\mA}&:=\chi_{\mB\mA},\qquad \mB,\mA=2,\dots,N-1,\\
\hat\kappa_{11}&:=\kappa_{11}=\kappa_{1N}=\kappa_{NN},\\
\hat\kappa_{1\mA}&:=\kappa_{1\mA}=\kappa_{N\mA},\qquad \mA=2,\dots,N-1,\\
\hat\kappa_{\mB\mA}&:=\kappa_{\mB\mA},\qquad \mB,\mA=2,\dots,N-1,\\
\hat m_{1\mA}&:=m_{1\mA}=m_{N\mA},\qquad \mA=2,\dots,N-1,\\
\hat m_{\mB\mA}&:=m_{\mB\mA},\qquad \mB,\mA=2,\dots,N-1.
\end{align}
\end{subequations}
\end{theorem}

\begin{proof}
\subsubsection*{Free energy and chemical potentials}
We define the reduced free energy by
\begin{subequations}
\label{eq:Psi_merged}
\begin{align}
\hat\Psi(\hat{\boldsymbol\phi},\nabla\hat{\boldsymbol\phi})
&=
\hat\Psi_0(\hat{\boldsymbol\phi})
+\frac12\sum_{\alpha,\beta=1}^{N-1}\hat\kappa_{\alpha\beta}\nabla\hat\phi_\alpha\cdot\nabla\hat\phi_\beta,\\
\hat\Psi_0(\hat{\boldsymbol\phi})
&=
W\sum_{\alpha=1}^{N-1}\hat\phi_\alpha\log\hat\phi_\alpha
-
W\sum_{\alpha,\beta=1}^{N-1}\hat\chi_{\alpha\beta}\hat\phi_\alpha\hat\phi_\beta.
\end{align}
\end{subequations}
We denote
\begin{align}
\hat\mu_\alpha
=
\partial_{\hat\phi_\alpha}\hat\Psi_0
-\sum_{\beta=1}^{N-1}\hat\kappa_{\alpha\beta}\Delta\hat\phi_\beta,
\qquad
\hat g_\alpha=\frac{\hat\mu_\alpha+\lambda}{\hat\rho_\alpha}.
\end{align}

For every $\mA=2,\dots,N-1$, the chemical potential of the original $N$-phase system is
\begin{align}
\mu_\mA
=
W(\log\phi_\mA+1)
-2W\sum_{\beta=1}^{N}\chi_{\mA\beta}\phi_\beta
-\sum_{\beta=1}^{N}\kappa_{\mA\beta}\Delta\phi_\beta.
\end{align}
Using symmetry of $\chi$ and $\kappa$ together with the assumptions
\begin{align}
\chi_{\mA1}=\chi_{1\mA}=\chi_{N\mA}=\chi_{\mA N},
\qquad
\kappa_{\mA1}=\kappa_{1\mA}=\kappa_{N\mA}=\kappa_{\mA N},
\end{align}
we obtain
\begin{align}
\mu_\mA
&=
W(\log\hat\phi_\mA+1)
-2W\Big(\chi_{\mA1}(\phi_1+\phi_N)+\sum_{\beta=2}^{N-1}\chi_{\mA\beta}\phi_\beta\Big)\notag\\
&~~-\Big(\kappa_{\mA1}(\Delta\phi_1+\Delta\phi_N)+\sum_{\beta=2}^{N-1}\kappa_{\mA\beta}\Delta\phi_\beta\Big)\notag\\
&=
W(\log\hat\phi_\mA+1)
-2W\sum_{\beta=1}^{N-1}\hat\chi_{\mA\beta}\hat\phi_\beta
-\sum_{\beta=1}^{N-1}\hat\kappa_{\mA\beta}\Delta\hat\phi_\beta
=
\hat\mu_\mA.
\end{align}
Hence
\begin{align}
g_\mA=\hat g_\mA,
\qquad
\nabla g_\mA=\nabla\hat g_\mA,
\qquad \mA=2,\dots,N-1.
\end{align}

\subsubsection*{Effective chemical potential of the merged phase}
Since $\rho_1=\rho_N$, we have
\begin{align}
\phi_1\nabla g_1+\phi_N\nabla g_N
=
\frac{1}{\rho_1}\Big(\phi_1\nabla(\mu_1+\lambda)+\phi_N\nabla(\mu_N+\lambda)\Big).
\end{align}
Using the explicit form of the chemical potentials, we obtain
\begin{align}
\phi_1\nabla\mu_1+\phi_N\nabla\mu_N
&=
W\phi_1\nabla(\log\phi_1+1)+W\phi_N\nabla(\log\phi_N+1)\notag\\
&\quad
-2W\phi_1\sum_{\beta=1}^{N}\chi_{1\beta}\nabla\phi_\beta
-2W\phi_N\sum_{\beta=1}^{N}\chi_{N\beta}\nabla\phi_\beta\notag\\
&\quad
-\phi_1\sum_{\beta=1}^{N}\kappa_{1\beta}\nabla\Delta\phi_\beta
-\phi_N\sum_{\beta=1}^{N}\kappa_{N\beta}\nabla\Delta\phi_\beta.
\end{align}
By the assumptions on the coefficients and the identities
\begin{subequations}
\begin{align}
\phi_1\nabla(\log\phi_1+1)+\phi_N\nabla(\log\phi_N+1)
=
\nabla\phi_1+\nabla\phi_N
=
\nabla\hat\phi_1,\\
\chi_{11}\nabla\phi_1+\chi_{1N}\nabla\phi_N
=
\hat\chi_{11}\nabla(\phi_1+\phi_N)
=
\hat\chi_{11}\nabla\hat\phi_1,\\
\kappa_{11}\nabla\Delta\phi_1+\kappa_{1N}\nabla\Delta\phi_N
=
\hat\kappa_{11}\nabla\Delta(\phi_1+\phi_N)
=
\hat\kappa_{11}\nabla\Delta\hat\phi_1,
\end{align}
\end{subequations}
this becomes
\begin{align}
\phi_1\nabla\mu_1+\phi_N\nabla\mu_N
=
W\nabla\hat\phi_1
-2W\hat\phi_1\sum_{\beta=1}^{N-1}\hat\chi_{1\beta}\nabla\hat\phi_\beta
-\hat\phi_1\sum_{\beta=1}^{N-1}\hat\kappa_{1\beta}\nabla\Delta\hat\phi_\beta.
\end{align}
On the other hand, the reduced chemical potential of the merged phase is
\begin{align}
\hat\mu_1
=
W(\log\hat\phi_1+1)
-2W\sum_{\beta=1}^{N-1}\hat\chi_{1\beta}\hat\phi_\beta
-\sum_{\beta=1}^{N-1}\hat\kappa_{1\beta}\Delta\hat\phi_\beta,
\end{align}
and therefore
\begin{align}
\hat\phi_1\nabla\hat\mu_1
=
W\nabla\hat\phi_1
-2W\hat\phi_1\sum_{\beta=1}^{N-1}\hat\chi_{1\beta}\nabla\hat\phi_\beta
-\hat\phi_1\sum_{\beta=1}^{N-1}\hat\kappa_{1\beta}\nabla\Delta\hat\phi_\beta.
\end{align}
Hence
\begin{align}
\phi_1\nabla\mu_1+\phi_N\nabla\mu_N
=
\hat\phi_1\nabla\hat\mu_1,
\end{align}
and consequently
\begin{align}
\phi_1\nabla g_1+\phi_N\nabla g_N
=
\hat\phi_1\nabla\hat g_1.
\label{eq:merged_g_identity}
\end{align}

\subsubsection*{Diffusive flux of the merged phase}
Adding the phase equations for $\phi_1$ and $\phi_N$ gives
\begin{align}
\partial_t\hat\phi_1+\operatorname{div}(\hat\phi_1\vv)+\rho_1^{-1}\operatorname{div}(\bJ_1+\bJ_N)=0.
\end{align}
We now show that
\begin{align}
\bJ_1+\bJ_N
=
-\sum_{\beta=1}^{N-1}\hat M_{1\beta}\nabla\hat g_\beta,
\end{align}
where the reduced mobility has the same form,
\begin{align}
\hat M_{\alpha\beta}
=
\begin{cases}
-\hat m_{\alpha\beta}\hat\phi_\alpha\hat\phi_\beta,&\alpha\neq\beta,\\[0.4em]
\displaystyle\sum_{\gamma\neq\alpha}\hat m_{\alpha\gamma}\hat\phi_\alpha\hat\phi_\gamma,&\alpha=\beta.
\end{cases}
\end{align}

For $\beta=2,\dots,N-1$, using $m_{1\beta}=m_{N\beta}$ we obtain
\begin{align}
M_{1\beta}+M_{N\beta}
=
-m_{1\beta}\phi_1\phi_\beta-m_{N\beta}\phi_N\phi_\beta
=
-\hat m_{1\beta}\hat\phi_1\hat\phi_\beta
=
\hat M_{1\beta}.
\end{align}
For the diagonal part,
\begin{align}
M_{11}+M_{N1}
&=
\sum_{\gamma\neq 1}m_{1\gamma}\phi_1\phi_\gamma-m_{N1}\phi_N\phi_1=
\sum_{\gamma=2}^{N-1}m_{1\gamma}\phi_1\phi_\gamma,
\end{align}
and similarly
\begin{align}
M_{1N}+M_{NN}
=
\sum_{\gamma=2}^{N-1}m_{1\gamma}\phi_N\phi_\gamma.
\end{align}
Therefore
\begin{align}
(M_{11}+M_{N1})\nabla g_1+(M_{1N}+M_{NN})\nabla g_N =&~
\left(\sum_{\gamma=2}^{N-1}m_{1\gamma}\hat\phi_\gamma\right)
\big(\phi_1\nabla g_1+\phi_N\nabla g_N\big)\notag\\
=&~
\left(\sum_{\gamma=2}^{N-1}\hat m_{1\gamma}\hat\phi_\gamma\right)
\hat\phi_1\nabla\hat g_1
=
\hat M_{11}\nabla\hat g_1,
\end{align}
where we used \eqref{eq:merged_g_identity}. Altogether,
\begin{align}
\bJ_1+\bJ_N
=
-\sum_{\beta=1}^{N-1}\hat M_{1\beta}\nabla\hat g_\beta.
\end{align}
Thus the merged phase $\hat\phi_1$ satisfies exactly the reduced phase equation.

\subsubsection*{Diffusive fluxes of the remaining phases}
Let $\alpha\in\{2,\dots,N-1\}$. Then
\begin{align}
\bJ_\alpha
=
-\sum_{\beta=1}^{N}M_{\alpha\beta}\nabla g_\beta
=
-\sum_{\beta=2}^{N-1}M_{\alpha\beta}\nabla g_\beta
-M_{\alpha1}\nabla g_1
-M_{\alpha N}\nabla g_N.
\end{align}
Using $m_{\alpha1}=m_{\alpha N}$ and \eqref{eq:merged_g_identity}, we find
\begin{align}
M_{\alpha1}\nabla g_1+M_{\alpha N}\nabla g_N
&=
-m_{\alpha1}\phi_\alpha\phi_1\nabla g_1-m_{\alpha N}\phi_\alpha\phi_N\nabla g_N\notag\\
&=
-\hat m_{\alpha1}\hat\phi_\alpha\big(\phi_1\nabla g_1+\phi_N\nabla g_N\big)\notag\\
&=
-\hat m_{\alpha1}\hat\phi_\alpha\hat\phi_1\nabla\hat g_1
=
\hat M_{\alpha1}\nabla\hat g_1.
\end{align}
Since $M_{\alpha\beta}=\hat M_{\alpha\beta}$ and $\nabla g_\beta=\nabla\hat g_\beta$ for $\beta=2,\dots,N-1$, it follows that
\begin{align}
\bJ_\alpha
=
-\sum_{\beta=1}^{N-1}\hat M_{\alpha\beta}\nabla\hat g_\beta,
\qquad
\alpha=2,\dots,N-1.
\end{align}
Hence all remaining phase equations also reduce to the $(N-1)$-phase system.

\subsubsection*{Capillary force}
Finally, the capillary force in the momentum equation satisfies
\begin{align}
\sum_{\beta=1}^{N}\phi_\beta\nabla(\mu_\beta+\lambda)
&=
\phi_1\nabla(\mu_1+\lambda)+\phi_N\nabla(\mu_N+\lambda)
+\sum_{\beta=2}^{N-1}\phi_\beta\nabla(\mu_\beta+\lambda)\notag\\
&=
\hat\phi_1\nabla(\hat\mu_1+\lambda)
+\sum_{\beta=2}^{N-1}\hat\phi_\beta\nabla(\hat\mu_\beta+\lambda)\notag\\
&=
\sum_{\beta=1}^{N-1}\hat\phi_\beta\nabla(\hat\mu_\beta+\lambda).
\end{align}
This is precisely the capillary force of the reduced $(N-1)$-phase system.

The above identities show that the merged variables $\hat{\boldsymbol\phi}$ satisfy exactly the $(N-1)$-phase NSCH system with the reduced constitutive coefficients. This proves the claim.
\end{proof}

\begin{theorem}[Reduction for absent phase]\label{thm:face_reduction appendix}

If a solution of the $N$-phase NSCH system satisfies $\phi_\gamma\equiv0$,
then the restriction of the $N$-phase NSCH system to $\mathcal G_\gamma$ coincides with the $(N-1)$-phase NSCH system obtained by deleting phase $\gamma$ and canonically relabeling the remaining indices.
\end{theorem}

\begin{proof}
\subsubsection*{Free energies}
We first identify the reduced free energy. Denote by
\begin{align}
\widehat{\boldsymbol\phi}
:=
(\phi_1,\dots,\phi_{\gamma-1},\phi_{\gamma+1},\dots,\phi_N)
\end{align}
the vector of active phases. Define
\begin{subequations}
\label{eq:Psi_reduced}
\begin{align}
\widehat\Psi(\widehat{\boldsymbol\phi},\nabla\widehat{\boldsymbol\phi})
&=
\widehat\Psi_0(\widehat{\boldsymbol\phi})
+\frac12\sum_{\substack{\alpha,\beta=1\\ \alpha,\beta\neq \gamma}}^{N}
\kappa_{\alpha\beta}\nabla\phi_\alpha\cdot\nabla\phi_\beta,\\
\widehat\Psi_0(\widehat{\boldsymbol\phi})
&=
W\sum_{\substack{\alpha=1\\ \alpha\neq \gamma}}^{N}\phi_\alpha\log\phi_\alpha
-
W\sum_{\substack{\alpha,\beta=1\\ \alpha,\beta\neq \gamma}}^{N}
\chi_{\alpha\beta}\phi_\alpha\phi_\beta.
\end{align}
\end{subequations}
Because $\phi_\gamma=0$ and $\nabla\phi_\gamma=0$, the restriction of \eqref{eq:Psi_final} to $\mathcal G_\gamma$ is exactly \eqref{eq:Psi_reduced}.

\subsubsection*{Chemical potentials} For every active index $\alpha\neq \gamma$, the chemical potential of the $N$-phase system is
\begin{align}
\mu_\alpha
= W(\log\phi_\alpha+1)
-2W\sum_{\beta=1}^{N}\chi_{\alpha\beta}\phi_\beta
-\sum_{\beta=1}^{N}\kappa_{\alpha\beta}\Delta\phi_\beta.
\end{align}
Restricting to $\mathcal G_\gamma$ gives
\begin{align}
\mu_\alpha
=
W(\log\phi_\alpha+1)
-2W\sum_{\substack{\beta=1\\ \beta\neq \gamma}}^{N}\chi_{\alpha\beta}\phi_\beta
-\sum_{\substack{\beta=1\\ \beta\neq \gamma}}^{N}\kappa_{\alpha\beta}\Delta\phi_\beta,
\qquad \alpha\neq \gamma.
\end{align}
On the other hand, the reduced chemical potential associated with \eqref{eq:Psi_reduced} is
\begin{align}
\widehat\mu_\alpha
=
W(\log\phi_\alpha+1)
-2W\sum_{\substack{\beta=1\\ \beta\neq \gamma}}^{N}\chi_{\alpha\beta}\phi_\beta
-\sum_{\substack{\beta=1\\ \beta\neq \gamma}}^{N}\kappa_{\alpha\beta}\Delta\phi_\beta,
\qquad \alpha\neq \gamma,
\end{align}
As a consequence, the chemical potentials coincide: $
\mu_\alpha=\hat{\mu}_\alpha,
\alpha\neq \gamma$, and through the identification $\lambda = \hat{\lambda}$ we have $
g_\alpha=\hat{g}_\alpha, \alpha\neq \gamma$, so the effective chemical potentials of all active phases agree with those of the reduced system.

\subsubsection*{Diffusive fluxes of active phases} Next, we analyze the diffusive fluxes. For the active phases we have:
\begin{align}
\bJ_\alpha=-\sum_{\beta=1}^{N}M_{\alpha\beta}\nabla g_\beta
=
-\sum_{\substack{\beta=1\\ \beta\neq \gamma}}^{N}M_{\alpha\beta}\nabla g_\beta
-
M_{\alpha \gamma}\nabla g_\gamma, \qquad \alpha\neq \gamma. 
\end{align}
Invoking $M_{\alpha \gamma}=-m_{\alpha \gamma}\phi_\alpha\phi_\gamma$, we find
\begin{align}
M_{\alpha \gamma}\nabla g_\gamma
&=
-\frac{m_{\alpha \gamma}\phi_\alpha}{\rho_\gamma}
\left(
W\nabla\phi_\gamma
-2W\phi_\gamma\sum_{\beta=1}^{N}\chi_{\gamma\beta}\nabla\phi_\beta
-\phi_\gamma\sum_{\beta=1}^{N}\kappa_{\gamma\beta}\nabla\Delta\phi_\beta
+\phi_\gamma\nabla\lambda
\right).
\end{align}
Since $\phi_\gamma=\nobreak0$ and $\nabla\phi_\gamma=\nobreak0$, we obtain
$M_{\alpha \gamma}\nabla g_\gamma\equiv0$, and thus
\begin{align}
\bJ_\alpha
=
-\sum_{\substack{\beta=1\\ \beta\neq \gamma}}^{N}M_{\alpha\beta}\nabla g_\beta,
\qquad \alpha\neq \gamma.
\end{align}
For $\alpha,\beta\neq \gamma$, the coefficients $M_{\alpha\beta}$ are exactly the mobility coefficients of the reduced system, because the deleted phase contributes nothing when $\phi_\gamma=0$. Since also $\nabla g_\beta=\rho_\beta^{-1}\nabla(\widehat\mu_\beta+\lambda)$ for $\beta\neq \gamma$, the phase equations for all active phases coincide with those of the reduced $(N-1)$-phase model.

\subsubsection*{Diffusive fluxes of the absent phase} We now consider the $\mA$th phase equation. The diffusive flux is given by
\begin{align}
\bJ_\mA=-\sum_{\beta=1}^{N}M_{\mA\beta}\nabla g_\beta,
\end{align}
where for $\beta\neq \mA$ we have $
M_{\mA\beta}=-m_{\gamma\beta}\phi_\mA\phi_\beta=0$ and for the diagonal term, since
$M_{\mA\mA}=\sum_{\delta\neq \mA}m_{\mA\delta}\phi_\gamma\phi_\delta$, we compute 
\begin{align}
M_{\mA\mA}\nabla g_\mA
%&=
%\frac{\phi_\mA}{\rho_\mA}\left(\sum_{\gamma\neq \mA}m_{\mA\gamma}\phi_\gamma\right)
%\left(
%W\frac{\nabla\phi_\mA}{\phi_\mA}
%-2W\sum_{\beta=1}^{N}\chi_{\mA\beta}\nabla\phi_\beta
%-\sum_{\beta=1}^{N}\kappa_{\mA\beta}\nabla\Delta\phi_\beta
%+\nabla\lambda
%\right)\\
&=
\frac1{\rho_\gamma}\left(\sum_{\delta\neq \gamma}m_{\gamma\delta}\phi_\delta\right)\nonumber\\
&~~~\times\left(
W\nabla\phi_\gamma
-2W\phi_\gamma\sum_{\beta=1}^{N}\chi_{\gamma\beta}\nabla\phi_\beta
-\phi_\gamma\sum_{\beta=1}^{N}\kappa_{\gamma\beta}\nabla\Delta\phi_\beta
+\phi_\gamma\nabla\lambda
\right).
\end{align}
Again $\phi_\gamma=\nobreak0$ and $\nabla\phi_\gamma=\nobreak0$, hence
\begin{align}
M_{\gamma\gamma}\nabla g_\gamma\equiv0,
\qquad
\bJ_\gamma\equiv0.
\end{align}
Thus the $\gamma$th phase equation decouples completely.

\subsubsection*{Capillary force} Finally, the capillary force in the momentum equation reduces as
\begin{align}
\sum_{\beta=1}^{N}\phi_\beta\nabla(\mu_\beta+\lambda)
=
\sum_{\substack{\beta=1\\ \beta\neq \gamma}}^{N}\phi_\beta\nabla(\mu_\beta+\lambda)
=
\sum_{\substack{\beta=1\\ \beta\neq \gamma}}^{N}\phi_\beta\nabla(\widehat\mu_\beta+\lambda),
\end{align}
because $\phi_\gamma=0$ and $\mu_\beta=\widehat\mu_\beta$ for all $\beta\neq \gamma$. This is precisely the capillary force of the reduced system.

The above conclusions show that after deleting the inactive phase $\gamma$, the restricted $N$-phase NSCH system coincides with the $(N-1)$-phase NSCH system.
\end{proof}

\begin{corollary}[Absent phase remains absent]\label{cor:absence_invariance appendix}
If a phase is absent, i.e. $
\phi_\gamma(\cdot,0)\equiv 0$,
then it remains absent, i.e. $
\phi_\gamma(\cdot,t)\equiv 0
\qquad\text{for all }t\ge 0$.
\end{corollary}

\begin{proof}
Consider the $\gamma$th phase equation,
\begin{align}
\partial_t\phi_\gamma+\operatorname{div}(\phi_\gamma\vv)+\rho_\gamma^{-1}\operatorname{div}\bJ_\gamma=0,
\qquad
\bJ_\gamma=-\sum_{\beta=1}^{N}M_{\gamma\beta}\nabla g_\beta.
\end{align}
By Theorem~\ref{thm:face_reduction appendix}, whenever $\phi_\gamma\equiv 0$ the restriction of the system to the face $\mathcal G_\gamma$ is well defined and the $\gamma$th diffusive flux vanishes identically: $
\bJ_\gamma\equiv 0$. Hence, on $\mathcal G_\gamma$, the $\gamma$th phase equation reduces to
\begin{align}
\partial_t\phi_\gamma+\operatorname{div}(\phi_\gamma\vv)=0.
\end{align}
The identically zero function $\phi_\gamma\equiv 0$ is a solution of this equation and attains the initial datum
$\phi_\gamma(\cdot,0)\equiv 0$. Therefore, by uniqueness of classical solutions, it follows that
\begin{align}
\phi_\gamma(\cdot,t)\equiv 0
\qquad\text{for all }t\ge 0.
\end{align}
%Thus $\mathcal G_\gamma$ is an invariant manifold of the NSCH flow.
\end{proof}

\begin{proposition}[Boyer--Lapuerta free energies are not mixture-aware]\label{appendix prop: Boyer}
For $N=3$, the Boyer--Lapuerta free energy is not mixture-aware. 
\end{proposition}
\begin{proof}
For $N=3$, the free energy of Boyer--Lapuerta is
\begin{align}
\Psi^{\mathrm{B}}(\phi,\nabla \phi)
=
\frac{12}{\varepsilon}F_\Lambda(\phi_1,\phi_2,\phi_3)
+
\frac{3\varepsilon}{8}\sum_{\alpha=1}^3 \Sigma_\alpha |\nabla \phi_\alpha|^2,
\end{align}
with
\begin{align}
F_\Lambda
=&~
\sigma_{12}\phi_1^2\phi_2^2+\sigma_{13}\phi_1^2\phi_3^2+\sigma_{23}\phi_2^2\phi_3^2\nn\\
&~+\phi_1\phi_2\phi_3(\Sigma_1\phi_1+\Sigma_2\phi_2+\Sigma_3\phi_3)
+\Lambda \phi_1^2\phi_2^2\phi_3^2,
\end{align}
and
\begin{align}
\Sigma_1=\sigma_{12}+\sigma_{13}-\sigma_{23},\qquad
\Sigma_2=\sigma_{12}+\sigma_{23}-\sigma_{13},\qquad
\Sigma_3=\sigma_{13}+\sigma_{23}-\sigma_{12}.
\end{align}
If phases $1$ and $2$ are physically identical, then capillary-force reduction consistency would require that, after merging,
\begin{align}
\phi_1\nabla\mu_1+\phi_2\nabla\mu_2
=
\hat\phi_1\nabla\hat\mu_1,
\qquad
\hat\phi_1=\phi_1+\phi_2.
\end{align}
This does not hold. Indeed, on the face $\phi_3=0$ the bulk term reduces to
\begin{align}
F_\Lambda(\phi_1,\phi_2,0)=\sigma_{12}\phi_1^2\phi_2^2,
\end{align}
so that already at the bulk level
\begin{align}
\partial_{\phi_1}F_\Lambda(\phi_1,\phi_2,0)=2\sigma_{12}\phi_1\phi_2^2,
\qquad
\partial_{\phi_2}F_\Lambda(\phi_1,\phi_2,0)=2\sigma_{12}\phi_1^2\phi_2.
\end{align}
Hence the capillary-force contribution
\begin{align}
\phi_1\nabla\partial_{\phi_1}F_\Lambda+\phi_2\nabla\partial_{\phi_2}F_\Lambda = 6\sigma_{12}\phi_1^2\phi_2^2\nabla(\phi_1\phi_2)
\end{align}
still depends on the redistribution between $\phi_1$ and $\phi_2$, rather than only on the merged variable $\hat\phi_1=\phi_1+\phi_2$. The separate gradient terms $|\nabla\phi_1|^2$ and $|\nabla\phi_2|^2$ also depend on the redistribution between $\phi_1$ and $\phi_2$.

\end{proof}

\begin{proposition}[Steinbach free energy is not mixture-aware]\label{appendix prop: Steinbach}
The free energy structure
\begin{subequations}
   \begin{align}
\Psi =&~ \Psi_0(\boldsymbol{\phi}) + \Psi_{\nabla}(\boldsymbol{\phi},\nabla \boldsymbol{\phi}),\\
\Psi_{\nabla}(\boldsymbol{\phi},\nabla \boldsymbol{\phi}) = &~ \sum_{\alpha<\beta}
\omega_{\alpha\beta} \bigl|\phi_\alpha\nabla \phi_\beta-\phi_\beta\nabla \phi_\alpha\bigr|^2,
\end{align}
\end{subequations}
for coefficients $\omega_{\alpha\beta}\geq 0$ is not mixture-aware.
\end{proposition}

\begin{proof}
Consider the ternary case $N=3$. We consider the phases $1$ and $3$ to represent the same phase and set 
\begin{align}
\hat\phi_1:=\phi_1+\phi_3,
\qquad
\hat\phi_2:=\phi_2.
\end{align}
The capillary force should satisfy
\begin{align}
\phi_1\nabla\mu_1^{\nabla}+\phi_2\nabla\mu_2^{\nabla}+\phi_3\nabla\mu_3^{\nabla}
=
\hat\phi_1\nabla\hat\mu_1^{\nabla}+\hat\phi_2\nabla\hat\mu_2^{\nabla},
\label{eq:SG_RC}
\end{align}
where $\mu_\alpha^\nabla$ and $\hat\mu_\alpha^\nabla$ denote the chemical potentials induced by $\Psi_\nabla$ alone. Restricting to the face $\phi_2=0$ we have
\begin{align}
\hat\phi_1=\phi_1+\phi_3=1,
\qquad
\hat\phi_2=\phi_2=0, \qquad \nabla \hat\phi_1=\nabla \hat\phi_2=0.
\end{align}
Hence the reduced gradient chemical potentials are evaluated at a constant state, so their gradients vanish $
\nabla \hat\mu_1^\nabla=\nabla \hat\mu_2^\nabla=0$, i.e. the capillary force vanishes as well:
\begin{align}
\hat\phi_1\nabla\hat\mu_1^\nabla+\hat\phi_2\nabla\hat\mu_2^\nabla=0.
\label{eq:SG_reduced_zero}
\end{align}
On the face $\phi_2=0$ we have $\nabla \phi_2=0$, and therefore on this face we have
\begin{align}
  \Psi_{\nabla} = \omega_{13}\bigl|\phi_1\nabla \phi_3-\phi_3\nabla \phi_1\bigr|^2    
\end{align}
The corresponding chemical potentials are
\begin{subequations}
    \begin{align}
\mu_1^\nabla
=&~
2\omega_{13}
\bigl(\phi_1\nabla \phi_3-\phi_3\nabla \phi_1\bigr)\cdot\nabla\phi_3
+
2\omega_{13}\div\!\Big(\phi_3(\phi_1\nabla \phi_3-\phi_3\nabla \phi_1)\Big),
\\
\mu_3^\nabla
=&~
-2\omega_{13}
\bigl(\phi_1\nabla \phi_3-\phi_3\nabla \phi_1\bigr)\cdot\nabla\phi_1
-
2\omega_{13}\div\!\Big(\phi_1(\phi_1\nabla \phi_3-\phi_3\nabla \phi_1)\Big).
\end{align}
\end{subequations}
On the face $\phi_2=0$, one has $\phi_3=1-\phi_1$ and hence
\begin{align}
\nabla\phi_3=-\nabla\phi_1,
\qquad
\phi_1\nabla \phi_3-\phi_3\nabla \phi_1
=
-(\phi_1+\phi_3)\nabla\phi_1
=
-\nabla\phi_1.
\end{align}
Substituting gives
\begin{align}
\phi_1\nabla\mu_1^\nabla+\phi_3\nabla\mu_3^\nabla
&=
4\omega_{13}\nabla |\nabla\phi_1|^2
+
2\omega_{13}(\Delta\phi_1)\nabla\phi_1,
\end{align}
which does in general not vanish nor can it be absorbed into the pressure.
\end{proof}

\section{Connection to the standard two-phase mobility tensor}\label{sec: connection with binary model}
In the binary case ($N=2$), $M\mathbf{1}=\mathbf{0}$ implies $M_{11}=-M_{12}$ and $M_{22}=-M_{21}$, so that
\begin{align}\label{eq:binary_flux_exchange}
  \bJ_1 = M_{12}\nabla(g_1-g_2),
  \qquad
  \bJ_2 = -\bJ_1.
\end{align}
Introducing the order parameter $\varphi:=\phi_1-\phi_2$ (hence $\phi_1=\tfrac12(1+\varphi)$ and $\phi_2=\tfrac12(1-\varphi)$),
subtracting the phase equations yields a single evolution equation
for $\varphi$ of the form
\begin{align}\label{eq:binary_phi_order_parameter}
  \partial_t\varphi+\div(\varphi\vv)
  +\div\!\Big((\rho_1^{-1}+\rho_2^{-1})\bJ_1\Big)=0.
\end{align}
Substituting \eqref{eq:binary_flux_exchange} gives
\begin{align}\label{eq:binary_phi_mobility_identification}
  \partial_t\varphi+\div(\varphi\vv)
  +\div\!\Big((\rho_1^{-1}+\rho_2^{-1})M\nabla(g_1-g_2)\Big)=0,
\end{align}
where $M=M_{12}=M_{21}$. In order to compare with the
standard two-phase order-parameter formulation \cite{eikelder2023unified,brunk2026simple}, we introduce a rescaled mobility:
\begin{align}\label{eq:binary_mobility_rescaling}
  \breve{M} := -(\rho_1^{-1}+\rho_2^{-1})^2 M,
\end{align}
and rewrite the exchange driving force in terms of $\breve{\mu}+\omega\lambda$. A direct calculation yields
\begin{align}\label{eq:gdiff_breve_mu}
  g_1-g_2
  =(\rho_1^{-1}+\rho_2^{-1})(\breve{\mu}+\omega\lambda),
  \qquad
  \breve{\mu}
  :=(\rho_1^{-1}+\rho_2^{-1})^{-1}\big(\rho_1^{-1}\mu_1-\rho_2^{-1}\mu_2\big),
\end{align}
where $\omega:=(\rho_1^{-1}-\rho_2^{-1})/(\rho_1^{-1}+\rho_2^{-1})$. Substituting \eqref{eq:gdiff_breve_mu} into \eqref{eq:binary_phi_mobility_identification} gives
\begin{align}\label{eq:binary_phi_standard_form}
  \partial_t\varphi+\div(\varphi\vv)
  -\div\!\Big(\breve{M}\nabla(\breve{\mu}+\omega\lambda)\Big)=0,
\end{align}
which matches the form of the phase-field equation used in \cite{eikelder2023unified}. 

For the mobility choice \eqref{eq: mob choice} we find:
\begin{align}
  \breve{M}
  = \epsilon m_0(1-\varphi^2).
\end{align}
The constant $\epsilon m_0$ therefore sets the diffusion time scale in the binary reference case. Furthermore, this form reproduces the standard degenerate factor $(1-\varphi^2)$ in the two-phase
order-parameter formulation.

For the Maxwell--Stefan choice \eqref{eq:MS_offdiag_simplified}, we obtain
\begin{align}\label{eq:MS_binary_M12}
  M=M_{12}  = -\epsilon m_0\rho_1\rho_2\phi_1\phi_2
  = -\frac{\epsilon m_0}{4}\rho_1\rho_2(1-\varphi^2),
\end{align}
and hence
\begin{align}\label{eq:breveM_binary}
  \breve{M}
  = \frac{\epsilon m_0}{4}\frac{(\rho_1+\rho_2)^2}{\rho_1\rho_2}(1-\varphi^2).
\end{align}
Up to the overall constant prefactor, this also reproduces the standard degenerate factor $(1-\varphi^2)$. In the case of identical densities we retrieve $ \breve{M} =  \epsilon m_0 (1-\varphi^2)$.

\section{Numerical parameters}\label{sec appendix parameters}

Here we report the calibrated capillarity matrices, interface widths, and scale factors. For the details of the computations we refer to \cite{surfacetension2026}.

For the computations in \cref{subsec: num red} we adopted for Case 1:
\begin{subequations}
    \begin{align}    \bar{\boldsymbol{\kappa}} =&~10^{2}\begin{pmatrix}
  2.6806101676  & -5.3612203352 & 2.6806101676 \\
  -5.3612203352 & 10.722440670  & -5.3612203352 \\
  2.6806101676 & -5.3612203352 & 2.6806101676
\end{pmatrix}\\ \varepsilon_{12} = &~24.3042511, \quad \varepsilon_{13} = \text{nonexistent}, \quad \varepsilon_{23} = 24.3042511, \\
\varepsilon_0 = &~ 3.2144625 \times 10^{-4},
\end{align}
\end{subequations}
and for Case 2:
\begin{subequations}
    \begin{align}
  \bar{\boldsymbol{\kappa}} =&~\begin{pmatrix}
  1.7135806699 &  -3.4271613399 & 1.7135806699\\
-3.4271613399 & 6.8543226797 & -3.4271613399\\
1.7135806699 & -3.4271613399 & 1.7135806699
\end{pmatrix},\\  \varepsilon_{12} =&~ 1.94308232, \quad \varepsilon_{13} = \text{nonexistent}, \quad \varepsilon_{23} = 1.94308232,\\
\varepsilon_0 =&~  4.02066622 \times 10^{-3}.
\end{align}
\end{subequations}
Next, for the computations in \cref{sec:results:rising_bubble} we have used:
\begin{subequations}
    \begin{align}    \bar{\boldsymbol{\kappa}} =&~10^{-3}\begin{pmatrix}
4.1361833328 & -6.5504823549  & -1.7218843107 \\
-6.5504823549  & 6.5504823549
 & -6.5504823549 \\
-1.7218843107  &  -6.5504823549  &  4.1361833328 
\end{pmatrix}\\ \varepsilon_{12} = &~0.0819004278, \quad \varepsilon_{13} = 0.0581144728, \quad \varepsilon_{23} = 0.0819004278, \\
r_1 : &~ \varepsilon_0 = 1.55942257 \times 10^{-3}, \qquad r_2 : \varepsilon_0 =  1.07546309  \times 10^{-3}. 
\end{align}
\end{subequations}
Finally, for the computations in \cref{subsec: Quaternary rising bubble simulation} we have used:
\begin{subequations}
    \begin{align}   
    \bar{\boldsymbol{\kappa}}= &{\scriptsize~10^{-2}\begin{pmatrix}
1.1005813937 & -5.0644662183\times 10^{-1} & -1.5390251666 & -1.5569099892\times 10^{-1}\\
-5.0644662183\times 10^{-1} & 7.9867010573\times 10^{-1} & -1.1291527134 & 3.8259123747\times 10^{-2}\\
-1.5390251666 & -1.1291527134 & 1.9579824238 & -1.2477869676\\
-1.5569099892\times 10^{-1} & 3.8259123747\times 10^{-2} & -1.2477869676 & 6.8260942136\times 10^{-1}
\end{pmatrix}}\\ \varepsilon_{12} = &~0.103652639, \quad \varepsilon_{13} = 0.149825882, \quad \varepsilon_{14} = 0.0846653410, \nn\\
\varepsilon_{23} = &~0.130170775, \quad \varepsilon_{24} = 0.0700575184, \quad \varepsilon_{34} = 0.129812413, \\
\varepsilon_0 =&~ 3.10458733 \times 10^{-3}.
\end{align}
\end{subequations}

%\noindent {\small \textbf{Funding.} 
\subsection*{Funding}
MtE acknowledges support from the German Research Foundation (Deutsche Forschungsgemeinschaft DFG), project number 566600860. A.B. acknowledges support by the German Research Foundation (DFG) via TRR 146, subproject C3, project number 233630050 and via SPP 2256 within the project ``Variational quantitative phase-field modeling and simulation of powder bed fusion additive manufacturing'' project number 441153493. MtE and AB acknowledges support by the RMU via the joined project "MULTIMIX".

\bibliography{jfm}

\end{document}